\documentclass[a4paper, 11pt, twoside]{Thesis}  % Use the "Thesis" style, based on the ECS Thesis style by Steve Gunn
\graphicspath{Figures/}  % Location of the graphics files (set up for graphics to be in PDF format)

\usepackage[square, numbers, comma, sort&compress]{natbib}  % Use the "Natbib" style for the references in the Bibliography
\usepackage{verbatim}  % Needed for the "comment" environment to make LaTeX comments
\hypersetup{urlcolor=navy, colorlinks=true}  % Colours hyperlinks in blue, but this can be distracting if there are many links.

\usepackage[utf8x]{inputenc}
\usepackage{ae,aecompl}
\usepackage{graphicx}
\usepackage{ucs}
\usepackage{dsfont}
\usepackage{subfigure}
\usepackage{enumerate}
\usepackage{comment}
\usepackage{braket}
\usepackage{amsmath}
\usepackage{amsthm}
\usepackage{chngcntr}
\usepackage{float}
\usepackage{wallpaper}
\usepackage{setspace}
\usepackage{accents}
\usepackage{caption}
\usepackage{youngtab}
\usepackage{mathbbol}
\usepackage{color}
\usepackage{amsfonts}
\usepackage{amssymb}

\newlength{\dhatheight}

\theoremstyle{plain}
\newtheorem{defn}{Definition}
\newtheorem{thm}[defn]{Theorem}
\newtheorem{prop}[defn]{Proposition}
\newtheorem{cor}[defn]{Corollary}
\newtheorem{conj}[defn]{Conjecture}
\newtheorem{lem}[defn]{Lemma}

\newtheorem{ex}[defn]{Example}
\theoremstyle{definition}
\newtheorem{rem}[defn]{Remark}

\newenvironment{changemargin}[2]{%
	\begin{list}{}{%
			\setlength{\topsep}{0pt}%
			\setlength{\leftmargin}{#1}%
			\setlength{\rightmargin}{#2}%
			\setlength{\listparindent}{\parindent}%
			\setlength{\itemindent}{\parindent}%
			\setlength{\parsep}{\parskip}%
		}%
		\item[]}{\end{list}}

\numberwithin{equation}{chapter}
\numberwithin{defn}{chapter}
\counterwithout{figure}{chapter}

\newcommand{\tr}{\mathrm{tr}}

\newcommand{\diag}{\mathrm{diag}}

\renewcommand{\Re}{\mathrm{Re}}
\let\ForAll\forall
\renewcommand{\forall}{\ \ForAll}
\newcommand{\rank}{\mathrm{rk}\,}

\newcommand{\R}{\mathbb{R}}
\renewcommand{\C}{\mathbb{C}}
\newcommand{\N}{\mathbb{N}}
\newcommand{\Z}{\mathbb{Z}}

\newcommand{\D}{\mathrm{d}\,}
\newcommand{\ketbra}[2]{\left|#1\right\rangle\!\!\left\langle #2\right|}
\newcommand{\proj}[1]{\ketbra{#1}{#1}}
\newcommand{\bracket}[2]{\left\langle #1| #2\right\rangle}

\newcommand{\supp}{\mathrm{supp}}

\newcommand{\spa}{\mathrm{span}}
\newcommand{\spec}{\mathrm{spec}}

\newcommand{\Hi}{\mathcal{H}}
\newcommand{\1}{\mathds{1}}
\newcommand{\id}{\mathrm{id}}
\newcommand{\St}{\mathcal{S}}
\newcommand{\hi}{\Hi}
\newcommand{\sos}[1]{\mathcal{S}\left(#1\right)}
\newcommand{\End}[1]{\mathrm{End}\left(#1\right)}
\newcommand{\Hom}[2]{\mathrm{Hom}\left(#1,#2\right)}

\newcommand{\eps}{\varepsilon}
\newcommand{\relent}[4]{D_{#1}\left( #3 #2\| #4\right)}

\newcommand{\half}{\frac{1}{2}}

\newcommand{\hchapter}[1]{\chapter{#1}\lhead{\emph{#1}}}
\newcommand{\hchapterstar}[1]{\chapter*{#1}\lhead{\emph{#1}}}

\newcommand{\gl}{\mathrm{GL}}

\newcommand{\Enc}{\ensuremath{\mathsf{Enc}}}
\newcommand{\Dec}{\ensuremath{\mathsf{Dec}}}

\newcommand{\poly}{\operatorname{poly}}

\newcommand{\SKQES}{\textsf{QES}}
\newcommand{\ITS}{\textsf{ITS}}
\newcommand{\IND}{\textsf{IND}}
\newcommand{\ITNM}{\textsf{NM}}
\newcommand{\ABWNM}{\textsf{ABW-NM}}
\newcommand{\ABW}{\textsf{ABW}}
\newcommand{\DNS}{\textsf{DNS}}
\newcommand{\GYZ}{\textsf{GYZ}}
\newcommand{\acc}{\textsf{acc}}
\newcommand{\rej}{\textsf{rej}}
\newcommand{\doublehat}[1]{%
	\settoheight{\dhatheight}{\ensuremath{\hat{#1}}}%
	\addtolength{\dhatheight}{-0.35ex}%
	\hat{\vphantom{\rule{1pt}{\dhatheight}}%
		\smash{\hat{#1}}}}

\newcommand{\opr}{\mathrm{End}}

\newcommand{\one}{\mathds 1}

\begin{document}
\frontmatter      % Begin Roman style (i, ii, iii, iv...) page numbering

% Set up the Title Page
\title{Entropy in Quantum Information Theory~-- Communication and Cryptography}
\authors  {Christian Majenz
            }
        \UNIVERSITY{UNIVERSITY OF COPENHAGEN}
        \faculty{Faculty of Science}
        \group{QMATH}
        \department{Institute for Mathematical Sciences}
        \university{University of Copenhagen}
\addresses  {\groupname\\\deptname\\\univname}  % Do not change this here, instead these must be set in the "Thesis.cls" file, please look through it instead
\date       {\today}
\subject    {}
\keywords   {}

\maketitle
%% ----------------------------------------------------------------

\setstretch{1.3}  % It is better to have smaller font and larger line spacing than the other way round

% Define the page headers using the FancyHdr package and set up for one-sided printing
\fancyhead{}  % Clears all page headers and footers
\rhead{\thepage}  % Sets the right side header to show the page number
\lhead{}  % Clears the left side page header

%\pagestyle{fancy}  % Finally, use the "fancy" page style to implement the FancyHdr headers

%%% ----------------------------------------------------------------
%% Declaration Page required for the Thesis, your institution may give you a different text to place here
%\Declaration{
%
%\addtocontents{toc}{\vspace{1.5em}}  % Add a gap in the Contents, for aesthetics
%
%This thesis and the work presented in it are my own, except where explicitly stated otherwise in the text.
%}
%
%\clearpage  % Declaration ended, now start a new page

%% ----------------------------------------------------------------
% page with author infos etc.
\pagestyle{empty}  % No headers or footers for the following pages
Christian Majenz\\
Department of Mathematical Sciences\\
Universitetsparken 5\\
2100 Copenhagen\\
Denmark\\
{}\\
\htmladdnormallink{christian.majenz@gmail.com}{mailto:christian.majenz@gmail.com}\\
\vfill
\begin{tabular}{ll}
PhD Thesis&\\
Date of submission:& 31.05.2017\\
Advisor:& Matthias Christandl, University of Copenhagen\\
Assessment committee: & Berfinnur Durhuus, University of Copenhagen\\
&Omar Fawzi, ENS Lyon\\
&Iordanis Kerenidis, University Paris Diderot 7\\
\textcopyright~by the author&\\
ISBN:&978-87-7078-928-8
\end{tabular}

\clearpage

%% ----------------------------------------------------------------
% The "Funny Quote Page"
\null\vfill
% Now comes the "Funny Quote", written in italics
\textit{``You should call it entropy, for two reasons. In the first place your uncertainty function has been used in statistical mechanics under that name, so it already has a name. In the second place, and more important, no one really knows what entropy really is, so in a debate you will always have the advantage.''}

\begin{flushright}
John von Neumann, to Claude Shannon
\end{flushright}

\vfill\vfill\vfill\vfill\vfill\vfill\null
\clearpage  % Funny Quote page ended, start a new page
%% ----------------------------------------------------------------

% The Abstract Page

\addtotoc{Abstract}  % Add the "Abstract" page entry to the Contents

	\begin{changemargin}{-1cm}{-1.7cm}
		
\abstract{Abstract}{
\addtocontents{toc}{\vspace{.8em}}  % Add a gap in the Contents, for aesthetics
Entropies have been immensely useful in information theory. In this Thesis, several results in quantum information theory are collected, most of which use entropy as the main mathematical tool.

 The first one concerns the von Neumann entropy. While a direct generalization of the Shannon entropy to density matrices, the von Neumann entropy behaves differently. The latter does not, for example, have the monotonicity property that the latter possesses: When adding another quantum system, the entropy can decrease. A long-standing open question is, whether there are quantum analogues of unconstrained non-Shannon type inequalities. Here, a new constrained non-von-Neumann type inequality is proven, a step towards a conjectured unconstrained inequality by Linden and Winter.
 
  Like many other information-theoretic tasks, quantum source coding problems such as coherent state merging have recently been analyzed in the one-shot setting. While the case of many independent, identically distributed quantum states has been treated using the decoupling technique, the essentially optimal one-shot results in terms of the max-mutual information by Berta et al. and Anshu at al. had to bring in additional mathematical machinery. We introduce a natural generalized decoupling paradigm, catalytic decoupling, that can reproduce the aforementioned results when applied in a manner analogous to the application of standard decoupling in the asymptotic case. 
  
  Quantum teleportation is one of the most basic building blocks in quantum Shannon theory. While immensely more entanglement-consuming, the variant of port based teleportation is interesting for applications like instantaneous non-local computation and attacks on quantum position-based cryptography. Port based teleportation cannot be implemented perfectly, and the resource requirements diverge for vanishing error. We prove several lower bounds on the necessary number of output ports $N$ to achieve port based teleportation for given dimension and error. One of them shows for the first time that $N$ diverges uniformly in the dimension of the teleported quantum system, for vanishing error. As a byproduct, a new lower bound for the size of the program register for an approximate universal programmable quantum processor is derived. 
  
  Finally, the mix is completed with a result in quantum cryptography. While quantum key distribution is the most well-known quantum cryptographic protocol, there has been increased interest in extending the framework of symmetric key cryptography to quantum messages. We give a new definition for information-theoretic quantum non-malleability, strengthening the previous definition by Ambainis et al. We show that quantum non-malleability implies secrecy, analogous to quantum authentication. Furthermore, non-malleable encryption schemes can be used as a primitive to build authenticating encryption schemes. We also show that the strong notion of authentication recently proposed by Garg et al. can be fulfilled using 2-designs.

}	\end{changemargin}

\clearpage  % Abstract ended, start a new page
%% ----------------------------------------------------------------

% The Abstract Page

\addtotoc{Abstract in Danish}  % Add the "Abstract" page entry to the Contents

\begin{changemargin}{-1.7cm}{-1cm}
	
	\abstract{Resum\'e}{
		\addtocontents{toc}{\vspace{.8em}}  % Add a gap in the Contents, for aesthetics
		Entropibegrebet har vist sig ekstremt nyttigt i informationsteori. Denne afhandling samler en håndfuld resultater fra kvanteinformationsteori, hvoraf de fleste bruger entropier som matematisk hovedværktøj.
		
		Det første resultat omhandler von-Neumann-entropien. Selvom von-Neumann-entropien er en direkte generalisering af Shannonentropien til tæthedsmatricer, opfører den sig anderledes. Den mangler for eksempel Shannonentropiens monotonicitetsegenskab: Entropien kan blive mindre når et kvantesystem gøres større ved at tilføje flere systemer. Et spørgsmål, som har været åbent i mange år, er, om der er uligheder for von-Neumann-entropien, der svarer til ubetingede uligheder af ikke-Shannon type. Her beviser vi en ny betinget ulighed af ikke-von-Neumann type, som er et skridt mod en formodet ubetinget ulighed af Linden og Winter.
		
		Ligesom mange andre informationsteoretiske problemer, er kvantekildekodningsproblemer, for eksempel fusion af kvantetilstande, fornyligt blevet analyseret i det såkaldte \emph{one-shot} ("enkelinstans") scenarie. Mens kvantetilstandsfusionsproblemet blev løst ved hjælp af afkoblingsmetoden i tilfældet af mange uafhængige, identisk fordelte kvantetilstande, krævede de næsten optimale \emph{one-shot} resultater af Berta et al. og Anshu et al. yderligere matematiske redskaber. Vi indfører en naturlig generalisering af afkoblingsmetoden, katalytisk afkobling, som kan bruges til at genbevise de førnævnte resultater analogt til hvordan standardafkoblingsmetoden bruges i den asymptotiske sammenhæng.
		
		Kvanteteleportation er en af de mest basale byggesten i kvante-Shannonteorien. En variant af teletransportationsprotokollen, kaldet portbaseret teleportation, kræver meget større mængder af sammenfiltring, men er interessant for f. eks. øjeblikkelig ikke-lokal beregning. Portbaseret teleportation kan ikke realiseres perfekt, og de nødvendige ressources divergerer for små fejlparametre. Vi beviser flere nedre grænser for antallet af udgangsporte $N$, som kræves for at realisere portbaseret teleportation for en given fejltolerance og et kvantesystem af en given dimension. En af dem viser for første gang at $N$ divergerer uniformt i kvantesystemets dimension, når fejltolerancen går mod nul. Som biprodukt beviser vi også en ny nedre grænse for en tilnærmet universel programm\'erbar kvanteprocessors programregisterstørrelse.
		
		Endelig runder vi af med et resultat fra kvantekryptografien. Mens kvante-\emph{key-distribution} (deling af krypteringsnøgler v.h.a. kvantesystemer) er den mest kendte kvantekryptografiske protokol, er der opstået en voksende interesse i at generalisere symmetrisk kryptering til kvantebeskeder. Vi giver en ny definition for ubetinget sikker \emph{non-malleability}, som forbedrer den hidtidige definition af Ambainis et al. Vi viser at kvante-\emph{non-malleability} medfører hemmelighed, analogt til kvanteautentificering. Yderligere kan kvante-\emph{non-malleable} krypteringsprotokoller bruges til at bygge kvanteautentificeringsprotokoller. Vi viser derudover at den seneste stærke definition af kvanteautentificering forslået af Garg et al. kan opfyldes ved brug af 2-designs.

}	\end{changemargin}

\clearpage  % Abstract ended, start a new page
%% ----------------------------------------------------------------

\setstretch{1.3}  % Reset the line-spacing to 1.3 for body text (if it has changed)

% The Acknowledgements page, for thanking everyone
\acknowledgements{
\addtocontents{toc}{\vspace{1.5em}}  % Add a gap in the Contents, for aesthetics

I first want to thank my advisor Matthias Christandl for his encouragement, his unique viewpoint on, and his contagious excitement about, quantum information theory, our collaborations, and his help with a variety of problems ranging from math to interim furniture storage. 

I want to thank my collaborator Gorjan Alagic for many great discussions and our joint project on quantum non-malleability and authentication that forms the last chapter of this thesis. Furthermore, I want to thank my collaborators Mario Berta, Fr\'ed\'eric Dupuis, and Renato Renner, for an enjoyable and smooth project on catalytic decoupling that forms the first half of the third chapter of this thesis. Finally I would like to thank Florian Speelmann for joining the port based team leading to a great improvement of our bounds. Beyond the mentioned collaborations I enjoyed the priviledge of many enlightening discussions with different colleagues. The first ``thank you" in this category goes to all members of the QMath center and its predecessor. Special thanks go to Alexander Müller-Hermes for not shouting at me when I didn't stop asking him questions about different norms on operator spaces (and also for introducing the arXiv-beer tradition). Furthermore I would like to thank Fr\'ed\'eric Dupuis, Felix Leditzky, Christopher Portmann, Michael Walter, Michael Wolf, Jan Bouda and the Brno quantum information group, the quantum information group at IQIM, Caltech, the groups of Jens Eisert in Berlin, David Gross in Cologne, Ottfried Gühne in Siegen, Renato Renner in Zürich and Stefan Wolf in Lugano as well as the members of QuSoft center in Amsterdam for discussions about topics of this thesis and valuable feedback and questions after talks I gave about catalytic decoupling as well as quantum non-malleablity and authentication. Special thanks go to Mario Berta, Fernando Brandao and all of the quantum information theorist population of the Annenberg building at Caltech for hosting me for a couple of months.

The structure, style and language of this thesis benefitted from the thorough proof reading by Jonatan Bohr-Brask, Birger Brietzke, Matthias Christandl and Christopher Perry.

Finally I want to thank Maj Rørdam Nielsen, as well as my family, whose love and support helped me to work on my PhD project with joy an overwhelming fraction of the time.

I acknowledge financial support from the European Research Council (ERC Grant Agreement no 337603), the Danish Council for Independent Research (Sapere Aude) and VILLUM FONDEN via the QMATH Centre of Excellence (Grant No. 10059).
}
\clearpage  % End of the Acknowledgements
%% ----------------------------------------------------------------

\pagestyle{fancy}  %The page style headers have been "empty" all this time, now use the "fancy" headers as defined before to bring them back

%% ----------------------------------------------------------------
\lhead{\emph{Contents}}  % Set the left side page header to "Contents"
\tableofcontents  % Write out the Table of Contents
\clearpage  %Start a new page
\lhead{\emph{Notation}}  % Set the left side page header to "Notation"
\listofnomenclature{ll}  % Include a list of Symbols (a three column table)
{
% symbol & description\\
$x:=y$, $y=:x$ & $x$ is defined to be equal to $y$\\
 $\N$& the nonnegative integers\\
 $\Z$ & the integers\\
 $\R$ & the reals\\
 $\C$ & the complex numbers\\
 $n\mod m$ & remainder when dividing $n$ by $m$\\
 $\log x, \log(x)$ & logarithm with base 2 \\
 $\overline z$ & the copmplex conjugate of $z\in \C$ \\
 $[n]$ & $\{1,2,...,n\}$\\
 $\mathcal X_n$ & $\{0,...,n-1\}$ \\
 $I\cap J$ & intersection of the sets $I$ and $J$\\
 $I\cup J$ & union of the sets $I$ and $J$\\
 $I\setminus J$ & set difference\\
 $I^c$ & complement of a set $I$ with respect to a base set\\& that is mentioned or clear from context\\
 &  \\ % Gap to separate the Roman symbols from the Greek
 $\mathcal X,\mathcal Y,\mathcal Z$& Alphabets, i.e. finite sets\\
 $X,Y,Z$& classical systems\\
 $\mathbb{P}_p[E]$ & probability of an event $E$ over the probability distribution $p$\\
 $\mathbb E_p[f]$& expectation of a quantity $f$ over a probability distribution $p$\\
 $(v,w)$ & inner product on a real vector space\\
 
 &  \\ % Gap to separate the Roman symbols from the Greek
 $A,B,C$ & quantum systems \\
 $\hi_A,\hi_B,\hi_B...$& Hilbert spaces of the systems $A,B,C...$\\
 $\hi^*$ & dual hilbert space of $\hi$\\
 $|A|$ & the dimension of $\hi_A$\\
 $AB$ & composite quantum system consisting of subsystems $A$ and $B$\\
 $\Hom{\hi_A}{\hi_B}$& Space of linear maps from $\hi_A\to \hi_B$\\
 $\End{\hi_A}$&$\Hom{\hi_A}{\hi_A}$\\
 $X, X_A$ & element of $\End{\hi_A}$\\
 $X, X_{A\to B}$& element of $\Hom{\hi_A}{\hi_B}$\\
 $X^\dagger$ & the adjoint of a matrix $X\in\Hom{\hi_A}{\hi_B}$
}
%% ----------------------------------------------------------------
% End of the pre-able, contents and lists of things
% Begin the Dedication page

\setstretch{1.3}  % Return the line spacing back to 1.3

\pagestyle{empty}  % Page style needs to be empty for this page
\dedicatory{I dedicate this thesis to Sarah-Katharina Meisenheimer, who has left us in January before handing in hers.}

\addtocontents{toc}{\vspace{1.5em}}  % Add a gap in the Contents, for aesthetics

%% ----------------------------------------------------------------
\mainmatter	  % Begin normal, numeric (1,2,3...) page numbering
\pagestyle{fancy}  % Return the page headers back to the "fancy" style

% Include the chapters of the thesis, as separate files
% Just uncomment the lines as you write the chapters

\hchapter{Prologue}\label{chap:prolog}
Information theory, the theory of information processing with physical systems, has revolutionized our lives in the last 70 years. Without it, you would likely be reading this thesis written on a typewriter, with the formulas inserted by hand, and, of course, on a different topic. Mobile phones, hard drives and digital radio stations would be unthinkable without error correction codes, to name just one achievement of information theory and a few examples of its applications.

 A central and ubiquitous concept in information theory is entropy. The question that different entropy functions try to answer is one that everybody living through these times should ask themselves when, for example, writing another e-mail requiring a chunk of the recipient's precious time: How much information does a certain system contain? In the example, this question has to be asked from the perspective of the recipient. When she or he is told that they will receive an e-mail from you, and are asked to model it before it arrives, they will have to resort to probability theory. There are many possibilities for the content of the e-mail, and in general they would consider it more or less likely that they receive one or another kind of e-mail from you.
 
 A more basic example is the role of a dice. When rolling a dice, the outcome is unknown. The dice roll contains a lot of information from the perspective of the spectators \emph{before} the dice has actually been rolled. This is, because, ideally, we would expect any outcome from the set $\{1,2,3,4,5,6\}$ with equal probability.
 
 Each of the two  examples have precisely defined analogues in information theory, that are analyzed with the help of different entropy functions. The analogue of the first situation is the task of \emph{data compression}, also called \emph{source coding}. In this task, some raw data, that we imagine coming from an information source, is to be encoded in a way such that (i) the original data can be recovered, and (ii) the memory that the encoded message takes up is minimized. Following the founder of information theory, Claude Elwood Shannon, we take English text as an example \cite{Shannon1951}. A very coarse model of somebody writing English text is a source that spits out random letters independently according to the probability distribution given by the frequency distribution in some large sample of English text, i.e. the letters are \emph{independent} and \emph{identically distributed}. This situation is called the IID setting. In particular, it will for example  spit out the letter `e' more often than the letter `q' on average. This simple model already enables us to store a (large) piece of English text using approximately the \emph{Shannon entropy} of the letter distribution, $4.14$ bits, of memory per letter instead of $\log 26\approx 4.7$ bits -- a significant saving.
 
 The analogue of the second example is the problem of randomness extraction, which is mainly relevant in cryptography. There are many ways to obtain seemingly random numbers. One can roll a dice a couple of times, or shuffle some cards and distribute them, as is done in games for entertainment. Or one can harvest random numbers from the RAM of a computer. In all of these cases there are several obvious questions about the quality of the random data. How random is it? Are there other observers from whose perspective the numbers are less random? If the data is not uniformly random, can I extract uniformly random numbers from it, and if so, at which rate? The last question is answered by the construction of randomness extractors, and the achievable rate of uniformly random numbers is given by the \emph{min-entropy} of our original data. The min-entropy is also connected to the probability of \emph{guessing} the original data correctly. This probability is the same as for guessing the outcome of the extracted uniform random number, which has Shannon entropy equal to the min-entropy of the original data. 
 
 The basis of classical information theory as introduced by Shannon is classical physics, a dice, the text on a piece of paper or the current through a transistor in a computer behave classically. It is therefore a more or less obvious question to ask how information theory changes when using the more fundamental theory of quantum mechanics as an underlying physical theory. When Alexander Holevo asked this question and answered it for some tasks in the 70's, \emph{quantum information theory} was born.
 
 Since then quantum information theory underwent fast development, from a research  field in some niche between physics and computer science, to a field transitioning from pure academic activity to business development. There are already quantum information processing devices on the market, like quantum random number generators, and it seems likely that more will follow soon. Today, most experts think that a practical quantum computer, is a midterm perspective on the order of 10s of years. The contribution of quantum information theory is to lay the foundation for quantum information processing and enable the users of quantum devices to make informed and prudent use of them.
 
There are two general kinds of problems in quantum information theory. On the one hand, given a certain task or problem from classical information theory, one can ask how it changes if some or all systems involved are assumed to be quantum instead of classical. On the other hand, there are tasks that have no classical analogue. This is because certain tasks can be achieved with quantum resources that are impossible classically, like unconditionally secure key distribution, and there are tasks that reduce to simpler ones in the classical setting, like source coding with side information at the encoder. There are also the paradigms of \emph{quantum teleportation} and \emph{superdense coding}, that concern the interplay of quantum and classical information. A quantum teleportation protocol is a way to send quantum information by using classical communication only by making use of preexisting quantum correlations between sender and receiver. In the second half of Chapter \ref{chap:one-shot}, \emph{port based teleportation} is analyzed, a variant of standard quantum teleportation.

Entropy plays an equally central role in quantum information theory as it does in the classical one. The two examples above, data compression and the optimal guessing probability, both have quantum versions. The quantum generalization of Shannon's source coding theorem is the so called Schumacher compression. Given a source that spits out independent ``quantum letters", i.e. quantum systems in some mixed quantum state, the resulting quantum data can be compressed to only use an amount of quantum memory per letter equal to the von Neumann entropy of the quantum state. This coding problem becomes more involved, when the decoder has some \emph{side information} about the encoded data. This is a special case of the classical scenario of Slepian-Wolf distributed source coding. In the quantum setting, this problem was famously solved by Horodecki, Oppenheim and Winter using the so-called \emph{decoupling technique}. In the first half of Chapter \ref{chap:one-shot}, this technique is generalized to be aplicable to the \emph{one shot} setting, where information processing of data that does not have the simple structure of the IID setting.

In a multi-party setting, every non-empty subset of the parties can be assigned an information-theoretic state. An example is the task of \emph{network coding}. Here, information has to be redistributed over a certain network of communication channels. While this problem is solved by a simple \emph{store-and-forward} routing protocol in the case of the internet, more complex protocols will be used in the 5G mobile communication standard, and there are conceivable scenarios where non-trivial network coding is necessary. The fundamental limits for communication rates over a network are given by \emph{entropy inequalities}. While infinitely many such inequalities have been proven for the case of four or more classical parties, only the basic inequalities that can be derived from the strong subadditivity inequality are known in the quantum setting, except for some inequalities that only hold for a special set of quantum states fulfilling some extra \emph{constraints}. In Chapter \ref{chap:entropy}, we prove a new inequality of the latter kind by removing one of the constraints from an inequality by Linden and Winter. This is a step towards proving an unconstrained inequality conjectured by these authors.

In cryptography, entropy plays a role as well. One application, randomness extraction, was already mentioned. Another example is entropic security definitions in symmetric key cryptography. Here, two parties share a secret, the key, and use it to communicate in a way that is protected against malicious adversaries. If they want to keep their message secret, they have to make sure that the encrypted message does not contain any information about its content from the perspective of somebody who does not have the key. This can be formulated in terms of an entropic quantity called the \emph{mutual information}. A different security paradigm is \emph{non-malleability}. Here, an active adversary should not be able to \emph{modify} the message in a meaningful way. An entropic formulation is, that any adversary should essentially not be able to gain any correlations with the message content, as measured by the mutual information. In Chapter \ref{chap:crypto}, this formulation is shown to remedy the shortcomings of the previous definition in the quantum setting.

During the time of the PhD program at University of Copenhagen I had the pleasure to collaborate with many excellent researchers. The results of three of these collaborations are included in this Thesis, two of which have appeared in peer reviewed articles. It is stated explicitly in the text where results from collaborative research is presented. The section  on catalytic decoupling, Subsesection \ref{subs:catalytic}, presents results from the following article:
\begin{itemize}
	\item Christian Majenz, Mario Berta, Frédéric Dupuis, Renato Renner, and Matthias Christandl. \emph{Catalytic decoupling of quantum information}. Physical Review Letters, 118(8):080503, 2017.
\end{itemize} 
The chapter on quantum non-malleability and authentication, Chapter \ref{chap:crypto}, consists of results from the article
\begin{itemize}
	\item G. Alagic and C. Majenz. Quantum non-malleability and authentication. ArXiv e-prints, accepted for publication in Advances in Cryptology - CRYPTO 2017, October 2016.
\end{itemize}
The results in Subsubsection \ref{subsubs:nsbound} have been obtained in collaboration with Matthias Christandl and Florian Speelman.

\hchapter{Introduction}\label{chap:intro}
In this chapter, the basic mathematical and information-theoretical concepts are introduced that are used throughout the rest of this thesis. It is organized in four sections, introducing the basic formalism of quantum information theory, distance measures, entropies, and some fundamentals of representation theory of finite and compact groups, respectively.

\section{The formalism of quantum information theory}
Quantum information theory is the theory of processing information using physical systems that behave according to the laws of quantum mechanics. While building on the the theory of non-relativistic quantum mechanics, modern quantum information has developed its own theoretical framework which is presented in the following section. Good introductions to the topic can for example be found in References \cite{Nielsen2000,Wilde2013,Preskill1998b,Renner2012}.

\subsection{Systems and States}
As in standard Schrödinger quantum mechanics, a quantum system $A$ is described by a \emph{Hilbert space} $\hi_A$. While infinite-dimensional Hilbert spaces do play a role in quantum information theory, we will use only finite-dimensional Hilbert spaces in this thesis, i.e. $\dim\hi_A<\infty$. Therefore we have $\hi_A\cong\C^d\cong \hi_A^*$ with $d=|A|=\dim\hi_A$ and where $\hi_A^*$ is  the dual of $\hi_A$. In Dirac bra-ket-notation, an element of a Hilbert space $\hi$ is denoted $\ket\psi\in\hi$, and its dual vector is denoted by $\bra\psi\in\hi^*$. That way, the inner product of two vectors $\ket\phi,\ket\psi\in\hi$ is denoted by $\bracket{\phi}{\psi}$, and the outer product by $\ketbra{\phi}{\psi}$. We sometimes write $\ket\psi_A, \bra\psi_A, \ketbra{\phi}{\psi}_A$ and $\bracket{\phi}{\psi}_A$ to emphasize which quantum system and which Hilbert space these objects belong to.

Given two quantum systems $A$ and $B$, we can consider them jointly by defining the composite quantum system $AB$. Its Hilbert space is the tensor product of the Hilbert spaces of its parts, i.e. $\hi_{AB}=\hi_A\otimes \hi_B$. For elements of $\hi_A\otimes \hi_B$ we sometimes omit the tensor product symbol, $\ket\phi_A\ket\psi_B=\ket\phi_A\otimes\ket\psi_B$.

The \emph{state} of of a quantum system $A$ is given by a positive semidefinite matrix $\rho_A\in\End{\hi_A}$ that is normalized to have unit trace, i.e. $\rho_A\ge 0$ and $\tr\rho_A=1$. If clear from context, the subscript is sometimes omitted, i.e. we just write $\rho$ instead of $\rho_A$. We use the notation $\mathcal{S}(\hi_A)\subset\mathcal P(\hi_A)\subset \mathrm{Herm}(\hi_A)\subset \End{\hi_A}$ for the set of quantum states, the set of positive semidefinite matrices, the set of Hermitian matrices and the set of matrices, on $\hi_A$, respectively. A state is called \emph{pure} if it is a projector, i.e. $\rho_A=\proj{\psi}_A$ for a vector $\ket\psi_A$. States that are not pure are called \emph{mixed}. For a pure state $\rho_A=\proj{\psi}_A$, the vector $\ket\psi_A$ is called \emph{state vector}. A distinguished quantum state is the \emph{maximally mixed} state $\tau_A=\frac{\mathds{1}_A}{|A|}$.

Another special subset of $\End{\hi_A}$ is the unitary group. A matrix $U_A$ is unitary if $U_A\in\End{\hi_A}$ with $U_AU_A^\dagger =1$. A special unitary on $\hi_A\otimes \hi_{A'}$ for $\hi_A\cong\hi_{A'}$ is the \emph{swap} $F$, defined by $F\ket\phi_A\otimes\ket\psi_{A'}=\ket\psi_A\otimes\ket\phi_{A'}$. The following is a handy fact about the swap.

\begin{lem}[Swap trick]\label{lem:swap-trick}
	For any matrices $A, B$ we have $\tr [AB]=\tr [F  A\otimes B]$. 
\end{lem}

By the isomorphism $\hi_A\otimes\hi_B\cong \Hom{\hi_A}{\hi_B}$ and the singular value decomposition, for every pure state $\rho_{AB}=\proj{\psi}_{AB}$ there exist orthonormal bases $\{\ket{\alpha_i}_A\}_{i=0}^{|A|-1}$ and $\{\ket{\beta_i}_B\}_{i=0}^{|B|-1}$ such that

\begin{equation}
 \ket\psi_{AB}=\sum_{i=0}^{\min(|A|,|B|)-1}\sqrt{p_i}\ket{\alpha_i}_A\otimes\ket{\beta_i}_B
\end{equation}
for some probability distribution $p=\{p_i\}_{i=0}^{\min(|A|,|B|)-1}$. This representation is called the \emph{Schmidt decomposition}. A short calculation shows that $p$ is the spectrum of both $\psi_A=\tr_B\proj\psi_{AB}$ and $\psi_B==\tr_A\proj\psi_{AB}$, the \emph{marginals} $\proj\psi_{AB}$. Here, $\tr_A$ is the \emph{partial trace} over the system $A$ which is defined in Subsection \ref{subs:QChan} for pedagogical reasons. As it is convenient to switch between vectors and the projection onto them in the case of pure states, we write $\rho=\proj\rho$ for a given pure state $\rho$, or for a given vector $\ket \rho$.

Via an isomorphism $\hi_A\cong\C^d$, a Hilbert space $\hi_A$ can be assigned a standard basis $\{\ket i\}_{i=0}^{|A|-1}$, which is also called \emph{computational basis} in quantum information theory. A state that is diagonal in that basis is called \emph{classical}. When considering composite systems, such a basis is chosen for the atomic Hilbert spaces. The computational basis of the tensor product Hilbert space is the product bases induced by the local computational bases. This is necessary to ensure that classical states have classical marginals, i.e. they are classical on all subsystems and correspond to a joint probability distribution of a composite classical system. A bipartite quantum state $\rho_{AB}$ is called a quantum-classical state, or cq-state for short, if the $A$-system is classical. This means, that $\rho_{AB}$ has a basis of eigenvectors $\ket{\psi_{ij}}_{AB}=\ket i_A\otimes\ket{\psi_j}_B$. In other words,
\begin{equation}
\rho_{AB}=\sum_i p_i\proj i_A\otimes\rho^{(i)}_B
\end{equation}
for some probability distribution $p$ and a family of normalized quantum states $\rho^{(i)}_B$.

For every mixed state $\rho_A$ there exists a pure extension to another quantum system $A'$, i.e. a pure state $\psi_{AA'}=\proj\psi_{AA'}$ such that $\psi_A=\rho_A$. Such a pure state is called a \emph{purification} of $\rho_A$, and by the Schmidt decomposition theorem all purifications of $\rho_A$ can be transformed into each other by applying a (partial) isometry to the purifying system.

One of the most interesting and useful features of quantum theory is \emph{entanglement}. A bipartite pure state $\proj\psi_{AB}$ is called entangled, if it is not a product state, i.e. $\ket\psi_{AB}\neq \ket\alpha_A\otimes\ket\beta_B$ for all $\ket\alpha_A$ and $\ket\beta_B$. A mixed state is called entangled, if it is not a convex combination of product states, otherwise it is called \emph{separable}. A particular entangled state that is used frequently is the standard maximally entangled state $\proj{\phi^+}_{AA'}\in\sos{\hi_A\otimes\hi_{A'}}$ with 
\begin{equation}
\ket{\phi^+}_{AA'}=\frac{1}{\sqrt{|A|}}\sum_{i=0}^{|A|-1}\ket i_A\otimes \ket i_{A'}.
\end{equation}
If we want to emphasize the dimension of the local Hilbert space, we write $\ket{\phi^+_{|A|}}_{AA'}=\ket{\phi^+}_{AA'}$. For $\hi_A=\C^2$, $\proj{\phi^+}_{AA'}$ is called an entanglement bit, or ebit. The maximally entangled state has a property called the \emph{mirror lemma},
\begin{equation}
	X_A\ket{\phi^+}_{AA'}=X_{A'}^T\ket{\phi^+}_{AA'},
\end{equation}
where $X^T$ is the transpose of $X$ in the computational basis. A similar statement holds even when $X$ is a rectangular matrix.
\begin{lem}\label{lem:genmirr}
	Let $X_{A\to B}\in L(\hi_A,\hi_B)$ be a linear operator from $A$ to $B$. Then
	\begin{equation}
	X_{A\to B}\ket{\phi^+}_{AA'}=\sqrt{\frac{|B|}{|A|}}X^T_{B'\to A'}\ket{\phi^+}_{BB'}.
	\end{equation}
\end{lem}
\begin{proof}
	\begin{align}
	X_{A\to B}\ket{\phi^+}_{AA'}=& \frac{1}{\sqrt{|A|}}\sum_{i=0}^{|A|-1}\sum_{j=0}^{|B|-1}X_{ji}\ket{j}_B\otimes\ket{i}_{A'}\nonumber\\
	=&\frac{1}{\sqrt{|A|}}\sum_{i=0}^{|A|-1}\sum_{j=0}^{|B|-1}X^T_{ij}\ket{j}_B\otimes\ket{i}_{A'}\nonumber\\
	=&\sqrt{\frac{|B|}{|A|}}X^T_{B'\to A'}\ket{\phi^+}_{BB'}.
	\end{align}
\end{proof}
 By the Schmidt decomposition and the mirror lemma, any state vector $\ket\psi_{AB}$ can be expressed in terms of the maximally entangled state,
\begin{equation}
	\ket\psi_{AB}=\sqrt{|B|}\psi_{A}^{1/2}V_{B'\to A}\ket{\phi^+}_{B'B}.
\end{equation}

Sometimes it is convenient to use \emph{subnormalized} quantum states, i.e. $\rho\ge 0$, $\tr\rho\le 1$. The set of subnormalized states on a Hilbert space $\hi$ is denoted by $\St_{\le}(\hi)$. This is mainly a mathematical tool. One can, however, think about subnormalized states as the result of some protocol that has aborted with a certain probability. When we use subnormalized states in this thesis, we will explicitly mention it.

\subsection{Measurement}

The last section detailed how to describe quantum systems, containing quantum information, mathematically. As human beings, however, we interact with the world in a classical manner. This means that if we want to reap the benefits that the superior information processing power of quantum mechanical systems promises, we have to understand how to extract classical information from a quantum system. This process is called \emph{measurement}. While the interpretation of what exactly happens in the physical process of measurement is the subject of an intense ongoing debate, its mathematical description for the purpose of quantum information theory is well-established and will be described in the following.

A measurement on a quantum system $A$ is mathematically described by a \emph{positive operator-valued measure (POVM)}. A POVM is a family of positive semidefinite matrices $\{E^{(i)}_A\}_{i=0}^{r-1}$, $E^{(i)}_A\in\mathcal P(\hi_A)$ such that $\sum_i E^{(i)}_A=\mathds 1_A$. The set $\{0,1,...,r-1\}$ is the set of \emph{outcomes}. When the measurement is applied to a quantum system in state $\rho_A$, the probability of the measurement returning outcome $i$ is $p^{(i)}_\rho=\tr E^{(i)}_A\rho_A$. Note that $p^{(i)}_\rho=(E^{(i)}_A,\rho_A)_{\mathrm{HS}}$, where for $X,Y\in\Hom{\hi_A}{\hi_B}$, $(X,Y)_{\mathrm{HS}}=\tr X^\dagger Y$ is the Hilbert Schmidt inner product.

A given a POVM $\{E^{(i)}_A\}_{i=0}^{r-1}$ can be applied to the $A$-part of a composite system $AB$ by defining a POVM $\{\hat E^{(i)}_{AB}\}_{i=0}^{r}$ by setting $\hat E^{(i)}_{AB}=E^{(i)}_A\otimes\mathds 1_B$. As a shorthand we write $\{E^{(i)}_A\}_{i=0}^{r}$ instead of $\{\hat E^{(i)}_{AB}\}_{i=0}^{r}$, omitting the identity matrix $\one_B$.

\subsection{Quantum channels}\label{subs:QChan}

Transformations between quantum states are mathematically described by completely positive trace-preserving linear maps $\Lambda_{A\to B}\in\Hom{\End{\hi_A}}{\End{\hi_B}}$, called \emph{quantum channels}. As for quantum states, the subscript is omitted at times, i.e. we write $\Lambda \in\Hom{\End{\hi_A}}{\End{\hi_B}}$. To avoid confusion with the identity matrix, we denote the identity map in $\End{\End{\hi_A}}$ by $\id_A$. The tensor product of two maps $\Lambda_{A\to B}\in\Hom{\End{\hi_A}}{\End{\hi_B}}$ and $\Lambda'_{C\to D}\in\Hom{\End{\hi_C}}{\End{\hi_D}}$ is defined by its action on product matrices, $\Lambda_{A\to B}\otimes\Lambda'_{C\to D}(X_A\otimes Y_C)=\Lambda_{A\to B}(X_A)\otimes\Lambda'_{C\to D}(Y_C)$.

\begin{defn}[Quantum Channel]
  A map $\Lambda_{A\to B}\in\Hom{\End{\hi_A}}{\End{\hi_B}}$ is called
  \begin{enumerate}[i)]
   \item positive if $\Lambda(\mathcal P(\hi_A))\subset \mathcal P(\hi_B)$
   \item completely positive (CP) if $\Lambda_{A\to B}\otimes \mathrm{id}_C$ is positive for all $C$
   \item trace preserving (TP) if $\tr\Lambda_{A\to B}(X_A)=\tr X_A$ for all $X_A\in\End{\hi_A}$.
  \end{enumerate}
If a map has properties $ii)$ and $iii)$, it is called a CPTP map or quantum channel. The set of quantum channels from $A$ to $B$ is denoted by $\mathcal{CPTP}_{A\to B}$.
\end{defn}

Embedding a POVM into a composite system as described in the last section defines a CP-map $\iota_{A\to AB}(X_A)=X_A\otimes\1_B$. An important quantum channel is the \emph{partial trace} $\tr_B=\left(\iota_{A\to AB}\right)^\dagger$, the adjoint of the embedding map with respect to the Hilbert Schmidt inner product. On product matrices it acts as $\tr_B(X_A\otimes Y_B)=\tr(Y_B) X_A$. Note that the subscript denotes the system that is discarded. The partial trace is used to define the \emph{marginals} of a quantum state $\rho_{AB}$, $\rho_A=\tr_B\rho_{AB}$ and $\rho_B=\tr_A\rho_{AB}$.

There are two important characterization theorems for CPTP maps, the Kraus and Stinespring representation theorems. 

\begin{thm}[Kraus representation, \cite{Kraus1971}]
  Let $\Lambda_{A\to B}\in\Hom{\End{\hi_A}}{\End{\hi_B}}$. $\Lambda_{A\to B}$ is CP if and only if there exists a set $\{A_i\}_{i=0}^{r-1}\subset \Hom{\hi_A}{\hi_B}$ of matrices such that
  \begin{equation}
   \Lambda_{A\to B}(X)=\sum_{i=0}^{r-1}A_iXA_i^\dagger
  \end{equation}
  for all $X\in\End{\hi_A}$. $\Lambda$ is  in addition TP if and only if $\sum_{i=0}^{r-1}A_i^\dagger A_i=\one_A$.
 
\end{thm}

The reversible channels from a quantum system to itself are unitary channels $\mathcal U_{A\to A}(X_A)=U_AXU_A^\dagger$, where $U_A$ is a unitary matrix

\begin{thm}[Stinespring representation, \cite{Stinespring1955}]\ \ \\
   Let $\Lambda_{A\to B}\in\Hom{\End{\hi_A}}{\End{\hi_B}}$. $\Lambda_{A\to B}$ is CPTP if and only if there exists a Hilbert space $\hi_E$ and an isometry $V\in\Hom{\hi_A}{\hi_{BE}}$ such that
   \begin{equation}
   \Lambda_{A\to B}(X)=\tr_E VXV^\dagger
  \end{equation}
  for all $X\in\End{\hi_A}$. $V$ is called a \emph{Stinespring dilation} of $\Lambda$.
\end{thm}

The latter theorem is important to relate the notion of discrete dynamics in the form of quantum channels back to the physical picture of Schrödinger quantum mechanics.  Any quantum channel can be implemented by appending a quantum system, applying a unitary time evolution and then discarding part of the result.

An important tool when analyzing quantum channels is the Choi-Jamio\l kowski isomorphism.
\begin{thm}[Choi Jamio\l kowski isomorphism, \cite{Choi1975,Jamiolkowski1972}]
	Let $\hi_{A'}\cong \hi_A$. The map 
	\begin{align}
	\Phi: \Hom{\End{\hi_A}}{\End{\hi_B}}\to &\End{\hi_{A'}\otimes\hi_B}\nonumber\\
	\Lambda_{A\to B}\mapsto &\Lambda_{A\to B}(\proj{\phi^+}_{AA'})
	\end{align}
	is an isomorphism. We write $\eta_{\Lambda}=\Phi(\Lambda)$. Furthermore, $\Lambda$ is CP if and only if $\eta_\Lambda\ge 0$, and $\Lambda$ is TP if and only if $(\eta_{\Lambda})_{A'}=\tau_{A'}$.
\end{thm}

The state $\eta_{\Lambda}$ is called the \emph{Choi-Jamio\l kowski state} or \emph{CJ state} of $\Lambda$.

\section{Distance measures}

How similar or different are two quantum states? How about two quantum channels? To answer these question, metrics on the respective sets are introduced. The most important metric is the \emph{trace distance}. The trace distance is defined in terms of the Schatten 1-norm or trace norm,
\begin{equation}
	\|M\|_1=\tr\sqrt{M^\dagger M},
\end{equation}
for a matrix $M\in\Hom{\hi_A}{\hi_B}$. The trace distance is now defined as half the trace norm distance of two quantum states $\rho,\sigma\in\sos{\hi}$,
\begin{equation}
	\delta(\rho,\sigma)=\frac{1}{2}\|\rho-\sigma\|_1.
\end{equation}
On classical states, it is equal to the total variational distance of the corresponding probability distributions. The importance of the trace distance stems from the fact, that it has an operational interpretation. The optimal strategy of guessing whether an unknown state is $\rho$ or $\sigma$, when promised that it is one of the two with probability $\frac 1 2$ each, has a sucess probability of 
\begin{equation}
	p_{\mathrm{succ}}=\frac{1}{2}(1+\delta(\rho,\sigma)).
\end{equation}

The relevant norm for state vectors is the 2-norm which is induced by the inner product, $\|\ket\phi\|_2=\bracket{\phi}{\phi}$. The following lemma relates the 2-norm distance of two state vectors to the trace distance of their projectors.

\begin{lem}\label{lem:1-norm-2-norm}
	Let $\ket{\psi},\ket{\phi}\in\hi$ be two vectors. Then
	\begin{equation}
	\|\proj{\psi}-\proj{\phi}\|_1\le 2\|\ket{\psi}-\ket{\phi}\|_2.
	\end{equation}
\end{lem}
\begin{proof}
	The trace norm distance of two pure states is given by \cite{Nielsen2000}
	\begin{equation}
	\|\proj{\psi}-\proj{\phi}\|_1=2\sqrt{1-|\bracket{\phi}{\psi}|^2}.
	\end{equation}
	We bound
	\begin{align}
	\|\proj{\psi}-\proj{\phi}\|_1=&2\sqrt{1-|\bracket{\phi}{\psi}|^2}\nonumber\\
	=&2\sqrt{(1-|\bracket{\phi}{\psi}|)(1+|\bracket{\phi}{\psi}|)}\nonumber\\
	\le&2\sqrt{2(1-|\bracket{\phi}{\psi}|)}\nonumber\\
	\le&2\sqrt{2(1-\mathrm{Re}(\bracket{\phi}{\psi}))}\nonumber\\
	=&2\|\ket{\psi}-\ket{\phi}\|_2
	\end{align}
	
\end{proof}

Another important quantity that is used to define a metric is the \emph{fidelity}, a measure of similarity. The fidelity of two quantum states $\rho, \sigma\in\sos{\hi}$ is defined as
\begin{equation}
	F(\rho,\sigma)=\|\sqrt{\rho}\sqrt{\sigma}\|_1.
\end{equation}
If $\sigma=\proj{\sigma}$ is pure, then a simpler formula, holds,
\begin{equation}
	F(\rho,\sigma)=\sqrt{\bra{\sigma}\rho\ket\sigma},
\end{equation}
and therefore
\begin{equation}
	F(\rho,\sigma)=\big|\!\bracket{\sigma}{\rho}\!\big|
\end{equation}
if in addition $\rho=\proj{\rho}$. The fidelity can be used to define a metric on the set of quantum states. The purified distance of two quantum states $\rho, \sigma\in\sos{\hi}$ is defined as
\begin{equation}
	P(\rho,\sigma)=\sqrt{1-F(\rho,\sigma)^2}.
\end{equation}
For pure states $\proj\psi,\proj\phi$ we sometimes write $F(\ket\psi,\ket\phi):=F(\proj\psi,\proj\phi)$ and $P(\ket\psi,\ket\phi):=P(\proj\psi,\proj\phi)$.

The purified distance and the trace distance can be related in the following way,
\begin{equation}\label{eq:tr2fid}
	\frac{1}{2}P(\rho,\sigma)^2\le 1-\sqrt{1-P(\rho,\sigma)^2}\le\delta(\rho,\sigma)\le P(\rho,\sigma).
\end{equation}
These inequalities are called Fuchs van de Graaf inequalities.

Both metrics can be extended to the set of subnormalized states \cite{tomamichel2010duality}. For two sub-normalized quantum states $\rho,\sigma\in\St_{\le}(\hi)$, the trace distance is defined as
	\begin{align*}
	\delta(\rho,\sigma)=\frac 1 2 \left(\|\rho-\sigma\|_1+|\tr(\rho-\sigma)|\right).
	\end{align*}
	%where $\|\cdot\|_1$ is the Schatten 1-norm.
	%\footnote{The Schatten 1-norm is defined as $\|A\|=\tr\sqrt{A^\dagger A}$ for a matrix $A$.}
	Their purified distance is defined as
	\begin{align*}
	P(\rho,\sigma)=\sqrt{1-F(\rho,\sigma)^2},\quad\mathrm{where}\quad F(\rho,\sigma)=\|\sqrt{\rho}\sqrt{\sigma}\|_1 +\sqrt{(1-\tr\rho)(1-\tr\sigma)}
	\end{align*}
	is the \emph{generalized fidelity}. We extend these definitions to apply to pairs of probability distributions by considering the corresponding diagonal density matrices. $B_\varepsilon(\rho)\subset\St_{\le}(\hi)$ denotes the purified distance ball of radius $\varepsilon$ around $\rho\in\St_{\le}(\hi)$. The Fuchs-van-de-Graaf inequalities \eqref{eq:tr2fid} hold for these generalized versions of trace distance and purified distance as well \cite{Tomamichel2015}. These generalized measures reflect the interpretation of subnormalized states as ``partial states" in the sense of being the result of a probabilistic protocol that aborts sometime.
	
One of the main reasons why the fidelity and the purified distance are so handy in practice is Uhlmann's theorem.
\begin{thm}[Uhlmann's theorem, \cite{Uhlmann1985}]\label{thm:uhlmann}
	Let $\rho_{A},\sigma_{A}$ be quantum states. Then
	\begin{equation}
		F(\rho_A, \sigma_A)=\max_{\rho_{AE},\sigma_{AE}}F\left(\rho_{AE},\sigma_{AE}\right),
	\end{equation}
	where the maximum is taken over all purifications $\rho_{AE}$ of $\rho_A$ and $\sigma_{AE}$ of $\sigma_A$.
\end{thm}

Both metrics introduced above are non-increasing under quantum channels, and therefore in particular under the partial trace,
\begin{align}
	\delta(\rho_{AB},\sigma_{AB})\ge&\delta(\rho_A,\sigma_A), \text{ and}\nonumber\\
	P(\rho_{AB},\sigma_{AB})\le P(\rho_A,\sigma_A)\iff&F(\rho_{AB},\sigma_{AB})\le F(\rho_A,\sigma_A).
\end{align}
This property is also called monotonicity under the partial trace. Forgetting the eigenbases of two states does not increase their trace distance.
\begin{lem}\cite[Box 11.2]{Nielsen2000}\label{lem:spectr}
	We have
	\begin{align*}
	\delta(\rho,\sigma)\ge\delta(\spec(\rho),\spec(\sigma)),
	\end{align*}
	where $\spec(A)$ denotes the ordered spectrum of a Hermitian matrix $A$.
\end{lem}

The trace norm naturally induces a norm on $\Hom{\End{\hi_A}}{\End{\hi_B}}$, the completely bounded trace norm or \emph{diamond norm}. The diamond norm of a map $\Lambda\in\Hom{\End{\hi_A}}{\End{\hi_B}}$ is defined as
\begin{equation}
	\|\Lambda\|_\diamond=\sup_{\rho_{AE}}\|\Lambda\otimes\id_E(\rho)\|_1.
\end{equation}
By the monotonicity of the trace distance under partial trace, there is always a pure optimizer. Via the operational interpretation of the trace distance, the diamond norm has a corresponding interpretation. Suppose we are given a quantum channel that is promised to be equal to $\Lambda_1$ or $\Lambda_2$ with probability $\frac 1 2$ each. Then the optimal success probability of guessing which one it is, when we are allowed to use it once on part of an arbitrary quantum state, is
\begin{equation}
	p_{\mathrm{succ}}=\frac{1}{2}\left(1+\frac 1 2\left\|\Lambda_1-\Lambda_2\right\|_\diamond\right).
\end{equation}

Another measure of similarity between quantum channels is the so called entanglement fidelity,
\begin{equation}
	F(\Lambda_1,\Lambda_2)=F(\eta_{\Lambda_1},\eta_{\Lambda_2}).
\end{equation}
We write $F(\Lambda):=F(\Lambda,\id)$. The entanglement fidelity with respect to the identity can be easily related to the diamond norm.

\begin{lem}\label{lem:entanglementf2diamond}
	For a quantum channel $\Lambda\in\mathcal{CPTP}_{A\to A}$,
	\begin{equation}
		1-F(\Lambda)^2\le\|\Lambda-\id\|_{\diamond}/2\le |A|\sqrt{1-F(\Lambda)^2}.
	\end{equation}
\end{lem}
\begin{proof}
	For the first inequality we bound
	\begin{align}
			\|\Lambda-\id\|_\diamond\ge&\|\Lambda(\phi^+_{AA'})-\phi^+_{AA'}\|_1\nonumber\\
		\ge &2\left(1-F(\Lambda(\phi^+_{AA'}),\phi^+_{AA'})^2\right)\nonumber\\
		=&2\left(1-F(\Lambda)^2\right).
	\end{align}
	Here we used the definition of the diamond norm in the first, and a strengthened Fuchs van de Graaf inequality from Exercise 9.21 in \cite{Nielsen2000} in the second inequality.

	For the other inequality, let $\proj\psi_{AA'}$ be a pure quantum state such that $\|\Lambda-\id\|_\diamond=\|(\Lambda-\id)(\proj\psi)\|_1$. Let $\ket\psi_{AA'}=\sqrt{|A|}\psi_A'^{1/2}U_A'\ket{\phi^+}_{AA'}$. Now we bound
	 \begin{align}
	 	\|\Lambda-\id\|_\diamond=&\left\|(\Lambda-\id)(\proj\psi)\right\|_1\nonumber\\
	 	=&|A|\left\|\psi_A'^{1/2}(\Lambda-\id)(\proj{\phi^+})\psi_A'^{1/2}\right\|_1\nonumber\\
	 	\le&|A|\|\psi_A'^{1/2}\|_\infty^2\left\|\Lambda(\phi^+_{AA'})-\phi^+_{AA'}\right\|_1\nonumber\\
	 	\le&2|A|\sqrt{1-F(\Lambda)^2}.
	 \end{align}
	 For the first inequality we used Hölder's inequality twice, and the second inequality is a Fuchs van de Graaf inequality. Finally we use that $\|\psi_A'^{1/2}\|_\infty^2=\|\psi_A'\|_\infty\le \|\psi_A'\|_1=1$.
\end{proof}

If the CJ states of two quantum channels are close, the channels are also close in diamond norm. The approximation gets worse by a dimension factor, however.

\begin{lem}\label{lem:closeCJ2diamond}
	Let $\Lambda^{(i)}_{A\to B}$, $i=0,1$ be CPTP maps such that
	$$\left\|\eta_{\Lambda^{(0)}}-\eta_{\Lambda^{(1)}}\right\|_1\le \varepsilon.$$
	Then the two maps are also close in diamond norm,
	$$\left\|\Lambda^{(0)}_{A\to B}-\Lambda^{(1)}_{A\to B}\right\|_\diamond\le|A|\varepsilon.$$
\end{lem}
\begin{proof}
	The proof of this lemma is a simple application of the Hölder inequality. Let $\ket{\psi}_{AA'}=\sqrt{|A|}\psi_{A'}^{1/2}V_{A'}\ket{\phi^+}_{AA'}$ be an arbitrary pure state with $V_{A'}$ a unitary. Then we have
	\begin{align}
	\left\|(\Lambda^{(0)}_{A\to B}-\Lambda^{(1)}_{A\to B})(\proj{\psi})\right\|_1=&|A|\left\|\psi_{A'}^{1/2}V_{A'}(\eta_{\Lambda^{(0)}}-\eta_{\Lambda^{(1)}})V_{A'}^\dagger\psi_{A'}^{1/2}\right\|_1\nonumber\\
	\le&|A|\left\|\psi_{A'}^{1/2}\right\|_{\infty}^2\left\|V_{A'}\right\|_{\infty}^2\left\|\eta_{\Lambda^{(0)}}-\eta_{\Lambda^{(1)}}\right\|_1\le|A|\varepsilon.
	\end{align}
	
\end{proof}

\section{Entropies}

This section introduces the Shannon and von Neumann entropies, which are central to information theory, as well other notions of entropy used in this thesis.

\subsection{The Shannon and von Neumann Entropies}
From the conception of information theory by Claude Shannon \cite{Shannon1948} on, entropies have played an important role, in particular the Shannon entropy. The Shannon entropy of a probability distribution $p_X: \mathcal X\to[0,1]$ of a classical system $X$ is defined as
\begin{equation}
 H(X)_p=H(p_X)=-\sum_{x\in\mathcal X}p(x)\log(p(x)).
\end{equation}
For a binary distribution $p:\{0,1\}\to [0,1]$, the Shannon entry only depends on a single parameter, $p(0)$ due to the normalization of the probability distribution. One therefore defines the \emph{binary entropy} function
\begin{equation}
 h(q)=-q\log q-(1-q)\log (1-q),
\end{equation}
such that $H(p)=h(p(0))$ in the example.

For a joint distribution $p:\mathcal X\times \mathcal Y\to[0,1]$ of two classical systems $X$ and $Y$, the conditional entropy of $X$ given $Y$ is defined as
\begin{equation}
 H(X|Y)_p=\sum_{y\in\mathcal Y}p_Y(y)H(p_{X|Y}(\cdot|y))=H(XY)_p-H(Y)_p,
\end{equation}
where $p_{X|Y}(x|y)=\frac{p_{XY}(x,y)}{p_Y(y)}$ and $p_Y$ is the $Y$-marginal of $p_{XY}$. The conditional entropy is nonnegative, as it is a convex combination of Shannon entropies, which are nonnegative. This fact is called the \emph{monotonicity} inequality,
\begin{equation}\label{eq:monotonicity}
 H(X|Y)_p\ge 0
\end{equation}

While there are many more useful properties of the Shannon entropy worth mentioning, we will move on to the the quantum generalization of the Shannon entropy. This is because this thesis is mainly concerned with the quantum case, and whenever needed, the classical results can be stated within the quantum formalism by restricting it to classical states.

The von Neumann entropy is defined as the Shannon entropy of the spectrum of a quantum state, or, equivalently, $H(\rho)=-\tr\rho\log\rho$. Note that the von Neumann entropy is denoted by the same letter, $H$, as the Shannon entropy. This is only a very slight abuse of notation, as any probability distribution corresponds to a diagonal quantum state in a canonical way. To emphasize which marginal of a state the von Neumann entropy is evaluated on, we often write $H(AB)_\rho=H(\rho_{AB})$ and $H(B)_\rho=H(\rho_B)$ etc.

By the Schmidt decomposition, the spectra of the two marginals $\rho_A$ and $\rho_B$ of a pure state $\rho_{AB}$ are the same, therefore their entropies are as well, $H(A)_\rho=H(B)_\rho$.

\subsection{The quantum relative entropy}
Another entropic quantity is the the quantum relative entropy \cite{Umegaki1962},
\begin{equation}
 D(\rho\|\sigma)=\tr\rho(\log\rho-\log\sigma).
\end{equation}
The quantum relative entropy has the flavor of a distance measure, as it is nonnegative and $D(\rho\|\sigma)=0$ if and only if $\rho=\sigma$. It is, however, not symmetric under the exchange of its argument, so it is not a metric on the set of quantum states. Another important property of the quantum relative entropy that is expected from any reasonable distance measure on the set of quantum states is contraction under CPTP maps,
\begin{equation}\label{eq:relent-dataproc}
 D(\Lambda(\rho)\|\Lambda(\sigma))\le D(\rho\|\sigma),
\end{equation}
for all quantum channels $\Lambda$. This is referred to as the data processing inequality of the Relative entropy.

Despite not being a metric, the quantum relative entropy can be related to the trace distance.

\begin{lem}[Pinskers inequality, see \cite{Ohya1993}]\label{lem:pinsker}
	For quantum states $\rho_{AB}$ and $\sigma_{AB}$,
	\begin{align}
	D(\rho_A||\sigma_A)\ge&\frac 1 2\|\rho_A-\sigma_A\|_1^2.
	\end{align}
\end{lem}

\subsection{Information measures and entropy inequalities}

In analogy to the classical quantities, the quantum conditional entropy, the quantum mutual information, and the quantum conditional mutual information, of a quantum state $\rho_{ABC}$, are defined by
\begin{align}
 H(A|B)_\rho=&H(AB)_\rho-H(B)_\rho\\
 I(A:B)_\rho=&H(A)_\rho+H(B)_\rho-H(AB)_\rho\\
 I(A:B|C)_\rho=&H(AC)_\rho+H(BC)_\rho-H(ABC)_\rho-H(C)_\rho.\label{eq:SSA}
\end{align}

It turns out that the latter two quantities are nonnegative. This fact is referred to as subadditivity and strong subadditivity \cite{Lieb1973}, respectively. Like the monotonicity \eqref{eq:monotonicity} of the Shannon entropy, these inequalities are examples of entropy inequalities, which will be the topic of Chapter \ref{chap:entropy}, and can both be derived from the properties of the Relative entropy. For subadditivity it is sufficient to observe that $I(A:B)_\rho=D(\rho_{AB}\|\rho_A\otimes\rho_B)$ and to conclude the inequality from the nonnegativity of the quantum relative entropy. Strong subadditivity, on the other hand, follows from the data processing inequality, Equation \eqref{eq:relent-dataproc}. The quantum conditional mutual information can be expressed as a difference of mutual informations, $I(A:B|C)_\rho=I(A:BC)_\rho-I(A:C)_\rho$. Using the expression for the quantum mutual information in terms of the quantum relative entropy, it is easy to see that $I(A:B|C)\ge 0$ is equal to Equation \eqref{eq:relent-dataproc} for $\rho=\rho_{ABC}$, $\sigma=\rho_A\otimes\rho_{BC}$, and $\Lambda=\tr_C$.

The monotonicity inequality \eqref{eq:monotonicity} that holds for the Shannon entropy does not hold for the von Neumann entropy, as can be seen looking at any entangled bipartite pure state. It does, however, hold for all separable states \cite{Nielsen2001}.

The (conditional) quantum mutual information fulfills the following chain rule,
\begin{equation}\label{eq:chainrule}
 I(A:BC|D)=I(A:B|D)+I(A:C|BD),
\end{equation}
which holds for trivial $D$ as well.
Using this equation and strong subadditivity yields an alternative form of the strong subadditivity inequality,
\begin{equation}\label{eq:ssa-alt}
 I(A:BC|D)\ge I(A:B|D).
\end{equation}

Let $\Lambda_{B\to C}$ be a quantum channel and $V_{B\to CE}$ a Stinespring dilation of $\Lambda$. Then the above inequality yields the data processing inequality for the conditional quantum mutual information,
\begin{equation}\label{eq:dataproc-mut}
 I(A:C|D)_{\Lambda(\rho)}\le I(A:B|D)_\rho
\end{equation}
for all $\rho_{ABD}$.

The von Neumann information measures are continuous functions of the quantum state they are evaluated on. Explicit continuity bounds are given in the following lemma.

\begin{lem}[Fannes-Audenaert inequality, Alicki-Fannes inequality, \cite{Fannes1973,Audenaert2007,Alicki2004,Wilde2013}]\label{lem:fannes}
	\textcolor{white}{Fck LaTeX!}\\Let $\rho_{ABC}$ and $\rho'_{ABC}$ be tripartite quantum states such that
	\begin{equation}
	\|\rho_{ABC}-\rho'_{ABC}\|_1\le\varepsilon.
	\end{equation}
	Then the following continuity bounds hold for entropic quantities:
	\begin{align}
	|H(A)_\rho-H(A)_{\rho'}|\le &\frac \varepsilon 2 \log\left(|A|-1\right)+h\left(\frac \varepsilon 2\right)\nonumber\\
	|H(A|B)_\rho-H(A|B)_{\rho'}|\le &4\varepsilon \log\left(|A|\right)+2h(\varepsilon)\nonumber\\
	|I(A:B)_\rho-I(A:B)_{\rho'}|\le &5\varepsilon\log\left(\min(|A|, |B|)\right)+3h(\varepsilon)\nonumber\\
	|I(A:B|C)_\rho-I(A:B|C)_{\rho'}|\le &8\varepsilon \log\left(\min(|A|, |B|)\right)+4h(\varepsilon).
	\end{align}
\end{lem}

\subsection{R\'enyi entropies}

What the Neumann entropy is for quantum information theory in the setting of many independent, identically distributed (IID) copies of a task, are quantum R\'eniy-entropic quantities for one-shot quantum information theory. From the sprawling zoo of quantum R\'enyi entropies, the species that are introduced here are those that will be used in Chapter \ref{chap:one-shot}. In this subsection, we will use subnormalized states.

The most widely used quantum R\'enyi entropies are the conditional min- and max-entropy. They are defined with the help of a R\'enyi variant of the relative entropy, the max-relative entropy.

\begin{defn}[Max-relative entropy]
	The \emph{max-relative entropy} of a state $\rho\in\St_\le(\hi)$ with respect to a state $\sigma\in\St(\hi)$ is defined as
	\begin{align*}
	D_{\max}(\rho\|\sigma)=\min\left\{\lambda\in\R\Big|2^\lambda\sigma\ge\rho\right\}.
	\end{align*}
\end{defn}
The conditional min- and max-entropy are defined as follows.
\begin{defn}[Conditional min- and max-entropy, \cite{renner2005security,tomamichel2010duality}]
	The conditional min-entropy of a positive semidefinite matrix $\rho_{AB}\in\End{\hi_{A}\otimes\hi_B}$ is defined as
	\begin{align*}
	H_{\min}(A|B)_\rho&=\max_{\sigma}\max\left\{\lambda\Big|2^{-\lambda}1_A\otimes\sigma_B\ge \rho_{AB}\right\}\\
	&=\max_\sigma\left(-D_{\max}\left(\rho_{AB}\big\|\mathds 1_A\otimes\sigma_B\right)\right),
	\end{align*}
	where the maximum is taken over all normalized quantum states. The conditional max-entropy is defined as the dual of the conditional min-entropy in the sense that
	\begin{align*}
	H_{\max}(A|B)_\rho=-H_{\min}(A|C)_\rho,
	\end{align*}
	where $\rho_{ABC}$ is a purification of $\rho_{AB}$. 
\end{defn}

The conditional max-entropy can be expressed in terms of the fidelity.

\begin{lem}[\cite{konig2009operational}]
	We have
	\begin{align*}
	H_{\max}(A|B)_\rho=\max_{\sigma\in\St(\hi_B)}2\log \sqrt{|A|}\,F(\rho_{AB},\tau_A\otimes\sigma_B),
	\end{align*}
	where $|A|=\dim\hi_A$.
\end{lem}
	The unconditional min- and max-entropy are defined as their conditional counterparts with a trivial conditioning system.
	
	\begin{lem}\cite{konig2009operational}\label{lem:uncond}
		The min and max-entropy are given by
		\begin{align*}
		H_{\min}(\rho)=-\log\|\rho\|_\infty\quad\mathrm{and}\quad H_{\max}(\rho)=2\log\tr\sqrt{\rho}.
		\end{align*} 
	\end{lem}
Here, $\|X\|_\infty$ denotes the operator norm of $X$.

As many information processing tasks allow for a small error, and to be able to take the IID limit, it is often necessary to optimize certain protocols over a small ball around the actual input state. On the other hand, min- and max entropies can change much faster with the quantum state they are evaluated on than the von Neumann entropy. While the latter changes only by an amount of order $\varepsilon\log d-\varepsilon\log\varepsilon$ when a quantum state on a $d$-dimensional hilbert space changes by $\varepsilon$, the former can change by an amount of order $\varepsilon d$. Therefore one defines so-called smoothed versions of the min- and max-entropy (as well as for other R\'enyi-entropic quantities).

The smooth conditional min- and max-entropies are defined by maximizing and minimizing over a ball of sub-normalized states $\tilde\rho_{AB}$, respectively,
\begin{align*}
H_{\min}^\varepsilon(A|B)_\rho=&\max_{\tilde{\rho}\in B_\varepsilon(\rho)}H_{\min}(A|B)_{\tilde{\rho}}\\
H_{\max}^\varepsilon(A|B)_\rho=&\min_{\tilde{\rho}\in B_\varepsilon(\rho)}H_{\max}(A|B)_{\tilde{\rho}}.
\end{align*}

As an auxiliary quantity we also need the unconditional R\'enyi entropy of order $0$.

\begin{defn}
	For a quantum state $\rho_A\in\St(\hi_A)$ the R\'enyi entropy of order 0 is defined by
	\begin{align*}
	H_0(A)_\rho=\log\mathrm{rk}(\rho_A),
	\end{align*}
	where $\mathrm{rk}(X)$ denotes the rank of a matrix $X$. As in the case of the max-entropy, the smoothed version is defined by minimizing over an epsilon ball,
	\begin{align*}
	H_{0}^\varepsilon(A)_\rho=\min_{\tilde{\rho}\in B_\varepsilon(\rho)}H_{0}(A)_{\tilde{\rho}}.
	\end{align*}
\end{defn}

The smoothed $0$-entropy is almost equal to the smoothed max-entropy.

\begin{lem}\cite[Lemma 4.3]{renner2004smooth}\label{lem:H0Hmax}
	We have
	\begin{align*}
	H_{\max}^{2\varepsilon}(\rho)\le H_0^{2\varepsilon}(A)_\rho\le H_{\max}^\varepsilon(\rho)+2\log(1/\varepsilon).
	\end{align*}
\end{lem}

\section{Representation theory}

Group representation theory is the theory of group homomorphisms from a group $G$ to the general linear group $\gl(n,\mathbb{F})$ of $n\times n$ matrices over a field $\mathbb{F}$. In this section we will review some basics about complex (i.e. $\mathbb{F}=\C$) representations of finite and compact groups. This is to prepare the reader for Chapters \ref{chap:one-shot} and \ref{chap:crypto} which make use of this beautiful theory. 

Let $G$ and $H$ be groups. A \emph{group homomorphism} from $G$ to $H$ is a map $\phi: G\to H$ such that $\phi(gg')=\phi(g)\phi(g')$ for all $g,g'\in G$. A group homomorphism that is invertible is called a group isomorphism. A \emph{representation} of $G$ is a group homomorphism $\phi: G\to \gl(n,\C)$. Via the defining action of $\gl(n,\C)$ on $V=\C^n$, $\phi$ defines an action of $G$ on $V$ for this action we write $g\ket v:=g.\ket{v}:=\phi(g)\ket v$ for $\ket v\in V$.  Two representations $\phi$ on $V$ and $\psi$ on $W$ are equivalent, if there exists an isomorphism $X\in\Hom{V}{W}$ such that $X\phi(g)=\psi(g)X$ for all $g\in G$. $X$ is called an \emph{intertwiner}. A representation $\phi$ is called unitary if $\phi(g)$ is unitary for all $\gamma\in G$. Every representation of a finite or compact group is equivalent to a unitary representation, so we will only consider unitary representations. If clear from context, we will call $V$ a representation without mentioning $\phi$.

A representation $\phi$ on $V$ is called \emph{reducible}, if there exists a decomposition $V\cong V_1\oplus V_2$ such that $\phi(g)=\phi_1(g)\oplus \phi_2(g)$ for all $g\in G$, otherwise it is called \emph{irreducible}. It turns out that every representation of a finite or compact group can be expressed as a direct sum of irreducible representations.

\begin{thm}[Maschke's Theorem]
	Let $V$ be a representation of $G$. Then there exists a decomposition 
	\begin{equation}\label{eq:irrepsum}
		V=\bigoplus_i V_i\otimes M_i
	\end{equation}
	where the $V_i$ are pairwise inequivalent irreducible representations of $G$, $M_i$ are complex vector spaces called \emph{multiplicity spaces}, and this decomposition is unique up to isomorphism.
\end{thm}

Schur's lemma shows that homomorphisms between irreducible representations of a group $G$ have a very simple structure.

\begin{lem}[Schur's Lemma]
	Let $V$ and $W$ be irreducible representations of a group $G$, and let $X\in\Hom{V}{W}$ be an intertwiner. Then $X$ is either an isomorphism, or $X=0$. If $V=W$ (as representations), then $X=\lambda \one$ for some $\lambda\in C$. 
\end{lem}

As a corollary, we can characterize the algebra of matrices that commutes with a general representation.

\begin{cor}\label{cor:compositeSchur}
	Let 
	\begin{equation}
	V=\bigoplus_i V_i\otimes M_i
	\end{equation}
	be a decomposition of a representation $V$ of a group $G$ into irreducible representations. Let $X\in\End{V}$ commute with the action of $G$, i.e. $gX=Xg$. Then $X$ has the form
	\begin{equation}
		X=\bigoplus_i \one_{V_i}\otimes X_{M_i}
	\end{equation}
	for some matrices $X_i\in\End{M_i}$.
\end{cor}

One way to produce matrices that commute with a certain representation is by applying the \emph{twirl} corresponding to the representation. For a representation $V$ of a finite group $G$, We define $\mathcal T_V\in\End{\End{V}}$ by
\begin{equation}
	\mathcal T_V(X)=\frac{1}{|G|}\sum_{g\in G} g X g^\dagger.
\end{equation}
It is easy to see that $g\mathcal T_V(X)=\mathcal T_V(X) g$ for all $g\in G$, so Corollary \ref{cor:compositeSchur} applies. The twirl can be defined for compact groups as well. The normalized sum over group elements is here replaced by the integral over the group using the Haar measure.

Two particularly important representations in quantum information theory are the natural representations of the symmetric group $S_n$ and the special unitary group $\mathrm{SU}(d)$ on the Hilbert space $\hi=\left(\C^d\right)^n$. $S_n$ acts by permuting the tensor copies of $\C^d$, i.e. $\pi.\bigotimes_{i=1}^n\ket{\phi_i}=\bigotimes_{i=1}^n\ket{\phi_{\pi^{-1}(i)}}$ for $\pi\in S_n$ and $\ket\phi_i\in\C^d$, and $\mathrm{SU}(d)$ acts diagonally, i.e. by $U.\ket\psi =U^{\otimes n}\ket\psi$ for $\ket\psi\in\hi$ and $U\in\mathrm{SU}(d)$. These two representation commute, i.e. $\pi.(U.\ket\psi)=U.(\pi.\ket\psi)$. This implies that the two representations define a representation of the product group $S_n\times \mathrm{SU}(d)$. What is more, the algebras generated by the two representations are \emph{double commutants}. In other words, the decomposition of the resulting representation of $S_n\times \mathrm{SU}(d)$ into irreducible representations is multiplicity-free, i.e.  the spaces $M_i$ from equation \eqref{eq:irrepsum} when decomposing the $S_n$ representation into irreducible representations, are irreducible representations of $\mathrm{SU}(d)$, and vice versa,
\begin{equation}\label{eq:SchurWeyl}
	\hi=\bigoplus_{\Lambda\vdash(n,d)}[\lambda]\otimes V_\lambda.
\end{equation}
Here, the direct sum is over all Young diagrams $\lambda\vdash(n,d)$ with $n$ boxes and at most $d$ rows, which index the irreducible representations $[\lambda]$ of $S_n$ and the polynomial irreducible representations $V_\lambda$ of $\mathrm{SU(d)}$, respectively. The fact that the representations of $S_n$ and $\mathrm{SU}(d)$ on $\hi$ are double commutants, as well as the decomposition \eqref{eq:SchurWeyl}, are known as \emph{Schur-Weyl duality}. A good introduction to Schur-Weyl duality can be found in \cite{Christandl2006}. 

Two special terms in the direct sum decomposition \eqref{eq:SchurWeyl} are the symmetric subspace $\bigvee^n\C^d$ and the antisymmetric subspace $\bigwedge^n\C^d$. They are irreducible $\mathrm{SU}(d)$ representations, and the corresponding one-dimensional $S_n$-representations are the trivial and the sign representation, respectively.

As a last remark, note that  $\mathrm{U}(d)$ is isomorphic to $\mathrm{U}(1)\times\mathrm{SU}(d)$, i.e. any irreducible representation of  $\mathrm{U}(d)$ is specified by (and is the product of) a pair of irreducible representations of the Abelian group $\mathrm{U}(1)$ (the determinant) and $\mathrm{SU}(d)$. Therefore the difference between $\mathrm{U}(d)$ and $\mathrm{SU}(d)$ is immaterial in quantum information theory: the $U(1)$-part cancels whenever conjugating a matrix with a unitary.

 % Introduction

\hchapter{Entropy inqualities}\label{chap:entropy}

Entropy is a fundamental quantity in information theory and thermodynamics, both classical and quantum. When applying the concept of entropy to information theoretic problems like, for example, channel coding, entropy inequalities provide fundamental limits to information processing. Entropy inequalities are linear inequalities relating the Shannon or von Neumann entropies of different overlapping marginals of a quantum state or probability distribution. Basic facts such as the impossibility of increasing correlations of a system $A$ with another system $B$ by acting on $A$ alone can be compactly expressed in terms of entropy inequalities. In more complex communication scenarios like that of classical network coding, even so-called non-Shannon type inequalities constrain the rate region for communication through a network. The conic geometry of the set of possible entropies implies that entropy inequalities are, in fact, the \emph{only} constraints on such rate regions.

While for four or more random variables, infinite families of unconstrained entropy inequalities are known, finding such inequalities beyond strong subadditivity for the von Neumann entropy has been an open question for many years. One reason why the classical non-Shannon type inequalities resist quantum proof so far is that the only known classical proof technique relies crucially on the ability of \emph{copying} a random variable. This makes a direct quantum generalization impossible due to the no-cloning theorem. 

This chapter of my thesis is dedicated to presenting a new constrained inequality for the von Neumann entropy of four party quantum states. In the first section I will introduce the classical and quantum entropy cones and review some previous results, before presenting the new inequality and its properties in the second section. In this Chapter we omit the subscript for entropic quantities that indicates which state it is evaluated on.

\section{Entropy cones}

In this section, the entropy vector formalism and some important results about the classical and quantum entropic regions are introduced. More thorough introductions to the topic can be found in \cite{Yeung2008,Majenz2014}.

\subsection{The entropy vector formalism}

A multi-party quantum state $\rho\in\sos{\bigotimes_{i\in[n]}\hi_i}$ has $2^n-1$ marginals, one for each $\emptyset\neq I\subset [n]$. Each of these marginals has a von Neumann entropy, and these entropies are not independent. In particular, they are subject to entropy inequalities like the strong subadditivity inequality, Equation \eqref{eq:SSA}. The different marginal entropies can be collected in one real vector,
\begin{equation}
	h(\rho)=\left(H(\rho_I)\right)_{I\subset [n]}\in \R^{2^n}.
\end{equation}
Here $\rho_I=\tr_{I^c}\rho$ denotes the marginal of $\rho$ that includes the subsystems in $I\subset [n]$, and we include the empty set for notational convenience, while adopting the convention $H(\rho_{\emptyset})=0$. Note that we use the same letter for the entropy vector of a multi-party quantum state and for the binary entropy function. There should be, however, no danger of confusion, as quantum states are denoted by Greek letters and probabilities by Latin letters. We can now define the set of all quantum entropy vectors,
\begin{equation}
 \Gamma_n=\left\{h(\rho)\Bigg|\rho\in\sos{\bigotimes_{i\in[n]}\hi_i}\text{ for some Hilbert spaces }\hi_i\right\}.
\end{equation}
A subset of $\Gamma_n$ is the set of all classical entropy vectors,
\begin{equation}
 \Sigma_n=\left\{h(\rho)\Bigg|\rho\in\sos{\bigotimes_{i\in[n]}\hi_i}\text{ classical for some Hilbert spaces }\hi_i\right\}.
\end{equation}
A priori these subsets of $\R^{2^n}$ seem to have a complicated structure: they are the images of the non-linear entropy function. It turns out, however, that their closures are convex cones. In the following I will give an overview of the convex geometric properties of these two cones. To this end, we need some notions from convex geometry as they are introduced in, e.g., \cite{Barvinok2002a}.

\begin{defn}[Convex cone]
 Let $V$ be a real vector space. A convex set $C\subset V$ is called a \emph{convex cone}, if it is scale invariant, i.e. for all $v\in V$ and $r\in \R$, $rv\in V$.
\end{defn}

The dual of a convex cone is defined via the inner product on the underlying real vector space.

\begin{defn}[Dual cone]
 Let $C\subset V$ be a convex cone. Its \emph{dual cone} $C^\circ$ is defined as the set of vectors that have nonnegative inner product with all vectors in the \emph{primal cone} $C$,
 \begin{equation}
  C^\circ=\left\{v\in V\Big|(v,w)\ge 0\ \forall w\in C\right\}.
 \end{equation}

\end{defn}
\begin{figure}
	\centering
	\includegraphics[width=.5\textwidth]{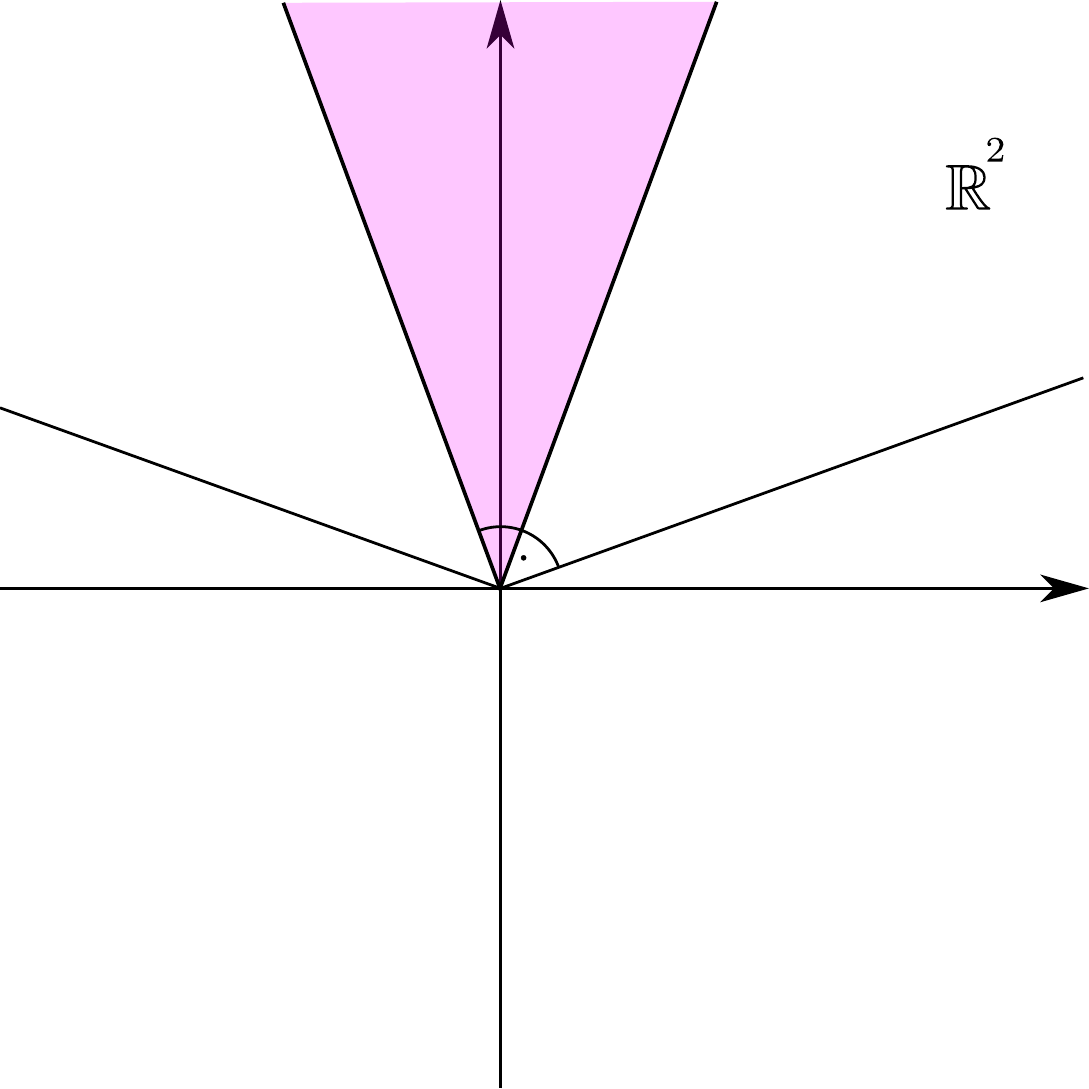}
	\caption{A convex cone and its dual}
\end{figure}
The double dual of a convex cone is equal to its closure.

\begin{lem}\label{lem:dd=closure}
 Let $C\subset V$ be a convex cone. Then $C^\circ$ is closed, and $(C^\circ)^\circ=\overline C$.
\end{lem}

As a consequence, the set of closed convex cones is bijectively mapped into itself by the dualization map, in other words, a closed convex cone is completely specified by its dual.

We are now ready to explore the convex geometry of the entropic sets. Let us first state the fact that $\overline{\Gamma}_n$ and $\overline{\Sigma}_n$ are convex cones

\begin{theorem}[Zhang and Yeung, \cite{Zhang1997}; Pippenger, \cite{Pippenger2003}]\label{thm:entcone}
 The topological closure of the set of classical entropy vectors, $\overline{\Sigma}_n$, is a convex cone, and so is the closure of the set of quantum entropy vectors, $\overline{\Gamma}_n$.
\end{theorem}
$\overline{\Gamma}_n$ and $\overline{\Sigma}_n$ are called the quantum and classical entropy cones, respectively. As closed convex cones, according to Lemma \ref{lem:dd=closure}, they are completely characterized by their dual cones. But what are these dual cones, $\overline\Sigma_n^\circ$ and $\overline\Gamma_n^\circ$? Taking a closer look, we see that they are the sets of valid entropy inequalities. As an example, let $n=2$ and take the subadditivity inequality, $I(1:2)\ge 0$ for all $\rho\in\sos{\hi_1\otimes \hi_2}$. We can express this inequality as an inner product, $(f,h(\rho))\ge 0$, with $f\in\R^{4}$, $f_{\{1\}}=f_{\{2\}}=1$ and $f_{\{1,2\}}=-1$.

For a general $n$-party quantum state $\rho$, note that by definition $h(\rho)\in\Gamma_n\subset \overline\Gamma_n$. Therefore, for any vector $f\in\Gamma_n^{\circ}$,
\begin{align*}
 0\le& (f,h(\rho))\\
 =&\sum_{I\subset [n]}f_IH(\rho_I),
\end{align*}
i.e. $f$ defines a linear entropy inequality.

Theorem \ref{thm:entcone} shows that the closures of the entropic sets $\Sigma_n$ and $\Gamma_n$ have a simple geometric structure. But nevertheless, these sets themselves could be quite complicated. It turns out, however, that the only differences between the entropic sets and their closures, i.e. the entropy cones, are situated on the boundary of the latter.
\begin{thm}[Mat\' u\v s, Theorem 1 in \cite{Matus2007b}]\label{thm:ri}
 The relative interior of the entropy cones $\overline\Sigma_n$ and $\overline\Gamma_n$ is contained in the coresponding entropic set, i.e.
 \begin{align*}
  \mathrm{ri}(\overline\Sigma_n)&\subset\Sigma_n\text{ and }\\
  \mathrm{ri}(\overline\Gamma_n)&\subset\Gamma_n.
 \end{align*}

\end{thm}
Note that Theorem 1 in \cite{Matus2007b} is stated for $\Sigma_n$ only, but the proof given there works for $\Gamma_n$ as well.

\subsection{The Shannon and von Neumann cones}

In Chapter 1, we already came across some quantum entropy inequalities, the most improtant of which is the strong subadditivity of the von Neumann entropy \cite{Lieb1973}, 
\begin{equation}
 I(A:B|C)=H(AC)+H(BC)-H(ABC)-H(C)\ge 0
\end{equation}
for all tripartite quantum states $\rho_{ABC}$. For an $n$ party state $\rho$, there is therefore a (strong) subadditivity inequality for any pair of nonempty subsets $\emptyset\neq I,J\subset [n]$,
\begin{align}\label{eq:allSSAs}
 (\Delta[I,J], h(\rho)):=&I(I\setminus J: J\setminus I|I\cap J)\nonumber\\
 =&H(I)+H(J)-H(I\cup J)-H(I\cap J)\ge 0.
\end{align}
Another quantum entropy inequality is obtained when considering strong subadditivity for a four party pure state $\rho_{ABCD}$. Using the fact that the entropies of complementary marginals of a pure state are equal, i.e. $H(AB)=H(CD)$, $H(ACD)=H(B)$ etc., we obtain
\begin{align}\label{eq:WM}
 0\le&I(B:D|C)\nonumber\\
 =& H(BC)+H(CD)-H(BCD)-H(C)\nonumber\\
 =&H(BC)+H(AB)-H(A)-H(C)\nonumber\\
 =&H(B|A)+ H(B|C).
\end{align}
Because every mixed quantum state has a purification, this inequality, which is called \emph{weak monotonicity}, holds for arbitrary quantum states. For an $n$ party quantum state $\rho$, there hence is a weak monotonicity inequality for any pair of subsets $I,J\subset[n]$ with nonempty intersection,
\begin{align}\label{eq:allWMs}
 (E[I,J],h(\rho)):=&H(I\cap J|I\setminus J)+H(I\cap J|J\setminus I)\nonumber\\
 =&H(I)+H(J)-H(I\setminus J)-H( J\setminus I)\ge 0.
\end{align}
The transformation that yields weak monotonicity from strong subadditivity is, in fact, part of the $S_{n+1}$-symmetry of $\Gamma_n$ that was described in \cite{Majenz2014}.

In the following we allow us at times to write $I(A:B|C)_{h(\rho)}:=I(A:B|C)_\rho$ etc. and extend this notation to arbitrary vectors outside the quantum entropy region.

The entropy inequalities \eqref{eq:allSSAs} and \eqref{eq:allWMs} are elements of $\overline\Gamma_n^\circ$, and hence their conic hull is a subcone of $\overline\Gamma_n^\circ$. The closed convex cone $\Xi_n\supset \Gamma_n$ such that $\Xi_n^\circ$ is the conic hull of the vectors $\Delta[I,J]$ and $E[I,J]$ for all suitable $I,J\subset [n]$, is called the \emph{von Neumann cone}. Elements of $\Xi_n^\circ$ are called von Neumann type inequalities.

In the classical case, the monotonicity inequality, $H(A|B)_\sigma\ge 0$ for a classical state $\sigma_{AB}$, implies and therefore replaces the weak monotonicity inequality \eqref{eq:WM}. For an $n$ party classical state $\sigma$, for each pair of sets $J\subset I\subset [n]$, there is a monotonicity inequality

\begin{align}\label{eq:allMs}
 (m[I,J],h(\sigma)):=&H(I| J)\ge 0.
\end{align}

The closed convex cone $\Theta_n\supset \Sigma_n$ such that $\Theta_n^\circ$ is the conic hull of the vectors $\Delta[I,J]$ and $m[I,J]$ for all suitable $I,J\subset [n]$, is called the \emph{Shannon cone},  elements of $\Theta_n^\circ$ are called Shannon type inequalities.

\subsection{Non-Shannon type and Non-von-Neumann type inequalities}

While the Shannon type inequalities were proven by Claude Shannon in his seminal work \cite{Shannon1948} in 1948, and the strong subadditivity of the von Neumann entropy was proven by Elliot Lieb and Mary Beth Ruskai in 1973 \cite{Lieb1973}, it was not until 1997 that Zhen Zhang and Raymond W. Yeung discovered the first non-Shannon type inequality \cite{Zhang1997}. This inequality, however, does not hold unconditionally, but only under some additional assumptions called ``constraints''. In other words, this inequality only shows that some part of the \emph{boundary} of $\Theta_n$ is not part of $\Sigma_n$. A year later, the two authors discovered an unconstrained non-shannon type inequality as well \cite{Zhang1998}, thereby showing that $\overline\Sigma_n\subsetneq\Theta_n$.

While constrained non-von-Neumann type inequalities have been found by Noah Linden and Andreas Winter \cite{Linden2005} and Cadney et al. \cite{Cadney2012}, no unconstrained non-von-Neumann type inequality is known to date, and hence the inclusion $\overline \Gamma_n\subset\Xi_n$ is not known to be strict. In the following section we will formally define constrained and unconstrained non-von-Neumann type inequalities. The constrained inequality for the von Neumann entropy by Linden and Winter \cite{Linden2005} will be explained in detail, as well as the techniques used to prove it.

\subsubsection{Constrained and unconstrained inequalities.} The Shannon and von Neumann type inequalities that we have seen in the last section are all unconstrained, i.e. they hold for all entropy vectors. As mentioned above, the first non-shannon type inequality, as well as the only known non-von-Neumann type inequalities are \emph{constrained} in the sense that they only hold for entropy vectors satisfying some aditional equations. More formally, a constrained non-von-Neumann type inequality for $n$ party quantum states is a tuple of $k+1$ vectors $f,g_1,...,g_k\in\R^{2^n}$ such that the following holds. For all vectors $v\in\Gamma_n$ with $(g_i,v)=0$ for $i=1,...,k$, the inequality $(f,v)\ge 0$ holds. Moreover, $f\not\in\Xi^\circ+\R \{g_1,...,g_k\}$. As such an inequality can only possibly cut out a lower-dimensional slice of $\overline\Gamma_n$, its existence is not enough to conclude that $\overline\Gamma_n\subsetneq \Xi_n$

% It is easy to see that the convexity of $\overline\Gamma_n$, together with Theorem \ref{thm:ri}, implies that for any non-trivial constrained inequality, for each constraint $g_i$ one of the following holds:
% \begin{enumerate}[i)]
%  \item The inequality still holds when $g_i$ is removed.
%  \item Any $v\in\overline \Gamma_n$ with $(g,v)=0$ lies on the boundary of $\overline\Gamma_n$.
% \end{enumerate}

An unconstrained non-von-Neumann inequality for $n$ party quantum states, on the other hand, is a vector $f\in\R^{2^n}$ such that $f\in\overline\Gamma_n^\circ\setminus\Xi_n^\circ$, implying $\overline\Gamma_n^\circ\supsetneq\Xi_n^\circ$. Dualizing, this yields $\overline\Gamma_n\subsetneq\Xi_n$, by Lemma \ref{lem:dd=closure}.

\subsubsection{The Linden-Winter inequality and exact quantum Markov chains.}

In \cite{Linden2005} Linden and Winter prove the following constrained inequality for the von Neumann entropy.
\begin{thm}\label{thm:LiWi}
 Let $\rho_{ABCD}$ be a quadripartite quantum state such that $I(C:A|B)=I(C:B|A)=I(A:B|D)=0$. Then
 \begin{equation}\label{eq:liwi}
  I(C:D)\ge I(C:AB),
 \end{equation}
and this inequality is independent from the von Neumann type inequalities.
\end{thm}
The proof relies heavily on a characterization result for \emph{exact quantum Markov chains}, i.e. tripartite quantum states $\rho_{ABC}$ such that $I(A:C|B)=0$. In analogy to the notation for classical Markov chains we say $A-B-C$ is an exact quantum Markov chain in $\rho$, if we want to emphasize which conditional quantum mutual information vanishes.
\begin{thm}[Hayden et al. \cite{Hayden2004}]\label{thm:SSA-sat}
 Let $\rho_{ABC}$ be a tripartite quantum state that is an exact quantum markov chain, i.e. $I(A:C|B)=0$. Then there is an orthogonal direct sum decomposition
 \begin{equation}
  \mathcal{H}_B=\bigoplus_{i=0}^{k-1}  \mathcal{H}_{B_a^{(i)}}\otimes  \mathcal{H}_{B_c^{(i)}}
 \end{equation}
such that
\begin{equation}
 \rho=\bigoplus_{i=0}^{k-1}p_i \rho_{AB_a^{(i)}}\otimes \rho_{B_c^{(i)}C}.
\end{equation}
for some probability distribution $p=\{p_i\}_{i=0}^{k-1}$ and quantum states $\rho_{AB_a^{(i)}}$ and $\rho_{B_c^{(i)}C}$.
\end{thm}

It is instructive to put this decomposition into a form more intuitive to the (quantum) information theorist. To this end, let $B_\gamma$ be a quantum system with $|B_\gamma|\ge \dim\hi_{B_\gamma}^{(i)}$, and $V_\gamma^{(i)}\in\Hom{\hi_{B_\gamma}^{(i)}}{\hi_{B_\gamma}}$ be an isometry, for all $i=0,...,k-1$ and $\gamma=a,c$. Furthermore, let $K_B$ be a classical register with $|K_B|=k$. Now we can define an isometry $V_{B\to K_BB_aB_c}$ by setting $V_{B\to K_BB_aB_c}\ket\psi=\ket i_{K_B}\otimes\left[\left(V_a^{(i)}\otimes V_c^{(i)}\right)\ket\psi\right]$ for $\ket\psi\in\mathcal{H}_{B_a^{(i)}}\otimes  \mathcal{H}_{B_c^{(i)}}\subset\hi_B$. Applying this local isometry to the $B$ system of $\rho_{ABC}$ as in the theorem yields
\begin{equation}
 \hat\rho_{K_BAB_aB_cC}=V_{B\to K_BB_aB_c}\rho_{ABC}\left(V_{B\to K_BB_aB_c}\right)^\dagger=\sum_{i=0}^{k-1}p_i\proj i_{K_B}\otimes\hat\rho^{(i)}_{AB_a}\otimes \hat\rho^{(i)}_{B_cC},
\end{equation}
where $\hat\rho^{(i)}_{AB_a}=V_a^{(i)}\rho_{AB_a^{(i)}}\left(V_a^{(i)}\right)^\dagger$, and $\hat\rho^{(i)}_{B_cC}$ is defined analogously. 

Another way to look at the characterization Theorem \ref{thm:SSA-sat} is, that the vanishing quantum conditional mutual information $I(A:C|B)$ implies that there exists a quantum channel $\mathcal R_{B\to BC}$ such that $\rho_{ABC}= R_{B\to BC}(\rho_{AB})$, i.e. the system $C$ can be recovered from $B$ alone when lost.

An important step in the proof of Theorem \ref{thm:LiWi} as it appears in \cite{Linden2005} is to observe what Theorem \ref{thm:SSA-sat} implies for a state $\rho_{ABC}$ where both $A-B-C$ and $B-A-C$ are exact quantum Markov chains in $\rho$. In this case, using the isometry discussed above, $\rho$ can be assumed to have the form

\begin{equation}\label{eq:double-ssasat}
 \rho_{K_AK_BA_bA_cB_aB_cC}=\sum_{i=0}^{k-1}\sum_{j=0}^{l-1}p_{ij}\proj{i,j}_{K_AK_B}\otimes\rho^{(i,j)}_{A_bB_a}\otimes\rho^{(i)}_{A_c}\otimes\rho^{(j)}_{B_c}\otimes\rho_{C}^{(i,j)}.
\end{equation}
The fact that $C$ is product with $A_c$ and $B_c$ given $i,j$ follows from the fact that $C$ can be recovered from either $A$ or $B$. The same argument shows that $\rho_C^{(i,j)}$ can only depend on $i$ because of the Markov condition $B-A-C$, and at the same time it can only depend on $j$ because of the Markov condition $A-B-C$. Hence there exists an extension of $\rho$ to a register $\hat K$ such that 
$$\rho_{K_AK_B\hat K  C}=\sum_{i=0}^{k-1}\sum_{j=0}^{l-1}p_{ij}\proj{k(i,j)}_{\hat K}\proj{i,j}_{K_AK_B}\otimes\rho_C^{k(i,j)},$$
and $k(i,j)=f(i)=g(j)$ for some functions $f$ and $g$. In other words, the register $\hat K$ can be created locally from $A$ as well as from $B$. Let us denote the two copies of $\hat K$ created from $A$ and $B$ by $\hat K_A$ and $\hat K_B$, respectively.

\subsubsection{Other constrained inequalities for $\Gamma_4$.}

In \cite{Cadney2012}, the authors prove several additional families of constrained inequalities. We record the ones that apply to the four party case here for comparison with the new inequality that will be proven in the next section.
\begin{thm}[Cadney, Linden and Winter; three and four party versions of Theorems 1', 5, 5' \cite{Cadney2012}]\label{thm:Cadney}
 Let $\rho_{ABCD}$ be a four party quantum state with $I(A:C|B)=I(B:C|A)=0$. Then the inequalities
 \begin{align}
  I(A:B|D)+I(A:B|CD)+H(D|ABC)-I(AB:C)&\ge0\\
  I(A:B|D)+I(A:B|C)+I(C:D)-I(C:AB)&\ge 0\\
  I(A:B|D)+I(A:B|C)+I(A:B|CD)&\nonumber\\
  +H(D)  +H(C)+H(CD|AB)-2I(AB:C)&\ge 0\\
  2I(A:B|C)+H(C)+H(C|AB)-I(AB|C)&\ge0
 \end{align}
 hold.
\end{thm}

Cadney et al. also exhibit the vector $v\in\R^{16}$ with components
\begin{table}[h!]
\centering
\begin{tabular}{c||c|c|c|c|c|c|c|c|c|c|c|c}
 $I$&$\emptyset$&$A$&$B$&$C$&$D$&$AB$&$AC$&$AD$&$BC$&$BD$&$CD$&$ABC$\\
 \hline\hline
 $v_I$&0&5&5&2&4&6&5&5&5&5&6&6
\end{tabular}\\
\vspace{.6cm}
\begin{tabular}{c||c|c|c|c}
 $I$&$ABD$&$ACD$&$BCD$&$ABCD$\\
 \hline\hline
 $v_I$&6&5&5&4
\end{tabular}.

\caption[Table \ref{tab:Cadney-vector}]{}
\label{tab:Cadney-vector}
\end{table}

This vector respects all von Neumann type inequalities. Furthermore, it fulfils the constraints of Theorem \ref{thm:LiWi} and therefore also the constraints of Theorem \ref{thm:Cadney}. It also respects the inequalities from the latter but it violates the inequality from the former Theorem. This shows that inequality \eqref{eq:liwi} is not implied by the inequalities in Theorem \ref{thm:Cadney}.

\section{Generalized Linden Winter inequality}

In \cite{Linden2005}, Linden and Winter make the following conjecture.

\begin{conj}\label{conj:liwi}
 There exist constants $\kappa_1, \kappa_2$ and $\kappa_3$ such that for all four party quantum states $\rho_{ABCD}$,
 \begin{equation}
  \kappa_1 I(A:C|B)+\kappa_2 I(B:C|A)+\kappa_3 I(A:B|D)+I(C:D)-I(C:AB)\ge 0.
 \end{equation}
\end{conj}
Such an inequality could, were it true, be seen as a ``robust version'' of Theorem \ref{thm:LiWi} in the sense that when the constraints are fulfilled approximately, then the inequality \eqref{eq:liwi} would only be violated by a little bit. Note that the conjectured inequality is unconstrained, i.e. it would show that $\overline{\Gamma}_n\subsetneq\Xi_n$.

The following theorem is a step towards proving Conjecture \ref{conj:liwi} in that it treats the case of $\rho_{ABCD}$ with $A-B-C$ and $B-A-C$ exact quantum Markov chains.

\begin{thm}
 Let $\rho_{ABCD}$ be a four party quantum state with $I(A:C|B)=I(B:C|A)=0$. Then the inequality
 \begin{equation}\label{eq:genLiWi}
  I(C:AB)\le I(C:D)+I(A:B|D)
 \end{equation}
 holds, and is independent of the Shannon type inequalities as well as the inequalities from Theorem \ref{thm:LiWi} and Theorem \ref{thm:Cadney}.
\end{thm}

\begin{proof}
 Let us first prove that the inequality holds. The constraints imply that the state can be extended to the perfectly correlated registers $\hat K_A$ and $\hat K_B$ that are locally created from $A$ and $B$, respectively, as described below Equation \eqref{eq:double-ssasat}. Let $M^A_{A\to A\hat K_A}$ and $M^B_{B\to B\hat K_B}$ be the maps that create the registers   $\hat K_A$ and $\hat K_B$ from $A$ and $B$, respectively, without disturbing the state $\rho_{ABC}$. Define the state 
 \begin{equation}
 	\sigma_{ABCD\hat K_A\hat K_B}=M^A\otimes M^B(\rho_{ABCD}).
 \end{equation}
 Note that $\rho_{ABC}=\sigma_{ABC}$ as discussed above, but also $\rho_{CD}=\sigma_{CD}$, as $M^A\otimes M^B$ only acts on the systems $AB$.
 
 For the following argument, we reintroduce the notation of subscripts for entropic quantities that indicate the state they are evaluated on. We begin by observing that
 \begin{eqnarray}\label{eq:liwiproof1}
	I(C:AB)_\rho&=&I(C:AB)_\sigma\nonumber\\
	&=&I(C:ABK_A)_\sigma\nonumber\\
	&=& I(C:\hat K_A)_\sigma+I(C:AB|\hat K_A)_\sigma\nonumber\\
	&=& I(C:\hat K_A)_\sigma.
\end{eqnarray}
The first equality holds because $\rho_{ABC}=\sigma_{ABC}$. The second equality is due to the fact that the register $\hat K_A$ can be created and discarded without disturbing $\rho_{ABC}=\sigma_{ABC}$. The third equality is the chain rule \eqref{eq:chainrule}, and the last equality follows from the fact that $C$ can be recovered from the register $\hat K_A$ with respect to $AB$, i.e. $AB-\hat K_A-C$ is an exact quantum Markov chain in $\sigma$. We go on by bounding
\begin{eqnarray}\label{eq:liwiproof2}
	I(C:\hat K_A)_\sigma&\le& I(C:\hat K_AD)_\sigma\nonumber\\
	&=&I(C:D)_\sigma+I(C:\hat K_A|D)_\sigma.
\end{eqnarray}
The inequality is the alternative form of strong subadditivity \eqref{eq:ssa-alt}, and the equality is another application of the chain rule \eqref{eq:chainrule}. The last term can be bounded further,
\begin{eqnarray}\label{eq:liwiproof3}	
	I(C:\hat K_A|D)_\sigma&=&H(\hat K_A|D)_\sigma-H(\hat K_A|CD)_\sigma\nonumber\\
	 &\le&H(\hat K_A|D)_\sigma,
\end{eqnarray}
where the inequality is due to the fact that the conditional entropy of classical systems is nonnegative, even when conditioning on a quantum system. For the last step, notice that $\hat K_A$ and $\hat K_B$ are perfectly correlated, i.e. $H(\hat K_A|S)=I(\hat K_A:\hat K_B|S)$ with respect to any conditioning system. Therefore we get
\begin{eqnarray}\label{eq:liwiproof4}		 
	H(\hat K_A|D)_\sigma&=&I(\hat K_A:\hat K_B|D)_\sigma\nonumber\\
	&\le&I(A:B|D)_\rho
\end{eqnarray}
by a final application of the data processing inequality \eqref{eq:dataproc-mut}. Putting together Equations \eqref{eq:liwiproof1}, \eqref{eq:liwiproof2}, \eqref{eq:liwiproof3}, and \eqref{eq:liwiproof4} shows that the inequality \eqref{eq:genLiWi} holds.
%The register $C$ can be perfectly recreated from $\hat K_A$, implying $I(C:AB|K_A)=0$. Together with the chain rule \eqref{eq:chainrule}, this yields the first equation. The first inequality is Equation \eqref{eq:ssa-alt}, and third inequality is the data processing inequality \eqref{eq:dataproc-mut}. For the second inequality, first observe that $I(C:\hat K_A|D)=H(\hat K_A|D)-H(\hat K_A|CD)$. But the conditional entropy $H(\hat K_A|CD)$ of a classical register with respect to a quantum register is nonnegative, as a classical-quantum state is, in particular, separable. In the second to last line we used that $\hat K_A=\hat K_B=\hat K$.

Let us now prove that the inequality is independent. Let $v\in\R^{16}$ be the vector in Table \ref{tab:Cadney-vector} and $v'=h(\proj\psi_C\otimes \sigma_{ABD})$ for a pure state $\proj\psi_C$ and some state $\sigma$ with $I(A:B|D)_\sigma\neq 0$. For $\varepsilon>0$, $w=v'+\varepsilon v''$ fulfills the constraints $I(A:C|B)_W=I(B:C|A)_W=0$. It also fulfils all inequalities from Theorem \ref{thm:Cadney} and all von Neumann type inequalities. This is because $v$ does according to \cite{Cadney2012}, and $v'$ is an actual quantum entropy vector fulfilling the constraints of Theorem \ref{thm:Cadney}. It also trivially fulfills the constrained inequality from Theorem \ref{thm:LiWi}: it is not applicable, as $w$ does not have $I(A:B|D)_w=0$. On the other hand, it violates the inequality \ref{eq:genLiWi} for $\varepsilon$ small enough, showing that the latter is not implied by the other known inequalities.
\end{proof}

\fbox{\begin{minipage}{\textwidth}\vspace{.3cm}
		\begin{center}\large{\textbf{Summary of Chapter \ref{chap:entropy}}}\end{center}
		\vspace{.3cm}
		\begin{itemize}
			\item A new constrained quantum entropy inequality, Equation \eqref{eq:genLiWi} was proven.
			\item It is independent from previously known non-von-Neumann type inequalities, as well as from the von Neumann type inequalities.
			\item The new inequality is a step towards proving the unconstrained quantum entropy inequality that was conjectured in \cite{Linden2005}.
		\end{itemize}
\end{minipage}}

 % Entropy inequalities

\hchapter{One-shot quantum Shannon theory}\label{chap:one-shot}
Quantum Shannon theory, the theory of quantum information transmission and quantum data compression, is arguably the corner stone of quantum information theory. Starting from the work of Alexander Holevo in the 70's \cite{Holevo1973,Holevo1979}, quantum Shannon theory in the setting of many independent, identically distributed copies of a resource (IID) has been developed to an impressive extent. In recent years, there has been a rising interest in the so called one-shot setting, where the theoretical challenge is to characterize the resource requirements to carry out a certain task only once. This is particularly interesting in the quantum case, as the first quantum information processing devices capable of performing some of the protocols proposed by quantum Shannon theory are, and will be, small scale and therefore not be able to approach the asymptotic regime.

 In this chapter we will be concerned with two uniquely quantum topics in information theory. The first one is the decoupling technique \cite{Horodecki2007}. Due to the existence and equivalence of purifications and the Stinespring dilation (a kind of purification for maps), this technique can be used to achieve optimal coding results by decoupling a quantum system from a reference system it contains information about. In other words, the transmission of the correlated parts of a quantum system is, counterintuitively, achieved by finding a map that \emph{destroys} these correlations. Here, catalytic decoupling, a generalization of this technique, is presented. This generalization is necessary to make the decoupling technique fit for one-shot quantum Shannon theory, as the standard decoupling technique fails to give optimal protocols in this setting \cite{Berta2016c,Datta2011}. Catalytic decoupling on the other hand achieves optimal one-shot results in \emph{state merging}, i.e. quantum source coding with side information.
 
 The second topic is quantum Teleportation \cite{Bennett1993}. In this protocol, quantum information can be transmitted from one party to another using only a classical communication channel and an entangled resource state that the two players share. Together with superdense coding \cite{Bennett1992}, it shows the interconvertibility of quantum and classical communication, if one has a sufficient number of ebits at hand. A variant of quantum teleportation is the task of port based teleportation. While regular teleportation requires the receiver to perform a unitary correction operation, in port based teleportation he only has to select one of several output ports, a classical operation. This task, however, has steep resource requirements. The known protocols use an amount of entanglement exponential in the size of the system that is teleported. This fact puts port based teleportation in the one-shot corner of quantum Shannon theory as well: this task becomes \emph{harder} when bundling many IID quantum systems into one larger system compared to teleporting them one by one. Here we explore different approaches to bound the resource requirements for port based teleportation from below, in the most general setting. As a byproduct of efforts towards this goal, a lower bound on the size of the program register for approximate universal programmable quantum processors is derived, as well.

\section{Decoupling}

In this section, the decoupling technique and its applications are described, as well as the problems that arise when trying to apply it to one-shot quantum source coding problems. Subsequently, the known results that overcome these difficulties are reviewed. In Subsection \ref{subs:catalytic}, these techniques are unified and simplified into a new generalized decoupling paradigm: catalytic decoupling.

\subsection{Standard decoupling and one-shot coherent quantum state merging}
\subsubsection{Coherent quantum state merging}
To understand how the decoupling technique is used in quantum information theory, let us have a look at one of first problems where decoupling has been applied, \emph{coherent quantum state merging}, or fully quantum Slepian Wolf (FQSW) coding \cite{Abeyesinghe2009}. This task can be described as follows. One player, Alice, wants to send a quantum system $A$  to another player, Bob, who already has a quantum system $B$ that is possibly correlated with $A$. Let their joint quantum state be $\rho_{AB}$. To do this, the two have a perfect qubit quantum channel, that they want to use as few times as possible to complete the task. To make sure that they do not cheat in that Bob just creates a fresh copy of $\rho_{AB}$ locally and claims Alice had sent it, and because we cannot compare Alice's initial and Bob's final copy of the $A$ system directly, let $\rho_{ABR}$ be a purification of $\rho_{AB}$ and imagine $R$ in the hands of a Referee. Now we require that the final state of a FQSW protocol is close to the initial state, including the reference system $R$, except that Bob has both systems $A$ and $B$.

A trivial protocol is for Alice to send all of $A$, using the qubit channel $\log|A|$ times. But this can be highly suboptimal. If, e.g., $\rho_{AB}$ is pure, the protocol where Alice and Bob discard their systems and Bob prepares a fresh copy of $\rho_{AB}$ works, making no use of the quantum channel at all. This is analogous to the classical Slepian Wolf coding scenario when the two players have fully correlated random variables. The question is therefore: How much quantum communication is necessary to complete the described task for a given state $\rho_{AB}$, allowing for a small error $\varepsilon>0$?

It is not necessary to specify at this point whether we are looking at the IID or one-shot scenario, The following reasoning applies to both.

It turns out we can construct an optimal protocol using \emph{decoupling}. A quantum channel $\Lambda_{A\to A'}$ is said to $\varepsilon$-decouple $A$ from $R$ in $\rho_{AR}$, if $P\left(\Lambda_{A\to A'}(\rho_{AR}), \sigma_{A'}\otimes\sigma_{R}\right)\le\varepsilon$ for some quantum state $\sigma$. In most applications of decoupling, $\Lambda$ is a unitary followed by either a partial trace over a subsystem of $A$, or a projective measurement. For the purpose of constructing a FQSW protocol, we will use the former. 

Let $\rho_{ABR}$ be a pure state and Alice wants to send $A$ to Bob who holds $B$. Looking at the marginal $\rho_{AR}$, let $\hi_A\cong\hi_{A_1}\otimes \hi_{A_2}$ be a decomposition of $A$ into two subsystems and let $U_A$ be a unitary such that
\begin{equation}
	P\left(\hat\rho_{A_1R},\sigma_{A_1}\otimes\sigma_{R}\right)\le\varepsilon,
\end{equation}
where $\hat\rho_{A_1A_2R}=U_A\rho_{AR}U_A^{\dagger}$ and therefore by the triangle inequality
\begin{equation}\label{eq:decoupling-cond}
P\left(\hat\rho_{A_1R},\sigma_{A_1}\otimes\rho_{R}\right)\le 2\varepsilon.
\end{equation}
Now consider the following protocol. Alice applies the unitary $U_A$. Instead of applying the partial trace over $A_2$, she sends it to Bob. Now Alice's and the referee's joint state is $\rho_{A_1R}$. By Equation \eqref{eq:decoupling-cond} and Uhlmann's theorem, there exists a purification $\sigma_{A_1A_2BR}$ of $\sigma_{A_1}\otimes \rho_R$ such that 
  \begin{equation}
  P\left(\hat\rho_{A_1A_2BR},\sigma_{A_1A_2BR}\right)\le 2\varepsilon.
  \end{equation}
  But all purifications are isometrically equivalent, therefore there exists an isometry $V_{A_2B\to AB \overline A_1}$ such that
  \begin{equation}
  P\left(\doublehat{\rho}_{ABRA_1\overline A_1},\eta_{A_1\overline A_1}\otimes \rho_{ABR}\right)\le 2\varepsilon,
  \end{equation}
  where $\doublehat{\rho}_{ABRA_1\overline A_1}=V_{A_2B\to AB \overline A_1}\hat\rho_{A_1A_2BR}\left(V_{A_2B\to AB \overline A_1}\right)^\dagger $, and $\eta_{A_1\overline A_1}$ is a purification of $\sigma_{A_1}$. Therefore, to finish the protocol, Bob applies $V$, and the final state of the protocol, $\doublehat{\rho}$, is $2\varepsilon$-close to the target state (plus a state  $\eta_{A_1\overline A_1}$ that Alice and Bob now share).
  
  The communication cost of this protocol is $\log|A_2|$, as Alice sends the system $A_2$ to Bob. We call $A_2$ the remainder system, and define a quantity $R^\eps(A:R)_\rho$ which is the logarithm of the minimal $|A_2|$ for $\varepsilon$-decoupling $A$ from $R$ in $\rho$. More formally, we make the following definition.
  \begin{defn}[Minimal remainder system sizes for standard decoupling]
  	Let $\rho_{AE}\in\St(\hi_{A}\otimes\hi_E)$ be a bipartite quantum state and $1\ge \varepsilon\ge 0$ an error parameter. The minimal remainder system size $R^\varepsilon(A:E)_\rho$ for decoupling $A$ from $E$ in $\rho_{AE}$ up to an error $\varepsilon$ is defined as the minimal number $r$ such that there exists a unitary $U_A\in\mathrm{U}(\hi_A)$ and a decomposition $\hi_A\cong \hi_{A_1}\otimes \hi_{A_2}$ with $\log|A_2|=r$, as well as states $\omega_{A_1}\in\St(\hi_{A_1})$ and $\tilde\omega_E\in\St(\hi_E)$ with the following property:
  	\begin{equation*}
  	P\left(\tr_{A_2}\left(U_A\otimes\mathds 1_E\right)\rho_{AE}\left(U_A^\dagger \otimes\mathds 1_E\right), \omega_{A_1}\otimes\tilde\omega_E\right)\le \varepsilon.
  	\end{equation*}
  \end{defn}	
   How small we can choose the remainder system, while still achieving decoupling? This question is answered by decoupling theorems. 
  
  \subsubsection{A standard decoupling theorem}
  
  For the one-shot scenario, Dupuis, Berta, Wullschleger and Renner have proven a decoupling theorem for general CPTP maps, for which the partial trace special case of the partial trace can be stated as follows:
  
  \begin{thm}[One-shot decoupling by partial trace, \cite{Dupuis2014,Berta2011,Majenz2017}]\label{thm:one-shot-decoup}
  	For a bipartite quantum state $\rho_{AB}$, let $A=A_1A_2$ be a decomposition such that
  	\begin{align}\label{eq:achiev_Hmin}
  	R^{\eps}(A:E)_\rho\lesssim \frac{1}{2}\Big(H^{\eps'}_{\max}(A)_\rho-H^{\eps'}_{\min}(A|E)_\rho\Big)\,\mathrm{with}\,\eps'=\frac{\eps}{5}.
  	\end{align}
  	Then there exists a unitary $U_A$ such that
  	\begin{equation}
  		P\left(\hat\rho_{A_1R},\tau_{A_1}\otimes\rho_{R}\right)\le \varepsilon.
  	\end{equation}
  \end{thm}

The main proof technique that is used in decoupling theorems like the above is the probabilistic method. Instead of exhibiting a unitary that achieves the desired task, it is shown that a random protocol achieves the task \emph{on average}. In the case of decoupling, it is shown that applying a random unitary and subsequently taking the partial trace over $A_2$, yields a state which has distance at most $\varepsilon$ from the desired product state \emph{on average}.

An important caveat about the above decoupling theorem is, that it achieves more than what we asked for. With the application to FQSW in mind, we only ask for the final state of the decoupling protocol being close to any product state. The above theorem, however, also randomizes the remaining part of $A$. As a consequence, it requires $A_2$ to be large even for states $\rho_{AR}=\rho_A\otimes \rho_R$ that are already decoupled, if $\rho_A$ is far from $\tau_A$. This means that the theorem yields a protocol for FQSW that can be highly suboptimal depending on the input state. As an example, suppose $\rho_{ABR}=\proj 0_A\otimes\proj 0_B\otimes \proj 0_R$. There are no correlations present in this state whatsoever, Bob can just prepare a $\proj 0$-state and Alice discards hers. The protocol described above together with Theorem \ref{thm:one-shot-decoup}, however, will use $\frac{\log|A|}{2}$ qubits of communication.

\subsubsection{Optimal one-shot coherent quantum state merging}

There are two works that overcome this problem. In \cite{Berta2011}, Berta, Christandl and Renner get around this problem by using \emph{embezzling states}.

\begin{defn}[Embezzling state \cite{van2003universal}]
	A state $\ket\mu\in\mathcal{H}_{A}\otimes\hi_B$ is called a \emph{universal $(d,\delta)$-embezzling state} if for any state $\ket{\psi}\in\hi_{A'}\otimes\hi_{B'}$ with $\dim \hi_A=\dim\hi_B\le d$ there exist isometries $V_{\psi,X}:\hi_{X}\to\hi_X\otimes\hi_{X'}$, $X=A,B$ such that\footnote{The original concept was defined using the trace distance instead of the purified distance \cite{van2003universal}. We use the purified distance here as it fits our task, the definitions are equivalent up to a square according to Equation \eqref{eq:tr2fid}.}
	\begin{align*}
	P\left(V_{\psi,A}\otimes V_{\psi,B}\ket\mu,\ket\mu\otimes\ket\psi\right)\le\delta.
	\end{align*}
\end{defn}
Such embezzling states exist for all parameters.
\begin{prop}\cite{van2003universal}
	Universal $(d,\delta)$-embezzling states exist for all $d$ and $\varepsilon$.
\end{prop}
The quantum communication cost of coherent state merging is expressed in terms of a R\'enyi version of the mutual information, the \emph{max-mutual information},
\begin{align}
	I_{\max}^{\eps}(E:A)_{\rho}&:=\min_{\bar{\rho}}I_{\max}(E:A)_{\bar{\rho}}\quad\mathrm{with}\label{eq:Imax1}\\
	I_{\max}(E:A)_{\bar{\rho}}&:=\min_{\sigma_A} D_{\max}(\bar\rho_{AE}\|\sigma_A\otimes\bar{\rho}_E),\label{eq:Imax2}.
\end{align}
Here, the minimum in the first equation is taken over subnormalized states that are $\varepsilon$-close to $\rho$, and the minimum in the second equation is taken over all quantum states $\sigma_A$.

More precisely, Berta et al. prove the following theorem
\begin{thm}[Berta, Christandl and Renner, \cite{Berta2011}]\label{thm:merge-Berta}
	Let $\rho_{ABR}$ be a quantum state, where Alice has $A$, Bob has $B$ and $R$ is in the hands of a referee. Using a universal $(|A|,\delta)$-embezzling state, coherent state merging with error $\varepsilon+\varepsilon'+\delta|A|+|A|^{-1/2}$ can be achieved with communication cost
	\begin{equation}
		c=\frac 1 2 I_{\max}^{\varepsilon'}(A:R)_\rho+2\log\frac 1 \varepsilon+\log\log|A|.
	\end{equation}
	Also, any protocol with error $\varepsilon$ and communication cost 
	\begin{equation}
		c<\frac{1}{2} I_{\max}^{\varepsilon}(A:R)_\rho
	\end{equation}
	must fail.
\end{thm}

Another work that overcomes the problem of suboptimality of Theorem \ref{thm:one-shot-decoup} in the one-shot setting is the work by Anshu, Devabathini and Jain \cite{Anshu2014}. The main topic of the paper being quantum state redistribution, they tighten the above achievability result for state merging in the following way as a side result.

\begin{thm}[Anshu, Devabathini and Jain, \cite{Anshu2014}]\label{thm:merge-Anshu}
	Let $\rho_{ABR}$ be a quantum state, where Alice has $A$, Bob has $B$ and $R$ is in the hands of a referee. If arbitrary entangled resource states between Alice and Bob are available, coherent quantum state merging with error $\varepsilon$ can be achieved with communication cost
	\begin{equation}
	c=\frac 1 2 \left[I_{\max}^{\varepsilon}(R:A)_\rho+\log\log I_{\max}^{\varepsilon}(R:A)_\rho\right]+\mathcal O\left(\log\frac{1}{\varepsilon}\right).
	\end{equation}
\end{thm} 

To prove this result, they use a lemma called the convex split lemma.

\begin{lem}[\cite{Anshu2014}]\label{lem:convsplit}
	Let $\rho\in\St(\hi_{A}\otimes\hi_E)$ and $\sigma\in\St(\hi_E)$ be quantum states, $k=D_{\text{max}}(\rho_{AE}\|\rho_A\otimes\sigma_E)$ and $0<\delta<\frac{1}{6}$. Define
	\begin{align*}
	n=\begin{cases}
	1 & k\le 3\delta\\ \left\lceil\frac{8\cdot 2^{k}\log\left(\frac{k}{\delta}\right)}{\delta^3}\right\rceil&\mathrm{else.}
	\end{cases}
	\end{align*}
	For the state
	\begin{align}\label{eq:taudef}
	\tau_{AE_1...E_n}=\frac{1}{n}\sum_{j=1}^n\rho_{AE_j}\otimes\left(\sigma^{\otimes (n-1)}\right)_{E_{j^c}}
	\end{align}
	$E$ is decoupled from $A$ in the following sense:
	\begin{align*}
	I(A:E_1...E_n)_\tau\le3\delta\quad\text{as well as}\quad P\left(\tau_A\otimes\tau_{E_1...E_n},\tau_{AE_1...E_n}\right)\le\sqrt{6\delta},
	\end{align*}
	where $E_{j^c}$ denotes $\{E_i\}_{i\neq j}$.
\end{lem}

\subsection{Catalytic decoupling}\label{subs:catalytic}
This section presents results that were obtained together with Mario Berta, Fr\'ed\'eric Dupuis, Renato Renner and my supervisor Matthias Christandl and were published in \cite{Majenz2017}. Parts of this publication are used here unaltered.

\begin{figure}
	\includegraphics[width=\textwidth]{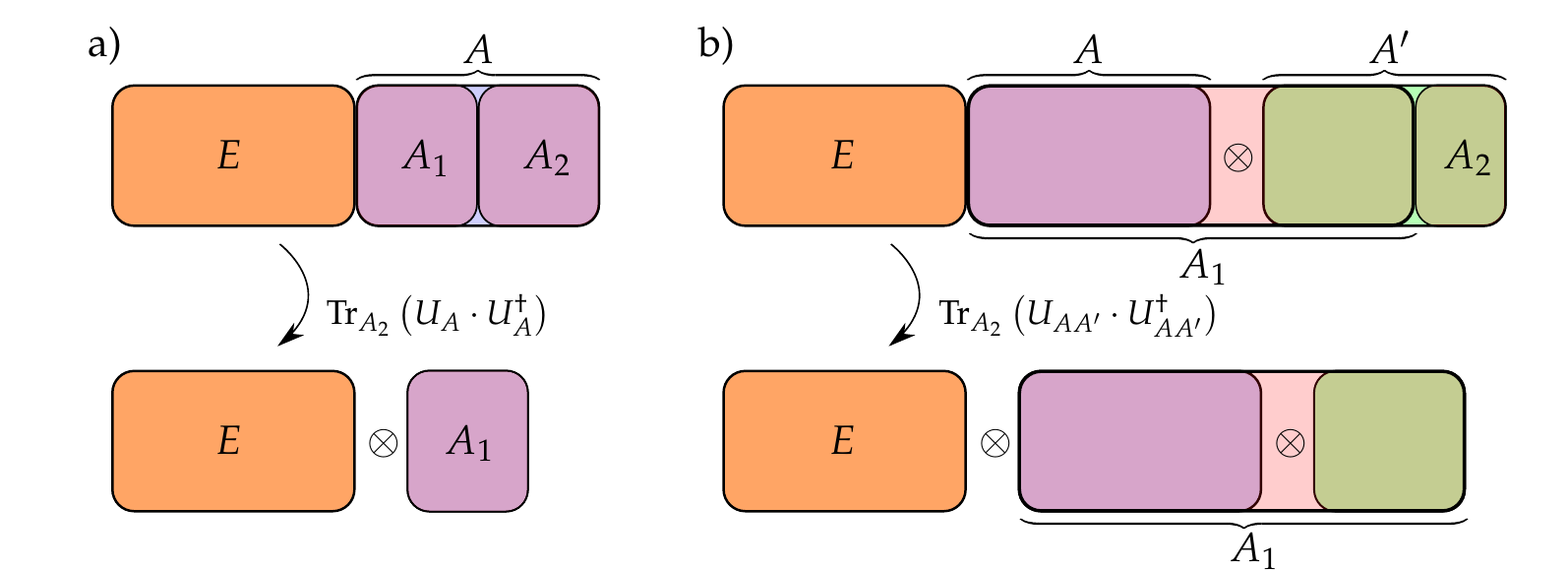}
	\caption[Catalytic decoupling]{Schematic representation of a) standard and b) catalytic decoupling: tracing out a system $A_2$ leaves the remaining state decoupled. While there is no ancilla for standard decoupling as in a), catalytic decoupling as in b) allows to make use of an additional, already decoupled system $A'$. The basic question is how large we have to choose the system $A_2$ such that the remaining system $A_1$ is decoupled from $E$.}\label{fig:schematic}
\end{figure}

\subsubsection{Definition}

A natural question to ask is if decoupling can be achieved more efficiently, using a unitary and a partial trace, in the presence of an already uncorrelated ancilla system (see Figure \ref{fig:schematic}). We call this generalized notion of decoupling \emph{catalytic decoupling}. This generalization is inspired by both results on one-shot coherent quantum state merging, \cite{Berta2011,Anshu2014}. Indeed, in both works, the protocol uses a large amount of extra entanglement, which is handed back almost unaltered. In the first work, i.e. Theorem \ref{thm:merge-Berta}, an embezzling state is needed, while applying Theorem \ref{thm:merge-Anshu} requires a purification of $\sigma_E^{\otimes (n-1)}$ from Lemma \ref{lem:convsplit}. The main contribution that the following results constitute is not so much of a technical, but more of a conceptual nature. Catalytic decoupling unifies the two mentioned techniques, and brings one-shot coherent quantum state merging home into the set of problems that can be solved optimally using a decoupling theorem.

Let us formally define catalytic decoupling. We say that a bipartite quantum state $\rho_{AE}$ is $\eps$-decoupled catalytically by the ancilla state $\sigma_{A'}$ and the partial trace map $\mathcal T_{\bar A\to A_1}(\cdot)=\tr_{A_2}[\cdot]$ with $\bar A\equiv AA'\cong A_1A_2$ if there exists unitary operation $U_{\bar A}$ such that,

% we say that $\eps$-decoupling can be achieved catalytically for a bipartite quantum state $\rho_{AE}$ if there exists an ancilla state $\rho_{A'}$, a unitary $U_{AA'}$ and a decomposition $\bar A\equiv AA'\cong A_1A_2$ such that
\begin{align}\label{eq:catalytic_decoupling}
\min_{\omega_{A_1}\otimes\omega_E}\!P\big(\mathcal T_{\bar A\to A_1}(U_{\bar A}\rho_{\bar AE}U_{\bar A}^\dagger ),\omega_{A_1}\otimes\omega_E\big)\le \eps&\\
\mathrm{where}\quad\rho_{\bar AE}=\rho_{AE}\otimes \sigma_{A'}&.
\end{align}
Again, we call the $A_1$-system the decoupled system and the $A_2$-system the remainder system. The term catalytic means that the share of the initially uncorrelated ancilla system $A'$, that becomes part of the decoupled system $A_1$, stays decoupled (see Figure \ref{fig:schematic}).

Now we are interested in the minimal size of the remainder system $A_2$ in order to achieve $\eps$-decoupling catalytically. We denote the minimal remainder system size for catalytically decoupling $A$ from $E$ in a state $\rho_{AE}$ by $R^\eps_c(A:E)_\rho$. Formally, we make the following

\begin{defn}[Minimal remainder system size for catalytic decoupling]
	Let $\rho_{AE}\in\St(\hi_{A}\otimes\hi_E)$ be a bipartite quantum state and $1\ge \varepsilon\ge 0$ an error parameter. The minimal remainder system size $R_c^\varepsilon(A:E)_\rho$ for catalytically decoupling $A$ from $E$ in $\rho_{AE}$  up to an error $\varepsilon$ is defined as the minimum over all finite-dimensional ancilla Hilbert spaces $\hi_{A'}$ and all ancilla states $\sigma_{A'}\in\St(\hi_{A'})$ of the minimal remainder system size for (standard) decoupling $AA'$ from $E$ in $\rho_{AE}\otimes \sigma_{A'}$. As a formula
	\begin{equation}\label{eq:catalytic1}
	R_c^\varepsilon(A:E)_\rho=\min_{\sigma_{A'}}R^\varepsilon(AA':E)_{\rho\otimes \sigma},
	\end{equation}
\end{defn}
where the minimum is taken over all quantum states on arbitrary finite-dimensional Hilbert spaces. The minimum in Equations \eqref{eq:catalytic1} always exists. This is because any remainder system size is the logarithm of an integer greater than or equal to one i.e. the minimum is taken over a discrete set that is bounded from below (by zero). Clearly, we have $R^\eps_c(A:E)_\rho\le R^\eps(A:E)_\rho$, as we can always choose a trivial ancilla. 

When arbitrary entangled resources are available, catalytic decoupling can be used to devise a protocol for coherent quantum state merging in exactly the same way that standard decoupling can. The only difference is, that Alice and Bob initially share a pure entangled resource state $\sigma_{A'B'}$ that is a purification of $\sigma_{A'}$. The initial state now being $\rho_{ABR}\otimes \sigma_{A'B'}$ where Alice holds systems $AA'$, Bob holds systems $BB'$ and the referee has system $R$. Alice divides her system according to the subdivision $AA'=A_1A_2$ as given by the catalytic decoupling protocol, applies the unitary $U_{\bar A}$ and sends system $A_2$ to Bob. Now the situation is exactly the same as at the corresponding stage in the standard-decoupling-based protocol: The joint state of Alice and the referee is $\varepsilon$-decoupled, and the overall state is still pure. Therefore Bob can use Uhlmann's theorem in the exact same way as before to find a decoding isometry and complete the coherent state merging protocol.

\subsubsection{Converse}
As a first step, we want to prove a converse bound, showing in advance that the achievability results derived from the embezzling and the convex split technique are almost optimal. To this end, we need a few lemmas about the max-mutual information. 

The max-mutual information has a data processing inequality.
\begin{lem}[Lemma B.17 in \cite{Berta2011}]\label{lem:Imax-dp}
	For all bipartite quantum states $\rho_{AB}$ and quantum channels $\Lambda_{A\to A'}$, $\Gamma_{B\to B'}$, the smooth max-mutual information fulfills the data processing inequality
	\begin{equation}\label{lem:dp-imax}
		I_{\max}(A:B)_\rho\ge I_{\max}(A':B')_{\Lambda\otimes\Gamma(\rho)}.
	\end{equation}
\end{lem}

The max-mutual information is invariant under local isometries.
\begin{lem}\label{lem:iso-inva}
	For a bipartite quantum state $\rho_{AB}$ and isometries $V_{A\to A'},\ W_{B\to B'}$,
	\begin{align*}
	I_{\max}^{\varepsilon}(A:B)_\rho=I_{\max}^{\varepsilon}(A':B')_{\tilde{\rho}},
	\end{align*}
	where $\tilde{\rho}_{A'B'}=V\otimes W\rho_{AB} V^\dagger \otimes W^\dagger $.
\end{lem}

\begin{proof}
	The data processing inequality implies $I_{\max}^{\varepsilon}(A:B)_\rho\ge I_{\max}^{\varepsilon}(A':B')_{\tilde{\rho}}$. Now define quantum channels $\Lambda_{A'\to A}$, $\Gamma_{B'\to B}$ such that $\Lambda(VXV^\dagger )^=X$ and $\Gamma(WYW^\dagger )=Y$. Using the data processing inequality for these maps as well implies $I_{\max}^{\varepsilon}(A:B)_\rho\le I_{\max}^{\varepsilon}(A':B')_{\tilde{\rho}}$.
\end{proof}

Tensoring a local ancilla does not change the max-mutual information.

\begin{lem}\label{lem:tens-inva}
	Let $\rho_{AB}$, $\sigma_C$ be quantum states. The smooth max-mutual information is invariant under adding local ancillas,
	\begin{align*}
	I_{\max}^{\varepsilon}(A:B)_\rho=I_{\max}^{\varepsilon}(A:BC)_{\rho\otimes\sigma}.
	\end{align*}
\end{lem}

\begin{proof}
	According to Lemma \ref{lem:dp-imax}, the max-mutual information decreases under local CPTP maps. Adding and removing an ancilla are such  maps, which implies the claimed invariance.
\end{proof}
The max mutual information also has the non-locking property.
\begin{lem}[Lemma B.12 in \cite{Berta2010}]\label{lem:Imax-non-locking}
	For a tripartite quantum state $\rho_{ABC}$,
	\begin{equation}
	I_{\max}^{\varepsilon}(A:BC)_\rho\le I_{\max}^{\varepsilon}(A:B)_\rho +2\log|C|.
	\end{equation}
\end{lem}

There are several ways to define the max-mutual information \cite{ciganovic2014smooth}, one of the alternative definitions will be useful for catalytic decoupling.

\begin{defn}\label{defn:alt-max-mut}
	An alternative definition of the max-mutual information of a quantum state $\rho_{AB}$ is given by
	\begin{align*}
	I_{\max}(A:B)_{\rho,\rho}=D_{\max}(\rho\|\rho_A\otimes\rho_B).
	\end{align*}
	The smooth version $I_{\max}^\varepsilon(A:B)_{\rho,\rho}$ is defined analogously to $I^\varepsilon_{\max}(A:B)_{\rho}$,
	\begin{align*}
	I_{\max}^\varepsilon(A:B)_{\rho,\rho}=\min_{\tilde{\rho}\in B_\varepsilon(\rho)}I_{\max}(A:B)_{\tilde{\rho},\tilde{\rho}}.
	\end{align*}
\end{defn}

This alternative definition has some disadvantages, in particular the non-smooth version is not bounded from above for a fixed Hilbert space dimension. The  two different smooth max-mutual informations, however, are quite similar, in particular they can be approximated up to a dimension independent error.

\begin{lem}[\cite{ciganovic2014smooth}, Theorem 3]\label{lem:ciganovic}
	For a bipartite quantum state $\rho_{AB}$,
	\begin{equation}
	I_{\max}^{\varepsilon+2\sqrt{\varepsilon}+\varepsilon'}(A:B)_\rho\lesssim I_{\max}^{\varepsilon+2\sqrt{\varepsilon}+\varepsilon'}(A:B)_{\rho,\rho}\lesssim I_{\max}^{\varepsilon'}(A:B)_\rho,
	\end{equation}
	where the notation $\lesssim$ hides errors of order $\log(1/\varepsilon)$.
\end{lem}

We are now ready to prove the a converse bound for the minimal remainder system size for catalytic decoupling.

\begin{thm}
	The minimal remainder system size for catalytic $\varepsilon$-decoupling is lower bounded by half the smooth max-mutual information,
	\begin{equation}
		R^\eps_c(A:E)_\rho\ge\frac{1}{2}I_{\max}^\varepsilon(A:E)_\rho
	\end{equation}
\end{thm}

\begin{proof}
	A catalytic decoupling protocol is specified by an ancilla state $\sigma_{A'}$, a decomposition $AA'=A_1A_2$ and a unitary $U_{AA'}$. The decoupling condition reads
	\begin{equation}
		P(\hat\rho_{A_1E}, \gamma_{A_1}\otimes\gamma_E)\le\varepsilon
	\end{equation}
	for some quantum state $\gamma$, with $\hat\rho=V\left(\rho\otimes\sigma\right) V^{\dagger}$. This implies $ I_{\max}^{\varepsilon}(A_1:E)_{\hat{\rho}}=0$, as $\gamma_{A_1}\otimes \gamma_E$ is a point in the minimization that defines $I_{\max}^{\varepsilon}(A_1:E)_{\hat{\rho}}$. Now we bound
	\begin{align}
		I_{\max}^\varepsilon(A:E)_\rho=&I_{\max}^\varepsilon(AA':E)_{\rho\otimes\sigma}\nonumber\\
		=&I_{\max}^\varepsilon(AA':E)_{\hat\rho}\nonumber\\
		\le&I_{\max}^{\varepsilon}(A_1:E)_{\hat{\rho}}+2\log|A_2|\nonumber\\
		=&2\log|A_2|.
	\end{align}
	In the first equality we have used Lemma \ref{lem:tens-inva}, and Lemma \ref{lem:iso-inva} implies the second equality. The inequality is Lemma \ref{lem:Imax-non-locking}, and the last equality is due to the decoupling condition as discussed above. This finishes the proof.
\end{proof}

\subsubsection{Achievability}
In the following we will show that the smooth max-information is achievable up to lower order terms. In fact, we will give two proofs, one abstracted from each of the two state merging results, Theorems \ref{thm:merge-Berta} and \ref{thm:merge-Anshu}. We begin with a proof based on the latter.

\begin{thm}[Catalytic decoupling]\label{thm:anc-decoup-supp}
	Let $\hat\rho_{AE}\in\St(\hi_A\otimes \hi_E)$ be a quantum state. Then for any $0<\delta\le \varepsilon$, catalytic decoupling with error $\varepsilon$ can be achieved with remainder system size
	\begin{align*}
	\log|A_2|\le\frac 1 2 \left(I_{\max}^{\varepsilon-\delta}(E:A)_{\hat\rho}+\left\{\log\log I_{\max}^{\varepsilon-\delta}(E:A)_{\hat\rho}\right\}_+\right)+\mathcal{O}(\log\frac{1}{\delta}),
	\end{align*}
	where we define $\{x\}_+$ to be equal to $x$ if $x\in\R_{\ge0}$ and $0$ otherwise.
\end{thm}

\begin{proof}
	Let $\gamma=\varepsilon-\delta$. Take $\rho\in B_\gamma(\hat{\rho})$ such that $I_{\max}(E:A)_\rho=I^\gamma_{\max}(E:A)_{\hat\rho}$.
	Let $\sigma_A$ be the minimizer in 
	\begin{align*}
	k=I_{\max}(E:A)_\rho=\min_{\sigma_A\in\St(\hi_A)}\relent{\max}{\big}{\rho_{AE}}{\sigma_A\otimes\rho_E}.
	\end{align*}
	If $k\le \frac{\delta^2}{2}$ the state is already decoupled according to Lemma \ref{lem:convsplit} and the statement is trivially true, so let us assume $k>\frac{\delta^2}{2}$.
	We want to use Lemma \ref{lem:convsplit} so let
	\begin{align*}
	n=\left\lceil\frac{8\cdot 2^{k}\log\left(\frac{k}{\delta'}\right)}{\delta'^3}\right\rceil
	\end{align*}
	with $\delta'=\frac{\delta^2}{6}$, $\hi_{A'}=\hi_A^{\otimes (n-1)}\otimes\hi_{\bar A}$ with $\hi_{\bar A}\cong\C^n$ and define the state $\tilde{\rho}_{A^{(2)}...A^{(n)}\bar A}=\sigma^{\otimes (n-1)}\otimes\tau_{\bar A}$, where $\tau_{\bar{A}}=\mathds{1}_{\bar A}/|\bar A|$ denotes the maximally mixed state on $\hi_{\bar A}$. We can now define a unitary that permutes the $A$-systems conditioned on $\bar A$ and thus creates an extension of the state $\tau$ from Equation \eqref{eq:taudef} when applied to $\rho_{AE}\otimes \tilde\rho_{A'}$,
	\begin{align*}
	U^{(1)}_{AA'}=\sum_{j=1}^n(1j)_{A^{(1)}...A^{(n)}}\otimes\ketbra{j-1}{j-1}_{\overline A},
	\end{align*}
	where $(1j)_{A^{(1)}...A^{(n)}}$ is the transposition $(1j)\in S_n$ under the representation $S_n\looparrowright \hi_A^{\otimes n}$ of the symmetric group that acts by permuting the tensor factors, and $(11)=1_{S_n}$. Now, we are almost done, as Lemma \ref{lem:convsplit} implies that
	\begin{align*}
	P\left(\xi_{EA^{(1)}...A^{(n)}},\xi_E\otimes\xi_{A^{(1)}...A^{(n)}}\right)\le \delta,
	\end{align*}
	where $\xi=U^{(1)}_{AA'}\rho_{AE}\otimes\tilde\rho_{A'} \left(U^{(1)}_{AA'}\right)^\dagger $. The register $\bar A$, however, is still a factor of two larger than the claimed bound for $|A_2|$. We can win this factor of two by using superdense coding, as $\bar A$ is classical. Let us therefore slightly enlarge $\hi_{\bar A}$ such that $\dim(\hi_{\bar A})=m^2$ for $m=\lceil\sqrt n\rceil$. We now rotate the standard basis of $\hi_M$ into a Bell basis
	\begin{align*}
	\ket{\psi_{kl}}=\frac{1}{\sqrt m}\sum_{s=0}^{m-1}e^{\frac{2\pi i k s}{m}}\ket s\otimes\ket{s+l\mod m}
	\end{align*}
	of $\hi_{\bar A_1}\otimes\hi_{\bar A_2}$, with $\hi_{\bar A_i}\cong\C^m$. That is done by the unitary
	\begin{align*}
	U^{(2)}_{\bar A}: \hi_{\bar A}\to\hi_{\bar A_1}\otimes\hi_{\bar A_2}\quad\mathrm{with}\quad U^{(2)}_{\bar A}=\sum_{k,l=0}^{m-1}\ketbra{\psi_{kl}}{m k+l}.
	\end{align*}
	As $\tr_{\bar A_2}\proj{\psi_{kl}}=\tau_{\bar A_1}$ for all $k,l\in\{0,...,m-1\}$, the unitary $V_{AA'\to A_1A_2}=U^{(2)}_{\bar A}U^{(1)}_{AA'}$ and the definitions $\hi_{A_2}=\hi_{\bar A_2}$ and $\hi_{A_1}=\hi_A^{\otimes n}\otimes\hi_{\bar A_1}$ achieve $P\left(\eta_{A_1 E},\eta_{A_1}\otimes\rho_E\right)\le \delta$, with $\eta=V_{AA'\to A_1A_2}\rho\otimes\tilde{\rho}V_{AA'\to A_1A_2}^\dagger $. Using the triangle inequality for the purified distance we finally arrive at
	\begin{align*}
	P\left(\hat\xi_{A_1E},\eta_{A_1}\otimes\rho_E\right)\le \gamma+\delta=\varepsilon,
	\end{align*}
	for $\hat\xi=V_{AA'\to A_1A_2}\hat\rho\otimes\tilde\rho V_{AA'\to A_1A_2}^\dagger $.
	The size of the remainder system is 
	\begin{align*}
	\log|A_2|=\frac 1 2\log n\le\frac 1 2 \big(I_{\max}^\gamma(E:A)_{\hat\rho}+\left\{\log\log I_{\max}^\gamma(E:A)_{\hat\rho}\right\}_+\big)+\mathcal{O}\left(\log\frac{1}{\delta}\right).
	\end{align*}
\end{proof}

	Using the alternative definition of the max-mutual information, Definition \ref{defn:alt-max-mut}, and the notation from the above theorem and proof, we can prove in the same way that
	\begin{align*}
	P(\hat\xi_{A_1E}, \eta_{A_1}\otimes\rho_E)\le \varepsilon
	\end{align*}
	with $\hat\xi=V_{AA'\to A_1A_2}\hat\rho_{AE}\otimes\rho_A^{\otimes n} V_{AA'\to A_1A_2}^\dagger $ and $\eta=V_{AA'\to A_1A_2}\rho_{AE}\otimes\rho_A^{\otimes n} V_{AA'\to A_1A_2}^\dagger $ in this case, and $n$ defined in the same way as above, just with $k=I_{\max}^\varepsilon(A:E)_{\hat\rho, \hat{\rho}}$. This achieves a stronger notion of decoupling, as the catalyst can be approximately handed back in the same state, except the $A_2$ part, which is lost in the partial trace,
	\begin{align*}
	\eta_{A_1}=\rho_{A}^{\otimes n}\otimes \tau_{\bar A_1}.
	\end{align*}
	By Lemma \ref{lem:ciganovic} this still implies
	\begin{align*}
	\log|A_2|&=\frac 1 2\log n\\
	&\le\frac 1 2 \big(I_{\max}^{\varepsilon-\delta}(E:A)_{\hat\rho,\hat\rho}+\left\{\log\log
	I_{\max}^{\varepsilon-\delta}(E:A)_{\hat\rho,\hat\rho}\right\}_+\big)+\mathcal{O}\left(\log\frac{1}{\delta}\right)\\
	&\le\frac 1 2 \big(I_{\max}^{\varepsilon-2\delta-2\sqrt\delta}(E:A)_{\hat\rho}+\left\{\log\log
	I_{\max}^{\varepsilon-2\delta-2\sqrt\delta}(E:A)_{\hat\rho}\right\}_+\big)+\mathcal{O}\left(\log\frac{1}{\delta}\right)\\
	&\le\frac 1 2 \big(I_{\max}^{\varepsilon-\delta'}(E:A)_{\hat\rho}+\left\{\log\log
	I_{\max}^{\varepsilon-\delta'}(E:A)_{\hat\rho}\right\}_+\big)+\mathcal{O}\left(\log\frac{1}{\delta'}\right),
	\end{align*}
	having defined $\delta'=2\delta+2\sqrt{\delta}$.
	
	The second proof makes use of standard decoupling. The following standard decoupling theorem is slightly different from Theorem \ref{thm:one-shot-decoup} in that it has the non-smooth $\log|A|$ term instead of the smooth max-entropy term, and is more convenient for the proof we are about to commence.
	
	\begin{lem}\cite[Theorem 3.1]{Berta2011}, \cite[Table 2]{Dupuis2014}\label{cor:bcr-decoup}
		Let $\rho_{AE}\in\St(\hi_A\otimes \hi_E)$ be a quantum state. Then, there exists a decomposition $\hi_A\cong\hi_{A_1}\otimes\hi_{A_2}$ with
		\begin{align*}
		\log(|A_2|)\le \frac{1}{2}\left(\log(|A|)-H_{\min}(A|E)_\rho\right)+2\log\frac 1 \varepsilon+1,
		\end{align*}
		such that
		\begin{align*}
		P\left(\rho_{A_1E},\frac{1_{A_1}}{|A_1|}\otimes\rho_E\right)\le\varepsilon.
		\end{align*}
	\end{lem}
	
	The difference between the bound given here and the bound from \cite[Theorem 3.1]{Berta2011} stems from the fact that we define decoupling using the purified distance. We are now ready for the alternative achievability proof for catalytic decoupling.
	
	\begin{thm}[Non-smooth catalytic decoupling from standard decoupling and embezzling states]\label{thm:nonsmooth-anc-decoup-bcr}
		Let $\rho_{AE}\in\St(\hi_A\otimes \hi_E)$ be a quantum state. Then, $\varepsilon$-catalytic decoupling can be achieved with remainder system size
		\begin{align*}
		\!\log|A_2|\!\le\frac{1}{2}I_{\max}(A;E)_{\rho}+\log H_0(A)_{{\rho}}+\mathcal{O}\left(\log\left(\frac 1 \varepsilon\right)\right).
		\end{align*}
		In addition, if we allow for the use of isometries instead of unitaries, the ancilla system's final state is $\varepsilon$ close to its initial state.
	\end{thm}
	
	\begin{proof}
		For notational convenience let $\ket{\rho}_{AER}$ be a purification of $\rho$. 
		Also in slight abuse of notation we replace $\hi_A$ by $\mathrm{supp}(\rho_A)$ so that $|A|\le2^{ H_0(A)_{{\rho}}} $. The idea is to decompose the Hilbert space $\hi_A$ into a direct sum of subspaces where the spectrum of $\rho_A$ is almost flat. Let $Q=\left\lceil\log|A|+2\log\left(\frac 1 \varepsilon\right)-1\right\rceil$ and define the projectors $P_i,\ i=0,...,Q+1$ such that $P_{Q+1}$ projects onto the eigenvectors of $\rho_A$ with eigenvalues in $\left[0,2^{-(Q+1)}\right]$ and $P_i$ projects onto the eigenvectors of $\rho_A$ with eigenvalues in $\left[2^{-(i+1)},2^{-i}\right]$ for $i=0,...,Q$. We can now write the approximate state $\ket{\bar\rho}_{AER}=\frac{1}{\sqrt \alpha}(\mathds 1_A-P_{Q+1})\ket{\rho}_{AER}, \alpha=\tr(\mathds 1_A-P_{Q+1})\rho$ as a superposition of states with almost flat marginal spectra on $A$,
		\begin{align*}
		\ket{\bar\rho}=\sum\sqrt{p_i}\ket{\rho^{(i)}},
		\end{align*}
		with $p_i=\tr\bar\rho P_i$ and  $\ket{\rho^{(i)}}=\frac{1}{\sqrt{p_i}}P_i\ket{\rho}$. This decomposition corresponds to the direct sum decomposition 
		\begin{align*}
		\hi_A\cong\bigoplus_{i=0}^{Q+1}\hi_{A^{(i)}},
		\end{align*}
		where $\hi_{A^{(i)}}=\mathrm{supp}(P_i)$. Note that we have $P(\rho,\bar\rho)=\sqrt{1-\alpha}$ and $$1-\alpha\le |A|2^{-\left(\log|A|+2\log\left(\frac 1 \varepsilon\right)\right)}=\varepsilon^2,$$ i.e. $P(\rho,\bar\rho)\le\varepsilon$. Now, we have a family of states, $\{\rho^{(i)}_{A^{(i)}E}\}$ to each of which we apply  Lemma \ref{cor:bcr-decoup}. This yields decompositions $\hi_{A^{(i)}}\cong\hi_{A_1^{(i)}}\otimes\hi_{A_2^{(i)}}$ such that
		\begin{align}\label{eq:flat-decoup}
		P\left(\rho^{(i)}_{A^{(i)}_1E},\tau_{A^{(i)}_1}\otimes \rho^{(i)}_E\right)\le\varepsilon
		\end{align}
		and
		\begin{align*}
		\log(|A^{(i)}_2|)\ge \frac{1}{2}\left(\log(|A^{(i)}|)+H_{\min}(A|E)_{\rho^{(i)}}\right)+2\log\left(\frac 1 \varepsilon\right)+1,
		\end{align*}
		where $\tau_A=\frac{\mathds{1}_A}{|A|}$ is the maximally mixed state on a quantum system $A$.
		
		At this stage of the protocol the situation can be described as follows. Conditioned on $i$, $A_1^{(i)}$ is decoupled from $E$. If $\rho^{(i)}_E\neq\rho^{(j)}_E$ and $\left|A^{(i)}\right|\neq\left|A^{(j)}\right|$ however, there are still correlations left between $A_1$ and $E$. To get rid of this problem, we hide the maximally mixed states of different dimensions in an embezzling state by ``un-embezzling" them. Let us therefore first isometrically embed all these states in the same Hilbert space. To do that, define
		\begin{align*}
		d_2=\max_i\left|A^{(i)}_2\right|\quad\text{and}\quad d_1=\max\left(\max_i\left|A^{(i)}_1\right|, \left\lceil\frac{\left|A^{(Q+1)}\right|}{d_2}\right\rceil\right). 
		\end{align*}
		Now, let $\hi_{\tilde A_\alpha}\cong\C^{d_\alpha}$ and choose isometries $U^{(\alpha,i)}_{A^{(i)}_\alpha\to\tilde A_\alpha}$ for $\alpha=1,2$, define $U^{(i)}_{A^{(i)}\to\tilde A_1\otimes \tilde A_2}=U^{(1,i)}_{A^{(i)}_1\to\tilde A_1}\otimes U^{(2,i)}_{A^{(i)}_2\to\tilde A_2}$ for $i=1,...,Q$. In addition, choose an isometry $U^{Q+1}_{A^{(Q+1)}\to \tilde A_1\otimes \tilde A_2}$. Let $\ket{\mu}_{A'B'}\in\hi_{A'}\otimes\hi_{B'}$ be a $(d_1,\varepsilon)$-embezzling state, and let $\sigma_{A'}=\tr_{B'}\proj \mu$. Define the isometries $\bar V^{(i)}_{A'\to A'\tilde A_1}$ that would embezzle a state
		\begin{align*}
		\tau^{(i)}_{\tilde A_i}=U^{(1,i)}_{A^{(i)}_1\to\tilde A_1}\tau_{A^{(i)}_1}\left(U^{(1,i)}_{A^{(i)}_1\to\tilde A_1}\right)^\dagger 
		\end{align*}
		from $\sigma_{A'}$. Taking some state $\ket 0_{\tilde A_1}\in\hi_{\tilde A_1}$ we can pad these embezzling isometries to unitaries $V^{(i)}_{A'\tilde A_1}$ such that
		\begin{align}\label{eq:padembezz}
		P\left(V^{(i)}_{A'\tilde A_1}\sigma_{A'}\otimes\proj 0_{\tilde A_1}\left(V^{(i)}_{A'\tilde A_1}\right)^\dagger ,\sigma_{A'}\otimes \tau^{(i)}_{\tilde A_i}\right)\le\varepsilon.
		\end{align}
		We can combine the above isometries and unitaries now to un-embezzle the states that are approximately equal to $\tau_{A^{(i)}_1}$ conditioned on $i$. Define $\tilde{A}_3\cong\C^{Q+1}$ and 
		\begin{align*}
		W^I_{A\to\tilde A_1\tilde A_2\tilde A_3}=\sum_i U^{(i)}_{A^{(i)}\to \tilde A_1\otimes \tilde A_2}P_i\otimes\ket i_{\tilde A_3},\quad W^{II}_{A'\tilde A_1\tilde A_3}=\sum_i \left(V^{(i)}_{A'\tilde A_1}\right)^\dagger \otimes\proj i_{\tilde A_3}.
		\end{align*}
		The final state of our decoupling protocol is
		\begin{align*}
		\rho^f_{A_1A_2E}=W^{II}W^I\rho_{AE}\otimes\sigma_{A'} \left(W^I\right)^\dagger \left(W^{II}\right)^\dagger ,
		\end{align*}
		where we omitted the subscripts on the $W$s for compactness and have defined $A_1=A'\tilde{A}_1$ and $A_2=\tilde A_2\tilde A_3$. Let us show that this protocol actually decouples $A_1$ from $E$. We bound, omitting the subscripts on unitaries and isometries,
		\begin{align*}
		&P\left(\rho^f_{A_1E},\sigma_{A'}\otimes\proj 0_{\tilde A_1}\otimes\rho_E\right)\nonumber\\
		=&\sum_ip_iP\Big(\left(V^{(i)}\right)^\dagger  U^{(1,i)}\sigma_{A'}\otimes\rho^{(i)}_{A^{(i)}_1E}\left(U^{(1,i)}\right)^\dagger  V^{(i)},\sigma_{A'}\otimes\proj 0_{\tilde A_1}\otimes\rho_E\Big)\nonumber\\
		=&\sum_ip_iP\Big( U^{(1,i)}\sigma_{A'}\otimes\rho^{(i)}_{A^{(i)}_1E}\left(U^{(1,i)}\right)^\dagger ,V^{(i)}\sigma_{A'}\otimes\proj 0_{\tilde A_1}\otimes\rho_E\left(V^{(i)}\right)^\dagger \Big)\nonumber\\
		\le&\sum_ip_iP\Big( U^{(1,i)}\sigma_{A'}\otimes\rho^{(i)}_{A^{(i)}_1E}\left(U^{(1,i)}\right)^\dagger ,\sigma_{A'}\otimes\tau^{(i)}_{\tilde A_1}\otimes\rho_E\Big)\nonumber\\
		&\quad+\sum_ip_iP\Big(\sigma_{A'}\otimes\tau^{(i)}_{\tilde A_1}\otimes\rho_E,V^{(i)}\sigma_{A'}\otimes\proj 0_{\tilde A_1}\otimes\rho_E\left(V^{(i)}\right)^\dagger \Big)\nonumber\\
		\le&\sum_ip_iP\Big( U^{(1,i)}\sigma_{A'}\otimes\rho^{(i)}_{A^{(i)}_1E}\left(U^{(1,i)}\right)^\dagger ,\sigma_{A'}\otimes\tau^{(i)}_{\tilde A_1}\otimes\rho_E\Big)+\varepsilon.
		\end{align*}
		The first inequality is the triangle inequality, the second one is Equation \eqref{eq:padembezz}. It remains to bound the first summand,
		\begin{align*}
		&\sum_ip_iP\Big( U^{(1,i)}\sigma_{A'}\otimes\rho^{(i)}_{A^{(i)}_1E}\left(U^{(1,i)}\right)^\dagger ,\sigma_{A'}\otimes\tau^{(i)}_{\tilde A_1}\otimes\rho_E\Big)\nonumber\\
		=&\sum_ip_iP\Big(\sigma_{A'}\otimes\rho^{(i)}_{A^{(i)}_1E},\sigma_{A'}\otimes\left(\left(U^{(1,i)}\right)^\dagger \tau^{(i)}_{\tilde A_1}U^{(1,i)}\right)\otimes\rho_E\Big)\nonumber\\
		=&\sum_ip_i P\Big(\sigma_{A'}\otimes\rho^{(i)}_{A^{(i)}_1E},\sigma_{A'}\otimes\tau_{\tilde A^{(i)}_1}\otimes\rho_E\Big)\nonumber\\
	\le	&\sum_ip_iP\left(\rho^{(i)}_{A^{(i)}_1E},\tau_{\tilde A^{(i)}_1}\otimes\rho_E\right)\nonumber\\
		\le&\varepsilon,
		\end{align*}
		where the first inequality is the triangle inequality again, and the second one is Equation \eqref{eq:flat-decoup}. This shows that we achieved $2\varepsilon$-decoupling, i.e.
		\begin{align}\label{eq:non-smooth-decoup}
		P\left(\rho^f_{A_1E},\sigma_{A'}\otimes\proj 0_{\tilde A_1}\otimes\rho_E\right)\le2\varepsilon.
		\end{align}
		We also have to bound $\log|A_2|$, i.e. we need to make sure that
		\begin{align*}
		\max_{i=1,...,Q}\left(H_0(A)_{\rho^{(i)}}-H_{\min}(A|E)_{\rho^{(i)}}\right)\le I_{\max}(E;A)_\rho+\mathcal{O}\left(\log\left(\frac 1 \varepsilon\right)\right).
		\end{align*}
		This is shown in \cite{Berta2011} in the last part of the proof of Theorem 3.10.
		Thereby the size of the remainder system is bounded by
		\begin{align*}
		\log|A_2|=\frac{1}{2}I_{\max}(A;E)_{\hat\rho}+\log H_{0}(A)_{\hat{\rho}}+\mathcal{O}\left(\log\left(\frac 1 \varepsilon\right)\right).
		\end{align*}
		If we only want to use unitaries, we can complete all involved isometries to unitaries by adding an appropriate additional pure ancilla system.
	\end{proof}

As an easy corollary we can derive a bound on the remainder system that involves the smooth max-mutual information.

\begin{thm}[Catalytic decoupling]\label{prop:anc-decoup-bcr}
	Let $\hat\rho_{AE}\in\St(\hi_A\otimes \hi_E)$ be a quantum state. Then, $\varepsilon$-catalytic decoupling can be achieved with remainder system size
	\begin{align*}
	\!\log|A_2|\!\le\frac{1}{2}I_{\max}^{\varepsilon-\delta}(A;E)_{\rho}+\log H_0(A)_{{\rho}}+\mathcal{O}\left(\log\left(\frac 1 \delta\right)\right).
	\end{align*}
	In addition, if we allow for the use of isometrics instead of unitaries, the ancilla systems final state is $\varepsilon$ close to its initial state.
\end{thm}
\begin{proof}
	Let $\hat{\rho}\in B_{\eta}(\rho)$ with $\eta=\varepsilon-\delta$ such that
	\begin{equation}
	I_{\max}^{\eta}(A:E)_\rho=I_{\max}(A:E)_{\hat{\rho}}.
	\end{equation}
	Define $\rho'=\Pi\hat{\rho}\Pi$, where $\Pi$ is the orthogonal projector onto the support of $\rho$. It follows from Uhlmann's theorem that $P(\rho,\rho')\le P(\rho,\hat{\rho})$. As the max-mutual information is non-increasing under projections (cf. \cite[Lemma B.19]{Berta2011}), it follows that
	\begin{equation}
	I_{\max}^{\eta}(A:E)_\rho=I_{\max}(A:E)_{\rho'}
	\end{equation}
	as well. Applying Theorem \ref{thm:nonsmooth-anc-decoup-bcr} to $\rho'$ and an application of the triangle inequality yields the claimed bound.
\end{proof}

If we accept a slightly worse smoothing parameter for the leading order term, i.e. the max-mutual information, we can smooth the second term as well and replace the R\'enyi-0 entropy by the max-entropy.

\begin{cor}
	Let $\rho_{AE}\in\St(\hi_A\otimes \hi_E)$ be a quantum state. Then, $\varepsilon$-catalytic decoupling can be achieved with remainder system size
	\begin{align*}
	\!\log|A_2|\!\le\frac{1}{2}I_{\max}^{\varepsilon'}(A;E)_{\rho}+\log H_{\max}^{{\varepsilon'}^2/2}(A)_{\rho}+\mathcal{O}(\log{\varepsilon'}\!),
	\end{align*}
	where $\varepsilon'=\varepsilon/6$.
\end{cor}

\begin{proof}
	To get the bound involving smooth entropy measures we will find a state $\hat\rho\in B_{2{\varepsilon'}}(\rho)$ such that $I_{\max}(E;A)_{\hat\rho}\le I_{\max}^{\varepsilon'}(E;A)_{\rho}$ and $H_0(A)_{\hat\rho}\le H^{{\varepsilon'}^2/2}_0(A)_{\rho}$. Let $\rho'_{AE} \in B_{\varepsilon'}(\rho_{AE})$ such that $I_{\max}(E;A)_{\rho'}=I_{\max}^{\varepsilon'}(E;A)_{\rho}$. Let $\Pi_A$ be a projection of minimal rank such that $H_0^{{\varepsilon'}^2/2}(A)_{\rho}\ge H_0(A)_{\rho''}$, with $\rho''=\Pi_A \rho_{AE}\Pi_A\in B_{\varepsilon'}(\rho_{AE})$. To see why such a projection exists, note that Lemma \ref{lem:spectr} implies that there exists a state $\rho''_A$ such that
	\begin{align}\label{eq:H0bound}
	H_0(A)_{\rho''_A}=H_0^{{\varepsilon'},\tr}(A)_{\rho}\le H_0^{{\varepsilon'}^2\!/2}(A)_{\rho}
	\end{align}
	and $[\rho_A,\rho''_A]=0$, where the inequality is due to the Fuchs van de Graaf inequality. For the case of commuting density matrices, i.e. the classical case, it is clear that the density matrix in a given trace distance neighborhood of $\rho$, that has minimal rank, is just equal to $\rho$ with the smallest eigenvalues set to zero. This implies that $\rho''_A$ can be chosen to have the form $\rho''_A=\Pi_A\rho_A\Pi_A$. It is easy to see that $P(\rho_{AE},\rho''_{AE})\le{\varepsilon'}$ where $\rho''_{AE}=\Pi_A\rho_{AE}\Pi_A$: Pick a purification $\ket{\rho''}_{AER}=\Pi_A\ket{\rho}_{AER}$ and observe that
	\begin{align*}
	F(\rho_{AE},\rho''_{AE})=\max_{\ket\sigma_{AER}}\left|\braket{\sigma}{\rho''}\right|=\max_{\ket\sigma_{AER}}\left|\bra{\sigma}\Pi\ket{\rho}\right|=\tr\Pi\rho=F(\rho_A,\rho''_A),
	\end{align*}
	where the fist equation is Uhlmann's theorem and the third equation follows from the saturation of the Cauchy-Schwarz inequality. We also use that $[\rho_A,\rho''_A]=0$ in the last equation. Now, we define $\hat\rho=\Pi_A\rho'\Pi_A$ and bound
	\begin{align}\label{eq:annoying-dist-ineq}
	P(\hat\rho_{AE},\rho_{AE})=&P(\Pi_A\rho'_{AE}\Pi_A,\rho_{AE})\nonumber\\
	=& P( \rho'_{AE}, \Pi_A\rho_{AE}\Pi_A)\nonumber\\
	=&P( \rho'_{AE}, \rho''_{AE})\le P(\rho'_{AE}, \rho_{AE})+P(\rho_{AE}, \rho''_{AE})\le 2{\varepsilon'}.
	\end{align}
	The second equation follows easily by Uhlmann's theorem. According to \cite[Lemma B.19]{Berta2011} the max-mutual-information decreases under projections, i.e. we have
	\begin{align*}
	I_{\max}(E;A)_{\hat\rho}\le I_{\max}(E;A)_{\rho'}=I_{\max}^{\varepsilon'}(E;A)_{\rho}.
	\end{align*}
	Our choice of $\Pi_A$ gives
	\begin{align*}
	H^{{\varepsilon'}^2/2}_0(A)_{\rho}\ge H_0(A)_{\rho''_A}=\log\mathrm{rk}(\Pi_A)\ge H_0(A)_{\hat\rho},
	\end{align*}
	where the first inequality is Equation \eqref{eq:H0bound}.
	Now, we apply Theorem \ref{thm:nonsmooth-anc-decoup-bcr}, to $\hat\rho_{AE}$. Let $\rho^{(f)}_{A_1A_2E}$ be the final state when applying the resulting protocol to $\rho_{AE}$. Then, we get
	\begin{align*}
	P\left(\rho^{(f)}_{A_1E},\sigma_{A'}\otimes\proj 0_{\tilde A_1}\otimes\rho_E\right)\le 6{\varepsilon'}
	\end{align*}
	by using Equations \eqref{eq:annoying-dist-ineq}, \eqref{eq:non-smooth-decoup} , the triangle inequality and the monotonicity of the purified distance under CPTP maps. 	 	
\end{proof}

\subsubsection{Equivalence between catalytic decoupling and catalytic erasure of correlations}

Before the conception of the decoupling technique as we know it, Groisman, Popescu and Winter proposed an operational interpretation of the quantum mutual information \cite{Groisman2005}. They show, that the quantum mutual information $I(A:B)$ is equal to the rate of noise that has to be applied to a system $A$ in the form of a random unitary channel to render it independent from $B$.

We begin by formally define what we mean by a random unitary channel.
\begin{defn}
	A CPTP map $\Lambda: \End{\hi}\to\End{\hi}$ is called a \emph{random unitary channel}, or a \emph{mixture of $N$ unitaries}, if it has a Kraus representation where all $N$ Kraus operators are multiples of unitaries, i.e.
	\begin{equation}
	\Lambda(X)=\sum_{i=1}^{N}p_i U_iXU_i^\dagger .
	\end{equation}
\end{defn}

The results from \cite{Groisman2005} concern the asymptotic setting, i.e. the case where the task is to efficiently locally erase the correlations in the limit of many IID copies of a state $\rho_{AB}$, and can be summarized in the following way:

\begin{thm}[Local erasure of correlations, \cite{Groisman2005}]
	Let $\rho_{AB}\in\mathcal{S}(\hi_A\otimes\hi_B)$ be a bipartite quantum state. For all $\varepsilon>0$ there exists $n\in \N$ and a mixture of $N$ unitaries $\Lambda: \End{\hi_A^{\otimes n}}\to\End{\hi_A^{\otimes n}}$ with $\log N\le n(I(A:B)_\rho+\varepsilon)$ such that 
	\begin{equation}
	\left\|\Lambda\otimes\mathrm{id}_{B^n}\rho_{AB}^{\otimes n}-\Lambda\rho_A^{\otimes n}\otimes\rho_B^{\otimes n}\right\|_1\le \varepsilon.
	\end{equation}
	Conversely, if $\Lambda': \End{\hi_{A'}}\to\End{\hi_{A'}}$ is a mixture of $N'$ unitaries such that it $\varepsilon$-decorrelates a state $\rho'\in\mathcal{S}(\hi_{A'}\otimes\hi_{B'})$, i.e.
	\begin{equation}
	\left\|\Lambda\otimes\mathrm{id}_{B'}\sigma_{A'B'}-\Lambda\rho_{A'}\otimes\rho_{B'}\right\|_1\le \varepsilon,
	\end{equation}
	then
	\begin{equation}
	\log N\ge I(A':B')_\sigma+3\varepsilon\log|A'|+1.
	\end{equation}
\end{thm}

It turns out that the tasks of decoupling and the local erasure of correlations using mixtures of unitaries are equivalent as soon as we allow for ancillas, which provides a one shot generalization of the above result.

\begin{thm}
	The correlations of a state $\rho_{AB}\in\mathcal{S}(\hi_A\otimes\hi_B)$ can be erased by a mixture of $N\le 2^{2k}$ unitaries up to an error of $\varepsilon$  if and only if $\varepsilon$-decoupling is possible with remainder system size $\log|A_2|\le k$, if we allow appending an ancilla system in an arbitrary product state to $A$ in both tasks.
\end{thm}
\begin{proof}
	First assume that we have a mixture of $2^{2k}$ unitaries
	\begin{eqnarray}
	\Lambda: \End{\hi_{A}}&\to&\End{\hi_{A}}\nonumber\\
	X\mapsto \sum_{i=1}^{2^{2k}}p_i U_iXU_i^\dagger 
	\end{eqnarray}
	that $\varepsilon$-decorrelates $A$ from $B$.
	Now take an ancilla state
	\begin{equation}
	\tilde\rho_M=\sum_{i=1}^{2^{2k}}p_i\proj i\in\mathcal{S}(\hi_M), \ \hi_M=\C^{2^k}
	\end{equation}
	and define the controlled unitary
	\begin{equation}
	U=\sum_{i=1}^{2^{2k}}U_i\otimes\proj i.
	\end{equation}
	This unitary achieves decoupling in the sense that
	\begin{eqnarray}
	\left\|\xi_{AB}-\xi_A\otimes\rho_B\right\|_1&\le&\varepsilon,\text{ with}\nonumber\\
	\xi_{ABM}&=& U\rho_{AB}\otimes\sigma_M U^\dagger .
	\end{eqnarray}
	But $M$ is classical, so we can apply superdense coding as in the proof of Theorem \ref{thm:anc-decoup-supp} to split $\hi_M\cong\hi_{M'}\otimes\hi_{M''}$ such that
	\begin{eqnarray}
	\left\|\xi_{ABM'}-\xi_A\otimes\rho_B\otimes\frac{1}{|M'|}\mathds{1}_{M'}\right\|_1&\le&\varepsilon,
	\end{eqnarray}
	with $\log|M''|\le k$.
	
	Conversely assume that we have an ancilla $\tilde{\rho}_T$ and a unitary $U_{AT\to A_1A_2}$ such that
	\begin{equation}
	\left\|\xi_{A_1B}-\xi_{A_1}\otimes\rho_B\right\|_1\le\varepsilon
	\end{equation}
	where $\xi_{A_1A_2B}=U_{AT\to A_1A_2}\rho_{AB}\otimes\tilde{\rho}_TU_{AT\to A_1A_2}^\dagger $ and $\log|A_2|\le k$. Let $|A_2|=N$. We define the generalized pauli operators
	\begin{eqnarray}\label{eq:genpauli}
	\Sigma&=&\sum_{j=0}^{N}e^{\frac{2 \pi i j}{N}}\proj j\nonumber\\
	\Xi&=&\sum_{j=0}^{N}\ketbra{j}{\left(j+1\right)\mod N}
	\end{eqnarray}
	on $\hi_{A_2}$. A short calculation shows that
	\begin{equation}
	\frac{1}{N^2}\sum_{i,j=0}^{N-1}\Xi^i\Sigma^jX\Sigma^{-j}\Xi^{-i}=\tr(X)\mathds 1_{A_2}
	\end{equation}
	for any $X\in\End{\hi_{A_2}}$ and therefore the mixture of $N^2$ unitaries
	\begin{eqnarray}
	\Lambda: \hi_A\otimes\hi_T&\to&\hi_{A_1}\otimes\hi_{A_2}\nonumber\\
	X&\mapsto&\frac{1}{N^2}\sum_{i,j=0}^{N-1}V_{ij}XV_{ij}^\dagger ,\nonumber\\
	V_{ij}&=&\left(\mathds 1_{A_1}\otimes \Xi^i\Sigma^j\right)U_{AT\to A_1A_2}
	\end{eqnarray}
	$\varepsilon$-decorrelates $AT$ from $B$.
\end{proof}

\begin{cor}[Catalytic erasure of correlations]
	The correlations of a state $\rho_{AB}$ can be erased catalytically up to an error $\varepsilon$ by locally acting on $A$ with a mixture of $N$ unitaries, with
	\begin{equation}
	\log N\le I_{\max}^\varepsilon(E:A)_{\hat\rho}+\left(\log\log I_{\max}^\varepsilon(E:A)_{\hat\rho}\right)_++\mathcal{O}(\log\varepsilon).
	\end{equation}
	
\end{cor}

\section{Bounds on the resource requirements of port based teleportation}

Port based teleportation \cite{Ishizaka2008,Ishizaka2009} is a variant of quantum teleportation, where the receiver, instead of applying a non-trivial correction unitary, only has to select one of several output ports. While having tremendous resource requirements compared to standard teleportation, it has some advantages. First, one can imagine a receiver that, while being able to make use of a certain piece of quantum information, is incapable of Pauli computation on his register. Then port based teleportation can still be used, as selecting one of many pieces of memory is a classical operation. But second, and more importantly, port based teleportation has a kind of \emph{covariance property}: If Bob wants to apply any operation to the output of the teleportation protocol, he can do so before Alice has even started sending it, by just applying the operation to \emph{every} output port. Because of this property, port based teleportation is useful for instantaneous non-local computation \cite{Beigi2011} and, using the resulting instantaneous non-local computation protocol, for breaking any scheme for position-based quantum cryptography \cite{Buhrman2014}. A disadvantage of PBT is, that it cannot be done perfectly with finite resources, i.e. any PBT protocol, in practice, has to be approximate.

A protocol for port based teleportation was given in the original work by Satoshi Ishizaka and Tohya Hiroshima \cite{Ishizaka2008,Ishizaka2009}, a proof for achievability for all input dimensions was given in \cite{Beigi2011}. The error bounds given in these works, are, however, given in terms of the entanglement fidelity, while channel simulation tasks are most naturally defined using a worst case error measure.\footnote{As clear already from the title of Bennett et al.'s seminal paper \cite{Bennett1993}, teleportation is about transfering \emph{unknow} quantum states.}

Furthermore, the protocol for port-teleporting a $d$-dimensional quantum state that was exhibited in \cite{Ishizaka2008} uses many maximally entangled states as a resource. Also it can be shown that optimizing Alice's measurement for the port based teleportation task is equivalent to a state discrimination problem when using the entanglement fidelity as a figure of merit. In the protocol from \cite{Ishizaka2008}, the \emph{pretty good measurement} (PGM) \cite{Holevo1978,Hausladen1994}, also called square root measurement, is employed. There are, however, only few situations where the PGM is known to be optimal, and the ``pretty good" property \cite{Barnum2002} is not enough to conclude any optimality property for port based teleportation in arbitrary dimensions. An important question is therefore, whether the resource requirements can be significantly reduced if a more general resource state and the optimal measurement are used. 

In the following sections, port based teleportation and its relation to universal programmable processors will be introduced. The symmetries of the port based teleportation problem are explored, leading to a diamond norm error bound for the standard protocol for PBT. Then different techniques for finding lower bounds on the number of ports are explored, that are necessary to achieve port based teleportation with a given error $\varepsilon$. As a side result, a new lower bound on the dimension of the program register of an $\varepsilon$-approximate universal programmable quantum processor is proven.
%
%
%Port based teleportation can be seen as one task in a family of teleportation tasks indexed by the set of correction operations that the receiver has at his disposal. This set is the permutation group $S_N$ for $N$-port port based teleportation, and the Pauli group for standard teleportation.

\subsection{Port Based teleportation}\label{subs:PBT}

Let us first formally define the task of standard teleportation. We give a slightly generalized definition that highlights the task that has to be achieved, allows for errors and thereby prepares the definition of port based teleportation.

\begin{defn}[Teleportation, \cite{Bennett1993}]\label{defn:teleport}
	A $(d,\rho_{A'B'},\varepsilon)$-teleportation protocol is a protocol that simulates the identity channel $\id\in\End{\End{\C^d}}$ with diamond norm error $\varepsilon$, using the resource state $\rho_{A'B'}$ as well as local operations and classical communication (LOCC).
\end{defn}

The standard protocol for $(2,\proj{\phi^+},0)$-teleportation is well known and described for example in \cite{Nielsen2000}. It can be roughly described in a few sentences. Alice makes a joint measurement on the input qubit and her part of the resource state. Her measurement consists of the projectors onto the elements of the so called Bell basis, $\{\ket{\phi^+},\ket{\phi^-},\ket{\psi^+},\ket{\psi^-}\}$, where $\ket{\phi^-}=\sigma_z\otimes \one\ket{\phi^+}$, $\ket{\psi^+}=\sigma_x\otimes \one\ket{\phi^+}$, $\ket{\psi^-}=\sigma_x\sigma_z\otimes \one\ket{\phi^+}$, and $\sigma_i$, $i=x,y,z$ are the Pauli matrices. She sends the outcome to Bob, who then applies $\one$, $\sigma_z$, $\sigma_x$, or $\sigma_z\sigma_x$, respectively, depending on the measurement outcome he gets from Alice. As a result, Alice succeeds in sending one qubit to Bob using 2 bits of classical communication and one ebit of entanglement.

In some sense the dual protocol to teleportation is superdense coding \cite{Bennett1992}. In this protocol, Alice can send 2 classical bits using one ebit and one qubit of quantum communication. While teleportation uses classical communication and entanglement to implement an ideal quantum channel, superdense coding does the opposite: it uses quantum communication and entanglement to implement an ideal classical  channel.

The optimality of the error-free protocol for quantum teleportation can be seen by, e.g. assuming the existence of an improved protocol, concatenating it with superdense coding and using the non-locking inequality for the quantum mutual information,
\begin{equation}
	I(A:BC)_\rho\le I(A:B)_\rho+2\log C.
\end{equation}
In Subsection \ref{subs:PBT-bounds}, we will prove a lower bound on the classical communication required for imperfect teleportation.

A family of teleportation tasks can be defined by restricting the set of local operations that Bob can perform on his part of the entangled resource.
\begin{defn}[Generalized \cite{strelchuk2013generalized} and port based teleportation \cite{Ishizaka2008,Ishizaka2009}]
	Let $\hi_A$, $\hi_{A'}$ and $\hi_{B'}$ be finite dimensional Hilbert spaces, and $G\subset \mathcal{CPTP}_{B'\to A}$. An $(|A|, \rho_{A'B'},\varepsilon)$-$G$-teleportation protocol is a protocol, where Alice initially has an input register $\hi_A$ and the $A'$ register of a resource state $\rho_{A'B'}$, while Bob has $B'$. Now Alice performs a measurement on her systems and sends a classical message to Bob. If Bob can now apply a channel $U\in G$ such that the whole protocol simulates the identity channel $\id_{A}$ up to diamond norm error $\varepsilon$, the protocol is successful. The size of the message is called the \emph{communication cost} of the protocol.
	
	The special case of $\hi_{B'}=\hi_B^{\otimes N}$ and $G=\{\tr_{B^{i^c}}|i=1,...,N\}$ is called \emph{port based teleportation} (PBT). The systems $B_i$ are called \emph{ports}. The task of port-teleporting a $d$-dimensional quantum system with error $\varepsilon$ is denoted $(d,\varepsilon)$-port based teleportation.
\end{defn}
 Note that in the case of port based teleportation, the only information that Bob can actually use is which output port to choose. Also if a certain port is never chosen, it can be removed from the resource state. This implies, that the communication cost is equal to the logarithm of the number of ports $c=\log N$. A protocol for port based teleportation with input register $A$ and $N$ ports is completely specified by a resource state $\rho_{A'B^{N}}$ and an $N$-outcome POVM $\{\left(E_i\right)_{A_0 A'}\}_{i=1}^N$. The quantum channel resulting from the protocol is
 \begin{equation}
 	\Lambda_{A_0\to B}(X_A)=\sum_{i=1}^N\tr_{A_0A'B_{i^c}}\left(E_i\right)_{A_0 A'}\left(X_{A_0}\otimes\rho_{A'B^N}\right).
 \end{equation}
 Here, $B_{i^c}=B_1...B_{i-1}B_{i+1}...B_N$ denotes all $B$ systems except the $i$th one, and it is understood that $B_i$ is renamed for $B$ in each summand.
 
A variant of PBT is probabilistic PBT, where the goal is to teleport a state perfectly but allowing for a certain failure probability. This thesis is exclusively concerned with the deterministic task.

port based teleportation has been shown to be achievable for all dimensions and arbitrarily small nonzero error.
\begin{thm}[Achievability of port based teleportation, \cite{Ishizaka2008,Ishizaka2009,Beigi2011}]\label{thm:PBT-Beigi}
	Let $d\in \N$ and $\varepsilon>0$. Then there exists a port based teleportation protocol that achieves an entanglement fidelity of $F=\sqrt{1-\varepsilon^2}$ for $N=\left\lceil\frac{d^2}{\varepsilon^2}\right\rceil$.
\end{thm}
See \cite{Beigi2011} for a proof of this theorem. Note the difference of a square root between our definition of the entanglement fidelity and the one in \cite{Beigi2011}. The explicit protocol this result is based on uses $N$ maximally entangled states as a resource. The expression for the entanglement fidelity of the resulting protocol can be related to a state discrimination problem, which is solved using the pretty good measurement. In the following sections we will prove that the above achievability result can be strengthened to give error bounds in terms of the diamond norm, using the same protocol.

\subsection{The symmetries of port based teleportation}\label{subs:sym-PBT}

A priori, one can imagine that the optimal port based teleportation protocol uses a complicated resource state and a complicated POVM on Alice's side. Intuitively, however, the symmetries of the problem should help: A port based teleportation protocol should work equally well for all input states, so there should be a unitary symmetry in the problem. Also it is not important at which of the $N$ ports Bob receives the message, i.e. one should be able to choose the resource state invariant under permutations. In  this subsection we prove precise statements reflecting these facts.

Let us begin by proving a Lemma about purifications of quantum states with a given symmetry. The special case for the standard representation of $S_N$ on $\left(\C^d\right)^{\otimes N}$ of this lemma was used in \cite{christandl2007one}.
\begin{lem}\label{lem:sympur}
	Let $\rho_A$ be a quantum state that is invariant under a unitary representation $\phi$ of a group $G$. Then there exists a purification $\proj{\rho}_{AA'}$ such that $\ket{\rho}_{AA'}$ is invariant under $\phi\otimes\phi^*$, where $\phi^*$ is the dual representation of $\phi$, i.e. $\phi^*(g)=\overline{\phi(g)}$.
\end{lem}
\begin{proof}
	This lemma follows easily from the mirror lemma. We show that $\ket{\rho}_{AA'}=\rho_{A}^{1/2}\ket{\phi^+}_{AA'}$ is invariant under $\phi\otimes\phi^*$. As $\phi$ is a unitary representation, the invariance of $\rho_{A}$ implies the invariance of $f(\rho)_A$ for all functions $f:[0,1]\to \R$, in particular $\rho_{A}^{1/2}$ is invariant. Therefore we get
	\begin{align}
		\phi(g)_A\otimes\overline{\phi(g)}_{A'}\ket\rho_{AA'}=&\phi(g)_A\otimes\overline{\phi(g)}_{A'}\rho_{A}^{1/2}\ket{\phi^+}_{AA'}\nonumber\\
		=&\phi(g)_A\rho_{A}^{1/2}\phi(g)_{A}^{\dagger}\ket{\phi^+}_{AA'}\nonumber\\
		=&\rho_{A}^{1/2}\ket{\phi^+}_{AA'}\nonumber\\
		=&\ket{\rho}_{AA'}.
	\end{align}
	Here we have used the mirror lemma in the second row and the invariance of $\rho_A^{1/2}$ in the third row.
\end{proof}
%\CM{If there is time, replace this proof and the one about the diamond norm of covariant channels by design-based ones. ref. for existence of designs: \cite{Kane2015}}

We also need the concept of $\varepsilon$-coverings.
\begin{defn}
	Let $(X,d)$ be a metric space. A subset $D\subset X$ is called \emph{$\varepsilon$-covering}, if for all $x\in X$ there exists a $y\in D$ such that $d(x,y)\le \varepsilon$.
\end{defn}
If $X$ is compact as a topological space, there exists a finite $\varepsilon$-cover for every $\varepsilon>0$. An example is the unitary group $\mathrm{U}(\hi_A)$ with the metric induced by the operator norm. In this case, an $\varepsilon$-covering can be used to approximate the Haar twirl.
\begin{lem}\label{lem:apprtwirl}
	Let $D\subset \mathrm{U}(\hi_A)$ be an $\varepsilon$-covering of $\mathrm{U}(\hi_A)$. Let $\nu: \mathrm{U}(\hi_A)\to D$ be a function such that $\|\nu(U)-U\|_\infty\le\|V-U\|_\infty$ for all $V\in D$, $M(V)=\Big\{U\in\mathrm{U}(\hi_A)|\nu(U)=V\Big\}$ and $p(V)=\mu(M(V))$, where $\mu$ is the Haar measure on $\mathrm{U}(\hi_A)$ normalized to one. Then the $N$-twirl
	\begin{equation}
		\mathcal{T}^{(N)}_D(X)=\sum_{V\in D}p(V) V^{\otimes N}X\left(V^\dagger \right)^{\otimes N}
	\end{equation}
	constructed from $D$ is $2N\varepsilon$-close to the Haar $N$-twirl
	\begin{equation}
		\mathcal{T}^{(N)}_{\mathrm{Haar}}(X)=\int_{\mathrm{U}(\hi_A)} U^{\otimes N}X\left(U^\dagger \right)^{\otimes N} \D U
	\end{equation}
	in diamond norm.
\end{lem}
\begin{proof}
	We bound
	\begin{align}
		&\left\|\mathcal{T}^{(N)}_D(\rho_{AB})-\mathcal{T}^{(N)}_{\mathrm{Haar}}(\rho_{AB})\right\|_1\nonumber\\=&\left\|\sum_{V\in D}p(V) V^{\otimes N}X\left(V^\dagger \right)^{\otimes N}-\int_{\mathrm{U}(\hi_A)} U^{\otimes N}X\left(U^\dagger \right)^{\otimes N} \D U\right\|_1\nonumber\\
		=&\left\|\int_{\mathrm{U}(\hi_A)} \left(\nu(U)\rho_{AB}\left(\nu(U)^\dagger \right)^{\otimes N}- U^{\otimes N}X\left(U^\dagger \right)^{\otimes N}\right) \D U\right\|_1\nonumber\\
		\le&\int_{\mathrm{U}(\hi_A)} \left\|\nu(U)\rho_{AB}\left(\nu(U)^\dagger \right)^{\otimes N}- U^{\otimes N}X\left(U^\dagger \right)^{\otimes N}\right\|_1 \D U\nonumber\\
		\le&\int_{\mathrm{U}(\hi_A)} \Bigg(\left\|\left(\nu(U)^{\otimes N}-U^{\otimes N}\right)\rho_{AB}\left(\nu(U)^\dagger \right)^{\otimes N}\right\|_1\nonumber\\
		&+\left\|U^{\otimes N}\rho_{AB}\left(\left(\nu(U)^\dagger \right)^{\otimes N}-\left(U^\dagger \right)^{\otimes N}\right)\right\|_1\Bigg) \D U\nonumber\\
		\le&\int_{\mathrm{U}(\hi_A)} \Bigg(\left\|\nu(U)^{\otimes N}-U^{\otimes N}\right\|_\infty\left\|\rho_{AB}\left(\nu(U)^\dagger \right)^{\otimes N}\right\|_1\nonumber\\
		&+\left\|U^{\otimes N}\rho_{AB}\right\|_1\left\|\left(\nu(U)^\dagger \right)^{\otimes N}-\left(U^\dagger \right)^{\otimes N}\right\|_\infty\Bigg) \D U\nonumber\\
		\le&2N\varepsilon.
	\end{align}
	The first two inequalities are due to the triangle inequality for the trace norm, and the third inequality is Hölder's inequality.
\end{proof}

Now we can prove that any protocol for port based teleportation can be transformed into one that uses a resource state that is invariant under permutations of the ports as well as under the unitary action $U^{\otimes N}\otimes \overline{U}^{\otimes N}$.

\begin{prop}\label{prop:symsuffice}
	Let $\rho_{A'B^N}$ be the resource state of an $(d,\rho_{A'B^N},\varepsilon)$-PBT protocol, where $d=|B|$. Let further $\hi_A\cong\hi_B$. Then there is another protocol performing at least as good as the original one, that uses a resource state $\proj\psi_{{A}^N{B}^N}$ with $\ket{\psi}_{{A}^N{B}^N}\in\bigvee^{N}(\hi_A\otimes\hi_B)$ that is a purification of a symmetric Werner state, i.e. in addition to the $S_N$-invariance it is invariant under $U^{\otimes N}\otimes \overline{U}^{\otimes N}$,
	\begin{align}
	U^{\otimes N}\otimes\overline{U}^{\otimes N}\ket{\psi}_{{A}^N{B}^N}=\ket{\psi}_{{A}^N{B}^N} \ \forall U\in \mathrm{U}(\hi_{A}).
	\end{align}
\end{prop}
\begin{proof}
	Let $\tilde{\rho}_{AB^NIR}$ be a purification of the state
	\begin{equation}
	\tilde{\rho}_{AB^NI}=\sum_{\tau\in S_N}\tau_{B^N}\rho_{AB^N}\tau_{B^N}^\dagger \otimes\proj{\tau}_I,
	\end{equation}
	where $\tau_{B^N}$ is the standard action of $S_N$ on $\hi_B^{\otimes N}$ that permutes the tensor factors. As $\tilde{\rho}_{B^N}$ is permutation invariant, according to Lemma \ref{lem:sympur} there exists another purification $\ket{\psi}\in \bigvee^{N}(\hi_{A}\otimes\hi_B)$. But all purifications are equivalent, therefore the following protocol achieves the same performance than the preexisting one: Alice and Bob start sharing $\ket{\psi}_{A^NB^N}$ as an entangled resource. Alice applies the isometry that creates $\tilde{\rho}_{AB^NIR}$ from $\psi$, then she measures $I$ in the standard basis. Suppose the outcome is $\tau$. Then she goes on to execute the original protocol, except that she applies $\tau$ to the index she is supposed to send to Bob after her measurement, which obviously yields the same result as the original protocol.
	
	Let now $D\subset \mathrm{U}(\hi_B)$ be a $\delta$-net. Using the same technique as for the symmetric group, we begin by defining the state vector
	\begin{equation}
		\ket{\psi'}_{JA^NB^N}=\sum{V\in D}\sqrt{p(V)}\ket{V}_J\otimes V_B^{\otimes N}\ket{\psi}_{A^NB^N},
	\end{equation}
	with $\hi_J=\C D$.
	Let $\ket{\psi''}_{A^NB^N}$ be another purification of $\psi'_{B^N}$. As before, there is a protocol for port based teleportation using the resource state $\proj{\psi''}_{A^NB^N}$. First, Alice applies the isometry that transforms $\psi''_{A^NB^N}$ into $\psi'_{JA^NB^N}$. Then she measures $J$ in the computational basis. On outcome $V$, she applies $V^\dagger $ to her input state and proceeds with the initial protocol. This modified protocol achieves the same error $\varepsilon$.
	
	Let $(\delta_n)_{n\in \N}$ be a sequence such that $\delta_n>0$ for all $n$ and $\lim_{n\to \infty}\delta_n=0$. For each $n$, the above gives a new protocol using a resource state $\proj{\psi^{(n)}}_{A^NB^N}$ such that $\|\psi^{(n)}_{B^N}-T^{(n)}_{\mathrm{Haar}}(\psi_{B^N})\|_1\le 2N\delta_n$. This follows from Lemma \ref{lem:apprtwirl}. Each of these protocols $\mathfrak{P}_n$ is completely specified by the resource state and the $N$-outcome POVM. These objects live in compact subsets of normed finite dimensional vector spaces, so there exists a converging subsequence $\mathfrak{P}_{n_k}$. The limit protocol uses a resource state $\proj{\psi^{(\infty)}}_{A^NB^N}$ such that $\psi^{(\infty)}_{B^N}=T^{(n)}_{\mathrm{Haar}}(\psi_{B^N})$, i.e. in particular $U_B^{\otimes N}\psi^{(\infty)}_{B^N}\left(U_B^\dagger \right)^{\otimes N}=\psi^{(\infty)}_{B^N}$. An application of Lemma \ref{lem:sympur} finishes the proof.
\end{proof}
The same result can also be proven using designs, when using exact designs (see e.g. \cite{Seymour1984,Kane2015} for existence results) it is not necessary to consider the limiting procedure.
Symmetric Werner states, which the present proposition shows to be sufficient as resource states for port based teleportation, can be explicitly parameterized. Using Schur-Weyl duality for the $A^N$ and $B^N$ marginals and the $U^{\otimes N}\otimes\overline{U}^{\otimes N}$ invariance, any such resource state vector $\ket\psi\in\bigvee^n\C^d$ can be written as
\begin{equation}
	\sum_{\lambda\vdash(n,d)}z_\lambda \ket{\phi^+}_{[\lambda]_A[\lambda]_B}\otimes \ket{\phi^+}_{(V_\lambda)_A(V_\lambda)_B}
\end{equation}
which is to be understood in the context of the decomposition 
 \begin{align}
 \left(\left(\C^d\right)^{\otimes N}\right)^{\otimes 2}\cong& \left(\bigoplus_{\Lambda\vdash(n,d)}[\lambda]\otimes V_\lambda\right)^{\otimes 2}\nonumber\\
 \cong&\bigoplus_{\Lambda,\mu\vdash(n,d)}\left([\lambda]_A\otimes [\mu]_B\right)\otimes \left(\left(V_\lambda\right)_A\otimes\left(V_\mu\right)_B\right).
 \end{align}
 The numbers $z_\lambda\in\C$ fulfil the normalization condition $\sum_{\lambda\vdash(n,d)}\left|z_\lambda\right|^2=1$.

Having proven that the resource state can be assumed to be symmetric, we can go on to prove that the POVM elements can be taken to have a number of symmetries as well.
\begin{prop}\label{prop:symPOVM}
	Let $\{\left(E_i\right)_{A_0A^N}\}_{i=1}^N$ be the POVM for a $(|A|, \proj\psi_{A^NB^N},\varepsilon)$-port based teleportation protocol where $\psi$ has the symmetries from Theorem \ref{prop:symsuffice}. Then there is another POVM $\{\left(E'_i\right)_{A_0A^N}\}_{i=1}^N$ achieving the same diamond norm error $\varepsilon$ as the original POVM $\{\left(E_i\right)_{A^NB^N}\}_{i=1}^N$, such that 
	\begin{equation}
		\tau_{A^N}\left(E'_i\right)_{A_0A^N}\tau_{A^N}^\dagger =\left(E'_{\tau(i)}\right)_{A_0A^N}
	\end{equation}
	for all $\tau\in S_N$, and
	\begin{equation}
		\overline{U}_{A_0}\otimes U_{A}^{\otimes N}\left(E'_i\right)_{A_0A^N}U_{A_0}^T\otimes \left(U_{A}^{\dagger}\right)^{\otimes N}=\left(E'_i\right)_{A_0A^N}.
	\end{equation}
	In addition, the resulting channel $\Lambda'$ is unitarily covariant, i.e.
	\begin{equation}
		\Lambda'_{A_0\to B}(X)=U_B^\dagger \Lambda'_{A_0\to B}(U_{A_0}XU_{A_0}^\dagger )U_B.
	\end{equation}
\end{prop}
\begin{proof}
	Because of the permutation invariance of the resource state, for any permutation $\tau\in S_N$, the channel resulting from the protocol with the modified POVM $\left\{\tau_{A^N}\left(E_{\tau^{-1}(i)}\right)_{A_0A^N}\tau_{A^N}^\dagger \right\}_{i=1}^N$ is equal to the one resulting from the original protocol. Therefore Alice can just as well pick a permutation at random and use the corresponding POVM. 
	
	Now define the modified POVM $\left\{\left(\overline{U}_{A_0}\otimes U_{A}^{\otimes N}\right)\left(E_i\right)_{A_0A^N}\left(U_{A_0}^T\otimes \left(U_{A}^{\dagger}\right)^{\otimes N}\right)\right\}_{i=1}^N$, and denote the resulting channel by $\Lambda^U$. We calculate
	\begin{align}
		&\Lambda^U_{A_0\to B}(X_A)\nonumber\\
		=&\sum_{i=1}^N\tr_{A_0A'B_{i^c}}\left(\overline{U}_{A_0}\otimes U_{A}^{\otimes N}\right)\left(E_i\right)_{A_0A^N}\left(U_{A_0}^T\otimes \left(U_{A}^{\dagger}\right)^{\otimes N}\right)\left(X_{A_0}\otimes\proj\psi_{A^NB^N}\right)\nonumber\\
		=&\sum_{i=1}^N\tr_{A_0A'B_{i^c}}\left(E_i\right)_{A_0A^N}\left(\left(U_{A_0}^TX_{A_0}\overline{U}_{A_0}\right)\otimes\left(\left(U_{A}^{\dagger}\right)^{\otimes N}\proj\psi_{A^NB^N} U_{A}^{\otimes N}\right)\right)\nonumber\\
		=&\sum_{i=1}^N\tr_{A_0A'B_{i^c}}\left(E_i\right)_{A_0A^N}\left(\left(U_{A_0}^TX_{A_0}\overline{U}_{A_0}\right)\otimes\left(\overline U_{B}^{\otimes N}\proj\psi_{A^NB^N} \left(U^T_{B}\right)^{\otimes N}\right)\right)\nonumber\\
		=&\sum_{i=1}^N\tr_{A_0A'B_{i^c}}\left(E_i\right)_{A_0A^N}\left(\left(U_{A_0}^TX_{A_0}\overline{U}_{A_0}\right)\otimes\left(\overline U_{B_i}\proj\psi_{A^NB^N} U^T_{B_i}\right)\right)\nonumber\\
		=&\overline U_{B}\Lambda\left(U_{A_0}^TX_{A_0}\overline{U}_{A_0}\right)U^T_{B}.
	\end{align}
	Here we used the cyclicity of the trace in the second line, the $U^{\otimes N}\otimes\overline{U}^{\otimes N}$-invariance of $\proj\psi_{A^NB^N}$ in the third line, and the cyclicity of the trace again in the fourth line. The diamond norm is unitarily invariant, and the identity commutes with unitary conjugation, i.e. the above calculation shows that
	\begin{equation}
		\left\|\Lambda^U_{A_0\to B}-\id_{A_0\to B}\right\|_\diamond=\left\|\Lambda_{A_0\to B}-\id_{A_0\to B}\right\|_\diamond.
	\end{equation}
	Therefore we define
	\begin{equation}
		\left(E'_i\right)_{A_0A^N}=\int_{\mathrm{U}(\hi_A)}\sum_{\tau\in S_N}\left(\overline{U}_{A_0}\otimes U_{A}^{\otimes N}\right)\tau_{A^N}\left(E_{\tau^{-1}(i)}\right)_{A_0A^N}\tau_{A^N}^\dagger \left(U_{A_0}^T\otimes \left(U_{A}^{\dagger}\right)^{\otimes N}\right)\D U.
	\end{equation}
	The channel resulting from the protocol using this POVM is $\int_{\mathrm{U}(\hi_A)}\Lambda^U_{A_0\to B}\D U$. An application of the triangle inequality finishes the proof,
	\begin{equation}
			\left\|\int_{\mathrm{U}(\hi_A)}\Lambda^U_{A_0\to B}\D U-\id_{A_0\to B}\right\|_\diamond\le	\int_{\mathrm{U}(\hi_A)}\left\|\Lambda^U_{A_0\to B}-\id_{A_0\to B}\right\|_\diamond\D U.
	\end{equation}
\end{proof}
The maximizer for the diamond norm distance of unitarily covariant channels is the maximally entangled state.

\begin{lem}
	Let $\Lambda^{(i)}_{A\to A}$, $i=1,2$ be unitarily covariant maps. Then
	\begin{equation}
		\left\|\Lambda^{(1)}_{A\to A}-\Lambda^{(2)}_{A\to A}\right\|_\diamond=\left\|(\Lambda^{(1)}_{A\to A}-\Lambda^{(2)}_{A\to A})(\proj{\phi^+}_{AA'})\right\|_1.
	\end{equation}
\end{lem}
\begin{proof}
	Let $\ket\psi_{AA'}$ be a state such that
	\begin{equation}
	\left\|\Lambda^{(1)}_{A\to A}-\Lambda^{(2)}_{A\to A}\right\|_\diamond=\left\|(\Lambda^{(1)}_{A\to A}-\Lambda^{(2)}_{A\to A})(\proj{\psi}_{AA'})\right\|_1.
	\end{equation}
	Let furthermore $P$ be the analogue of the Pauli group in $|A|$ dimensions that is generated by the operators from Equation \eqref{eq:genpauli}. It is easy to check that this group is finite, like the Pauli group. Define a mixed state  
	\begin{equation}
		\rho_{AA'I}=\sum_{V\in P} p(V) V_A\proj{\psi}_{AA'}V_A^\dagger \otimes \proj{V}_I.
	\end{equation}
	$\rho$ achieves the diamond norm distance as well, as
	\begin{align}
		\left\|(\Lambda^{(1)}_{A\to A}-\Lambda^{(2)}_{A\to A})(\rho_{AA'I})\right\|_1=&\sum_{V\in D}p(V)\left\|(\Lambda^{(1)}_{A\to A}-\Lambda^{(2)}_{A\to A})( V_A\proj{\psi}_{AA'}V_A^\dagger )\right\|_1\nonumber\\
		=&\sum_{V\in D}p(V)\left\|V_A(\Lambda^{(1)}_{A\to A}-\Lambda^{(2)}_{A\to A})( \proj{\psi}_{AA'})V_A^\dagger \right\|_1\nonumber\\
		=&\left\|(\Lambda^{(1)}_{A\to A}-\Lambda^{(2)}_{A\to A})(\proj{\psi}_{AA'})\right\|_1\nonumber\\
		=&\left\|\Lambda^{(1)}_{A\to A}-\Lambda^{(2)}_{A\to A}\right\|_\diamond.
	\end{align}
	Here we used the unitary covariance of the channels in the second line and the unitary invariance of the trace norm in the third line. Let now $\proj{\rho}_{AA'IJ}$ be a purification of $\rho_{AA'I}$. $\rho_A=\tau_A$, and all purifications are isometrically equivalent, so there exists an isometry $W_{A'\to A'IJ}$ such that $\rho=W\proj{\phi^+}_{AA'} W^\dagger $. By the monotonicity of the trace distance under partial trace and its invariance under isometries we therefore get
		\begin{equation}
			\left\|\Lambda^{(1)}_{A\to A}-\Lambda^{(2)}_{A\to A}\right\|_\diamond=\left\|(\Lambda^{(1)}_{A\to A}-\Lambda^{(2)}_{A\to A})(\proj{\phi^+}_{AA'})\right\|_1.
		\end{equation}
\end{proof}

Together the above Lemma and the two propositions imply that Theorem \ref{thm:PBT-Beigi} holds for diamond norm error $\varepsilon$ as well:
\begin{thm}
	Let $d\in \N$ and $\varepsilon>0$. Then there exists a $(|A|,\proj{\phi^+}_A^{\otimes N},\varepsilon)$-port based teleportation protocol for $N=\left\lceil\frac{4d^2}{\varepsilon^2}\right\rceil$.
\end{thm}
\begin{proof}
	According to the proofs of Propositions \ref{prop:symsuffice} and \ref{prop:symPOVM}, the protocol achieving Theorem \ref{thm:PBT-Beigi} can be transformed into one resulting in a unitarily covariant channel $\Lambda$, keeping the same error and the same number of ports. It is easy to see that these constructions achieve the same when the error is measured in terms of the entanglement fidelity. But the identity is unitarily covariant as well, therefore we get
	\begin{align}
		\|\Lambda-\id\|_\diamond=&\left\|(\Lambda-\id)(\proj{\phi^+})\right\|_1\nonumber\\
	\le&2\sqrt{1-F(\Lambda)^2}\nonumber\\
	\le&2\varepsilon,
	\end{align}
	where the first inequality is a Fuchs van de Graaf inequality. Replacing $\varepsilon$ by $\varepsilon/2$ yields the claimed result.
\end{proof}
In fact, a closer look at the protocol used to prove Theorem \ref{thm:PBT-Beigi} in \cite{Ishizaka2008,Ishizaka2009,Beigi2011}, reveals that it already has all the necessary symmetries, i.e. the protocol based on maximally entangled states and the pretty good measurement itself is a $(|A|,\proj{\phi^+}_A^{\otimes N},\varepsilon)$-port based teleportation protocol for $N=\left\lceil\frac{d^2}{4\varepsilon^2}\right\rceil$.

\subsection{Approximate universal programmable quantum processors}

In this subsection, we first introduce the concept of programmable quantum processors (PQPs) \cite{Nielsen1997} and describe how they can be built from a port based teleportation scheme as done in \cite{Ishizaka2008}. We go on to improve the known lower bounds on the size of the program register necessary for an $\varepsilon$-approximate programmable quantum processor.

An errorless universal programmable quantum processor for a quantum system $D$ is a unitary $G\in\mathrm{U}(\hi_D\otimes \hi_P)$ for some Hilbert space $\hi_P$ such that for all $U\in\mathrm{U}(\hi_D)$ there exists a state vector $\ket U_P\in\hi_P$ such that $G\ket\psi_D\otimes \ket{U}_P=\left(U_D\ket\psi_D\right)\otimes\ket{U'}_P$ for some state vector $\ket{U'}_P$. More formally, we make the following

\begin{defn}\label{defn:eps-pqp}
	Let $\hi_D$, $\hi_P$ be Hilbert spaces. A unitary $G\in\mathrm{U}(\hi_D\otimes \hi_P)$ is called an $\varepsilon$-\emph{universal programmable quantum processor} ($\varepsilon$-uPQP), if the following holds:
	
	For each $U\in\mathrm{U}(\hi_D)$ there exists a state $\ket{U}_P\in\hi_P$ such that 
	\begin{equation}
		\|\tr_PG((\cdot)_{D}\otimes\proj{U}_P)G^\dagger -U(\cdot)_DU^\dagger \|_\diamond\le\varepsilon.
	\end{equation}
	The register $D$ is called the \emph{data register}, the register $P$ is called the \emph{program register}. If the above condition is only true for a subset $S\subset\mathrm{U}(\hi_D)$, $G$ is called an $\varepsilon$-\emph{$S$-programmable quantum processor}
\end{defn}
Note that the restriction to unitary processors is without loss of generality: the Stinespring dilating register of any CPTP processor can be included in the program register. Also it is easy to see by a standard argument that the maximum fidelity is reached for \emph{pure} program states, therefor this does not limit generality either.

Nielsen and Chuang prove in \cite{Nielsen1997} that a perfect universal programmable quantum processor is impossible using a finite-dimensional program register. This is essentially a consequence of the linearity of quantum mechanics and the fact that the set of possible programs, i.e. the projective unitary group $\mathrm{PU}(\hi_D)=\mathrm{U}(\hi_D)/\mathrm{U(1)}$, is infinite. Let us review the impossibility proof by Nielsen and Chuang, both because it is beautifully simple and because the bounds on the program register size in the approximate case follow from the same reasoning.
\begin{thm} [Nielsen and Chuang, \cite{Nielsen1997}]
	A universal zero error PQP with finite dimensional program register does not exist.
\end{thm}
\begin{proof}
	If $G\in\mathrm{U}(\hi_D\otimes \hi_P)$ is a zero error uPQP, then we have that for all $U_D\in\mathrm{U}(\hi_D)$ there exists a state $\ket{U}_P\in\hi_P$ such that for all $\ket{\psi}_D\in\hi_D$ there exists $\ket{U'_\psi}_P\in\hi_P$ with
	\begin{equation}\label{eq:zero-err-proc}
	G\ket{\psi}_D\otimes\ket{U}_P=U_D\ket{\psi}_D\otimes\ket{U'_\psi}_P.
	\end{equation}
	Taking the inner product between Equation \eqref{eq:zero-err-proc} for states $\ket{\psi}_D$ and $\ket{\phi}_D$ with $\bracket{\psi}{\phi}\neq 0$ we get $\ket{U'_\psi}_P=\ket{U'_\phi}=:\ket{U'}$, i.e. the resulting state on the program register only depends on the program. Taking the inner product between Equation \eqref{eq:zero-err-proc} for state $\ket{\psi}$ and unitaries $U,V\in\mathrm{U}(\hi_D)$ we get
	\begin{equation}
	\bracket{U}{V}_P=\bra{\psi}U^\dagger _DV_D\ket{\psi}_D\bracket{U'}{V'}_P.
	\end{equation}
	Supposing $\bracket{U'}{V'}_P\neq 0$ this implies $U^\dagger _DV_D=z\mathds 1$ for some $z\in \C$, as the left hand side does not depend on $\ket{\psi}$. This shows that if $U$ and $V$ differ by more than a global phase, then we have $\bracket{U}{V}_P=0$. As the projective unitary group is infinite for $|D|\ge 2$, $|\mathrm{U}(\hi_D)\slash \mathrm U(1)|=\infty$, any program register for a universal PQP has to be infinite-dimensional. 
\end{proof}

Any port based teleportation scheme can be used to construct a universal programmable quantum processor.

\begin{prop}[\cite{Ishizaka2008}]\label{prop:PBT2PQP}
	Any $(d,\proj\psi_{A^NB^N},\varepsilon)$-port based teleportation scheme can be used to construct a $\varepsilon$-uPQP with program register dimension $|P|=\binom{N+d^2-1}{N}$.
\end{prop}
\begin{proof}
	Let us assume we have a PBT protocol that uses an entangled state $\ket\psi\in\hi_{A'B'}$, where $\hi_{X'}=\hi_X^{\otimes N}$ for $X=A,B$, and let $G\in \mathrm{U}(\hi_{AA'BB'})$ be the unitary that executes the Stinespring dilation of the entire PBT protocol. In particular the final partial trace is just not applied and instead the output port is swapped with the first port. Then $G$ commutes with $\left(U^{\otimes N}\right)_{B'}$ for all $U\in\mathrm{U}(\hi_B)$. Setting $\hi_P=\hi_{A'B'}$ and $\ket{U}=\left(U^{\otimes N}\right)_{B'}\ket\psi$ we see that $G$ acts as an $\varepsilon$-uPQP. According to Proposition \ref{prop:symsuffice} the program register can be chosen to have dimension $|P|=\dim\bigvee^{N}(\hi_{A'}\otimes\hi_{B'})=\binom{N+d^2-1}{N}$.
\end{proof}

This connection between PBT and uPQPs is how Ishizaka and Hiroshima showed in \cite{Ishizaka2008} that perfect port based teleportation is impossible with finite resources. Therefore, a lower bound on the size of the program register of an $\varepsilon$-uPQP yields a lower bound on the number of ports necessary for $\varepsilon$-PBT.

Hillery et al. derived such a bound \cite{Hillery2006}. They use, however, the entanglement fidelity as an error measure for their definition of an approximate PQP. The entanglement fidelity is an average case error measure, while for a processor arguably a worst-case error measure should be applied to ensure composability. Their result is the following:
\begin{thm}[Hillery, Ziman and Bu\v zek \cite{Hillery2006}]\label{thm:appr-proc}
	Let $G\in\mathrm{U}(\hi_D\otimes \hi_P)$ be a PQP capable of executing a set of unitaries $M\subset \mathrm U(\hi_D)$ with error $\varepsilon$ in the following sense: For all $U\in M$ the entanglement fidelity is greater than\footnote{The square root comes from the fact, that in \cite{Hillery2006} a different definition of the fidelity is used.} $\sqrt{1-\varepsilon^2}$, i.e.  $P(\tr_PG\proj{\phi^+}_{DE}\otimes\proj{U}_PG^\dagger ,U_D \proj{\phi^+}_{DE}U_{D}^\dagger )\le \varepsilon$. Then the overlap of the program states satisfies
	\begin{equation}\label{eq:hillery-bound}
	\bracket{U}{V}_P\le \min\left(1, \frac{\varepsilon^2 d+\varepsilon \sqrt{d}}{\eta}\right)\frac{\left|\tr U^\dagger  V\right|}{d}+2\varepsilon+\varepsilon^2,
	\end{equation}
	where $\eta=1-\min_\psi|\bra{\psi}UV^\dagger \ket{\psi}|^2$ and $d=|D|$.
\end{thm}

Let us prove a lemma that relates the quantity $\eta$ in the above theorem to the operator norm distance on  the projective unitary group $\mathrm{PU}(\hi_D)=\mathrm{U}(\hi_D)/\mathrm{U(1)}$.

\begin{figure}
	\begin{center}
	\begin{minipage}{.5\textwidth}
		\includegraphics[width=\textwidth]{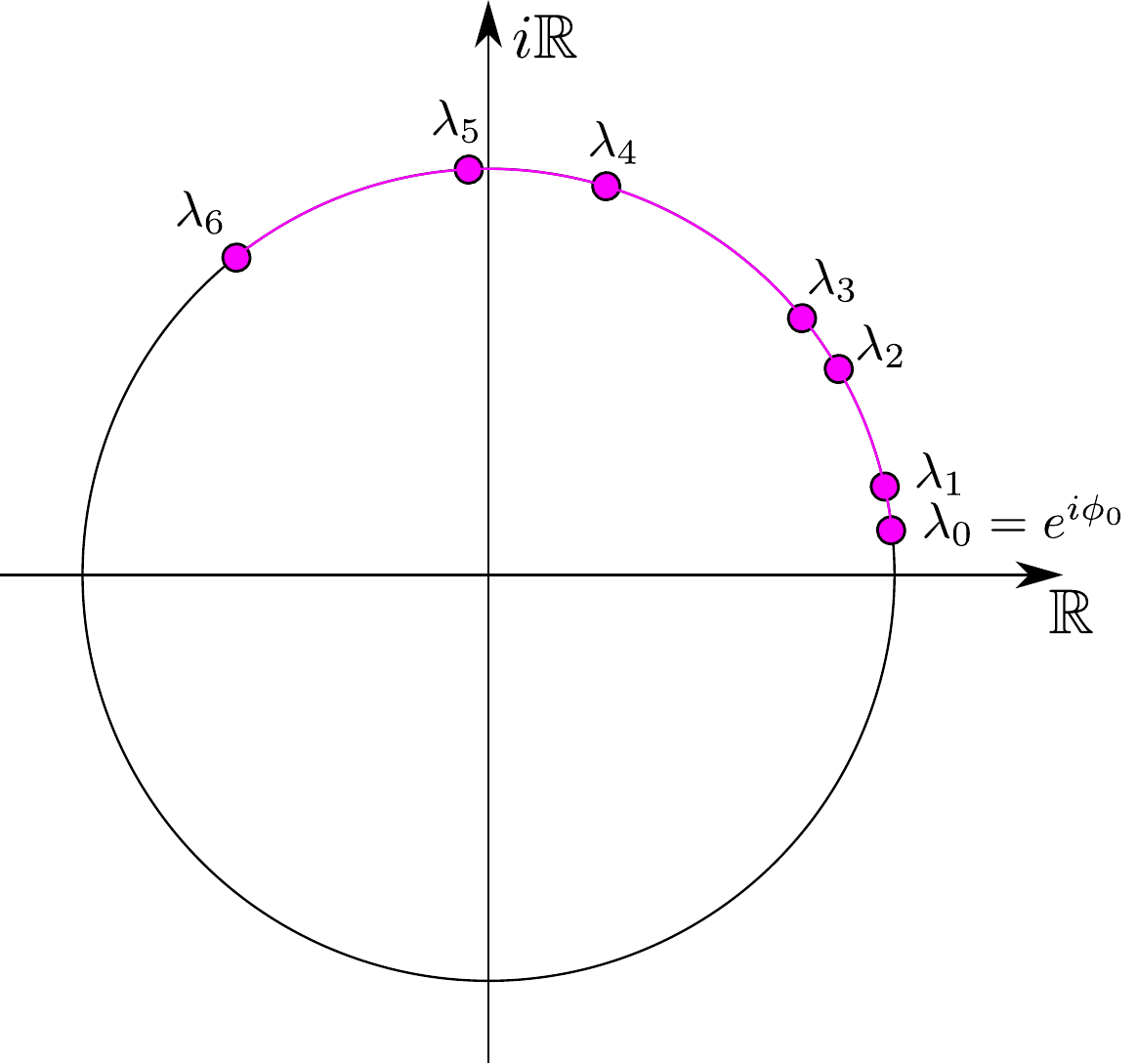}
		\caption*{a)}
		%\label{fig:prob1_6_2}
	\end{minipage}%
	\begin{minipage}{.5\textwidth}
		\includegraphics[width=\textwidth]{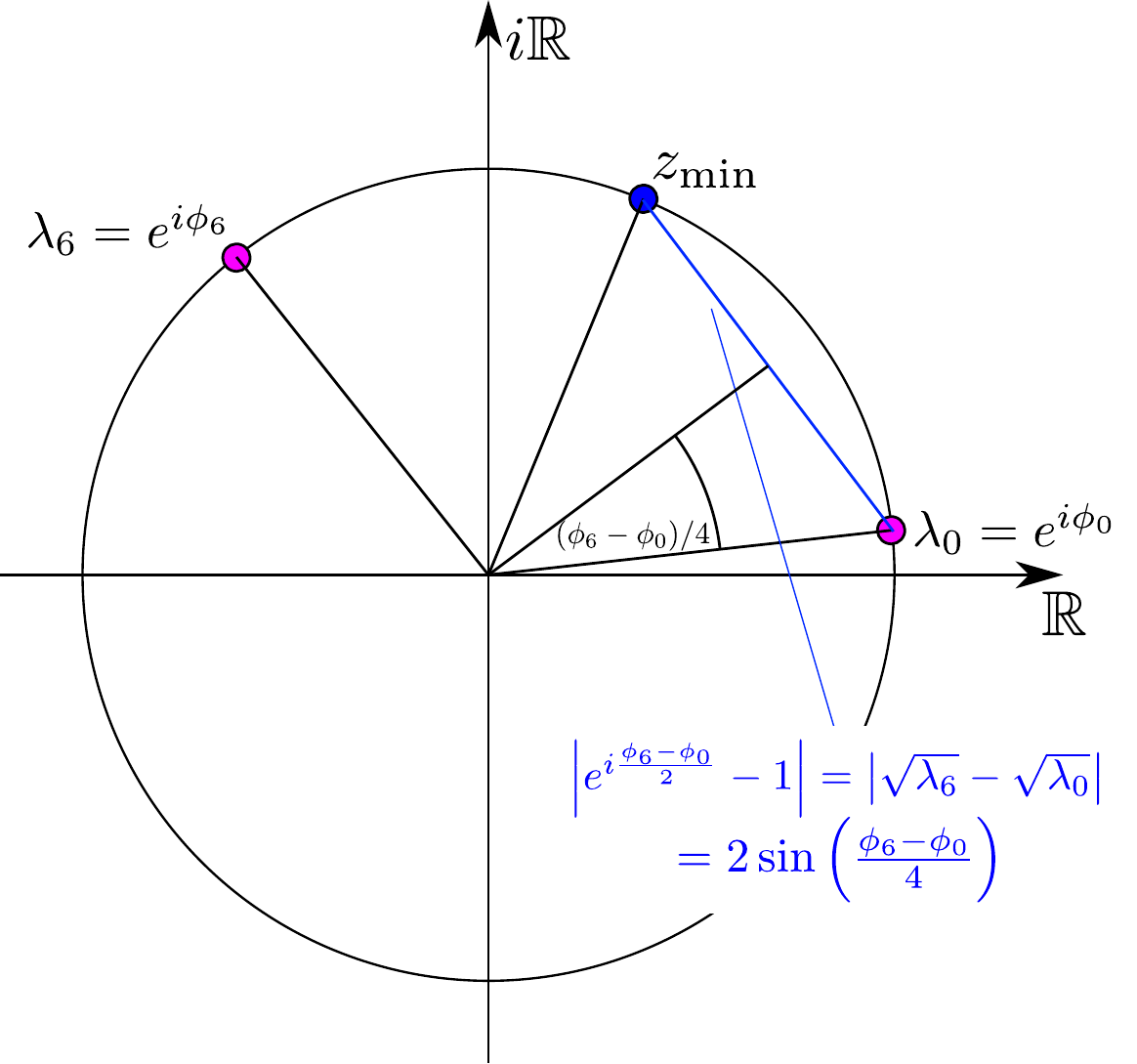}
		\caption*{b)}
		%\label{fig:prob1_6_1}
	\end{minipage}\\
	\begin{minipage}{.5\textwidth}
		\includegraphics[width=\textwidth]{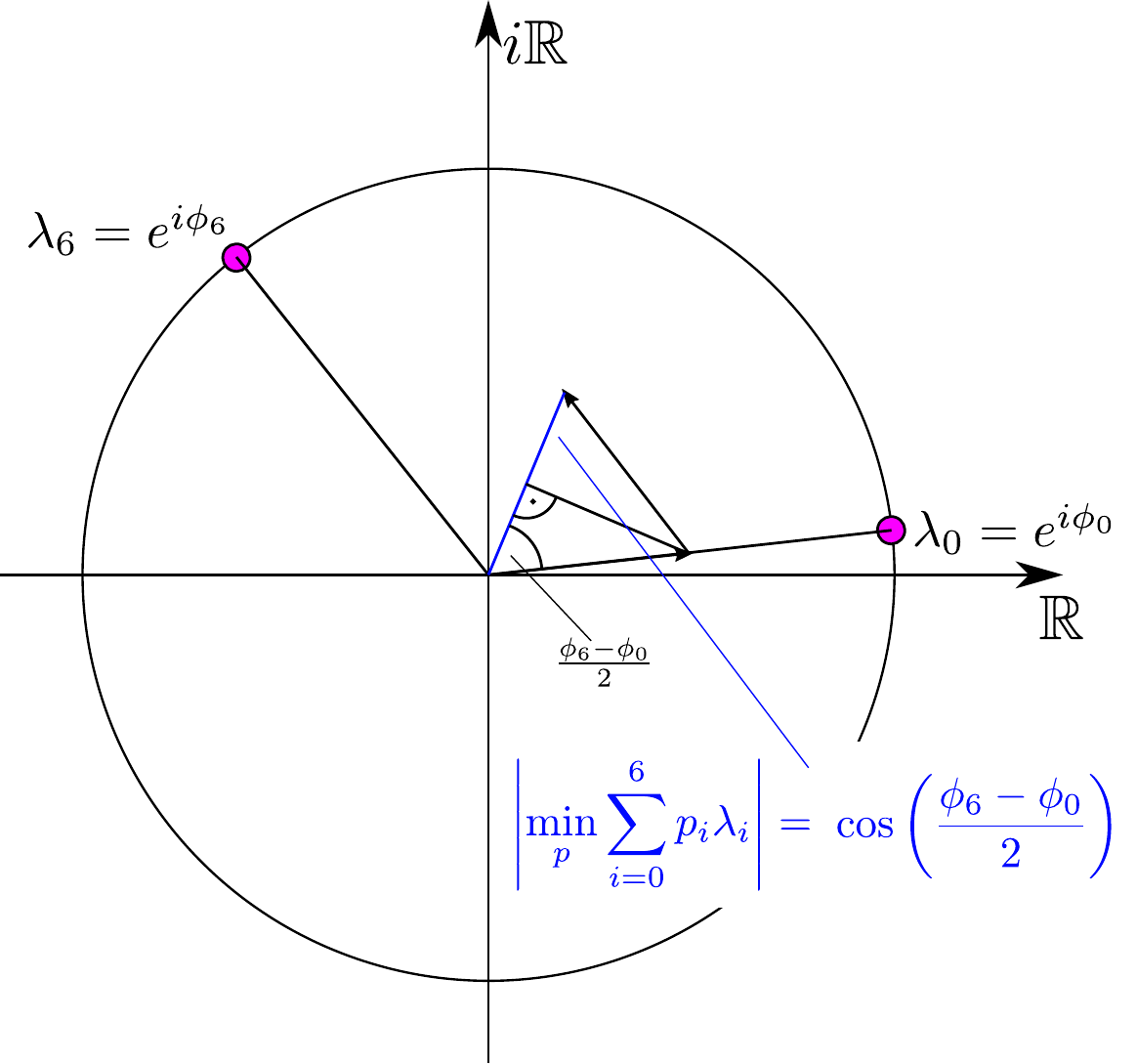}
		\caption*{c)}
		%\label{fig:prob1_6_1}
	\end{minipage}
\end{center}
	\caption{Example of the geometric considerations from the proof of Lemma \ref{lem:hillery-to-opnorm}. a) The purple arc is the minimal arc that contains all eigenvalues $\lambda_i$ of $UV^\dagger$. in other words, $\lambda_0$ and $\lambda_{d-1}$ (which is $\lambda_6$ in the example) are the pair of eigenvalues with the largest distance. The eigenvalues are ordered anticlockwise along the unit circle. We can take the corresponding angles $\phi_i$ to be in $[0,2\pi)$ without loss of generality, as we are dealing with the projective unitary group. b) The minimizer of the minimization in \eqref{eq:opnormdist} is clearly the point on the unit circle in the middle between $\lambda_0$ and $\lambda_{d-1}$. c) For creating the shortest vector by means of convex combination, one has to take the "most diametral" pair of vecors and add half of each. }\label{fig:geom}
\end{figure}

\begin{lem}\label{lem:hillery-to-opnorm}
	For $U,V\in\mathrm{U}(d)$ we have
	\begin{equation}\label{eq:ineq-hillery2opnorm}
	\min_{z,|z|=1}\left\|U-zV\right\|_\infty\ge\sqrt{2\left(1-\min_\psi|\bra{\psi}UV^\dagger \ket{\psi}|\right)}=\sqrt{2\left(1-\sqrt{1-\eta}\right)}
	\end{equation}
	with equality whenever the origin is not in the convex hull of the eigenvalues of $UV^\dagger $ in $\C$.
\end{lem}
\begin{proof}
	Let $UV^\dagger =W\diag(\lambda_0,...,\lambda_{d-1})W^\dagger $ and let the eigenvalues $\lambda_i$ be ordered along the unit circle in a way that the arc from $\lambda_0$ to $\lambda_{d-1}$ containing the eigenvalues has minimal length. Then we have
	\begin{equation}
	\left\|U-zV\right\|_\infty=\left\|UV^\dagger -z\mathds{1}\right\|_\infty,
	\end{equation}
	and
	\begin{equation}
	UV^\dagger -z\mathds{1}=W\diag(\lambda_i-z)W^\dagger .
	\end{equation}
	The $\lambda_i$ are roots of unity, i.e. $\lambda_i=e^{i \phi_i}$ for some $\phi_i\in\R$, and W.L.O.G. $\phi_i\ge\phi_{i+1}$. It follows from basic geometric considerations (see Figure \ref{fig:geom}) that
	\begin{eqnarray}
	\min_{z,|z|=1}\left\|U-zV\right\|_\infty^2&=&\min_{z,|z|=1}\max_i\left|\lambda_i-z\right|^2\nonumber\\
	&=&\left|\sqrt{\lambda_0}-\sqrt{\lambda_{d-1}}\right|^2\nonumber\\
	&=&\left|e^{i(\phi_0-\phi_{d-1})/2}-1\right|^2\nonumber\\
	&=&2-2\cos((\phi_0-\phi_{d-1})/2).\label{eq:opnormdist}
	\end{eqnarray}
	Now consider the other expression. Using the eigendecomposition $UV^\dagger =\sum_ie^{i\phi_i}\proj{\psi_i}$ we get
	\begin{equation}
	\min_\psi|\bra{\psi}UV^\dagger \ket{\psi}|^2=\min_p\left|\sum_ip_ie^{i\phi_i}\right|^2,
	\end{equation}
	where the minimum on the right hand side is taken over all probability distributions on $\{0,...,d-1\}$. If $0$ is in the convex hull of the $e^{i\phi_i}$, the minimum is $0$. Otherwise, it again follows from basic geometric considerations (see Figure \ref{fig:geom}) that the optimal probability distribution is $p_i=(\delta_{i0}+\delta_{i(d-1)})/2$, i.e.
	\begin{eqnarray}
	\min_\psi|\bra{\psi}UV^\dagger \ket{\psi}|^2&=&\frac 1 4\left|e^{i\phi_0}+e^{i\phi_{d-1}}\right|^2\nonumber\\
	&=&\cos^2((\phi_0-\phi_{d-1}/2).
	\end{eqnarray}
	If $0$ is in the convex hull of the $e^{i\phi_i}$, then the right hand side of the inequality \eqref{eq:ineq-hillery2opnorm} is equal to $\sqrt 2$, and the left hand side is not less than that according to Equation \eqref{eq:opnormdist}.
\end{proof}

	It is easy to prove that $\|[U]-[V]\|_\infty:=\min_{z,|z|=1}\left\|U-zV\right\|_\infty$ defines a metric on $\mathrm{PU}(\hi_D)$. $\mathrm{PU}(\hi_D)$ is isomorphic to $SU(d)$, and this metric is invariant. Note that in general $\|[U]-[V]\|_\infty<\|U\det(U)^{-1/d}-V\det(V)^{-1/d}\|_\infty$, i.e. $z_{opt}\neq\det(U^\dagger V)^{1/d}$. However, $2\|[U]-[V]\|_\infty\ge\|U\det(U)^{-1/d}-V\det(V)^{-1/d}\|_\infty$. One can define a metric constructed from any Schatten p-norm in the same way, and $\frac{\left|\tr U^\dagger V\right|}{d}=1-\frac{\left\|[U]-[V]\right\|_2^2}{2d}$. Expressed in these terms, whenever $0$ does not lie in the convex hull of the eigenvalues of $UV^\dagger $, the bound of Theorem \ref{thm:appr-proc} reads
\begin{eqnarray}
\bracket{U}{V}_P&\le& \min\left(1, \frac{\varepsilon d+\sqrt{\varepsilon d}}{\gamma(1-\gamma/4)}\right)\left(1-\frac{\left\|[U]-[V]\right\|_2^2}{2d}\right)+2\sqrt{\varepsilon}+\varepsilon\nonumber\\
&\le& \frac{\varepsilon d+\sqrt{\varepsilon d}}{\zeta(U,V)}+2\sqrt{\varepsilon}+\varepsilon, \text{with}\\
\zeta(U,V)&=&\frac{\|[U]-[V]\|^2_\infty(1-\|[U]-[V]\|^2_\infty/4)}{1-\frac{\left\|[U]-[V]\right\|_2^2}{2d}}\label{eq:zeta}
\end{eqnarray}
where $\gamma =\|[U]-[V]\|_\infty^2$

We want to use the dimension bounds on the program registers of approximate universal PQPs to get good bounds on the communication cost of PBT. To this extent we need a clean bound of the inner product in terms of some metric on $\mathrm{PU}(\hi_D)$ that is  equivalent to the Euclidean metric on the underlying space $\End{\hi_D}$. Then we can express the asymptotic size of maximal subsets of $\mathrm{PU}(\hi_D)$ with a certain minimal pairwise distance in terms of the real dimension of the group, $\dim_\R(\mathrm{PU}(\hi_D))=d^2-1$. The function $\zeta(U,V)$ from Equation \ref{eq:zeta}, however, is not a metric. Bounding it by the operator norm induced metric yields
\begin{equation}
\zeta(U,V)\ge\frac{1}{2}\|[U]-[V]\|^2_\infty,
\end{equation}
and thereby
\begin{equation}\label{eq:hillerycor}
\bracket{U}{V}_P\le \frac{2(\varepsilon d+\sqrt{\varepsilon d})}{\|[U]-[V]\|^2_\infty}+2\sqrt{\varepsilon}+\varepsilon
\end{equation}
whenever $0$ does not lie in the convex hull of the eigenvalues of $UV^\dagger $.

Let us forget about our application for a second: Suppose we are interested in an approximate PQP, most generally able to implement unitaries from a subset $S\subset\mathrm U(\hi_D)$. If we allow an error $\varepsilon$ in diamond norm, then it might make sense to try to implement a $\varepsilon/2$-PQP for a diamond norm $\varepsilon/2$-net in $S$ when having a bound similar to Equation \eqref{eq:hillery-bound} in mind. Therefore it would be nice to have a clean bound referring to a metric on $\mathrm{PU}(\hi_D)$ that can be related to the diamond norm in this case as well.

 There exists a closed expression for the diamond norm difference of two unitaries:
\begin{thm}[Johnston et al., \cite{johnston2009computing}, Theorem 26]\label{thm:diamond-unitaries}
	Let $U\in\mathrm{U}(\C^d)$, and define $\Phi(X)=X-UXU^\dagger $. Then $\|\Phi\|_\diamond$ is equal to the diameter of the smallest disk that contains all eigenvalues of $U$, i.e. if the convex hull of the eigenvalues $\lambda_0,...,\lambda_{d-1}$ of $U$ does not contain the origin, then we have
	\begin{equation}
	\|\Phi\|_\diamond=|\lambda_0-\lambda_{d-1}|,
	\end{equation}
	and $\|\Phi\|_\diamond=2$ otherwise, where we assumed in the first case that the eigenvalues of $U$ are ordered along the unit circle such that the arc from $\lambda_0$ to $\lambda_{d-1}$ containing all the eigenvalues has minimal length.
\end{thm}
See figure \ref{fig:johnston} for an example of the geometry of the eigenvalues.
\begin{figure}
	\begin{center}
		\begin{minipage}{.5\textwidth}
			\includegraphics[width=\textwidth]{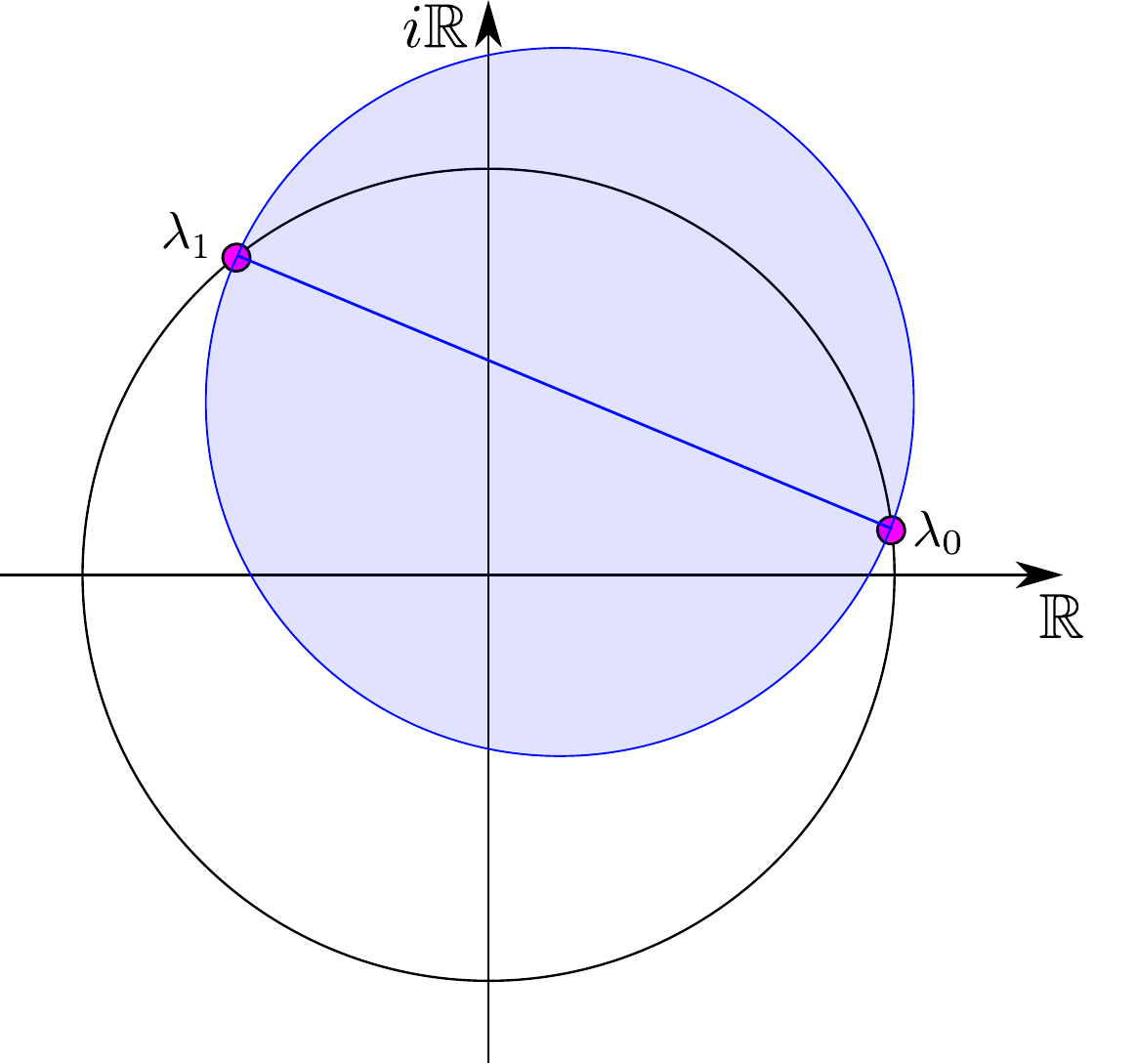}
			%\caption*{a)}
			%\label{fig:prob1_6_2}
		\end{minipage}%
		\begin{minipage}{.5\textwidth}
			\includegraphics[width=\textwidth]{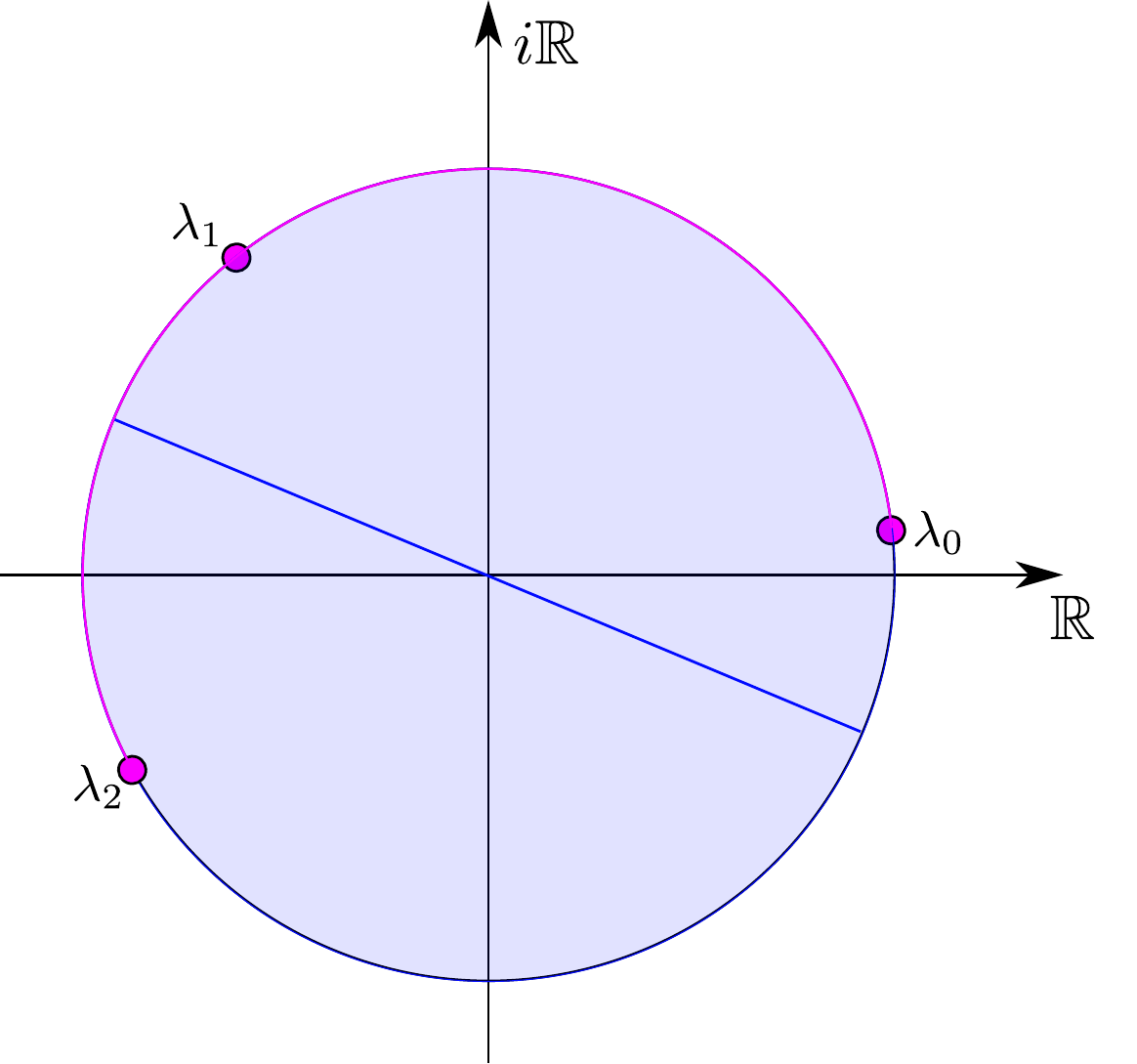}
			%\caption*{b)}
			%\label{fig:prob1_6_1}
		\end{minipage}\\
	\end{center}
	\caption{The two cases for the diamond norm distance of two unitary channels. If and only if the convex hull of the eigenvalues contains the origin, then the smallest disk containing all eigenvalues is the unit disc. }\label{fig:johnston}
\end{figure}
It follows as an easy corollary from the above and the proof of Lemma \ref{lem:hillery-to-opnorm} that the diamond norm difference of two unitaries is related to the operator norm difference of their cosets in $\mathrm{PU}(\hi_D)$ by a simple formula.
\begin{cor}
	Let $U, V\in\mathrm{U}(\hi_D)$ and define $\phi_W(X)=WXW^\dagger $ for $W=U,V$. Then
	\begin{equation}
	\|\phi_U-\phi_V\|_\diamond=2 \|[U]-[V]\|_\infty\sqrt{1-\|[U]-[V]\|_\infty^2/4}
	\end{equation}
	if  the origin is not in the complex hull of the eigenvalues of $UV^\dagger $ in $\C$, and $\|\phi_U-\phi_V\|_\diamond=2$ otherwise.
\end{cor}
\begin{proof}
%	The diamond norm is unitarily invariant, i.e. it suffices to prove
%	\begin{equation}
%	\|\phi_{UV^\dagger }-\mathrm{id}\|_\diamond=2\min(1, \|[U]-[V]\|_\infty\sqrt{1-\|[U]-[V]\|_\infty^2/4}).
%	\end{equation}
	
	Suppose first that the smallest disk containing all eigenvalues of $UV^\dagger $ is the unit disc, i.e. the convex hull of the eigenvalues of $UV^\dagger $ contains the origin. Then The statement is true by Theorem \ref{thm:diamond-unitaries}.
%	 we have $\|\phi_{UV^\dagger }-\mathrm{id}\|_\diamond=2$ according to Theorem \ref{thm:diamond-unitaries}. On the other hand, let $e^{i\phi_0},...,e^{i\phi_{d-1}}$ be the eigenvalues of $UV^\dagger $ ordered as in the Proof of Lemma \ref{lem:hillery-to-opnorm}, and $\phi_i\ge\phi_{i+1}$. Then we have, according to Equation \eqref{eq:opnormdist},
%	\begin{eqnarray}
%	\|[U]-[V]\|_\infty&=&\sqrt{2-2\cos((\phi_0-\phi_{d-1})/2)}\nonumber\\
%	&\ge&\sqrt{2},
%	\end{eqnarray}
%	where we used that $\pi\ge(\phi_0-\phi_{d-1})/2\ge \pi/2$ by assumption. Therefore we get
%	\begin{equation}
%	\|[U]-[V]\|_\infty\sqrt{1-\|[U]-[V]\|_\infty^2/4}\ge 1,
%	\end{equation}
%	as $x\mapsto x\sqrt{1-x^2/4}$ is monotone increasing for $x\in[0,1]$.
%	
	Let $e^{i\phi_0},...,e^{i\phi_{d-1}}$ be the eigenvalues of $UV^\dagger $ ordered as in the Proof of Lemma \ref{lem:hillery-to-opnorm}, and $\phi_i\ge\phi_{i+1}$. Assume now $(\phi_0-\phi_{d-1})<\pi$. Then we get by Theorem \ref{thm:diamond-unitaries} that
	\begin{eqnarray}
	\|\phi_{UV^\dagger }-\mathrm{id}\|_\diamond&=&2\sin((\phi_0-\phi_{d-1})/2)\nonumber\\
	&=&2\sqrt{1-\cos^2((\phi_0-\phi_{d-1})/2)}\nonumber\\
	&=&2\sqrt{1-\left(1-\frac{\|[U]-[V]\|_\infty^2}{2}\right)^2}\nonumber\\
	&=&2\|[U]-[V]\|_\infty\sqrt{1-\frac{\|[U]-[V]\|_\infty^2}{4}}
	\end{eqnarray}
	where we used Equation \eqref{eq:opnormdist} in the third line.
\end{proof}
%This implies in particular that 
%\begin{equation}
%\|\phi_U-\phi_V\|_\diamond=2\|[U]-[V]\|_\infty+\mathcal{O}(\|[U]-[V]\|_\infty^3),
%\end{equation}
%i.e. the diamond norm difference is in good approximation equal to twice the operator norm difference on $\mathrm{PU}(\hi_D)$ in the regime of small distances.
%

We go on to derive a bound on the inner product of program states in terms of $\|[U]-[V]\|_\infty$ that improves over \eqref{eq:hillerycor} in that it gets rid of the dimension dependence.
\begin{lem}\label{lem:appr-poststates}
	Let $G\in\mathrm{U}(\hi_D\otimes\hi_P)$ be an $\varepsilon$-universal PQP, and $\ket{\psi_1},\ket{\psi_2}\in\hi_D$ orthonormal. Then we have that for any $U\in\mathrm U(\hi_D)$
	\begin{equation}
	P(\ket{U'_{\psi_1}}, \ket{U'_{\psi_2}})\le 2\left(\varepsilon+\sqrt{2\varepsilon}\right) ,
	\end{equation}
	where $\ket{U'_{\psi_i}}$ is the vector that exists by Uhlmann's theorem such that
	\begin{equation}
	P(G\ket{\psi_i}_D\ket{U}_P,U_D\ket{\psi_i}_D\ket{U'_{\psi_i}}_P)\le \varepsilon.
	\end{equation}
\end{lem}
\begin{proof}
	As $G$ is a $\varepsilon$-uPQP, there exists, again by Uhlmann's theorem, a vector $\ket{U'_{\psi_1+\psi_2}}$ such that
	\begin{equation}\label{eq:overlap1}
	P\left(\frac{1}{\sqrt{2}}G(\ket{\psi_1}_D+\ket{\psi_2}_D)\ket U_P, \frac{1}{\sqrt{2}}(U\ket{\psi_1}_D+U\ket{\psi_2}_D)\ket{U'_{\psi_1+\psi_2}}_P\right)\le\varepsilon.
	\end{equation}
	On the other hand we have 
	\begin{equation}
	G\ket{\psi_i}_D\ket{U}_P=\sqrt{1-\varepsilon'^2}U_D\ket{\psi_i}_D\ket{U'_{\psi_i}}_P+\varepsilon' \ket{R_{\psi_i}}_{DP}
	\end{equation}
	for some normalized states $\ket{R_{\psi_i}}_{DP}$, $i=1,2$ and $\varepsilon'\le\varepsilon$. It follows that
	\begin{eqnarray}
	&&\frac{1}{\sqrt{2}}G(\ket{\psi_1}_D+\ket{\psi_2}_D)\ket U_P\nonumber\\
	&=&\frac{1}{\sqrt{2}}\Bigg(\sqrt{1-\varepsilon'^2}U_D\bigg(\ket{\psi_1}_D\ket{U'_{\psi_i}}_P
	\nonumber\\&&
	+\ket{\psi_2}_D\ket{U'_{\psi_i}}_P\bigg)+\varepsilon'\left( \ket{R_{\psi_1}}_{DP}+\ket{R_{\psi_2}}_{DP}\right)\Bigg),
	\end{eqnarray}
	and therefore
	\begin{equation}
		\frac 1 2\left|\left( \bra{\psi_1}_D+\bra{\psi_2}_D\right)\bra{U}_PG^\dagger _{DP}U_D\left(\ket{\psi_1}_D\ket{U'_{\psi_i}}_P+\ket{\psi_2}_D\ket{U'_{\psi_i}}_P\right)\right|\ge\frac{1-\varepsilon'}{\sqrt{1-\varepsilon'^2}}.
	\end{equation}
	We bound
	\begin{eqnarray}\label{eq:overlap2}
		&&d\left(\frac{1}{\sqrt{2}}G(\ket{\psi_1}_D+\ket{\psi_2}_D)\ket U_P,U_D\frac{1}{\sqrt{2}}\left(\ket{\psi_1}_D\ket{U'_{\psi_i}}_P+\ket{\psi_2}_D\ket{U'_{\psi_i}}_P\right)\right)\nonumber\\
		&\le&\sqrt{1-\frac{(1-\varepsilon')^2}{1-\varepsilon'^2}}\nonumber\\
	&=&\sqrt{1-\frac{1-\varepsilon'}{1+\varepsilon'}}\nonumber\\
	&\le&\sqrt{1-\frac{1-\varepsilon}{1+\varepsilon}}\nonumber\\
	&\le&\sqrt{2\varepsilon}.
	\end{eqnarray}
	The second-to-last inequality holds as the function $x\mapsto \frac{1-x}{1+x}$ is monotonically increasing on $[0,1]$. By the triangle inequality for the purified distance we can combine Equations \eqref{eq:overlap1} and \eqref{eq:overlap2} to get
	\begin{equation}
	\frac 1 2\left(\left|\bracket{U'_{\psi_1}}{U'_{\psi_1+\psi_2}}\right|+\left|\bracket{U'_{\psi_2}}{U'_{\psi_1+\psi_2}}\right|\right)\ge\sqrt{1-\left(\varepsilon+\sqrt{2\varepsilon}\right)^2}.
	\end{equation}
	% 	It follows that
	% 	\begin{eqnarray}
	% 		\left|\bracket{U'_{\psi_i}}{U'_{\psi_1+\psi_2}}\right|^2&\ge&\left( 2\sqrt{1-\left(\varepsilon+\sqrt{2\varepsilon}\right)^2}-1\right)^2\nonumber\\
	% 		&=&5-4\sqrt{1-\left(\varepsilon+\sqrt{2\varepsilon}\right)^2}-4\left(\varepsilon+\sqrt{2\varepsilon}\right)^2,
	% 	\end{eqnarray}
	% 	i.e.
	% 	\begin{equation}
	% 		P(\ket{U'_{\psi_i}},\ket{U'_{\psi_1+\psi_2}})\le 2\sqrt{\sqrt{1-\left(\varepsilon+\sqrt{2\varepsilon}\right)^2}+\left(\varepsilon+\sqrt{2\varepsilon}\right)^2-1}
	% 	\end{equation}
	% 	and again by the triangle inequality
	% 	\begin{eqnarray}
	% 		|\bracket{U'_{\psi_1}}{U'_{\psi_2}}|&\ge& \sqrt{1-16\left(\sqrt{1-\left(\varepsilon+\sqrt{2\varepsilon}\right)^2}+4\left(\varepsilon+\sqrt{2\varepsilon}\right)^2-1\right)}\nonumber\\
	% 		&=&\sqrt{17-16\left(\sqrt{1-\left(\varepsilon+\sqrt{2\varepsilon}\right)^2}+4\left(\varepsilon+\sqrt{2\varepsilon}\right)^2\right)}
	% 	\end{eqnarray}
	Define $\gamma_1\in [0,\pi/2),\ i=1,2$ such that $\cos{\gamma_i}=\left|\bracket{U'_{\psi_i}}{U'_{\psi_1+\psi_2}}\right|$. Then we have
	\begin{eqnarray}\label{eq:bound1}
	\sqrt{1-\left(\varepsilon+\sqrt{2\varepsilon}\right)^2}&\le&\frac 1 2\left(\cos{\gamma_1}+\cos{\gamma_2}\right)\nonumber\\
	&=&\cos(\gamma_1+\gamma_2)\cos(\gamma_1-\gamma_2),
	\end{eqnarray}
	where the last line is a trigonometric identity. By the triangle inequality for the purified distance and another trigonometric identity we get
	\begin{eqnarray}
	|\bracket{U'_{\psi_1}}{U'_{\psi_2}}|^2&\ge& 1-\left(\sqrt{1-\cos^2\gamma_1}+\sqrt{1-\cos^2\gamma_1}\right)^2\nonumber\\
	&=&1-\left(\sin\gamma_1+\sin\gamma_2\right)^2\nonumber\\
	&=&1-4\sin^2(\gamma_1+\gamma_2)\cos^2(\gamma_1-\gamma_2)\\
	&=&1-4\cos^2(\gamma_1-\gamma_2)+4\cos^2(\gamma_1+\gamma_2)\cos^2(\gamma_1-\gamma_2)\nonumber\\
	&\ge&-3+4\left(1-\left(\varepsilon+\sqrt{2\varepsilon}\right)^2\right)\nonumber\\
	&\ge&1-4\left(\varepsilon+\sqrt{2\varepsilon}\right)^2.
	\end{eqnarray}
	The second-to-last inequality is due to Equation \eqref{eq:bound1} and the fact that $\cos^2(\gamma_1-\gamma_2)\le 1$.
\end{proof}

\begin{thm}
	Let $G\in\mathrm U(\hi_D\otimes \hi_P)$ be an $\varepsilon$-uPQP and let $\ket{U_i}_P\in\hi_P$, $i=1,2$ be the program states of $U_i\in\mathrm U(\hi_D)$ such that $[U_1^\dagger U_2]$ is not too far from the coset of the identity in $PU(\hi_D)$ in the sense that the convex hull of the eigenvalues of $U_1^\dagger U_2$ does not contain the origin. Then
	\begin{equation}
	|\bracket{U_1}{U_2}|\le \frac{2\left(\sqrt{2\epsilon }+7 \epsilon +2 \sqrt{2} \epsilon ^{3/2}+3 \epsilon ^2\right)}{\left\|\left[U_1\right]-\left[U_2\right]\right\|_\infty}+2\varepsilon+\varepsilon^2.
	\end{equation} 
	for all $\zeta\in(0,1)$, where $\delta=2\varepsilon+\varepsilon^2$ and $\tilde\varepsilon=14\left(\varepsilon+\sqrt{2\varepsilon}\right)$.
\end{thm}
\begin{proof}
	Let $\lambda_0,...,\lambda_{d-1}$ be the eigenvalues of $U_1^\dagger U_2$ ordered along the unit circle such that the arc from $\lambda_0$ to $\lambda_{d-1}$ containing all the $\lambda_i$ has minimal length, and such that $\lambda_{d-1}/\lambda_0=e^{i\phi}, \phi\in[0, \pi]$. $G$ is an $\varepsilon$-uPQP, so for all $\ket{\psi}_D\in\hi_D$ there exist vectors $\ket{U'_{i,\psi}}_P\in\hi_P$ and $\ket{R_{U_i,\psi}}_{DP}\in\hi_{DP}$ such that
	\begin{equation}
	G\ket{\psi}_D\ket{U_i}_P=\sqrt{1-\varepsilon'(\psi, U_i)^2}U_i\ket{\psi}_D\ket{U'_{i,\psi}}_P+\varepsilon'(\psi, U_i)\ket{R_{U_i,\psi}}_{DP},
	\end{equation}
	where $\varepsilon'(\psi, U_i)\le\varepsilon$. Let $\ket{\psi_0}_D$ and $\ket{\psi_{d-1}}_D$ be eigenvectors for the eigenvalues $\lambda_0$ and $\lambda_{d-1}$. As $\bracket{\psi_0}{\psi_{d-1}}=0$, Lemma \ref{lem:appr-poststates} implies that there exist $0\le\delta_i\le 2(\varepsilon+\sqrt{2\varepsilon})$ and states $\ket{S_{U_i}}_P$ such that
	\begin{equation}\label{eq:poststate-appr}
	\ket{U'_{i,\psi_0}}=\sqrt{1-\delta_i^2}\ket{U'_{i,\psi_{d-1}}}+\delta_i\ket{S_{U_i}}.
	\end{equation}
	We can now write down an expression for the desired inner product,
	\begin{eqnarray}
	\bracket{U_1}{U_2}&=&\bra{\psi_i}\bra{U_1}G^\dagger G\ket{\psi_i}\ket{U_2}\nonumber\\
	&=&\sqrt{\left(1-\varepsilon'(\psi, U_1)^2 \right)\left(1-\varepsilon'(\psi, U_2)^2 \right)}\bracket{U'_{1,\psi_i}}{U'_{2,\psi_i}}\bra{\psi_i}U_1^\dagger U_2\ket{\psi_i}+\eta_{i}\nonumber\\
	&=&\alpha_i\bracket{U'_{1,\psi_i}}{U'_{2,\psi_i}}\lambda_i+\eta_{i},\label{eq:innerprod}
	\end{eqnarray}
	where we defined $1\ge\alpha_i\ge1-\varepsilon^2$ and  $\eta_i\in\C$ with
	\begin{align}
	\left|\eta_i\right|\le&	\varepsilon'(\psi, U_1)\sqrt{1-\varepsilon'(\psi, U_2)^2}+\varepsilon'(\psi, U_2)\sqrt{1-\varepsilon'(\psi, U_2)^2}+\varepsilon'(\psi, U_1)\varepsilon'(\psi, U_2)\le 2\varepsilon+\varepsilon^2\nonumber\\
	=:&\delta.
	\end{align}
	We can use this to write down an expression for the quotient of the extremal eigenvalues of $U_1^\dagger U_2$,
	\begin{eqnarray}\label{eq:quot}
	\frac{\lambda_{d-1}}{\lambda_0}&=&\frac{\left(\bracket{U_1}{U_2}-\eta_{d-1}\right)\alpha_0\bracket{U'_{1,\psi_0}}{U'_{2,\psi_0}}}{\left(\bracket{U_1}{U_2}-\eta_{0}\right)\alpha_{d-1}\bracket{U'_{1,\psi_{d-a}}}{U'_{2,\psi_{d-1}}}}\nonumber\\
	&=&\frac{\left(\bracket{U_1}{U_2}-\eta_{d-1}\right)\alpha_0\left(\bracket{U'_{1,\psi_{d-a}}}{U'_{2,\psi_{d-1}}}+\gamma\right)}{\left(\bracket{U_1}{U_2}-\eta_{0}\right)\alpha_{d-1}\bracket{U'_{1,\psi_{d-a}}}{U'_{2,\psi_{d-1}}}},
	\end{eqnarray}
	where the last step follows from Equation \eqref{eq:poststate-appr} and we defined $\gamma\in\C$ with
	\begin{eqnarray}
	|\gamma|&\le& \delta_1\sqrt{1-\delta_2^2}+\delta_2\sqrt{1-\delta_1^2}+\delta_1\delta_2\nonumber\\
	&\le& 4\left(\varepsilon+\sqrt{2\varepsilon}\right)\left(1+\varepsilon+\sqrt{2\varepsilon}\right):=\tilde\varepsilon .
	\end{eqnarray}
	We can W.L.O.G. assume that both $\bracket{U_1}{U_2}$ and $\bracket{U'_{1,\psi_{d-a}}}{U'_{2,\psi_{d-1}}}$ are real, otherwise we can just cancel the phase from the fraction in Equation \eqref{eq:quot} and redefine $\eta_i$ and $\gamma$, without changing there magnitude. Now let
	\begin{align}
	\bracket{U_1}{U_2}-\eta_{d-1}=&r_1e^{i\phi_1},\nonumber\\
	\bracket{U_1}{U_2}-\eta_{0}=&r_2e^{-i\phi_2}, \text{ and}\nonumber\\
	\bracket{U'_{1,\psi_{d-a}}}{U'_{2,\psi_{d-1}}}+\gamma=&r_3e^{i \phi_3},
	\end{align}
	with $r_i\in\R_{\ge 0}$ and $\phi_i\in (-\pi,\pi]$. Furthermore let $\lambda_{d-1}/\lambda_0=e^{i\phi}$ with $\phi\in (-\pi,\pi]$. With these definitions we have $\phi=\phi_1+\phi_2+\phi_3\mod 2\pi$. 
	
	We assume now that for some $x >\delta$ to be determined later we have
	\begin{eqnarray}\label{eq:ass2}
	\bracket{U_1}{U_2}&\ge&x\text{ and}\nonumber\\
	\bracket{U'_{1,\psi_{d-a}}}{U'_{2,\psi_{d-1}}}&\ge& |\gamma|.
	\end{eqnarray}
	Basic trigonometry shows that
	\begin{eqnarray}\label{eq:argineqs}
	\sin(|\phi_1|/2)&\le& \frac{|\eta_{d-1}|}{2\bracket{U_1}{U_2}},\nonumber\\
	\sin(|\phi_2|/2)&\le& \frac{|\eta_{0}|}{2\bracket{U_1}{U_2}}\text{ and}\nonumber\\
	\sin(|\phi_3|/2)&\le& \frac{|\gamma|}{2\bracket{U'_{1,\psi_{d-a}}}{U'_{2,\psi_{d-1}}}}.
	\end{eqnarray}
	As the sine is concave on $[0,\pi]$, we can bound the sum of the three sines as
	\begin{eqnarray}\label{eq:sinconcav}
	\left(\sin(|\phi_1|/2)+\sin(|\phi_2|/2)+\sin(|\phi_3|/2)\right)&\ge& 3\sin\left(\frac{|\phi_1|+|\phi_2|+|\phi_3|}{6}\right)\nonumber\\
	&\ge&3\left|\sin\left(\frac{\phi}{6}\right)\right|\nonumber\\
	&=&3\sin\left(\frac{\phi}{6}\right),
	\end{eqnarray}
	where the last step follows by our choice of ordering for the eigenvalues $\lambda_i$.
	Putting things together we get
	\begin{eqnarray}
	3\sin\left(\frac{\phi}{6}\right)&\le&\left(\sin(|\phi_1|/2)+\sin(|\phi_2|/2)+\sin(|\phi_3|/2)\right)\nonumber\\
	&\le&\frac{|\eta_0|+|\eta_{d-1}|}{2\bracket{U_1}{U_2}}+\frac{|\gamma|}{2\bracket{U'_{1,\psi_{d-a}}}{U'_{2,\psi_{d-1}}}}\nonumber\\
	&\le&\frac{|\eta_0|+|\eta_{d-1}|}{2\bracket{U_1}{U_2}}+\frac{|\gamma|\alpha_{d-1}}{2\left|\bracket{U_1}{U_2}-\eta_{d-1}\right|}\nonumber\\
	&\le&\frac{|\eta_0|+|\eta_{d-1}|}{2\bracket{U_1}{U_2}}+\frac{|\gamma|\alpha_{d-1}}{2\left(|x|-\left|\eta_{d-1}\right|\right)},
	\end{eqnarray}
	where the inequalities are Equations \eqref{eq:sinconcav}, \eqref{eq:argineqs}, \eqref{eq:innerprod} and \eqref{eq:ass2}. Solving for $\bracket{U_1}{U_2}$, which we want to bound, yields
	\begin{eqnarray}
	\bracket{U_1}{U_2}&\le&\frac{|\eta_0|+|\eta_{d-1}|}{3\sin\left(\frac{\phi}{6}\right)-\frac{|\gamma|\alpha_{d-1}}{2\left(|x|-\left|\eta_{d-1}\right|\right)}}\nonumber\\
	&\le&\frac{2\delta}{3\sin\left(\frac{\phi}{6}\right)-\frac{\tilde\varepsilon}{2\left(x-\delta\right)}}.
	\end{eqnarray}
	Considering both the case where the assumption \eqref{eq:ass2} holds and where it does not we get
	\begin{equation}
	\bracket{U_1}{U_2}\le\max\left(x,\frac{2\delta}{3\sin\left(\frac{\phi}{6}\right)-\frac{\tilde\varepsilon}{2\left(x-\delta\right)}}\right).
	\end{equation}
	The two arguments of the maximum are monotonically increasing and decreasing respectively for $x\in (x_0,\infty)$, where $x_0$ is the value of $x$ such that the denominator vanishes, so there is at exactly one intersection point of the two in this range considering the boundary values. This will give the optimal bound. Defining $\alpha=3\sin\left(\frac{\phi}{6}\right)$ and solving the equation
	\begin{equation}
	x=\frac{2\delta}{\alpha-\frac{\tilde\varepsilon}{2\left(x-\delta\right)}}, x\in(x_0, \infty)
	\end{equation}
	for $x$ yields
	\begin{equation}
	x=\frac{\sqrt{(2 (\alpha +2) \delta +\tilde\varepsilon )^2-32 \alpha  \delta ^2}+2(\alpha +2) \delta +\tilde\varepsilon }{4 \alpha }.
	\end{equation}
	Now we bound this quantity by utilizing $\sqrt{x^2+y}\le x+\frac{y}{2x}$,
	\begin{eqnarray}
	x&\le& \frac{2(2(\alpha+2)\delta+\tilde{\varepsilon})-\frac{16\alpha\delta^2}{2(\alpha+2)\delta+\tilde\varepsilon}}{4\alpha}\nonumber\\
	&=& \frac{2(\alpha+2)\delta+\tilde{\varepsilon}-4\delta +\frac{16\delta^2+4\delta\tilde{\varepsilon}}{2(\alpha+2)\delta+\tilde\varepsilon}}{2\alpha}\nonumber\\
	&\le& \frac{2\alpha\delta+\tilde{\varepsilon} +\frac{16\delta^2+4\delta\tilde{\varepsilon}}{4\delta+\tilde\varepsilon}}{2\alpha}\nonumber\\
	&=& \frac{\tilde{\varepsilon} +8\delta}{2\alpha}+\delta,
	\end{eqnarray}
	where the last equation is obtained by expressing the $\tilde{\varepsilon}$ in the fraction in the numerator  in terms of $\delta$. Plugging in the expressions for $\tilde{\varepsilon}$ and $\delta$ in terms of $\varepsilon$ yields
	\begin{equation}
	\bracket{U_1}{U_2}\le\frac{2\left(\sqrt{2\epsilon }+7 \epsilon +2 \sqrt{2} \epsilon ^{3/2}+3 \epsilon ^2\right)}{\alpha}+2\varepsilon+\varepsilon^2.
	\end{equation}
	It remains to bound $\alpha$ in terms of the operator norm distance of $[U]$ and $[V]$. For $\beta\in[0,\pi/2]$ and $r\ge 1$ the sine obeys the inequality $\sin(r \beta)\le r\sin\beta$. Therefore we have
	\begin{eqnarray}
	\alpha&=&3\sin\left(\frac{\phi}{6}\right)\nonumber\\
	&=&2\left(\frac{3}{2}\sin\left(\frac{\phi}{6}\right)\right)\nonumber\\
	&\ge&2\sin\left(\frac{\phi}{4}\right)\nonumber\\
	&=&\left|\sqrt{\lambda_{d-1}/\lambda_1}-1\right|\nonumber\\
	&=&\left\|\left[U_1\right]-\left[U_2\right]\right\|_\infty.
	\end{eqnarray}
	%	Bounding this expression by setting the negative term under the square root to zero we get
	%	\begin{eqnarray}
	%		x&\le& \frac{2 (\alpha +2) \delta +\tilde\varepsilon }{2\alpha}\nonumber\\
	%		&=&\frac{4 \delta +\tilde\varepsilon }{2\alpha}+\delta\nonumber\\
	%		&=&\frac{4 \delta +\tilde\varepsilon }{2\alpha}+\delta
	%	\end{eqnarray}
\end{proof}
The last equality is Equation \eqref{eq:opnormdist}.

The plan from now on is as follows: The operator norm distance is a an invariant metric on the projective unitary group. Therefore we can bound the maximum number of unitaries $N^{ent}_r(\mathrm{PU}(d))$ that have pairwise distance at least $r$ by a volume argument using the Weyl integration formula. Then we bound the number of almost orthogonal vectors in a given finite-dimensional Hilbert space in terms of the dimension and the maximal value of the inner product, using results from \cite{Alon2009}. This will give an explicit lower  bound on the dimension of the program register of a uPQP. Let us start with a few definitions.

\begin{figure}
	\begin{center}
		\begin{minipage}{.5\textwidth}
			\includegraphics[width=\textwidth]{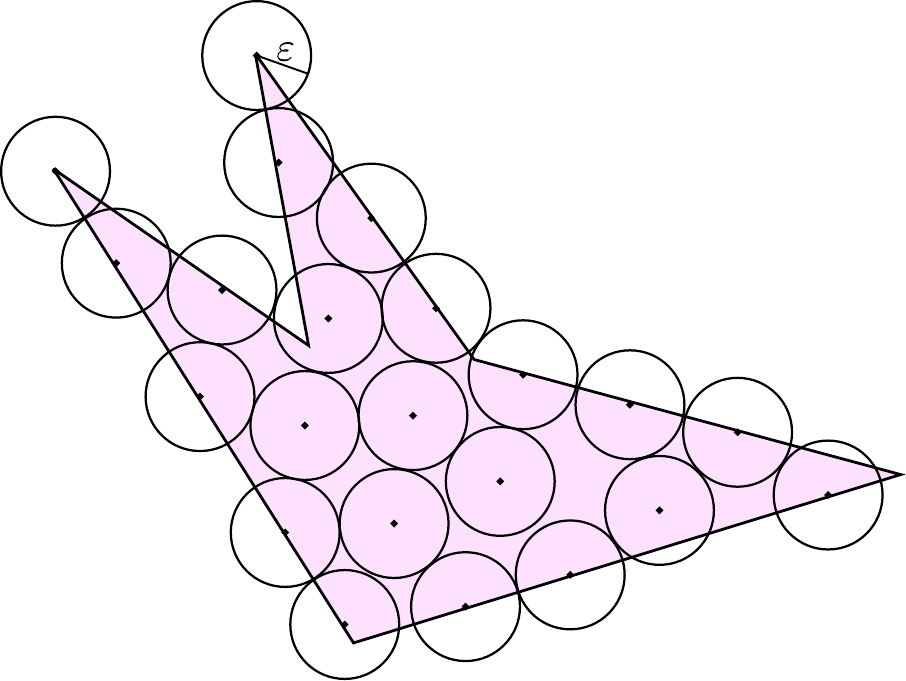}
			\caption*{a)}
			%\label{fig:prob1_6_2}
		\end{minipage}%
		\begin{minipage}{.5\textwidth}
			\includegraphics[width=\textwidth]{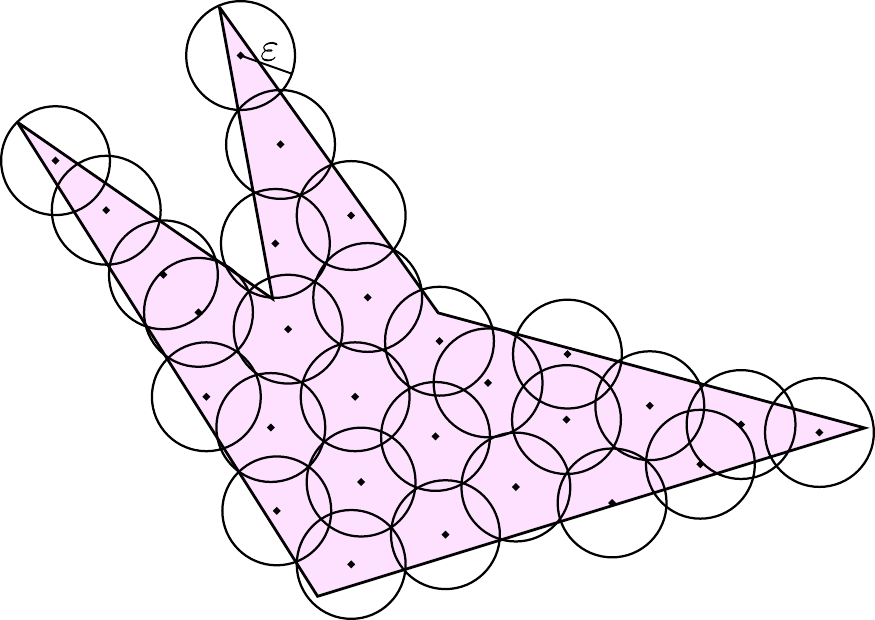}
			\caption*{b)}
			%\label{fig:prob1_6_1}
		\end{minipage}\\
		\begin{minipage}{.5\textwidth}
			\includegraphics[width=\textwidth]{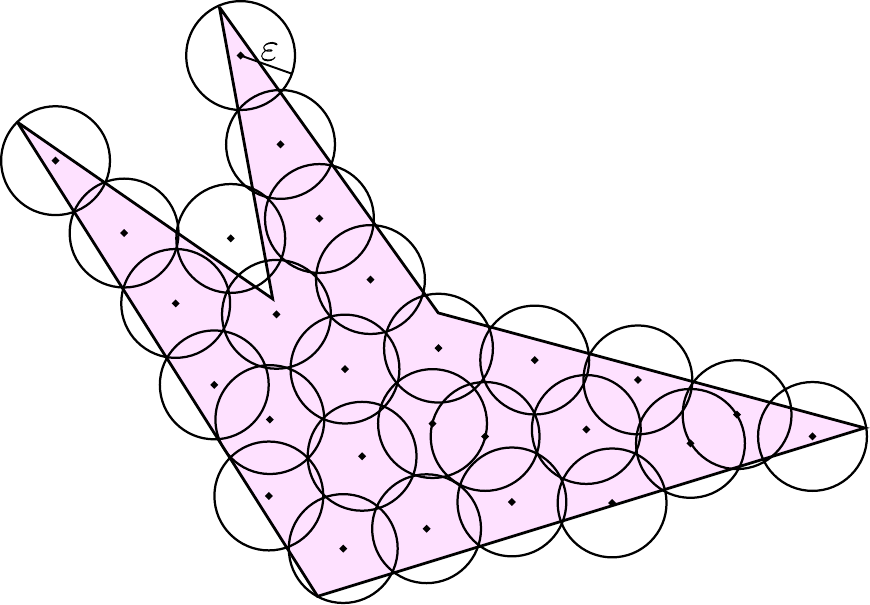}
			\caption*{c)}
			%\label{fig:prob1_6_1}
		\end{minipage}%
	\begin{minipage}{.5\textwidth}
		\includegraphics[width=\textwidth]{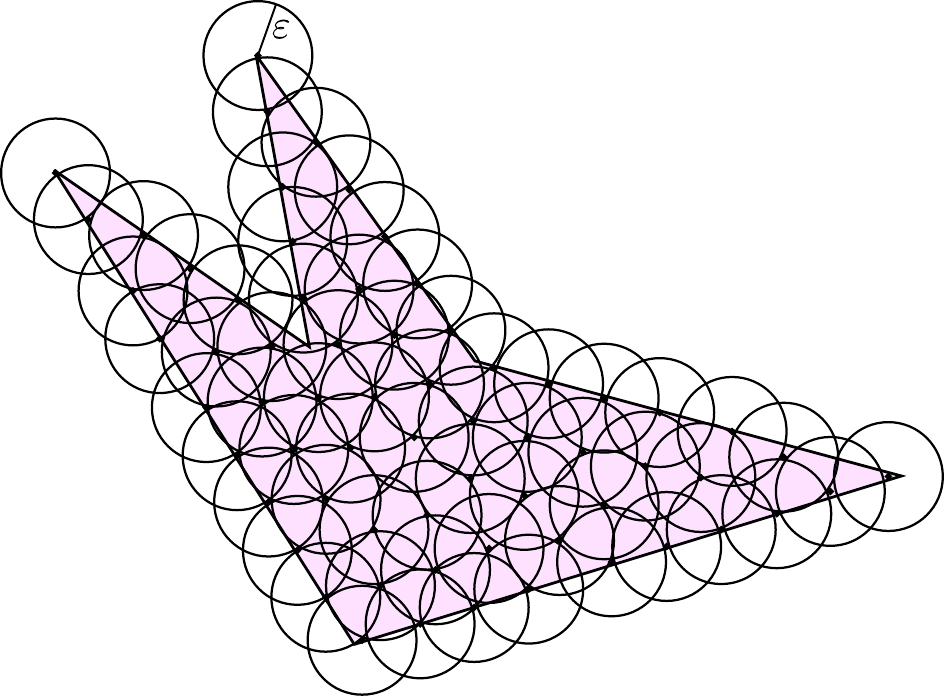}
		\caption*{d)}
		%\label{fig:prob1_6_1}
	\end{minipage}\\
\begin{minipage}{.5\textwidth}
	\includegraphics[width=\textwidth]{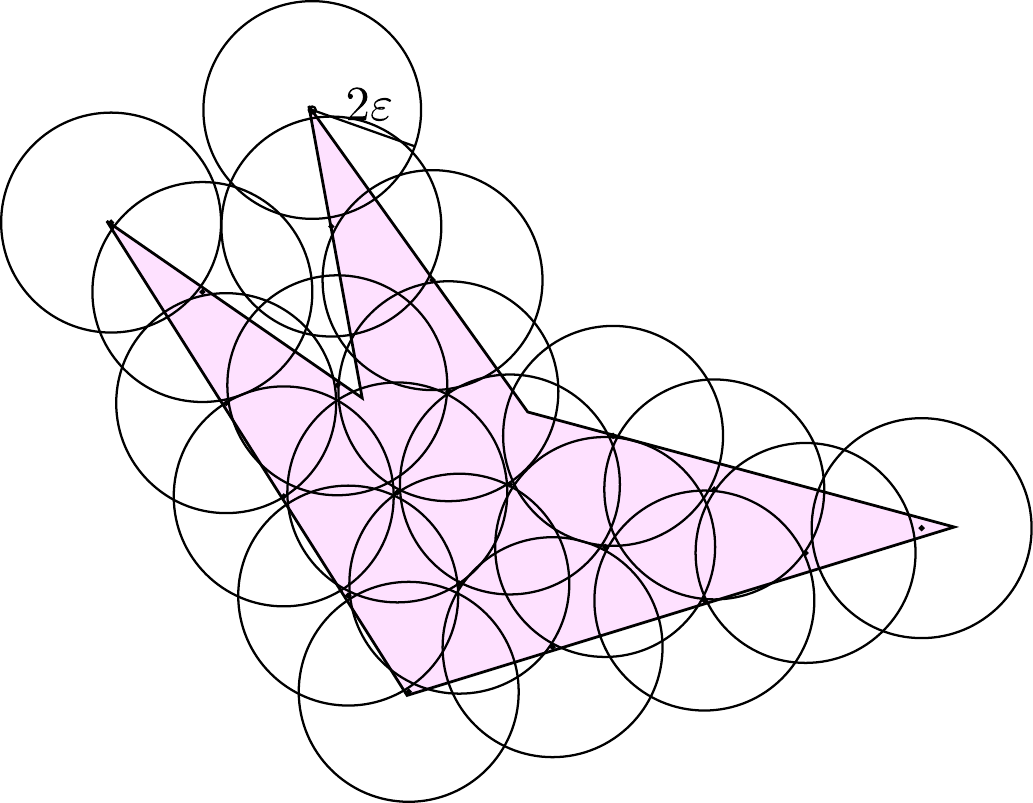}
	\caption*{e)}
	%\label{fig:prob1_6_1}
\end{minipage}
	\end{center}
	\caption{Candidates for the optimizations defining a) $N^{pack}_\varepsilon$, b) $N^{int}_\varepsilon$, c) $N^{ext}_\varepsilon$, d) $N^{ent}_\varepsilon$, and e) $N^{ent}_{2\varepsilon}$.}\label{fig:coverings}
\end{figure}

\begin{defn}	
	Let ${(X,d)}$ be a metric space, let ${E\subset X}$, and let ${r>0}$.	
	\begin{itemize}
		\item The packing number ${N^{pack}_r(E)}$ is the largest number $n$ of disjoint closed balls centered at points ${x_1,\dots,x_n}\in E$ .
		\item The internal covering number ${N^{int}_r(E)}$ is the smallest number $n$ of points ${x_1,\dots,x_n \in E}$ such that the closed balls ${B(x_1,r),\dots,B(x_n,r)}$ cover ${E}$.
		\item The external covering number ${N^{ext}_r(E)}$ is the smallest number $n$ of points ${x_1,\dots,x_n \in X}$ such that the closed balls ${B(x_1,r),\dots,B(x_n,r)}$ cover ${E}$.
		\item The metric entropy ${N^{ent}_r(E)}$ is the largest number of points ${x_1,\dots,x_n}\in{E}$ Such that $d(x_i,x_j) \geq r$ for all ${i \neq j}$.
	\end{itemize}
\end{defn}
See Figure \ref{fig:coverings} for examples.
It is easy to see that these numbers ar roughly the same.
\begin{prop}\label{prop:equiv-coverings}
	\begin{equation}
	N^{ent}_{2r}(E)\le N^{pack}_r(E)\le N^{ext}_r(E)\le N^{int}_r(E)\le N^{ent}_r(E).
	\end{equation}
\end{prop}
\begin{proof}
	The inequalities are obvious except for $N^{pack}_r(E)\le N^{ext}_r(E)$. To see why this is true, consider a packing of $r$-balls centered at $x_1,...,x_k$ and a cover of $r$-balls $B_1,...,B_l$. Suppose $k> l$. Then there exists an index $1\le i\le l$ and two indices $1\le j_1,j_2\le k$ such that $x_{j_1},x_{j_2}\in B_i$, which is a contradiction.
\end{proof}

For calculating the volume of an $\varepsilon$-ball in $\mathrm{PU}(\hi_D)$ we need the Weyl integration formula \cite{simon1996representations}.

\begin{thm}[Weyl integration formula]\label{thm:Weyl}
	Let $f: \mathrm U(d)\to\R$ be a class function, i.e. $f(VUV^\dagger )=f(U)$ for all $U,V\in \mathrm U(d)$. Then the integral of $f$ over the Haar measure	can be reduced to an integral over a maximal torus $T(d)\subset \mathrm (d)$ (diagonal subgroup),
	\begin{align}
	&\intop_{\mathrm U(d)}f(U)\D \nonumber\\
	=&\frac{1}{d!(2 \pi^d)}\intop_{T(d)}f(\diag(e^{i\phi_1}, e^{i \phi_2},...,e^{i \phi_d}))\prod_{1\le i<j\le d}\left|e^{i\phi_i}-e^{i\phi_j}\right|^2\D\phi_1\D\phi_2...\D \phi_d.
	\end{align}
\end{thm}

We can use this to bound the volume of the operator norm ball of radius $\varepsilon$ in $\mathrm{PU}(d)$.
\begin{lem}\label{lem:balls}
	Let $0<\varepsilon< \pi/2$. The Haar measure of $B_{\varepsilon}([V])\subset\mathrm{PU}(d)$, the ball centered at $[V]$ of radius $\varepsilon$ when distance is measured in the operator norm metric, is bounded by
	\begin{equation}
	\mu(B_{\varepsilon}([U]))\le \frac{(4\varepsilon)^{d^2-1}}{(\pi)^{d}d!},
	\end{equation}
	where $\mu$ denotes the normalized Haar measure on $\mathrm{PU}(d)$.
\end{lem}
\begin{proof}
	We start by observing that the operator norm metric is invariant, so we can just as well consider the ball around the identity. The integral over $\mathrm{PU}(d)$ can be expressed as an integral over $\mathrm{U}(d)$,
	\begin{eqnarray}
	\mu(B_{\varepsilon}([V]))&=&\intop_{\mathrm{PU}(d)}\vartheta\left(\varepsilon-\|[U]-[\mathds 1]\|_\infty\right)d[U]\nonumber\\
	&=&\intop_{\mathrm{PU}(d)}\vartheta\left(\varepsilon-\min_{z\in\C,\ |z|=1}\|zU-\mathds 1\|_\infty\right)d[U]\nonumber\\
	&=&\frac{1}{2\pi}\intop_{\mathrm{U}(d)}\vartheta\left(\varepsilon-\min_{z\in\C,\ |z|=1}\|zU-\mathds 1\|_\infty\right)dU,
	\end{eqnarray}
	where $\vartheta$ denotes the Heaviside step function and the last step follows, because by construction $\min_z\|zU-\mathds 1\|_\infty=\min_z\|zz'U-\mathds 1\|_\infty$ for all $z'\in\C$, $|z'|=1$, and $\mathrm U(d)\cong\mathrm{PU}(d)\times U(1)$. The integrand in the last expression is a class function, as the operator norm is unitarily invariant. Therefore we can apply Weyl's integration formula, Theorem \ref{thm:Weyl}, to get
	\begin{eqnarray}
	&&\mu(B_{\varepsilon}([V]))=\frac{1}{(2\pi)^{d+1}d!}\nonumber\\
	&&\cdot\intop_{[-\pi,\pi)^n}\vartheta\left(\varepsilon-\min_{z\in\C,\ |z|=1}\|z\diag(e^{i\phi_1},...,e^{i\phi_d})-\mathds 1\|_\infty\right)\prod_{1\le i<j\le d}\left|e^{i\phi_i}-e^{i\phi_j}\right|^2\D\phi_1\D\phi_2...\D \phi_d.\nonumber\\
	\end{eqnarray}
	The integrand does not depend on the phase of $U$, so a coordinate transformation and integrating over the phase yields
	\begin{eqnarray}
	&&\mu(B_{\varepsilon}([V]))=\frac{1}{(2\pi)^{d}d!}\nonumber\\
	&&\cdot\intop_{S_1}\vartheta\left(\varepsilon-\min_{z\in\C,\ |z|=1}\|z\diag(e^{i\phi_1},...,e^{i\phi_d})-\mathds 1\|_\infty\right)\prod_{1\le i<j\le d}\left|e^{i\phi_i}-e^{i\phi_j}\right|^2\D\phi_1\D\phi_2...\D \phi_d,\nonumber\\
	\end{eqnarray}
	with $S_1=\{\phi\in[-\pi,\pi)|\sum\phi_i=0\}$. Define two subsets $S_2,S_3\subset S_1$ by 
	\begin{align}
	S_2=&\left\{\phi\in S_1\Big|2\sin(|\phi_i|/2)\le2\varepsilon\right\}\text{ and}\nonumber\\
	S_3=&\left\{\phi\in S_1\Big|2\sin(|\phi_i-\pi|/2)\le2\varepsilon\right\}=S_2+(\pi,....,\pi).
	\end{align}
	If $\phi$ is in neither of the two sets, the integrand is zero. Therefore we get
	\begin{eqnarray}
	&&\mu(B_{\varepsilon}([V]))=\frac{2}{(2\pi)^{d}d!}\nonumber\\
	&&\cdot\intop_{S_2}\vartheta\left(\varepsilon-\min_{z\in\C,\ |z|=1}\|z\diag(e^{i\phi_1},...,e^{i\phi_d})-\mathds 1\|_\infty\right)\prod_{1\le i<j\le d}\left|e^{i\phi_i}-e^{i\phi_j}\right|^2\D\phi_1\D\phi_2...\D \phi_d,\nonumber\\
	\end{eqnarray}	
	having used that the integrand is translation invariant and therefore the integral over $S_3$ is equal to that over $S_2$. Now we can bound this in a straightforward manner. In $S_2$, $\left|e^{i\phi_i}-e^{i\phi_j}\right|\le 4\varepsilon$, so
	\begin{eqnarray}
	\mu(B_{\varepsilon}([V]))&\le&\frac{2}{(2\pi)^{d}d!}\intop_{S_2}(4\varepsilon)^{d(d-1)}\D\phi_1\D\phi_2...\D \phi_d\nonumber\\
	&\le&\frac{2(4\varepsilon)^{d(d-1)}\left(4 \arcsin(\varepsilon)\right)^{d-1}}{(2\pi)^{d}d!}\nonumber\\
	&\le&\frac{(4\varepsilon)^{d^2-1}}{(\pi)^{d}d!},
	\end{eqnarray}
	where we used in the last line that $\arcsin(x)\le 2x $ for $0\le x\le \pi/2$.
\end{proof}

As a corollary we can bound $N^{ent}(\mathrm{PU}(d))$.

\begin{lem}\label{lem:metricentr}
	\begin{equation}
	N^{ent}_\varepsilon(\mathrm{PU}(d))\ge \frac{(\pi)^{d}d!}{(4\varepsilon)^{d^2-1}} .
	\end{equation}
	for $0<\varepsilon<\pi/2$.
\end{lem}
\begin{proof}
	According to Proposition \ref{prop:equiv-coverings} we get
	\begin{eqnarray}
	N^{ent}_\varepsilon(\mathrm{PU}(d))&\ge& N^{int}_\varepsilon(\mathrm{PU}(d))\nonumber\\
	&\ge&\frac{1}{\mu(B_\varepsilon(\mathds{1}))},
	\end{eqnarray}
	and using Lemma \ref{lem:balls} concludes the proof.
\end{proof}

We also need a bound on the number of almost orthogonal unit vectors in a finite dimensional Hilbert space. This can be found in \cite{Alon2009} for the real case.

\begin{thm}[\cite{Alon2009}, Theorem 2.1]\label{thm:perturbed-identity}
	There exists a universal constant $\eta$ such that the following holds: Let $A\in\R^{n\times n}$ such that $A_{ii}=1$ and $|A_{ij}|\le\varepsilon$ for all $i\neq j$. For $\frac{1}{\sqrt n}\le \varepsilon< \frac 1 2$ the rank of $A$ can be lower bounded bounded as
	\begin{equation}
	\rank(A)\ge \frac{\eta\log n}{\varepsilon^2\log\left(1/\varepsilon\right)}.
	\end{equation}
\end{thm}

As a corollary we can bound the number of almost orthogonal vectors in a given Euclidean space.

\begin{cor}\label{cor:almost-orth}
	There exists a universal constant $\eta$ such that the following holds: Let $d\in\N$ and $S\subset\R^d$ such that for all $v,w\in S$ we have that $|(v,w)|\le\varepsilon$, where $(\cdot,\cdot)$ denotes the standard inner product in $\R^d$ and $\frac{1}{\sqrt{d}}\le \varepsilon< \frac 1 2$. Then the cardinality of $S$ can be upper bounded by
	\begin{equation}
	|S|\le \varepsilon^{-\frac 1 \eta\varepsilon^2 d}.
	\end{equation}
	For $0\le\varepsilon<\frac 1 d$ the best bound is
	\begin{equation}
	|S|\le d,
	\end{equation}
	and for $\frac 1 d\le\varepsilon<\frac{1}{\sqrt{d}}$ we get
	\begin{equation}
	|S|\le 2d.
	\end{equation}
\end{cor}
\begin{proof}
	For $\frac{1}{\sqrt{n}}<\varepsilon<1/2$, applying Theorem \ref{thm:perturbed-identity} to the Gram matrix of the vectors in $S$ directly yields the result, with the constant $\eta$ being the same as in Theorem \ref{thm:perturbed-identity}. For $0\le\varepsilon<\frac 1 n$ the observation that the Gram matrix can now only have positive eigenvalues yields the result, and the remaining case follows from Lemma 2.2 in \cite{Alon2009}.
\end{proof}

This also yields a bound for the complex analogue.

\begin{lem}\label{lem:almost-orth-comp}
	Let $S\subset \C^d$ be a set of unit vectors such that for $v,w\in S$ we have
	\begin{equation}
	|(v,w)|\le \varepsilon
	\end{equation}
	for $\frac{1}{\sqrt{2d}}\le \varepsilon< \frac 1 2$. Then $S$ cannot be too large,
	\begin{equation}
	|S|\le \varepsilon^{-\frac 2 \eta\varepsilon^2 d}.
	\end{equation}
	For $0\le\varepsilon<\frac{1}{2d}$ the best bound is
	\begin{equation}
	|S|\le 2d,
	\end{equation}
	and for $\frac{1}{2d}\le\varepsilon<\frac{1}{\sqrt{2d}}$ we get
	\begin{equation}
	|S|\le 4d.
	\end{equation}
\end{lem}
\begin{proof}
	Using the usual isomorphism $\C^n\cong\R^{2n}$ as $\R$-vector spaces, and observing that $\Re\bracket{\cdot}{\cdot}$ maps to the standard inner product on $\R^{2n}$ under this isomorphism, together with the fact that $|\Re(z)|\le|z|$ for all $z\in\C$, reduces the complex to the real case and an application of Corollary \ref{cor:almost-orth} concludes the proof.
\end{proof}

We are now ready to prove the dimension lower bound on the program register of an approximate uPQP.

\begin{thm}\label{thm:appr-uPQP-bound}
	Let $d\in\N$ and $\varepsilon>0$ such that $8\delta_1+\delta_2< \frac 1 2$ with $$\delta_1=2\left(\sqrt{2\epsilon }+7 \epsilon +2 \sqrt{2} \epsilon ^{3/2}+3 \epsilon ^2\right)$$ and $\delta_2=2\varepsilon+\varepsilon^2$. 
	For any $\varepsilon$-universal programmable quantum processor $G\in\mathrm U(\hi_D\otimes\hi_P)$ with $|D|=d$ the program register has dimension at least
	\begin{equation}
	|P|\ge \frac{1}{2(4\delta_1+\delta_2)^2}\min\left(1, \left(\left(\frac{\pi d}{e}\right)^d\left(4\delta_1+\delta_2\right)^2\right)^{\frac{2}{d^2+1}}\right).
	\end{equation}

	There exists a universal constant $\eta>0$ such that for
	\begin{equation}
	8\delta_1+\delta_2>\left[ \left(\frac{\pi d}{e}\right)^{-d}2^{-d^2+1}\right]^{\eta}
	\end{equation}
	the alternative bound 
	\begin{equation}\label{eq:alt-bound}
	|P|\ge \eta\frac{d\log\left(\frac{\pi d}{e}\right)+d^2-1}{2\delta_4^2\log\frac{1}{\delta_4}}
	\end{equation}
	holds.
\end{thm}
\begin{proof}
	Using Theorem \ref{lem:appr-poststates} we get
	\begin{equation}
	|\bracket{U_1}{U_2}|\le \frac{\delta_1}{\left\|\left[U_1\right]-\left[U_2\right]\right\|_\infty}+\delta_2
	\end{equation}
	for all $U_1,U_2\in\mathrm U(d)$, where we defined $\delta_1=2\left(\sqrt{2\epsilon }+7 \epsilon +2 \sqrt{2} \epsilon ^{3/2}+3 \epsilon ^2\right)$ and $\delta_2=2\varepsilon+\varepsilon^2$. According to Lemma \ref{lem:metricentr}, applied for $\varepsilon=\delta_3$, there exists a set $S\subset\mathrm U_d$ with
	\begin{equation}
	|S|\ge \pi^{d}d!(4\delta_3)^{-(d^2-1)}
	\end{equation}
	such that $\left\|\left[U_1\right]-\left[U_2\right]\right\|_\infty\ge \delta_3$ for all $U_1,U_2\in S$.
	Define $S'\subset \hi_P$ by
	\begin{equation}
	S'=\left\{\ket{U}_P\big|U\in S\right\}.
	\end{equation}
	Suppose now first that $K:=|P|< \frac{1}{2(4\delta_1+\delta_2)^2}$. Then we have by this assumption that $4\delta_1+\delta_2<\frac{1}{\sqrt{2K}}$. We set $\delta_3=\frac{\delta_1\sqrt{2K}}{1-\delta_2\sqrt{2K}}$. Applying Lemma \ref{lem:almost-orth-comp} yields
	\begin{equation}
	|S'|\le 4 K,
	\end{equation}
	i.e.
	\begin{eqnarray}
	4K&\ge& \pi^{d}d!\left(\frac{1-\delta_2\sqrt{2K}}{4\delta_1\sqrt{2K}}\right)^{(d^2-1)}\nonumber\\
	\iff (2K)^{\frac{d^2+1}{2}}&\ge&\frac 1 2 \pi^{d}d!\left(\frac{1-\delta_2\sqrt{2K}}{4\delta_1}\right)^{(d^2-1)}\nonumber\\
	&\ge&\frac 1 2 \pi^{d}d!\left(\frac{1}{4\delta_1+\delta_2}\right)^{(d^2-1)}.
	\end{eqnarray}
	bounding $d!\ge \left(\frac{d}{e}\right)^d$ we get
	\begin{equation}
	K\ge\frac 1 2 \left(\frac{\pi d}{e}\right)^{\frac{2d}{d^2+1}}\left(4\delta_1+\delta_2\right)^{-2\frac{d^2-1}{d^2+1}}.
	\end{equation}
	In summary,
	\begin{equation}
	K\ge \frac{1}{2(4\delta_1+\delta_2)^2}\min\left(1, \left(\left(\frac{\pi d}{e}\right)^d\left(4\delta_1+\delta_2\right)^2\right)^{\frac{2}{d^2+1}}\right).
	\end{equation}

	Suppose now that $K> \frac{1}{2(8\delta_1+\delta_2)^2}$. We set $\delta_3=\frac 1 8$ and use 
	\ref{lem:metricentr} to get a set $S\subset\mathrm U_d$ with
	\begin{equation}
	|S|\ge \pi^{d}d!2^{d^2-1}
	\end{equation}
	such that $\left\|\left[U_1\right]-\left[U_2\right]\right\|_\infty\ge \frac 1 8$ for all $U_1,U_2\in S$.
	On the other hand Lemma \ref{lem:almost-orth-comp} shows that
	\begin{equation}
	|S'|\le \delta_4^{-\frac{2}{\eta}\delta_4^2 K},
	\end{equation}
	with $\delta_4=8\delta_1+\delta_2$. Combining the two bounds and using $d!\ge \left(\frac{d}{e}\right)^d$ yields
	\begin{equation}
	K\ge \eta\frac{d\log\left(\frac{\pi d}{e}\right)+d^2-1}{2\delta_4^2\log\frac{1}{\delta_4}}.
	\end{equation}
	This bound is only meaningful if it exceeds the assumption, i.e. if
	\begin{equation}
	8\delta_1+\delta_2> \left[\left(\frac{\pi d}{e}\right)^{-d}2^{-d^2+1}\right]^{\eta}.
	\end{equation}
\end{proof}

\subsection{Lower bounds on the number of ports for $\varepsilon$-port based teleportation}\label{subs:PBT-bounds}

In this subsection, three lower bounds on the number of ports necessary for $(d,\varepsilon)$-port based teleportation are presented. The first one follows from a bound on the amount of classical communication needed for approximate teleportation without restrictions on Bob's correction unitaries. The second one follows from the bound on the size of the program register of an approximate universal programmable quantum processor, and is here only recorded to show that the reasoning of Nielsen and Chuang can be made approximate. It is made obsolete by the third bound, which uses the non-signaling principle to conclude that Bob's marginal cannot change when Alice's measurement is applied.

To the best of my knowledge, the only lower bound for the number of ports necessary for $(d,\varepsilon)$-port based teleportation that has appeared in the literature is
\begin{theorem}[\cite{Ishizaka2015}]
	port based teleportation of a $d$-dimensional quantum system with entanglement fidelity $F=\sqrt{1-\varepsilon^2}$ requires at least
	\begin{equation}
		N\ge \frac{1}{4(d-1)\varepsilon^2}
	\end{equation}
	ports.
\end{theorem}
This bound diverges with $\varepsilon\to 0$, but it vanishes for large dimensions unless $\varepsilon=o(d^{-\half})$.

\subsubsection{The communication bound}
We begin with the bound on the classical communication required for any kind of teleportation. This bound is proven in three steps. First, we bound the smooth max-mutual information of the maximally entangled state. Then we use this bound together with a lower bound on the communication cost of coherent state merging to derive a bound on the quantum communication cost of approximately sending half a maximally entangled state when arbitrary entanglement assistance is at hand. Finally, we use superdense coding to conclude a bound on the communication cost of approximate teleportation. A bound for teleportation with a certain allowed diamond norm error follows -- and is tight because any channel from a teleportation protocol can be made unitarily covariant with the techniques from Subsection \ref{subs:sym-PBT}.
\begin{prop}\label{prop:Imaxbound}
	Let
	\begin{equation}
	\ket{\phi^+}_{AB}=\frac{1}{\sqrt d}\sum_{i=0}^{d-1}\ket{ii}_{AB}\in\hi_A\otimes \hi_B
	\end{equation}
	be the standard maximally entangled state, with $\hi_A\cong\hi_B\cong \C^d$. Then
	\begin{equation}
	2\log\left\lceil d(1-\varepsilon^2)\right\rceil \ge I_{\max}^{\varepsilon}(A:B)_{\phi^+}\ge 2\log\left( d(1-\varepsilon^2)\right)
	\end{equation}
	
\end{prop}

\begin{proof}
	%	Let $\rho\in\End{\hi_A\otimes\hi_{B}}$ be a normalized state such that $I_{\max}^{\varepsilon}(A:B)_{\phi^+}=I_{\max}(A:B)_{\rho}$. Let $\ket{\psi_0}$ be the eigenvector of $\rho$ with the largest eigenvalue $r$. Define $x=|\bracket{\psi}{\phi^+}$. With these definitions we get
	%	\begin{eqnarray}
	%		1-\varepsilon^2&\le&\bra{\phi^+}\rho\ket{\phi^+}\nonumber\\
	%		&=& r x+\bra{\phi^+}\left(\mathds 1-\proj{\psi}\right)\rho\left(\mathds 1-\proj{\psi}\right)\ket{\phi^+}\nonumber\\
	%		&\le&r x+(1-r)(1-x).
	%	\end{eqnarray}
	%	Let
	
	Let $\rho\in\End{\hi_A\otimes\hi_{B}}$ be a normalized state such that $I_{\max}^{\varepsilon}(A:B)_{\phi^+}=I_{\max}(A:B)_{\rho}$. Let $\ket{\rho}_{ABE}$ be a purification of $\rho$. According to Uhlmann's theorem there exists a pure state $\ket{\alpha}_E$ such that
	\begin{equation}
	\sqrt{1-\varepsilon^2}\le F(\phi^+,\rho)=\bra{\phi^+}_{AB}\bra{\alpha}_E\ket{\rho}_{ABE},
	\end{equation} 
	as any phase of the right hand side can be included in $\ket{\alpha}$.
	Let
	\begin{equation}
	\ket{\rho}_{ABE}=\sum_{i=0}^{d-1}\sqrt{p_i}\ket{\phi_i}_{A}\otimes\ket{\psi_i}_{BE}
	\end{equation}
	be the Schmidt decomposition of $\ket{\rho}$ with respect to the bipartition $A:BE$. Let $U_A$ be the unitary such that $U_A\ket{i}_A=\ket{\phi_i}_{A}$. By the mirror lemma we get
	\begin{eqnarray}
	\ket{\phi^+}_{AB}&=&U_AU_A^\dagger \ket{\phi^+}_{AB}\nonumber\\
	&=&U_A\bar U_B\ket{\phi^+}_{AB}\nonumber\\
	&=&\frac{1}{\sqrt d}\sum_{i=0}^{d-1}\ket{\phi_i}_{A}\ket{\xi_i}_{B},
	\end{eqnarray}
	where $\ket{\xi_i}_{B}=\bar U_B\ket{i}_B$ and $\bar{U}$ denotes the complex conjugate in the computational basis. Hence we have
	\begin{eqnarray}\label{eq:uhl+schmidt}
	1-\varepsilon^2&\le& \bra{\phi^+}_{AB}\bra{\alpha}_E\ket{\rho}_{ABE}^2\nonumber\\
	&=&\Re \bra{\phi^+}_{AB}\bra{\alpha}_E\ket{\rho}_{ABE}^2\nonumber\\
	&=&\left(\sum_{i=0}^{d-1}\sqrt{\frac{p_i}{d}}\Re \bra{\xi_i}_B\bra{\alpha}_E\ket{\psi_i}_{BE}\right)^2\nonumber\\
	&\le&\frac 1 d\sum_{i=0}^{d-1}\left(\Re \bra{\xi_i}_B\bra{\alpha}_E\ket{\psi_i}_{BE}\right)^2\nonumber\\
	&\le&\frac 1 d\sum_{i=0}^{d-1}\Re\bra{\xi_i}_B\bra{\alpha}_E\ket{\psi_i}_{BE}.
	\end{eqnarray}
	Here the second inequality is the Cauchy Schwarz inequality and the third one follows from the fact that $\Re\bra{\xi_i}_B\bra{\alpha}_E\ket{\psi_i}_{BE}\le 1 $.
	
	We are now ready to bound the max-mutual information of $\rho$. Let $\lambda=I_{\max}(A:B)_\rho=I_{\max}^{\varepsilon}(A:B)_{\phi^+}$. Then by definition there exists a state $\sigma_B$ such that
	\begin{equation}
	2^\lambda=\left\|\rho_A^{-\half}\otimes\sigma_B^{-\half}\rho_{AB}\rho_A^{-\half}\otimes\sigma_B^{-\half}\right\|_\infty
	\end{equation}
	Let $\ket{\phi_\sigma}=\sqrt d\sigma_B^{\half}\ket{\phi^+}$, a normalized vector. Then we bound
	\begin{eqnarray}
	2^\lambda&=&\left\|\rho_A^{-\half}\otimes\sigma_B^{-\half}\rho_{AB}\rho_A^{-\half}\otimes\sigma_B^{-\half}\right\|_\infty\nonumber\\
	&\ge&\bra{\phi_\sigma}\rho_A^{-\half}\otimes\sigma_B^{-\half}\rho_{AB}\rho_A^{-\half}\otimes\sigma_B^{-\half}\ket{\phi_\sigma}\nonumber\\
	&=&\tr\bra{\phi_\sigma}\rho_A^{-\half}\otimes\sigma_B^{-\half}\proj{\rho}_{ABE}\rho_A^{-\half}\otimes\sigma_B^{-\half}\ket{\phi_\sigma}\nonumber\\
	&\ge&\bra{\phi_\sigma}_{AB}\bra{\alpha}_E\rho_A^{-\half}\otimes\sigma_B^{-\half}\proj{\rho}_{ABE}\rho_A^{-\half}\otimes\sigma_B^{-\half}\ket{\phi_\sigma}_{AB}\ket{\alpha}_E\nonumber\\
	&=&d\left|\bra{\phi^+}_{AB}\rho_A^{-\half}\bra{\alpha}_E\ket\rho_{ABE}\right|^2\nonumber\\
	&=&\left|\sum_i\bra{\xi_i}_B\bra{\alpha}_E\ket{\psi_i}_{BE}\right|^2\nonumber\\
	&\ge&\left(\sum_i\Re\bra{\xi_i}_B\bra{\alpha}_E\ket{\psi_i}_{BE}\right)^2\nonumber\\
	&\ge&d^2(1-\varepsilon^2)^2.
	\end{eqnarray}
	
	For the other inequality let ${d'}=\lceil d(1-\varepsilon^2)\rceil$ and 
	\begin{equation}
	\ket{\phi^+_{d'}}=\frac{1}{\sqrt{{d'}}}\sum_{i=0}^{{d'}-1}\ket{ii}_{AB}\in\hi_A\otimes \hi_B.
	\end{equation}
	Then we have
	\begin{eqnarray}
	{I_{\max}(A:B)_{\phi^+_{d'}}}&=&2\log {d'}=2\log\lceil d(1-\varepsilon^2)\rceil\\
	|\bracket{\phi^+}{\phi^+_{d'}}|^2&=&{d'}/d\ge 1-\varepsilon^2.
	\end{eqnarray}
\end{proof}

As a corollary we can bound the necessary quantum communication for simulating the identity channel with a given entanglement fidelity.

\begin{cor}\label{cor:sim-ident}
	Let $\mathcal E_{AA'\to B}$, $\mathcal D_{BB'\to A}$ be quantum channels such that there exists a state $\rho_{A'B'}$ achieving
	\begin{equation}
	F(\mathcal D\circ \mathcal E((\cdot)\otimes\rho_{A'B'}))=\sqrt{1-\varepsilon^2},
	\end{equation}
	$\dim\hi_A=d$ and $\dim\hi_B= {d'}$.
	Then
	\begin{equation}
	{d'}\ge d\left(1-\varepsilon^2\right).
	\end{equation}
\end{cor}
\begin{proof}
	This follows directly from Proposition \ref{prop:Imaxbound} together with the lower bound on the communication cost of one-shot state splitting from \cite{Berta2011}, i.e. the converse direction of Theorem \ref{thm:merge-Berta}.
\end{proof}

Via superdense coding this implies a lower bound on imperfect teleportation as well.

\begin{cor}\label{cor:telebound}
	If in the above corollary $\mathcal E$ is a $qc$-channel, then
	\begin{equation}
	{d'}\ge d^2\left(1-\varepsilon^2\right)^2.
	\end{equation}
\end{cor}
\begin{proof}
	This follows as any protocol with a lower classical communication in conjunction with superdense coding would violate Corollary \ref{cor:sim-ident}.
\end{proof}

Lemma \ref{lem:entanglementf2diamond} shows that the lower bounds hold for protocols with diamond norm error $2\varepsilon$ as well, and they are essentially as tight for the diamond norm as they are for the entanglement fidelity. This is because some extra entanglement can be used to create a random bit string shared by Alice and Bob. Alice then applies an element of the Pauli group to her input according to the shared random string, and Bob applies it again to his output.

As pointed out in Subsection \ref{subs:PBT}, the only kind of information that Alice can send to Bob in port based teleportation is which port to select. This implies that we have a lower bound on the number of ports necessary in PBT as an immediate corollary-
\begin{cor}\label{cor:PBT-combound}
	Any protocol for $(d,\varepsilon)$-port based teleportation uses at least
	\begin{equation}
	N\ge d^2\left(1-\varepsilon^2\right)^2
	\end{equation}
	output ports.
\end{cor}

\subsubsection{The bound based on approximate universal programmable quantum processors}

Here we record the fact that one can get a lower bound on the number of ports from the connection to approximate universal programmable quantum processors that diverges for $\varepsilon\to 0$ uniformly in the dimension. This bound is superseded by the last bound that is proven in Theorem \ref{thm:lowerbound-guessing}. 

\begin{thm}
	Any $\varepsilon$-PBT scheme for teleporting a $d$-dimensional state with $8\delta_1+\delta_2< \frac 1 2$ uses at least
	\begin{equation}
	N\ge \half\left(\log\left( \frac{1}{4\delta_1+\delta_2}\right)-1\right)
	\end{equation}
	output ports, where $\delta_1=2\left(\sqrt{2\epsilon }+7 \epsilon +2 \sqrt{2} \epsilon ^{3/2}+3 \epsilon ^2\right)$ and $\delta_2=2\varepsilon+\varepsilon^2$.
\end{thm}
\begin{proof} 
	According to Proposition \ref{prop:PBT2PQP} the existence of a $(d,\varepsilon)$-PBT scheme using $N$ output ports implies the existence of a $\varepsilon$-uPQP with program register size $|P|=\binom{N+d^2-1}{N}$. According to Theorem \ref{thm:appr-uPQP-bound} this implies that
	\begin{equation}
	\binom{N+d^2-1}{N}\ge \frac{1}{2(4\delta_1+\delta_2)^2}\min\left(1, \left(\left(\frac{\pi d}{e}\right)^d\left(4\delta_1+\delta_2\right)^2\right)^{\frac{2}{d^2+1}}\right).
	\end{equation}
	Using $\left(\frac e k\right)^{k+1/2} k!= O(1)$ and taking the logarithm we get
	\begin{eqnarray}
	&&N\log\left(\frac{N+d^2-1}{N}\right)+(d^2-1)\log\left(\frac{N+d^2-1}{d^2-1}\right)\nonumber\\
	&\ge& 2\log\left( \frac{1}{4\delta_1+\delta_2}\right)-1+\frac{2}{d^2+1}\min\left(0,d\log\left(\frac{\pi d}{e}\right)-2\log\left(\frac{1}{4\delta_1+\delta_2}\right)\right)\nonumber\\
	&\ge& 2\frac{d^2-1}{d^2+1}\log\left( \frac{1}{4\delta_1+\delta_2}\right)-1\nonumber\\
	&\ge& \log\left( \frac{1}{4\delta_1+\delta_2}\right)-1.
	\end{eqnarray}
	In the last line we used $d\ge 2$. This can be safely assumed, as otherwise teleportation is trivially possible without communication.
	Therefore we can bound
	\begin{eqnarray}
	&&N\log\left(\frac{N+d^2-1}{N}\right)+(d^2-1)\log\left(\frac{N+d^2-1}{d^2-1}\right)\nonumber\\
	&\le&N+(d^2-1)\log\left(\frac{N}{d^2-1}\left(1+\frac{1}{1-\varepsilon^2}-\frac{1}{N}\right)\right)\nonumber\\
	&\le&2 N,
	\end{eqnarray}
	where we used the bound on $N$ from Corollary \ref{cor:PBT-combound} in the first line and the upper bound on $\varepsilon$ and $\log x\le x$ for all $x>0$ in the second line. Combining the bounds we get
	\begin{equation}
	N\ge\half\left( \log\left( \frac{1}{4\delta_1+\delta_2}\right)-1\right).
	\end{equation}
\end{proof}

\subsubsection{The non-signaling bound}\label{subsubs:nsbound}
The results of this section have been obtained in collaboration with Florian Speelman and Matthias Christandl.

In the setting of PBT, non-signaling prevents perfect success as well. If Alice performs her measurement, but does not communicate the outcome, then Bob can at least guess the outcome. This gives him a chance of being right of $1/N$, where $N$ is the number of ports. On the other hand, Bob still has the marginal of the original resource state, as Alice only acted locally. For an intuitive understanding, let us imagine that perfect PBT were possible. Alice and Bob share a standard maximally entangled state $\phi^+$, and Alice begins to teleport her half to Bob, but does not communicate the measurement outcome. Bob picks a random port and measures whether his part of the preshared entangled state and the port are in a maximally entangled state. This gives him a success probability of
\begin{equation}
p_{\text{ideal}}=\frac{1}{N}+\frac{N-1}{N}\frac{1}{d^2}.
\end{equation}
However, his part of the maximally entangled state and any of his ports are W.L.O.G. in the maximally mixed state, i.e. $p_{\text{real}}=\frac{1}{d^2}$\footnote{We can assume that the resource state is locally maximally mixed by Proposition \ref{prop:symsuffice}.}. This immediately implies the no-go result of \cite{Ishizaka2008}. Let us derive a quantitative statement.
\begin{thm}\label{thm:lowerbound-guessing}
	Let $d\ge 2$ and $\frac{1}{\sqrt{2}}\ge\varepsilon>0$. Any $(d,\varepsilon)$-PBT protocol uses at least
	\begin{equation}
	N\ge \frac{d}{2\sqrt{2\varepsilon}}
	\end{equation}
	output ports.
\end{thm}
\begin{proof}
	Consider an $N$-port $\varepsilon$-PBT protocol, and the situation described above, where Alice and Bob initially share an extra maximally entangled state $\ket{\phi+}_{\bar A\bar B}$ and Alice wants to teleport her share to Bob. Let $\rho_{MB^N\bar B}$ be the state after Alice'ss measurement, but before Bob receives the outcome information. By assumption the protocol has error at most $\varepsilon$ in diamond norm, i.e. there exists a state $\sigma_{B^{N-1}}$ such that
	\begin{align}\label{eq:pbt-condition}
	F(\rho,\rho_{\text{ideal}})\ge& \sqrt{1-\varepsilon/2}\nonumber\\
	=:&\sqrt{1-\delta^2}
	\end{align}
	where
	\begin{equation}
	\left(\rho_{\text{ideal}}\right)_{MB^N\bar B}=\sum_{i=1}^N\proj i_M\otimes\proj{\phi^+}_{B_i\bar B}\otimes\sigma_{B_{i^c}}.
	\end{equation}
	This follows from Uhlmann's theorem: Given $M=i$, the $B_i\bar B$ marginals of the state fulfill the proposed fidelity lower bound, but this marginal is pure in the ideal case, so Uhlmann's theorem immediately implies the existence of the state $\sigma$. We can assume that $\sigma$ does not depend on the port $i$ by Proposition \ref{prop:symsuffice}.
	On the other hand,
	\begin{equation}\label{eq:resourcemarg}
	\rho_{B^n\bar B}=\psi_{B^N}\otimes\tau_{\bar B},
	\end{equation}
	where $\psi_{B^N}$ is Bob's marginal of the resource state $\psi_{A^NB^N}$ used in the protocol. We imagine now that Bob chooses a port $i$ at random and performs the measurement $\proj{\phi^+}_{B_i\bar B}$ vs. $\mathds 1_{B_i\bar B}-\proj{\phi^+}_{B_i\bar B}$. According to Equation \eqref{eq:resourcemarg}, the probability of getting the outcome $+$ is $p_{\text{real}}=1/d^2$. On the other hand we know that he would have a probability 
	\begin{equation}
	p_{\text{ideal}}=\frac{1}{N}+\frac{N-1}{N}\frac{1}{d^2}
	\end{equation}
	if he was measuring on the state $\rho_{\text{ideal}}$. By Equation \eqref{eq:pbt-condition} and the fact that the fidelity is non-decreasing under CPTP maps, we have
	\begin{equation}
	\sqrt{p_{\text{real}}p_{\text{ideal}}}+\sqrt{(1-p_{\text{real}})(1-p_{\text{ideal}})}\ge \sqrt{1-\delta^2}.
	\end{equation}
	Define angles $\alpha_s\in[0,\pi/2]$ by the equation $p_i=\sin^2\alpha_i$, for $i=\text{real},\text{ideal}$. Using a trigonometric angle sum identity we get
	\begin{align}
	\sqrt{1-\delta^2}\le& \sin\alpha_{\text{real}}\sin\alpha_{\text{ideal}}+\cos\alpha_{\text{real}}\cos\alpha_{\text{ideal}}\nonumber\\
	=&\cos(\alpha_{\text{real}}-\alpha_{\text{ideal}}),
	\end{align}
	and therefore
	\begin{equation}\label{eq:bound}
	\sin(\alpha_{\text{ideal}}-\alpha_{\text{real}})\le\delta.
	\end{equation}
	Note that $p_{\text{real}}<p_{\text{ideal}}$ and therefore $\alpha_{\text{real}}<\alpha_{\text{ideal}}$. 
	%	By the concavity of the sine function on the interval $[0, \pi/2]$ we get
	%	\begin{equation}
	%		\alpha_{\text{ideal}}-\alpha_{\text{real}}\le\frac{\pi}{2}\delta.
	%	\end{equation}
	We can bound the difference of the probabilities as follows,
	\begin{align}
	p_{\text{ideal}}-p_{\text{real}}=&\sin^2(\alpha_{\text{ideal}})-\sin^2(\alpha_{\text{real}})\nonumber\\
	=&\frac{\cos(2\alpha_{\text{real}})-\cos(2\alpha_{\text{ideal}})}{2}\nonumber\\
	=&\sin(\alpha_{\text{ideal}}+\alpha_{\text{real}})\sin(\alpha_{\text{ideal}}-\alpha_{\text{real}})\nonumber\\
	\le&\delta\sin(\alpha_{\text{ideal}}+\alpha_{\text{real}}).
	\end{align}
	Here we used trigonometric identities in the second and  third equality, and the inequality is Equation \eqref{eq:bound}.
	To bound the sine term, we use yet another trigonometric identity to get
	\begin{align}
	\sin(\alpha_{\text{ideal}}+\alpha_{\text{real}})=&\sin(\alpha_{\text{ideal}})\cos(\alpha_{\text{real}})+\sin(\alpha_{\text{ideal}})\cos(\alpha_{\text{real}})\nonumber\\
	\le& \sqrt{p_{\text{ideal}}}+\sqrt{p_{\text{real}}}
	\end{align}
	%		\le&\sqrt{\frac{1}{N}+\frac{N-1}{Nd^2}}+\frac 1 d\nonumber\\
	%		\le&\frac{2}{d}+\frac{1}{N},	 
	where we used the definition of $\alpha_i$ and $p_i$ and $\cos x\le 1$ for all $x\in\R$.

	% in the first inequality and $\sqrt{1+x}\le 1+\frac x 2$ and $1/d^2<1$ in the second inequality.
	Combining the last two inequalities yields
	\begin{equation}
	\sqrt{p_{\text{ideal}}}-\sqrt{p_{\text{real}}}\le \delta.
	\end{equation}
	Plugging in the expressions for $p_{\text{ideal}}$ and $p_{\text{real}}$ yields
	\begin{equation}
	N\ge \frac{d-1/d}{d\delta^2+2\delta}.
	\end{equation}
	If $d^2(1-\delta^2)^2>d/(4\delta)$, then the result follows from Corollary \ref{cor:telebound}. For the remaining case let $d^2\le d/(4\delta)$. Then we can bound
	\begin{align}
	\frac{d-1/d}{d\delta^2+2\delta}&\ge\frac{4(d-1/d)}{9\delta}\nonumber\\
	&\ge \frac{d}{3 \delta}\nonumber\\
	&\ge\frac{d}{4\delta}\nonumber\\
	&\ge\frac{d}{2\sqrt{2\varepsilon}},
	\end{align}
	where we used $d\ge 2$ in the second inequality. This finishes the proof.
	
	%	Combining everything we arrive at
	%	\begin{equation}
	%		\frac{1}{N}+\frac{1-N}{Nd^2}-\frac{1}{d^2}\le \delta\left(\frac{2}{d}+\frac{1}{N}\right),
	%	\end{equation}
	%	which implies
	%	\begin{equation}
	%		N\ge (1-1/d^2-\delta)\frac{d}{2\delta}.
	%	\end{equation}	
\end{proof}
We conclude this chapter with a corollary that combines the bounds of Corollary \ref{cor:PBT-combound} and Theorem \ref{thm:lowerbound-guessing} in one statement.
\begin{cor}
	Let $d\ge 2$ and $\frac{1}{\sqrt{2}}\ge\varepsilon>0$. Any $(d,\varepsilon)$-PBT protocol uses at least
	\begin{equation}\label{eq:finallowerbound}
	N\ge d\cdot\max\left(\frac{1}{2\sqrt{2\varepsilon}},d(1-\varepsilon^2)^2\right)
	\end{equation}
	output ports.
\end{cor}
This constitutes a significant improvement over previous bounds. In particular, this constitutes the first bound that diverges uniformly in $d$ for $\varepsilon\to 0$.

\fbox{\begin{minipage}{\textwidth}\vspace{.3cm}
				\begin{center}\large{\textbf{Summary of Chapter \ref{chap:one-shot}}}\end{center}\vspace{.3cm}
				In the first part of this chapter, the decoupling technique and its applicability in one-shot quantum information theory where discussed.
				\begin{itemize}
					\item Catalytic decoupling, a new generalized decoupling theorem was introduced.
					\item Catalytic decoupling allows appending a mixed catalyst that is returned in approximately the same state.
					\item It unifies two previously known techniques, and makes them available as a black box for applications.
					\item That way, one-shot state merging is added to the family of tasks that can be solved optimally using decoupling.
					\item The setting of free ancillary state is quite natural, as shown by the equivalence between erasure of correlations and decoupling that does not hold otherwise.
				\end{itemize}
			In the second part, several results concerning port based teleportation were proven.
			\begin{itemize}
				\item Any optimal protocol can be transformed into a protocol yielding a unitarily covariant channel. The standard protocol is unitarily covariant as well.
				\item This shows that the known results in terms of the entanglement fidelity imply similar bounds in terms of the diamond norm, without the loss of a dimension factor.
				\item Port based teleportation of a $d$-dimensional quantum system with diamond norm error $\varepsilon$ requires at least $d\cdot\max\left(\frac{1}{2\sqrt{2\varepsilon}},d(1-\varepsilon^2)^2\right)$ ports.
				\item As a byproduct of the efforts towards a lower bound on the number of ports, a lower bound on the dimension of the program register $P$ of an $\varepsilon$-approximate universal programmable quantum processor for a $d$-dimensional quantum system is obtained,  $|P|\ge \mathcal O(d^2/\varepsilon^2)$ for large $d$.
			\end{itemize}
\end{minipage}}

 % Quantum Shannon theory -- Catalytic decoupling and port based teleportation

\hchapter{Quantum Cryptography with classical keys -- Non-malleability and authentication}\label{chap:crypto}

The topic of this Chapter are the results of \cite{Alagic2016}, that have been obtained in collaboration with Gorjan Alagic, and large parts of this article are used here unaltered.

%%%%%%%%%%%%%%%%%%%%%%%%%
\section{Introduction}
%%%%%%%%%%%%%%%%%%%%%%%%%

%%%
In its most basic form, encryption ensures secrecy in the presence of eavesdroppers. Besides secrecy, another desirable property is \emph{non-malleability}, which guarantees that an active adversary cannot modify the plaintext by manipulating the ciphertext. In the classical setting, secrecy and non-malleability are independent: there are schemes which satisfy secrecy but are malleable, and schemes which are non-malleable but transmit the plaintext in the clear. If both secrecy and non-malleability is desired, then pairwise-independent permutations provide information-theoretically perfect (one-time) security~\cite{kawachi2011characterization}. In the computational security setting, non-malleability can be achieved by message authentication codes (MACs), and ensures chosen-ciphertext security for authenticated encryption.

In the setting of quantum information, encryption is the task of transmitting quantum states over a completely insecure quantum channel. Information-theoretic secrecy for quantum encryption is well-understood. Non-malleability, on the other hand, has only been studied in one previous work, by Ambainis, Bouda and Winter~\cite{Ambainis2009}. Their definition (which we will call \ABW-non-malleability, or \ABWNM) requires that the scheme satisfies secrecy, and that the ``effective channel'' $\Dec \circ \Lambda \circ \Enc$ of any adversary $\Lambda$ amounts to either the identity map or replacement by some fixed state. In the case of unitary schemes, $\ABWNM$ is equivalent to encrypting with a unitary two-design. Unitary two-designs are a natural quantum analogue of pairwise-independent permutations, and can be efficiently constructed in a number of ways (see, e.g., ~\cite{brandao2012local, dankert2009exact}.)

While quantum non-malleability has only been considered by~\cite{Ambainis2009}, the closely-related task of quantum authentication (where decryption is allowed to reject) has received significant attention (see, e.g.,~\cite{ABE10, barnum2002authentication, broadbent2016efficient, dupuis2012actively, garg2016new}.) The widely-adopted definition of Dupuis, Nielsen and Salvail asks that the averaged effective channel of any adversary is close to a map which does not touch the plaintext~\cite{dupuis2012actively}; we refer to this notion as \DNS-authentication. Recent work by Garg, Yuen and Zhandry~\cite{garg2016new} established another notion of quantum authentication, which they call ``total authentication.'' The notion of total authentication has two major differences from previous definitions: (i.) it asks for success with high probability over the choice of keys, rather than simply on average, and (ii.) it makes no demands whatsoever in the case that decryption rejects. We refer to this notion of quantum authentication as \GYZ-authentication. In~\cite{garg2016new}, it is shown that \GYZ-authentication can be satisfied with unitary eight-designs. 

%%%

In this Chapter, we devise a new definition of non-malleability (denoted \ITNM) for quantum encryption, improving on $\ABWNM$ in a number of ways. First, we consider more powerful adversaries, which can possess side information about the plaintext. Second, we remove the possibility of a ``plaintext injection'' attack, whereby an adversary against an $\ABWNM$ scheme can send a plaintext of their choice to the receiver. Finally, our definition does not demand secrecy; instead, we show that \emph{quantum secrecy is a consequence of quantum non-malleability.} This is a significant departure from the classical case, and is analogous to the fact that quantum authentication implies secrecy~\cite{barnum2002authentication}.

The primary consequence of our work is twofold: first, encryption with unitary two-designs satisfies all of the above notions of quantum non-malleability; second, when equipped with blank ``tag'' qubits, the same scheme also satisfies all of the above notions of quantum authentication. A more detailed summary of the results is as follows. For schemes which have unitary encryption maps, we prove that $\ITNM$ is equivalent to encryption with unitary two-designs, and hence also to $\ABWNM$. For non-unitary schemes, we prove a characterization theorem for $\ITNM$ schemes that shows that \ITNM~implies \ABWNM, and provide a strong separation example between $\ITNM$ and $\ABWNM$ (the aforementioned plaintext injection attack). In the case of \GYZ~authentication, we prove that two-designs (with tags) are sufficient, a significant improvement over the state-of-the-art, which requires eight-designs~\cite{garg2016new}. Moreover, the simulation of adversaries in this proof is efficient, in the sense of Broadbent and Wainewright~\cite{broadbent2016efficient}. Finally, we show that \GYZ authentication implies \DNS-authentication, and that equipping an arbitrary $\ITNM$ scheme with tags yields \DNS-authentication.

An independent work of C. Portmann on quantum authentication gives an alternative proof that \GYZ-authentication can be satisfied by the 2-design scheme \cite{Portmann2016}. 

We now recall the definition of unitary $t$-design, and some relevant variants. We begin by considering three different types of ``twirls.''
\begin{enumerate}
	\item For a finite subset $\mathrm D\subset \mathrm{U}(\hi)$ of the unitary group on some finite dimensional Hilbert space $\hi$, let
	\begin{equation}
	\mathcal T^{(t)}_{\mathrm D}(X)=\frac{1}{|\mathrm D|}\sum_{U\in\mathrm D}U^{\otimes t}X\left(U^\dagger\right)^{\otimes t}
	\end{equation}
	be the associated $t$-twirling channel. If we take the entire unitary group (rather than just a finite subset), then we get the Haar $t$-twirling channel
	\begin{equation}\label{eq:haar-twirl}
	\mathcal T^{(t)}_\mathsf{Haar}(X)=\int U^{\otimes t}X\left(U^\dagger\right)^{\otimes t} \D U.
	\end{equation}
	
	The case $t = 2$ is characterized as follows.	
	\begin{lem}\label{lem:Usquared}
		Let $M_{A^2B}$ be a matrix on $\hi_A^{\otimes 2}\otimes \hi_B$. Then we have the following formula for integration with respect to the Haar measure:
		\begin{equation}\label{eq:schur}
		\int U_A^{\otimes 2}M_{A^2B}(U_A^{\otimes 2})^\dagger \D U=\mathds 1_{A^2}\otimes R^{\mathds 1}_{B}+F_A\otimes R^F_{B},
		\end{equation}
		with 
		\begin{align}\label{eq:formulas}
		R^{\mathds 1}_{B}=&\frac{1}{d(d^2-1)}\left(d\tr_{A^2} M-\tr_{A^2} FM\right),\nonumber\\
		R^{F}_{B}=&\frac{1}{d(d^2-1)}\left(d\tr_{A^2} FM-\tr_{A^2} M\right).
		\end{align}
		
	\end{lem}
	\begin{proof}
		By Schur's lemma,
		\begin{equation}
		\int U_A^{\otimes 2}M_{A^2B}(U_A^{\otimes 2})^\dagger \D U=\Pi_{\bigwedge}\otimes R^{\bigwedge}+\Pi_{\bigvee}\otimes R^{\bigvee}
		\end{equation}
		for some matrices $R^{\bigwedge}$ and $R^{\bigvee}$, where $\Pi_{\bigwedge}$ is the projector onto the antisymmetric subspace and $\Pi_{\bigvee}$ the projector onto the symmetric subspace. As $\Pi_{\bigwedge}=(\mathds 1-F)/2$ and $\Pi_{\bigvee}=(\mathds 1+F)/2$, this implies the correctness of Equation \eqref{eq:schur}. The formulas \eqref{eq:formulas} follow by applying $\tr_A$ and $\tr_A(F(\cdot))$ to both sides of Equation \eqref{eq:schur} and solving the resulting system of 2 equations for $R^{\mathds 1}_{B}$ and $R^{F}_{B}$.
		
	\end{proof}

	\item We define the $U$-$\overline{U}$ twirl with respect to finite $\mathrm D\subset \mathrm{U}(\hi)$ by
	\begin{equation}\label{eq:bar-twirl}
	\overline{\mathcal{T}}_{\mathrm D}(X)
	=\frac{1}{|\mathrm D|}\sum_{U\in\mathrm D}\left(U\otimes\overline U \right)X\left(U\otimes \overline{U}\right)^\dagger.
	\end{equation}
	The analogous $U$-$\overline{U}$ Haar twirling channel is denoted by $\overline {\mathcal T}_{\mathsf{Haar}}$. 
	
	\item The third and final notion is called a channel twirl; this is defined in terms of $U$-$\overline{U}$-twirling, as follows. Given a channel $\Lambda$, let $\eta_\Lambda$ be the CJ state of $\Lambda$. The channel twirl $\mathcal T_{\mathrm{D}}^{ch}(\Lambda)$ of $\Lambda$ is defined to be the channel whose CJ state is $\overline{\mathcal{T}}_{\mathrm D}(\eta_\Lambda)$.
\end{enumerate}

Next, we define three notions of designs, corresponding to the three types of twirl defined above.

\begin{defn}\label{def:designs}
	Let $\mathrm D\subset \mathrm{U}(\hi)$ be a finite set. We define the following.
	\begin{itemize}
		\item If $\bigl\|\mathcal T^{(t)}_{\mathrm D}-\mathcal T^{(t)}_\mathsf{Haar}\bigr\|_\diamond\le\delta$ holds, then $\mathrm{D}$ is a $\delta$-approximate $t$-design.
		\item If $\bigl\|\overline{\mathcal T}_{\mathrm D}-\overline{\mathcal T}_\mathsf{Haar}\bigr\|_\diamond\le\delta$ holds, then $\mathrm{D}$ is a $\delta$-approximate $U$-$\overline{U}$-twirl design.
		\item If $\left\|\mathcal T^{ch}_{\mathrm D}(\Lambda)-T^{ch}_\mathsf{Haar}(\Lambda)\right\|_\diamond\le\delta$ holds for all CPTP maps $\Lambda$, then $\mathrm{D}$ is a $\delta$-approximate channel-twirl design.
	\end{itemize}
\end{defn} 
For all three of the above, the case $\delta = 0$ is called an ``exact design'' (or simply ``design''.) All three notions of design are equivalent in the exact case. In the approximate case they are still connected, but there are some nontrivial costs in the approximation quality. A lemma relating approximate $2$-designs and channel twirl designs can be found in \cite{low2010pseudo} (Lemma 2.2.14). The following Lemma completes the picture in the approximate case.

\begin{lem}\label{lem:channel-twirl-uubar-twirl} 
	Let $\dim\hi=d$, and let $\mathrm D\subset \mathrm{U}(\hi)$ be a finite set.
	\begin{itemize}
		\item If $\mathrm{D}$ is a $\delta$-approximate channel twirl, then it is also a $d\cdot\delta$-approximate $U$-$\overline{U}$ twirl design. 
		\item If $\mathrm{D}$ is a $\delta$-approximate $U$-$\overline{U}$ twirl design, then it is also a $d\cdot\delta$-approximate channel twirl design.
	\end{itemize}
\end{lem}

\begin{proof}
	Let $\mathrm{D}$ be a channel twirl design and $\rho_{AA'B}$ be a quantum state with $A'\cong A$ and $B$ arbitrary. Let furthermore $$\sigma_{AA'B}=(\mathds 1-\rho_{A'})^{1/2}\rho_{A'}^{-1/2}\rho_{AA'B}\rho_{A'}^{-1/2}(\mathds 1-\rho_{A'})^{1/2}.$$ Then $$\eta_{AA'BC}=\frac{1}{d}\left(\rho_{AA'B}\otimes\proj 0_C+\sigma_{AA'B}\otimes \proj 1_C\right)$$ is positive semidefinite, has trace 1 and $\eta_{A'}=\tau_{A'}$, i.e. it is the CJ-state of a quantum channel $\Lambda_{A\to ABC}$. By assumption we have that $$\left\|\mathcal T^{ch}_{\mathrm D}(\Lambda)-T^{ch}_\mathsf{Haar}(\Lambda)\right\|_\diamond\le\delta.$$ This implies in particular that the CJ states of the two channels have trace norm distance at most $\delta$, i.e.
	\begin{align}
	\delta\ge &\left\|\overline{\mathcal T}_{\mathrm D}(\eta)-\overline{\mathcal T}_\mathsf{Haar}(\eta)\right\|_1\nonumber\\
	=& \frac 1 d \left(\left\|\overline{\mathcal T}_{\mathrm D}(\rho)-\overline{\mathcal T}_\mathsf{Haar}(\rho)\right\|_1+\left\|\overline{\mathcal T}_{\mathrm D}(\sigma)-\overline{\mathcal T}_\mathsf{Haar}(\sigma)\right\|_1\right)\nonumber\\
	\ge& \frac 1 d \left\|\overline{\mathcal T}_{\mathrm D}(\rho)-\overline{\mathcal T}_\mathsf{Haar}(\rho)\right\|_1.
	\end{align}
	As $\rho$ was chosen arbitrarily this implies
	\begin{equation}
	\left\|\overline{\mathcal T}_{\mathrm D}-\overline{\mathcal T}_\mathsf{Haar}\right\|_{\diamond}\le d\cdot\delta.
	\end{equation}
	
	Conversely, let $\mathrm D$ be a $U$-$\overline{U}$-twirl design, and let $\Lambda_{A\to A}$ be a CPTP map and $\ket{\psi}_{AA'}\in \hi_A\otimes\hi_{A'}$ be an arbitrary state vector with $\hi_{A'}\cong\hi_A$. For calculating the diamond norm, we can choose $\ket{\psi}_{AA'}=\sqrt{d}\psi_{A'}\ket{\phi^+}_{AA'}$. We bound
	\begin{align*}
	&\left\|\left[\left(\mathcal T^{ch}_{\mathrm D}-\mathcal T^{ch}_{\mathsf{Haar}}\right)(\Lambda)\right]_{A\to A}\left(\proj{\psi}_{AA'}\right)\right\|_1\nonumber\\=&d \left\|\left[\left(\mathcal T^{ch}_{\mathrm D}-\mathcal T^{ch}_{\mathsf{Haar}}\right)(\Lambda)\right]_{A\to A}\left(\psi_{A'}^{\frac{1}{2}}\proj{\phi^+}_{AA'}\psi_{A'}^{\frac{1}{2}}\right)\right\|_1\nonumber\\
	=&d\left\|\psi_{A'}^{\frac{1}{2}}\left[\left(\overline{\mathcal T}_{\mathrm D}-\overline{\mathcal T}_{\mathsf{Haar}}\right)\left(\eta_\Lambda\right)_{AA'}\right]\psi_{A'}^{\frac{1}{2}}\right\|_1\nonumber\\
	\le&d \|\psi_{A'}\|_\infty \left\|\left(\overline{\mathcal T}_{\mathrm D}-\overline{\mathcal T}_{\mathsf{Haar}}\right)\left(\eta_\Lambda\right)_{AA'}\right\|_1\nonumber\\
	\le& d\cdot\delta.
	\end{align*}
	Here we used the mirror lemma in the second equality, the Hölder inequality twice in the first inequality, and the fact that $\|\rho\|_\infty\le 1$ for any quantum state $\rho$, and the assumption, in the last inequality. 
\end{proof}

It is well-known that $\varepsilon$-approximate $t$-designs on $n$ qubits can be generated by random quantum circuits of size polynomial in $n, t$ and $\log(1/\varepsilon)$~\cite{brandao2012local}. In particular, the size of these circuits is polynomial even for exponentially-small choices of $\varepsilon$. We emphasize this observation as follows.
\begin{rem}\label{rem:efficient}
	Fix a polynomial $t$ in $n$. Then, for any $\varepsilon > 0$, a random $n$-qubit quantum circuit consisting of $\poly(n, \log(1/\varepsilon))$ gates (from a universal set) satisfies every notion of $\epsilon$-approximate $t$-design in Definition \ref{def:designs}. 
\end{rem}

For exact designs, we point out two important constructions. First, the prototypical example of a unitary one-design on $n$ qubits is the $n$-qubit Pauli group. For exact unitary two-designs, the standard example is the Clifford group, which is the normalizer of the $n$-qubit Pauli group. Alternatively, the Clifford group is generated by circuits from the gate set $\{H, P, \text{CNOT}\}$. It is well-known that one can efficiently generate exact unitary two-designs on $n$-qubits by building appropriate circuits from this gate set, using $O(n^2)$ random bits~\cite{AG04, dankert2009exact}.

%----------------------------------------------------------------------------------------------------------------------------------------------------------------------------------------------------------------------------

%%%%%%%%%%%%%%%%%%%%%%%%%
\section{The zero-error setting}\label{sec:zero-error}
%%%%%%%%%%%%%%%%%%%%%%%%%

We begin with the zero-error. In the case of secrecy, zero-error means that schemes cannot leak any information whatsoever. In the case of non-malleability, zero-error means that the adversary cannot increase their correlations with the secret by even an infinitesimal amount (except by trivial means; see below.)

\subsection{Perfect secrecy}
%%%%%%%%%%%%%%%%%%%%%%%%%

We begin with a definition of symmetric-key quantum encryption. Our formulation treats rejection during decryption in a slightly different manner from previous literature.

\begin{defn}[Encryption scheme]\label{def:SKQES}
	A symmetric-key quantum encryption scheme (\SKQES) is a triple $(\tau_K,E,D)$ consisting of a classical state $\tau_K\in \opr(\hi_K)$ and a pair of channels
	\begin{align*}
	E &: \opr(\hi_A\otimes \hi_K) \longrightarrow \opr(\hi_C\otimes \hi_K)\\
	D &: \opr(\hi_C\otimes \hi_K)  \longrightarrow \opr(\left(\hi_A\oplus\C\ket{\bot}\right)\otimes \hi_K)
	\end{align*}
	satisfying $[D\circ E](\cdot\otimes \proj k)= (\mathrm{id}_{A} \oplus 0_\bot) \otimes \proj k$ for all $k$. 
\end{defn}
The Hilbert spaces $\hi_A$, $\hi_C$ and $\hi_K$ are implicitly defined by the triple $(\tau_K, E, D)$. The state $\ket{\bot}$ is an error flag that allows the decryption map to report an error. For notational convenience when dealing with these schemes, we set
\begin{alignat*}{2}
E_k &= E(\cdot\otimes \proj k)
\qquad \qquad 
&E_K &= \tr_KE(\cdot\otimes \tau_K)\\
D_k &= D(\cdot\otimes \proj k)
&D_K &= \tr_KD(\cdot\otimes \tau_K)\,.
\end{alignat*}
We will often slightly abuse notation by referring to decryption maps $D_k$ as maps from $C$ to $A$; in fact, the output space of $D_k$ is really the slightly larger space $\hi_{\bar A} := \hi_A \oplus \C\ket{\bot}$. 

In many situations it makes sense to restrict ourselves to \SKQES~that have identical plaintext and ciphertext spaces; due to correctness, this is equivalent to unitarity.
\begin{defn}[Unitary scheme]
	A \SKQES~$(\tau_K,E,D)$ is called unitary if the encryption and decryption maps are controlled unitaries, i.e., if there exists 
	$V = \sum_k U^{(k)}_A\otimes\proj k_K$ such that $E(X)=VXV^\dagger$.
\end{defn}
It is straightforward to prove that, for unitary schemes, \ITS~is equivalent to the statement that the encryption maps $\{E_k\}$ form a unitary 1-design. Note that unitarity of $E_k$ and correctness imply unitarity of $D_k$.

It is natural to define secrecy in the quantum world in terms of quantum mutual information. However, instead of asking for the ciphertext to be uncorrelated with the plaintext as in the classical case, we ask for the ciphertext to be uncorrelated from any reference system. 

\begin{defn}[Perfect secrecy]\label{def:ITS}
	A \SKQES~$(\tau_K, E, D)$ satisfies information - theoretic secrecy (\ITS) if, for any Hilbert space $\hi_B$ and any $\rho_{AB}\in \opr(\hi_A\otimes \hi_B)$, setting $\sigma_{CBK}=E(\rho_{AB}\otimes \tau_K)$ implies $I(C:B)_\sigma=0.$
\end{defn}

We note that, for perfect \ITS, adding side information is unnecessary: the definition already implies that the ciphertext is in product with \emph{any} other system. In particular, if the adversary has some auxiliary system $E$ in their possession, then $I(B:CE)_\sigma=I(B:E)_\sigma$. Several definitions of secrecy for symmetric-key quantum encryption have appeared in the literature, but the above formulation appears to be new. In the case of unitary schemes, information-theoretic quantum secrecy is equivalent to the unitary one-design property.

\begin{prop}\label{prop:one-design}
	A unitary \SKQES~$(\tau_K, E,D)$ is \ITS~if and only if $\mathrm D=\{U_k\}_{k\in \mathcal{K}}$ is a unitary 1-design, where $E_k(X)=U_kXU_k^\dagger$.
\end{prop}
\begin{proof}
	Let $\mathrm D$ be a unitary 1-design. Then the scheme is \ITS by Schur's Lemma (See, e.g., \cite{fulton1991representation}). Conversely, let $(\tau_K, E,D)$ be $\mathsf{ITS}$, i.e. $\mathcal T^{(1)}_{\mathcal D}(\rho_{AB})=(E_K)(\rho_{AB})=\mathcal T^{(1)}_{\mathcal D}(\rho_{A})\otimes\rho_B$. Suppose that there exist $\rho_A, \rho'_A$ such that $\mathcal T^{(1)}_{\mathcal D}(\rho_{A})\neq \mathcal T^{(1)}_{\mathcal D}(\rho'_{A})$. Then $I(A:B)_\sigma>0$ for $\sigma=E(\frac 1 2 \left(\rho_A\otimes\proj 0_B+\rho'_A\otimes \proj 1_B\right))$, which is a contradiction. This implies $T^{(1)}_{\mathcal D}(\rho_{A})=\sigma^0_A$ for all $\rho_A$. But $T^{(1)}_{\mathcal D}(\tau_{A})=\tau_A$, i.e. $\sigma^0_A=\tau_A$. The observation that the positive semidefinite matrices span the whole matrix space finishes the proof.
	
\end{proof}

A different formulation of information-theoretic secrecy is the notion of indistinguishability of ciphertexts, or \IND.

\begin{defn}[\IND]
	A \SKQES $(\tau_K,E,D)$ is \IND, if the encryption with a random key is a constant channel, i.e. $E_K(X)=\tr(X)E_K(\tau_A)$.
\end{defn}

The two notions of perfect secrecy (\ITS~and \IND) are equivalent. 
\begin{prop}
	A \SKQES~$(\tau_K,E,D)$ is \ITS~if and only if it is \IND.
\end{prop}
\begin{proof}
	Let $(\tau_K,E,D)$ be $\mathsf{ITS}$. Then there exists a $\sigma_0^C$ such that for all $\rho_{AB}$ $(E_K)_{KA\to KC}(\rho_{AB})=\sigma_C^0\otimes\rho_B$. This can be seen as follows: By definition every $\rho_{AB}$ is mapped to a product state. Suppose now $E_K(\rho_A)\neq E_K(\rho'_A)$, then $I(C:B)_{\sigma_{CB}}\neq 0$ with $\sigma_{CB}=E_K(\frac 1 2(\rho_A\otimes \proj 0_B+\rho'_A\otimes \proj 1_B))$, a contradiction. From this observation the $\mathsf{IND}$-property follows immediately.
	
	Conversely, let $(\tau_K,E,D)$ be $\mathsf{IND}$, that is in particular
	\begin{equation}
	\left\|(E_K)_{A\to C}(\rho^{(0)}_{A}-\rho^{(1)}_{A})\right\|_1=0
	\end{equation}
	for all $\rho^{(i)}_A$, $i=0,1$, i.e. $(E_K)_{A\to C}=\sigma^0_C\tr(\cdot)$ for some quantum state $\sigma^0_C$, as the set of quantum states spans all of $\opr(\hi_C)$. Now the $\mathsf{ITS}$-property follows immediately.
	
\end{proof}
%
%
% It can be shown that $\ITS$ is equivalent to perfect indistinguishability of ciphertexts (\IND). The latter notion is a special case of an early indistinguishability-based definition of Ambainis et al.~\cite{ambainis2000private}. 

\subsection{A new notion of non-malleability}
%%%%%%%%%%%%%%%%%%%%%%%%%

\subsubsection{Definition}
%%%%%%%%%%%%%%%%%%%%%%%%%
We consider a scenario involving a user Alice and an adversary Mallory. The scenario begins with Mallory preparing a tripartite state $\rho_{ABR}$ over three registers:  the plaintext $A$, the reference $R$, and the side-information $B$. The registers $A$ and $R$ are given to Alice, while Mallory keeps $B$. Alice then encrypts $A$ into a ciphertext $C$ and then transmits (or stores) it in the open. Mallory now applies an attack map
$$
\Lambda : \opr(\hi_C\otimes\hi_B) \to \opr(\hi_C\otimes\hi_{\tilde B})\,.
$$
Mallory keeps the (transformed) side-information $\tilde B$ and returns $C$ to Alice. Finally, Alice decrypts $C$ back to $A$, and the scenario ends. 
\begin{figure}[h]\label{fig:qnm}
	\begin{center}
		\includegraphics[width=0.5\textwidth]{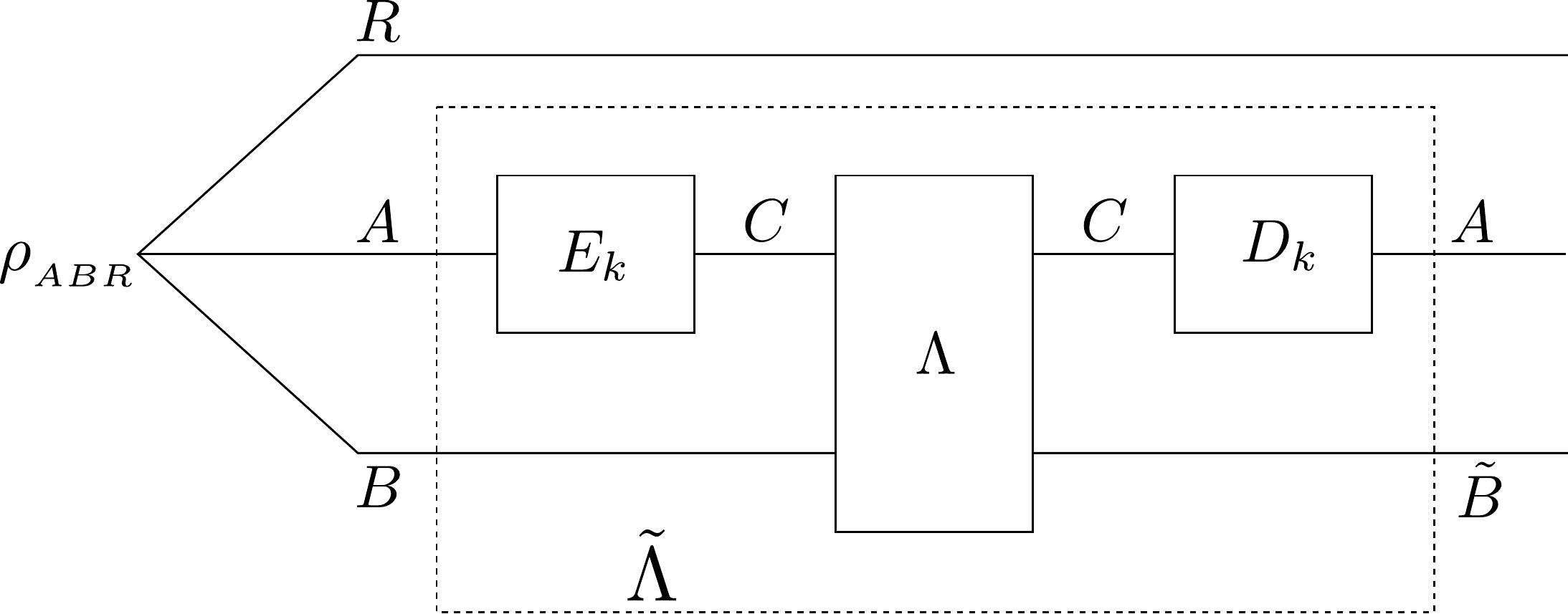}
	\end{center}
	\caption{The quantum non-malleability scenario.}
\end{figure}
We are now interested in measuring the extent to which Mallory was able to increase her correlations with Alice's systems $A$ and $R$. This can be understood by analyzing the mutual information $I(AR:\tilde B)_{\tilde\Lambda(\rho)}$ where $\tilde \Lambda_{AB \to A\tilde B}$ is the \emph{effective channel} corresponding to Mallory's attack:
\begin{equation}\label{eq:effective-channel}
\tilde\Lambda= \tr_K (D\circ\Lambda\circ E)((\cdot)\otimes \tau_K)\,.
\end{equation}
We point out one way in which Mallory can always increase these correlations, regardless of the structure of the encryption scheme. First, she flips a coin $b$, and records the outcome in $B$. If $b=1$, she replaces the contents of $C$ with some fixed state $\sigma_C$, and otherwise she leaves $C$ untouched. One then sees that Mallory's correlations have increased by $h(p_{=}(\Lambda,\rho))$, where $h$ denotes binary entropy and $p_{=}$ is a defined as follows.
\begin{equation}\label{eq:p-equals}
p_{=}(\Lambda, \rho) = 
\tr\left[ (\phi^+_{CC'}\otimes \mathds 1_{\tilde B}) \Lambda (\phi^+_{CC'} \otimes \rho_B)\right]\,.
\end{equation}
This quantity is the inner product between the identity map and the map $\Lambda((\,\cdot\,) \otimes \rho_B)$, expressed in terms of CJ states. Intuitively, it measures the probability with which Mallory chooses to apply the identity map; taking the binary entropy then gives us the information gain resulting from recording this choice.

We are now ready to define information-theoretic quantum non-malleability. Stated informally, a scheme is non-malleable if Mallory can only implement the attacks described above. 

\begin{defn}[Non-malleability]\label{def:qNM}
	A \SKQES~$(\tau_K,E, D)$ is non-malleable (\ITNM) if for any state $\rho_{ABR}$ and any CPTP map
	$\Lambda_{CB \to C{\tilde B}}$, we have
	\begin{equation}\label{eq:ITNM-condition}
	I(AR:\tilde B)_{\tilde\Lambda(\rho)} \leq I(AR:B)_\rho + h(p_{=}(\Lambda,\rho)).
	\end{equation}
\end{defn}

One might justifiably wonder if the term $h(p_{=}(\Lambda, \rho))$ is too generous to the adversary. However, as we showed above, every scheme is vulnerable to an attack which gains this amount of information. This term also appears (somewhat disguised) in the classical setting. In fact, if a classical encryption scheme satisfies Definition \ref{def:qNM} against classical adversaries, then it also satisfies classical information-theoretic non-malleability as defined in \cite{kawachi2011characterization}. 
Finally, as we will show in later sections, Definition \ref{def:qNM} implies $\ABWNM$ (see Definition \ref{def:ABWNM}), and schemes satisfying Definition \ref{def:qNM} are sufficient for building quantum authentication under the strongest known definitions.

\subsubsection{Relationship to classical non-malleability}\label{appendix:nm}

A desirable property of the quantum generalization of any classical security notion is, that the former implies the latter when restricting to the classical scenario.
We begin by recalling the following definition of classical, information-theoretic non-malleability. To this end we set down the notation in the classical case by giving a definition of a classical symmetric key encryption scheme.

\begin{defn}
	Let $\mathcal X, \mathcal C, \mathcal K$ be finite alphabets. A symmetric key encryption scheme (SKES) $(K,E,D)$ is a key random variable (RV) $K$ on $\mathcal{K}$ together with a pair of stochastic maps $E: \mathcal X\times \mathcal K\to\mathcal C$ and $D: \mathcal C\times \mathcal K\to \mathcal X\cup\{\bot\}$ such that
	\begin{equation}
	E(\cdot,k)\circ D(\cdot, k)=\mathrm{id}_{\mathcal X}\footnote{This again is a slight abuse of notation, more correctly the composition of encryption and decryption yields the canonical injection of $X$ into $X\cup\{\bot\}$.},
	\end{equation}
	where $\mathrm{id}_{\mathcal X}$ denotes the identity function of $\mathcal X$ (\emph{correctness}). We write $E_k=E(\cdot, k)$ and analoguously $D_k$.
	
\end{defn}

\begin{defn}[Classical non-malleability \cite{kawachi2011characterization}]\label{def:CNM}
	A classical SKES scheme $(K,E, D)$ is non-malleable if the following holds. For all RVs $X$ on $\mathcal X$ independent of the key, and all RVs $\tilde C$ on $\mathcal C$ independent of the key given $X, C=E(X,K)$ such that $\mathbb P[\tilde C=C]=0$,
	\begin{equation}
	I(\tilde X:\tilde C|XC)=0,
	\end{equation}
	where $\tilde X=D(\tilde C, K)$.
\end{defn}

The following proposition shows that, if a classical scheme satisfies our notion of information-theoretic quantum non-malleability, then it is also classically non-malleable according to Definition \ref{def:CNM} above.

\begin{prop}\label{prop:QNM-CNM}
	Let $(\tau_K,D,E)$ be a SKES~ embedded in to the quantum formalism in the standard way, which satisfies \ITNM. Then it is information theoretically non-malleable.
\end{prop}

\begin{proof}
	Let $B$ be a trivial system and $\rho_{AR}$ be a maximally correlated classical state. Let furthermore $\tilde B\cong CC'$, and let $\Lambda_{C\to C\tilde B}$ be a classical map from $C$ to $C$ that makes a copy of both input and output to the $\tilde B$ register and where $\tr \left(\proj i_C\otimes \mathds 1_{\tilde B}\right) \Lambda_{C\to C\tilde B}(\proj i_C)=0$ for all standard basis vectors $\ket i$ (this is the condition $\mathbb{P}[C=\tilde C]=0$). Then $p_=(\Lambda)=0$, and therefore according to the assumption
	\begin{equation}\label{eq:condmuinfonull}
	I(AR:\tilde B)_{\tilde\Lambda(\rho)}=0,
	\end{equation}
	as $B$ was trivial. Let $X,C,\tilde C, \tilde X$ be random variables corresponding to the systems $A$ in the beginning, $C$ after encryption, $C$ after the application of $\Lambda$ and $A$ after decryption, respectively. Then Equation \eqref{eq:condmuinfonull} reads
	\begin{align}
	0=&I(X\tilde X:C \tilde C)\nonumber\\=&I(X\tilde X:C)+I(X\tilde X:\tilde C|C)\\=&I(X\tilde X:C)+I(X:\tilde C|C)+I(\tilde X:\tilde C|CX),
	\end{align}
	as the input state was maximally classically correlated. The fact that all (conditional) mutual information terms above are non-negative finishes the proof.
	
\end{proof}

\subsubsection{Non-malleability implies secrecy}
%%%%%%%%%%%%%%%%%%%%%%%%%

In the classical case, non-malleability is independent from secrecy: the one-time pad is secret but malleable, and non-malleability is unaffected by appending the plaintext to each ciphertext. In the quantum case, on the other hand, we can show that $\ITNM$ implies secrecy. This is analogous to the fact that ``quantum authentication implies encryption''~\cite{barnum2002authentication}. The intuition is straightforward: (i.) one can only make use of one copy of the plaintext due to no-cloning, and (ii.) if the adversary can distinguish between two computational-basis states (e.g., $\ket{0}$ and $\ket{1}$) then they can also apply an undetectable Fourier-basis operation (e.g., mapping $\ket{+}$ to $\ket{-}$). The technical statement and proof follows.

\begin{prop}\label{thm:ITNMtoITS}
	Let $(\tau_K,E, D)$ be an \ITNM~\SKQES. Then $(\tau_K,E, D)$ is \ITS.
\end{prop}
\begin{proof}
	Let $B$, $\rho_{AB}$, and $\sigma_{CBK} = E(\rho_{AB} \otimes \tau_K)$ be as in the definition of \ITS~(Definition \ref{def:ITS}). We first rename $B$ to $R$. We then consider the non-malleability property in the following special-case scenario. The initial side-information register is empty, the final side-information register $\tilde B$ satisfies $\hi_{\tilde B} \cong \hi_C$, and the adversary map $\Lambda_{C\to C\tilde B}$ is defined as follows. Note that the ``ciphertext-extraction'' map $\Theta_{C\to C\tilde B}=\mathrm{id}_{C\to \tilde B}(\cdot)\otimes \tau_C$ has CJ state $\eta^{\Theta}_{CC'\tilde B}=\phi^+_{C'\tilde B}\otimes \tau_C$. We choose $\Lambda$ so that its CJ state satisfies
	\begin{equation}
	\eta^{\Lambda}_{CC'\tilde B}=\frac{d^2}{d^2-1}\Pi_{CC'}^- \,\eta^{\Theta}_{CC'\tilde B}\,\Pi_{CC'}^-\,.
	\end{equation}
	Applying the above projection to the CJ state of $\Theta$ ensures that $\Lambda$ will have $p_=({\Lambda})=0$ (note: $p_=(\Theta) > 0$.)
	
	Direct calculation of the $C' \tilde B$ marginal of the CJ state of $\Lambda$ yields
	\begin{equation}
	\eta^{\Lambda}_{C'\tilde B}=\frac{d^2-2}{d^2-1}\phi^+_{C'\tilde B}+\frac{1}{d^2-1}\tau_{C'}\otimes\tau_{\tilde B}.
	\end{equation}
	This implies that the output $\sigma_{AR\tilde B}=\tilde\Lambda_{A\to A\tilde B}(\rho_{AB})$ of the effective channel $\tilde \Lambda$ will satisfy
	\begin{equation}\label{eq:eq1}
	\sigma_{\tilde BR}=\frac{d^2-2}{d^2-1}\gamma_{\tilde BR}+\frac{1}{d^2-1}\tau_{\tilde B}\otimes\rho_R,
	\end{equation}
	where $\gamma_{CR}=(E_K)_{A\to C}(\rho_{AR})$ and we used the fact that $\hi_{\tilde B} \cong \hi_C$. By non-malleability, we have
	\begin{equation}\label{eq:eq2}
	I(\tilde B:R)_{\sigma}+I(\tilde B:A|R)_{\sigma}=I(\tilde B:AR)_{\sigma}=0.
	\end{equation}
	In particular, $I(\tilde B:R)_{\sigma}=0$ and thus $\sigma_{\tilde BR}=\sigma_{\tilde B}\otimes \rho_R.$	
	It follows by Equation \eqref{eq:eq1} that
	\begin{equation}
	\gamma_{\tilde BR}=\frac{d^2-1}{d^2-2}\left(\sigma_{\tilde B}-\frac{1}{d^2-1}\tau_{\tilde B}\right)\otimes\rho_R,
	\end{equation}
	i.e., $\gamma_{\tilde BR}$ is a product state. This is precisely the definition of information-theoretic secrecy.
	
\end{proof}

\subsubsection{Characterization of non-malleable schemes}\label{sec:effective-char}
%%%%%%%%%%%%%%%%%%%%%%%%%

Next, we provide a characterization of non-malleable schemes. First, we show that unitary schemes are equivalent to encryption with a unitary 2-design.

\begin{thm}\label{thm:USKQES-NM-2design}
	A unitary \SKQES~$(\tau_K, E, D)$ is \ITNM~if and only if $\{E_k\}_{k\in K}$ is a unitary 2-design.
\end{thm}
This fact is particularly intuitive when the 2-design is the Clifford group, a well-known exact 2-design. In that case, a Pauli operator acting on only one ciphertext qubit will be ``propagated'' (by the encryption circuit) to a completely random Pauli on all plaintext qubits. The plaintext is then maximally mixed, and the adversary gains no information. The Clifford group thus yields a perfectly non-malleable (and perfectly secret) encryption scheme using $O(n^2)$ bits of key~\cite{AG04}.

It will be convenient to prove Theorem \ref{thm:USKQES-NM-2design} as a consequence of our general characterization theorem. To prove it, we need a lemma, a kind of generalized data processing inequality.

\begin{lem}\label{lem:DP-CPTPtensCP}
	Let $\Lambda_{A\to A'}^{(i)}$ be CPTP maps and $\Lambda_{B\to B'}^{(i)}$, $i=1,...,k$ CP maps for $i=1,...,k$ such that $\sum_i\Lambda_{B\to B'}^{(i)}$ is trace preserving. Let $\Lambda^{(i)}_{AB\to A'B'}=\Lambda_{A\to A'}^{(i)}\otimes \Lambda_{B\to B'}^{(i)}$ and define the CPTP maps
	\begin{align}
	\Lambda_{AB\to A'B'C}=&\sum_{i=1}^k\Lambda_{AB\to A'B'}^{(i)}\otimes\proj i_C\text{ and}\nonumber\\
	\Lambda'_{B\to B'C}=&\sum_{i=1}^k\Lambda_{B\to B'}^{(i)}\otimes\proj i_C.
	\end{align}
	Then
	\begin{align}
	I(A':B')_{\Lambda(\rho)}\le I(A:B)_\rho+H(C|A)_{\Lambda'(\rho)}\le I(A:B)_\rho+H(C)_{\Lambda(\rho)}
	\end{align}
	for any quantum state $\rho_{AB}$.
\end{lem}
\begin{proof}
	Let $\rho_{AB}$ be a quantum state and define the following quantum states,  
	\begin{align}
	\sigma_{A'B'CC'}=&\sum_{i=1}^k\Lambda_{AB\to A'B'}^{(i)}(\rho_{AB})\otimes\proj i_C\otimes \proj i_{C'}\text{ and}\nonumber\\
	\sigma'_{AB'CC'}=&\sum_{i=1}^k\Lambda_{B\to B'}^{(i)}(\rho_{AB})\otimes\proj i_C\otimes \proj i_{C'},
	\end{align}
	i.e. $\sigma$ and $\sigma'$ are $\Lambda$ and $\Lambda'$ applied to $\rho$ with an extra copy of $C$.
	We bound
	\begin{align}
	I(A':B')_\sigma\le&I(A'C:B'C')_\sigma\nonumber\\
	\le&I(AC:B'C')_{\sigma'}\nonumber\\
	=&I(A:B'C')_{\sigma'}+I(C:B'C'|A)_{\sigma'}\nonumber\\
	=&I(A:B'C')_{\sigma'}+H(C|A)_{\sigma'}\nonumber\\
	\le&I(A:B)_\rho	+H(C|A)_{\sigma'}\nonumber\\
	\le&I(A:B)_\rho	+H(C)_{\sigma'}.
	\end{align}
	The first and second inequality are due to the data processing inequality of the quantum mutual information. The first equality is the chain rule for the quantum mutual information. The second equation is the fact that the mutual information of two copies of a classical system cannot be increased by adding systems on one side, relative to any conditioning system, and is equal to its entropy. The third inequality is due to the data processing inequality for the quantum mutual information again and the last inequality is the fact that conditioning can only decrease entropy. Using the definition of $\sigma$ and $\sigma'$ this implies the claim.
\end{proof}

Non-malleable quantum encryption schemes can be characterized in terms of the effective map an attack achieves.

\begin{thm}\label{thm:effective-char}
	Let $(\tau, E, D)$ be a \SKQES. Then $(\tau, E, D)$ is $\ITNM$ if and only if, for any attack $\Lambda_{CB\to C\tilde B}$, the effective map $\tilde \Lambda_{AB\to A\tilde B}$ has the form
	\begin{equation}\label{eq:effective-char}
	\tilde \Lambda =\id_A\otimes \Lambda'_{B\to\tilde B}+\frac{1}{|C|^2-1}\left(|C|^2\left\langle D_K(\tau)\right\rangle-\id\right)_A\otimes \Lambda''_{B\to\tilde B}
	\end{equation}
	where $\Lambda' =\tr_{CC'}[\phi^+_{CC'}\Lambda(\phi^+_{CC'}\otimes (\cdot))]$ and $\Lambda'' =\tr_{CC'}[\Pi^-_{CC'}\Lambda(\phi^+_{CC'}\otimes (\cdot))].$
	
\end{thm}

We remark that the forward direction holds even if $(\tau, E, D)$ only fulfills the $\ITNM$ condition (Equation \eqref{eq:ITNM-condition}) against adversaries with empty side-information $B$. The full proof is somewhat technical and will be given in the approximate setting (Theorem \ref{thm:eps-effective-char}). Here we give a proof sketch to convey the idea of the proof without the full bulky analysis being in the way.

\textit{Proof sketch.}
The first implication, i.e. $\ITNM$ implies Equation \eqref{eq:effective-char}, is best proven in the Choi-Jamoi\l kowski picture. Here, any $\SKQES$ defines a map
\begin{equation}
\mathcal E_{CC'\to AA'}=\frac{1}{|K|}\sum_k D_k\otimes E_k^T,
\end{equation}
where the transpose $E_k^T$ is the map whose Kraus operators are the transposes of the Kraus operators of $E_k$ (in the standard basis). Our goal is to prove that this map essentially acts like the $U\bar U$-twirl. We decompose the space $\hi_C^{\otimes 2}$ as
\begin{equation}
\hi_C^{\otimes 2}=\C\ket{\phi+}\oplus\supp\Pi^-
\end{equation}
which induces a decomposition of 
\begin{align}
\opr(\hi_C^{\otimes 2})
&=\C\proj{\phi^+}\oplus\left\{\ketbra{\phi^+}{v}\Big|\braket{\phi^+}{v}=0\right\}\nonumber\\
&\oplus\left\{\ketbra{v}{\phi^+}\Big|\braket{\phi^+}{v}=0\right\}\oplus\left\{X\in B\Big|\bra{\phi^+}X=X\ket{\phi^+}=0\right\}.
\end{align}
On the first and last direct summands, the correct behavior of $\mathcal E$ is easy to show: the first one corresponds to the identity, and the last one to the non-identity channels $\Lambda$ with $p_=(\Lambda)=0$. For the remaining two spaces, we employ Lemma \ref{lem:SKQES-char} which shows that the encryption map of any valid encryption scheme has the form of appending an ancillary mixed state and then applying an isometry. Evaluating $\mathcal E(\ketbra{\phi+}{v})$ for $\braket{\phi^+}{v}=0$ reduces to evaluating the adjoint of the average encryption map, $E^\dagger_K$, on traceless matrices. It is, however, easy to verify that $$\tr_A\mathcal E_{CC'\to AA'}(\sigma_C\otimes(\cdot)_{C'})=(E_K^T)_{C'\to A'}$$ for any $\sigma_C$. This can be used to prove $E_K=\langle\tau_C\rangle$ by observing that $\bra{\phi^+}_{CC'}\sigma_C\otimes\rho_{C'}\ket{\phi^+}_{CC'}=\tr(\sigma_C\rho_{C})$, so for rank-deficient $\rho$ we can calculate $\mathcal E_{CC'\to AA'}(\sigma_C\otimes(\cdot)_{C'})$ using what we have already proven.

The other direction is proven by a simple application of  Lemma \ref{lem:DP-CPTPtensCP}.

\hspace{.4cm}

The fact that $\ITNM$ is equivalent to 2-designs (for unitary schemes) is a straightforward consequence of the above.

\begin{proof} (of Theorem \ref{thm:USKQES-NM-2design})
	First, assume $(\tau_K, E, D)$ is a unitary $\ITNM$ $\SKQES$ with $E_k=U_k(\cdot)U_k^\dagger$. Then it has $|C|=|A|$, and $D_K(\tau_C)=\tau_A$, so the conclusion of Theorem \ref{thm:effective-char} in this case (i.e., Equation \eqref{eq:effective-char}) is exactly the condition for $\{U_k\}$ to be an exact channel twirl design and therefore an exact 2-design. If $(\tau_K, E, D)$, on the other hand, is a unitary $\SKQES$ and $\{U_k\}$ is a 2-design, then Equation \eqref{eq:effective-char} holds and the scheme is therefore $\ITNM$ according to Theorem \ref{thm:effective-char}.
\end{proof}

\subsubsection{Relationship to ABW non-malleability}\label{sec:ABW-exact}
%%%%%%%%%%%%%%%%%%%%%%%%%
Ambainis, Bouda and Winter give a different definition of non-malleability, expressed in terms of the effective maps that an adversary can apply to the plaintext by acting on the ciphertext produced from encrypting with a random key~\cite{Ambainis2009}. According to their definition, a scheme is non-malleable if the adversary can only apply maps from a very restricted class \emph{when averaging over the key, and without giving side information to the active adversary}. Let us recall their definition here. 

First, given a \SKQES~$(\tau_K, E,D)$, we define the set $S := \{ D_K(\sigma_C) \,|\, \sigma_C \in \opr(\hi_C)\}$ consisting of all valid average decryptions. We then define the class $C^S_A$ of all ``replacement channels''. This is the set of CPTP maps belonging to the space
\begin{equation}
\spa_{\R}\{\mathrm{id}_A, (X\mapsto \tr(X)\sigma_A) : \sigma_A\in S\}\,.
\end{equation}
We then make the following definition, which first appeared in~\cite{Ambainis2009}.
\begin{defn}[ABW non-malleability]\label{def:ABWNM}
	A \SKQES~$(\tau_K, E,D)$ is ABW-non-malleable (\ABWNM) if it is \ITS, and for all channels $\Lambda_{C\to C}$, we have
	\begin{equation}\label{eq:abw}
	\tr_K \left[D_{CK\to AK} \circ \Lambda_{C\to C} \circ E_{AK\to CK}(\,\cdot\,\otimes\tau_K)\right] \,\in\, C_A^S.
	\end{equation}
\end{defn}

As indicated in~\cite{Ambainis2009}, an approximate version of Equation \eqref{eq:abw} is obtained by considering the diamond-norm distance between the effective channel and the set $C_A^S$; this implies the possibility of an auxiliary reference system, which is denoted $R$ in $\ITNM$. We emphasize that this reference system is not under the control of the adversary. In particular, \ABWNM~does not allow for adversaries which maintain \emph{and actively use} side information about the plaintext system. 

Another notable distinction is that~\cite{Ambainis2009} includes a secrecy assumption in the definition of an encryption scheme; under this assumption, it is shown that a unitary \SKQES~is \ABWNM~if and only if the encryption unitaries form a 2-design. By our Theorem \ref{thm:USKQES-NM-2design}, we see that \ITNM~and \ABWNM~are equivalent in the case of unitary schemes. So, in that case, $\ABWNM$ actually ensures a much stronger security notion than originally considered by the authors of~\cite{Ambainis2009}.

In the general case, $\ITNM$ is strictly stronger than $\ABWNM$. First, by comparing the conditions of Definition \ref{def:ABWNM} to Equation \eqref{eq:effective-char}, we immediately get the following corollary of Theorem \ref{thm:effective-char}.
\begin{cor}\label{cor:NM-implies-ABWNM}
	If a $\SKQES$ satisfies $\ITNM$, then it also satisfies \ABWNM.
\end{cor}
Second, we give a separation example which shows that \ABWNM~is highly insecure; in fact, it allows the adversary to ``inject'' a plaintext of their choice into the ciphertext. This is insecure even under the classical definition of information-theoretic non-malleability of \cite{kawachi2011characterization}. We now describe the scheme and this attack.
\begin{ex}\label{ex:injection}
	Suppose $(\tau_K, E, D)$ is a \SKQES~that is both $\ITNM$ and $\ABWNM$. Define a modified scheme $(\tau_K, E', D')$, with enlarged ciphertext space $\hi_{C'} = \hi_{C}\oplus\hi_{\hat A}$ (where $\hi_{\hat A}\cong\hi_A$) and encryption and decryption defined by
	\begin{align*}
	E'(X) &= E(X)_{C}\oplus 0_{\hat A}\\
	D'(X) &= D_{CK\to AK}(\Pi_{C}X\Pi_{C})+ \mathrm{id}_{\hat AK\to AK}(\Pi_{\hat A}X\Pi_{\hat A})\,.
	\end{align*}
	Then $(\tau_K, E', D')$ is \ABWNM~but not \ITNM.
\end{ex}
While encryption ignores $\hi_{\hat A}$, decryption measures if we are in $C$ or $\hat A$ and then decrypts (in the first case) or just outputs the contents (in the second case.) This is a dramatic violation of $\ITNM$: set $\hi_{\tilde B}\cong\hi_{A}$, trivial $B$ and $R$, and 
\begin{equation}
\Lambda_{C'\to C' \tilde B}(X)=\tr(X)0_{C}\oplus \proj{\phi^+}_{\hat A\tilde B}\,;
\end{equation}
it follows that, for all $\rho$,
\begin{equation}
I(AR:\tilde B)_{\tilde\Lambda(\rho)}=2\log|A|\gg h(|C'|^{-2}) = h(p_=(\Lambda, \rho))\,.
\end{equation}

Now let us show that $(\tau, E', D')$ is still $\ABWNM$. Let $\Lambda_{C'\to C'}$ be an attack, i.e., an arbitrary CPTP map. Then the effective plaintext map is
\begin{equation}
\tilde{\Lambda}_{A\to A}=D\circ \Lambda^C_{C\to C}\circ E+\Lambda^{\hat A}_{C\to A}\circ E,
\end{equation}
where $\Lambda^C(X_C)=\Pi_C\Lambda(X_C\oplus 0_{\hat A})\Pi_C$ and $\Lambda^{\hat A}(X_C)=\mathrm{id}_{\hat{A}\to A}(\Pi_{\hat A}\Lambda(X_C\oplus 0_{\hat A})\Pi_{\hat A})$. Since $(\tau, E, D)$ is $\ITS$ (Theorem \ref{thm:ITNMtoITS}), there exists a fixed state $\rho^0_C$ such that $E_K(\rho_A)=\rho^0_C$ for all $\rho_A$. Since $(\tau, E, D)$ is $\ABWNM$, we also know that
$$
\tr_K\circ D\circ \Lambda^C_{C\to C}\circ E=\tilde{\Lambda}_1 \in C_A^S\,,
$$
with $S=\{ D_K(\sigma_C)\,|\,\sigma_C \in \opr(\hi_C)\}$. We therefore get
\begin{equation}
\tilde{\Lambda}_{A\to A}=\tilde{\Lambda}_1+\langle \Lambda^{\hat A}(\rho^0_C)\rangle\in C_A^{S'},
\end{equation}
with $S'=\{ D'_K(\sigma_{C'})\,|\,\sigma_{C'} \in \opr(\hi_{C'})\}.$ This is true because $S'$ contains all constant maps, as $D'_K(0_{C}\oplus\rho_{\hat A})=\rho_A$.

%%%%%%%%%%%%%%%%%%%%%%%%%
\section{The approximate setting}\label{sec:appr}
%%%%%%%%%%%%%%%%%%%%%%%%%

We now consider the case of approximate non-malleability. Approximate schemes are relevant for several reasons. First, an approximate scheme with negligible error can be more efficient than an exact one: the most efficient construction of an exact 2-design requires a quantum circuit of $O(n\log n\log\log n)$ gates \cite{cleve2016near}, where approximate 2-designs can be achieved with linear-length circuits \cite{dankert2009exact}. Second, in practice, absolutely perfect implementation of all quantum gates is too much to expect---even with error-correction. Third, when passing to authentication one must allow for errors, as it is always possible for the adversary to escape detection (with low probability) by guessing the secret key. 

For all these reasons, it is important to understand what happens when the perfect secrecy and perfect non-malleability requirements are slightly relaxed. In this section, we show that our definitions and results are stable under such relaxations, and prove several additional results for quantum authentication.
We begin with the approximate-case analogue of perfect secrecy. 

\begin{defn}[Approximate secrecy, $\varepsilon$-\ITS~and $\varepsilon$-\IND]\label{def:eps-ITS}
	Fix $\varepsilon > 0$. A \SKQES~$(\tau_K, E, D)$ is $\varepsilon$-approximately secret ($\epsilon$-\ITS) if, for any $\hi_B$ and any $\rho_{AB}$, setting $\sigma_{CBK}=E(\rho_{AB}\otimes \tau_K)$ implies $I(C:B)_\sigma \leq \varepsilon.$ $(\tau_K, E, D)$ is called $\varepsilon$-\IND, if for all pairs of quantum states $\rho_{AB}$ and $\sigma_{AB}$,
	\begin{equation}
		\frac 1 2\left\|\left(E_K\right)_{A\to C}(\rho_{AB}-\sigma_{AB})\right\|_1\le \frac 1 2\|\rho_{B}-\sigma_{B})\|_1+\varepsilon.
	\end{equation}
\end{defn}
Analogously to the exact case, unitary schemes satisfying approximate secrecy are equivalent to approximate one-designs.

Using Pinsker's inequality (Lemma \ref{lem:pinsker}) and the Alicki-Fannes inequality (Lemma \ref{lem:fannes}) one can also show that $\varepsilon$-\ITS~implies $(4\sqrt{2\varepsilon})$-\IND, and that $\varepsilon$-\IND~implies $(4h(\varepsilon)+6\varepsilon\log|A|)$-\ITS.

The following example shows that, in the approximate setting, there exist unitary schemes which can only be broken with access to side information. This is in contrast to the exact setting, where side information is unhelpful.

\begin{ex}
	Let $\mathrm{D}=\{\hat U^{(k)}\}_k$ be an exact unitary 2-design on a Hilbert space $\hi_A$ of even dimension $|A|=d$. Let $V_A$ be the unitary matrix with
	$$
	V_{j\,\,(d-j+1)}=i\,\cdot\,\mathrm{sign}(d-2j+1)
	$$
	for all $j$, and all other entries equal to zero. Set, for each $k$, 
	$$
	U^{(k)}=\hat U^{(k)}V\left(\hat U^{(k)}\right)^T\,.
	$$
	Define a \SKQES~by $E_k(X)=U^{(k)}X\left(U^{(k)}\right)^\dagger$ and $D_k=E_k^\dagger$.
\end{ex}

It is easy to check that $E_K(X_A)=\frac{1}{d-1}(d\tau_A-X^T_A)$ is the Werner-Holevo channel \cite{werner2002counterexample}. 
For any two quantum states $\rho_A, \rho'_A$,
\begin{align}
\left\|E_K(\rho_A)-E_K(\rho'_A)\right\|_1=&\frac{1}{d-1}\left\|{\rho'}^T_A-\rho^T_A\right\|_1\nonumber\\
\le&\frac{2}{d-1}.
\end{align}

On the other hand
\begin{align}
\left\|E_K-\langle\tau\rangle\right\|_\diamond=&\frac{1}{d-1}\left\|\langle\tau\rangle-\theta\right\|_\diamond\nonumber\\
\ge&\frac{1}{d-1}\left(\left\|\theta\right\|_\diamond-\left\|\langle\tau\rangle\right\|_\diamond\right)\nonumber\\
\ge&1.
\end{align}
The last step follows because the transposition map, here denoted by $\theta$, has diamond norm $d$\footnote{this is well known, the lower bound needed here can be obtained by applying the transposition to half of a maximally entangled state}.
By the definition of the diamond norm there exists a state $\rho_{AB}$ such that
\begin{equation}
\left\|E_K(\rho_{AB}-\tau_A\otimes\rho_B)\right\|_1=\left\|\left(E_K-\langle\tau_A\rangle\right)(\rho_{AB})\right\|_1\ge 1.
\end{equation}
In other words, we have exhibited a \SKQES~that is not $\varepsilon$-\IND~for any $\varepsilon<1/2$, but adversaries without side information achieve only negligible distinguishing advantage.

\subsection{Approximate non-malleability}
%%%%%%%%%%%%%%%%%%%%%%%%%

\subsubsection{Definition}
%%%%

We now define a natural approximate-case analogue of \ITNM, i.e., Definition \ref{def:qNM}. Let us briefly recall the context. The malleability scenario is described by systems $A$, $C$, $B$ and $R$ (respectively, plaintext, ciphertext, side-information, and reference), an initial tripartite state $\rho_{ABR}$, and an attack channel $\Lambda_{CB\to C\tilde B}$. Given this data, we have the effective channel $\tilde \Lambda_{AB \to A\tilde B}$  defined in Equation \eqref{eq:effective-channel} and the ``unavoidable attack'' probability $p_=(\Lambda, \rho)$ defined in Equation \eqref{eq:p-equals}. The new definition now simply relaxes the requirement on the increase of the adversary's mutual information.

\begin{defn}[Approximate non-malleability]\label{def:eps-qNM}
	A \SKQES~$(\tau_K,E, D)$ is $\varepsilon$-non-malleable ($\varepsilon$-\ITNM) if for any state $\rho_{ABR}$ and any CPTP map $\Lambda_{CB \to C\tilde B}$, we have 
	\begin{equation}\label{eq:eps-ITNM-condition}
	I(AR:\tilde B)_{\tilde\Lambda(\rho)}
	\leq I(AR:B)_\rho + h(p_{=}(\Lambda,\rho))+\varepsilon.
	\end{equation}
\end{defn}

We record the approximate version of Theorem \ref{thm:ITNMtoITS}, i.e., non-malleability implies secrecy. The proof is a straightforward adaptation of the exact case.
\begin{prop}\label{thm:eps-ITNMtoITS}
	Let $(\tau_K,E, D)$ be an $\varepsilon$-\ITNM~\SKQES. Then $(\tau_K,E, D)$ is $2\varepsilon$-\ITS.
\end{prop}

\subsubsection{Non-malleability with approximate designs}
%%%%

Continuing as before, we now generalize the characterization theorems of non-malleability (Theorem \ref{thm:effective-char} and Theorem \ref{thm:USKQES-NM-2design}) to the approximate case.

\begin{thm}\label{thm:eps-effective-char}
	Let $(\tau, E, D)$ be a \SKQES~with ciphertext dimension $|C|=2^{m}$ and $r>0$ a sufficiently large constant. Then the following holds:
	\begin{enumerate}
		\item If $(\tau, E, D)$ is $2^{-r m}$-$\ITNM$, then for any attack $\Lambda_{CB\to C\tilde B}$, the effective map $\tilde \Lambda_{AB\to A\tilde B}$ is $2^{-\Omega(m)}$-close (in diamond norm) to
		\begin{equation*}\label{eq:eps-effective-char}
		\tilde \Lambda^{\mathrm{exact}}_{AB\to A\tilde B}=\id_A\otimes \Lambda'_{B\to\tilde B}+\frac{1}{|C|^2-1}\left(|C|^2\left\langle D_K(\tau)\right\rangle-\id\right)_A\otimes \Lambda''_{B\to\tilde B},
		\end{equation*}
		with $\Lambda'$, $\Lambda''$ as in Theorem \ref{thm:effective-char}.
		
		\item Suppose that $\log|R| = O(2^m)$, where $R$ is the reference register in Definition \ref{def:eps-qNM}. Then there exists a constant $r$, such that if every attack $\Lambda_{CB\to C\tilde B}$ results in an effective map that is $2^{-r m}$-close to $\tilde \Lambda^{\mathrm{exact}}$, then the scheme is $2^{-\Omega(m)}$-\ITNM.
	\end{enumerate}
\end{thm}
The condition on $R$ required for the second implication is necessary, as the relevant mutual information can at worst grow proportional to the logarithm of the dimension according to the Alicki-Fannes inequality (Lemma \ref{lem:fannes}). This is not a very strong requirement, as it should be relatively easy for the honest parties to put a bound on their total memory.

Before commencing the lengthy proof of a version of Theorem \ref{thm:eps-effective-char}, we record the corollary which states that, for unitary schemes, approximate non-malleability is equivalent to encryption with an approximate 2-design. The proof proceeds as in the exact case, now starting from Theorem \ref{thm:eps-effective-char}.
\begin{thm}\label{thm:eps-USKQES-NM-2design}
	Let $\Pi = (\tau_K, E, D)$ be a unitary $\SKQES$ for $n$-qubit messages and $f:\N\to\N$ a function that grows at most exponential. Then there exists a constant $r>0$ such that
	\begin{enumerate}
		\item If $\{E_k\}$ is a $\Omega(2^{-rn})$-approximate 2-design and  $\log|R|\le f(n)$, then $\Pi$ is $2^{-\Omega(n)}$-$\ITNM$.
		\item If $\Pi$ is $\Omega(2^{-rn})$-$\ITNM$, then $\{E_k\}_{k\in K}$ is a $2^{-\Omega(n)}$-approximate 2-design.
	\end{enumerate}
\end{thm}

\subsubsection*{Proof of characterization theorem}\label{app:proofs}
%=======================================

This section is dedicated to proving the characterization theorem for non-malleable quantum encryption schemes, i.e., Theorem \ref{thm:eps-effective-char}. We begin with three preparatory lemmas.

The first lemma characterizes quantum channels that are invertible on their image such that the inverse is a quantum channel. This is exactly the set of possible encryption maps.

\begin{lem}\label{lem:SKQES-char}
	Let $(\tau_K,E, D)$ be a \SKQES. Then the encryption maps have the structure
	\begin{equation}
	\left(E_k\right)_{A\to C}=\left(V_k\right)_{A\hat C\to C}\left((\cdot)\otimes\sigma^{(k)}_{\hat C}\right)\left(V_k\right)_{A\hat C\to C}^\dagger,
	\end{equation}
	and the decryption maps hence must have the form
	\begin{align}\label{eq:Dkchar}
	&\left(D_k\right)_{C\to A}=\tr_{\hat C}\left[\Pi_{\supp \sigma^{k}}\left(V_k\right)_{A\hat C\to C}^\dagger\left(\cdot\right)\left(V_k\right)_{A\hat C\to C}\right]\nonumber\\
	&+\left(\hat D_k\right)_{ C\to A}\left[\Pi_f\left(\cdot\right)\Pi_f\right]
	\end{align}
	for some quantum states $\sigma^{(k)}_{\hat C}$, isometries $(V_k)_{C\to A\hat C}$, and some CPTP map $\hat D_k$. The projector $\Pi_f$ is the projector onto the subspace of invalid ciphertexts,
	\begin{equation}
		\Pi_f=\mathds 1_{ C}-\left(V_k\right)_{A\hat C\to C}(\Pi_{\supp \sigma^{k}})\left(V_k\right)_{A\hat C\to C}^\dagger,
	\end{equation}
	and $\Pi_{\supp \sigma^{k}}$ is the projector onto the support of $\sigma_k$.
\end{lem}

\begin{proof}
	Let $\ket{\phi}_{AA'}$ be some bipartite pure state. By the correctness of the scheme and the data processing inequality of the mutual information (see Lemma \ref{lem:DP-CPTPtensCP}),  we have that 
	\begin{align}
	2H(A')_\phi=&2H(A')_{E_k(\phi)} \ge I(A':C)_{E_k(\phi)}\nonumber\\
	\ge&I(A':A)_{D_k(E_k(\phi))}=2H(A')_{\phi},
	\end{align}
	i.e. $2H(A')_\phi = I(A':C)_\psi$. The first inequality is an easy-to-check elementary fact. It is easy to see that this only holds if the purification of $\psi_{A'}=\phi_{A'}$ lies entirely in $C$, i.e. there exists an isometry $U^{(k)}_{C\to A\hat C}$ such that
	\begin{equation}
	U^{(k)}_{C\to A\hat C}E_k(\phi)\left(U^{(k)}_{C\to A\hat C}\right)^\dagger=\phi_{AA'}\otimes \sigma^{(k)}.
	\end{equation}
	note that by the linearity of $E_k$ and $U^{(k)}$, $\sigma^{(k)}$ cannot depend on $\phi$. As the state $\phi$ was arbitrary, this implies that
	\begin{equation}
	E_k=\left(U^{(k)}_{C\to A\hat C}\right)^\dagger\left((\cdot)\otimes\sigma^{(k)}\right)U^{(k)}_{C\to A\hat C},
	\end{equation}
	i.e. $E_k$ has the claimed form with $V_k=\left(U^{(k)}\right)^\dagger$. The form of the decryption map then follows immediately by correctness.
	
\end{proof}

\begin{lem}\label{lem:E-of-offd}
	For any $\SKQES$ $(\tau, E,D)$ the map $\mathcal{E}:=|K|^{-1}\sum_{k}D_k\otimes E_{k}^T$ satisfies
	$$\mathcal{E}\left(\proj{\phi^+}_{CC'}(X_C\otimes \id_{C'})\right)=\frac{|A|}{|C|}\proj{\phi^+}_{AA'}(E_K^{\dagger}(X)\otimes \id_{A'})$$
\end{lem}
\begin{proof}
	Using  Lemma \ref{lem:SKQES-char} we derive an expression for $E_k^T$,
	\begin{align}
	\overline{\tr\left[E_k^T(\overline Y)\overline X\right]}=&\tr\left[E_k^\dagger(Y)X\right]\nonumber\\
	=&\tr\left[Y_CE_k(X_A)\right]\nonumber\\
	=&\tr\left[Y_CV_k(X_A\otimes\sigma^{(k)}_{\hat C})V_k^\dagger\right]\nonumber\\
	=&\tr\left[\tr_{\hat C}\left(\sigma^{(k)}_{\hat C}V_k^\dagger YV_k\right)X\right]\nonumber\\
	=&\overline{\tr\left[\tr_{\hat C}\left(\overline\sigma^{(k)}_{\hat C}V_k^T \overline Y\overline V_k\right)\overline X\right]}.
	\end{align}
	Here we use the definition of the adjoint in the second equality and the cyclicity of the trace in the third equality. Hence
	\begin{equation}\label{eq:E-transpose}
	E_k^T=\tr_{\hat C}\left(\overline\sigma^{(k)}_{\hat C}V_k^T (\cdot)\overline V_k\right)
	\end{equation}
	Define $\Pi_k=\Pi_{\supp\sigma^{(k)}}$ to be the projector onto the support of $\sigma_k$. In the following we omit the subscripts of CP maps and isometries to save space. 
	We start with one summand in the sum defining $\mathcal E$ and omit the second summand from the expression for $D_k$ in Equation \eqref{eq:Dkchar}
	\begin{align}\label{eq:summand1}
	&\tr_{\hat C\hat C'}\left[\left(\Pi_k\right)_{\hat C}\otimes\overline\sigma^{(k)}_{\hat C'}\right]\left[V_k^\dagger\otimes V_k^T\right]\phi^+_{CC'}X_C\left[V_k\otimes\overline V_k\right]\nonumber\\
	=&\frac{|A||\hat C|}{|C|}\tr_{\hat C\hat C'}\left[\left(\Pi_k\right)_{\hat C}\otimes\overline\sigma^{(k)}_{\hat C'}\right]\left[\phi^+_{AA'}\otimes \phi^+_{\hat C\hat C'}\right]V_k^\dagger X_CV_k\nonumber\\
	=&\frac{|A||\hat C|}{|C|}\tr_{\hat C\hat C'}\phi^+_{AA'}\otimes \phi^+_{\hat C\hat C'}V_k^\dagger X_CV_k\sigma^{(k)}_{\hat C}\nonumber\\
	=&\frac{|A||\hat C|}{|C|}\phi^+_{AA'} \bra{\phi^+}_{\hat C\hat C'}V_k^\dagger X_CV_k\sigma^{(k)}_{\hat C}\ket{\phi^+}_{\hat C\hat C'}\nonumber\\
	=&\frac{|A|}{|C|}\phi^+_{AA'}\tr_{\hat C}V_k^\dagger X_CV_k\sigma^{(k)}_{\hat C}\nonumber\\
	=&\frac{|A|}{|C|}\phi^+_{AA'}\left(E_k\right)_{C\to A}^{\dagger}(X_C).
	\end{align}
	Here we have used Lemma \ref{lem:genmirr} for the first  and the second equality, and in the fourth equality is due to the elementary fact that
	\begin{equation}
	\bra{\phi^+}_{\hat C\hat C'}Y_{\hat C}\ket{\phi^+}_{\hat C\hat C'}=\frac{1}{|\hat C|}\tr Y.
	\end{equation}
	Finally we have used the complex conjugate of Equation \eqref{eq:E-transpose} in the last equation.
	Now we look at the same expression but only taking the second summand from Equation \eqref{eq:Dkchar} into account.
	\begin{align}\label{eq:summand2}
	&\left\{\hat D_k\otimes \tr_{\hat C'}\right\}\left[\left(\mathds 1-\Pi_k\right)_{\hat C}\otimes\overline\sigma^{(k)}_{\hat C'}\right]\left[V_k^\dagger\otimes V_k^T\right]\phi^+_{CC'}X_C\left[V_k\otimes\overline V_k\right]\nonumber\\
	=&\frac{|A||\hat C|}{|C|}\left(\hat D_k\right)_{A\hat C\to A}\otimes \tr_{\hat C'}\left(\mathds 1-\Pi_k\right)_{\hat C}\otimes\overline\sigma^{(k)}_{\hat C'}\phi^+_{AA'}\otimes \phi^+_{\hat C\hat C'}V_k^\dagger X_CV_k\nonumber\\
	=&0
	\end{align}
	where the steps are the same as above and in the last equality we used that $\left(\mathds 1-\Pi_k\right)_{\hat C}\sigma_{\hat{C}}=0$. Adding Equations \eqref{eq:summand1} and \eqref{eq:summand2}, summing over $k$ and normalizing finishes the proof.
\end{proof}

\begin{lem}\label{lem:eps-like2des}
	Suppose $(\tau_K,E, D)$ satisfies Definition \ref{def:eps-qNM} for trivial $B$.
	Then $\mathcal{E}:=|K|^{-1}\sum_{k}D_k\otimes E_{k}^T$ satisfies
	\begin{align}
	&\Bigg\|\mathcal{E}(X)-\frac{|A|}{|C|}\bigg[\bra{\phi^+}X\ket{\phi^+}\proj{\phi^+}+\tr\left(\Pi^-X\right)\frac{1}{|C|^2-1}\left(|C|^2D_K(\tau_C)_A\otimes\tau_{A'}-\phi^+_{AA'}\right)\bigg]\Bigg\|_\diamond\nonumber\\
	&\le 2\sqrt{2\varepsilon}|A|\left(2\sqrt{|A|}|C|+1\right).
	\end{align}	
\end{lem}
\begin{proof}
	It follows directly from the fact that $(\tau_K,E, D)$ is a \SKQES~together with Lemma \ref{lem:genmirr} that
	\begin{equation}\label{eq:like2design1-in-eps-proof}
	\mathcal E(\phi^+_{CC'})=\frac{|A|}{|C|}\phi^+_{AA'}.
	\end{equation}
	
	Let $\Lambda^{(i)}_{C\to C\tilde B_1}$, $i=0,1$ be two attack maps such that $\eta_{\Lambda^{(i)}}\ket{\phi^+}=0$ for $i=0,1$ and define $$\Lambda_{C\to C\tilde B_1\tilde B_2}=\frac{1}{2}\sum_{i=0,1}\proj i_{\tilde B_2}\otimes \Lambda^{(i)}.$$ The the $\varepsilon$-$\ITNM$ property implies
	$$I(AA':\tilde B_1\tilde B_2)_{\eta_{\tilde\Lambda}}\le \varepsilon,$$
	and therefore, using Pinsker's inequality, Lemma \ref{lem:pinsker},
	\begin{align}
	&\Bigg\|	\frac{1}{2}\sum_{i=0,1}\proj i_{\tilde B}\otimes \left(\eta_{\tilde\Lambda^{(i)}}\right)_{CC'\tilde B_1}\nonumber\\
	&-\frac 1 4\left(\sum_{i=0,1}\proj i_{\tilde B}\otimes \left(\eta_{\tilde\Lambda^{(i)}}\right)_{\tilde B_1}\right)\otimes\left(\sum_{i=0,1} \left(\eta_{\tilde\Lambda^{(i)}}\right)_{CC'}\right)\Bigg\|_1\le\sqrt{2\varepsilon}.
	\end{align}
	Observe that
	\begin{align}\label{eq:tildecj}
	\eta_{\tilde \Lambda}=&\frac{1}{|K|}\sum_kD_k\circ\Lambda\circ E_k(\phi^+_{AA'})\nonumber\\
	=&\frac{|C|}{|A|}\frac{1}{|K|}\sum_k\left(D_k\otimes E_k^T\right)\circ\Lambda(\phi^+_{CC'})\nonumber\\
	=&\frac{|C|}{|A|}\mathcal E\circ\Lambda(\phi^+_{CC'}).
	\end{align}
	Setting $\left(\eta_{\Lambda^{(0)}}\right)_{CC'\tilde B_1}=\tau^-_{CC'}\otimes\left(\eta_{\Lambda^{(1)}}\right)_{\tilde B_1} $, 
	we get
	\begin{align}\label{eq:etanull}
	\eta_{\tilde\Lambda^{(0)}}=&\frac{|C|}{|A|}\mathcal E(\tau^-)\otimes\left(\eta_{\Lambda^{(1)}}\right)_{\tilde B_1} \nonumber\\
	=&\frac{|C|}{|A|}\frac{1}{|C|^2-1}\left(|C|^2 \mathcal{E}(\tau_{CC'})-\mathcal{E}(\phi^+_{CC'})\right)\otimes\left(\eta_{\Lambda^{(1)}}\right)_{\tilde B_1} \nonumber\\
	=&\frac{1}{|C|^2-1}\left(|C|^2 D_K(\tau_C)\otimes\tau_A -\phi^+_{AA'}\right)\otimes\left(\eta_{\Lambda^{(1)}}\right)_{\tilde B_1} .
	\end{align}
	and therefore
	\begin{align}\label{eq:eps-like2des1}
	&\Bigg\|\frac{1}{|C|^2-1}\left(|C|^2 D_K(\tau_C)\otimes\tau_A -\phi^+_{AA'}\right)\otimes\left(\eta_{\Lambda^{(1)}}\right)_{\tilde B_1}-\frac{|C|}{|A|}\mathcal{E}\left(\left(\eta_{\Lambda^{(1)}}\right)_{CC'\tilde B_1}\right)\Bigg\|_1\le 2\sqrt{2\varepsilon}
	\end{align}
	for all $\Lambda^{(1)}$.
	For any state $\rho_{CC'\tilde B_1}$ with $\rho_{CC'\tilde B}\ket{\phi^+}_{CC'}=0$, we define the state
	$$\rho'_{CC'\tilde B_1\tilde B_2}=\frac{1}{C}\left(\proj 0_{\tilde B_2}\otimes\rho_{CC'\tilde B_1}+\proj{1}_{\tilde B_2}\otimes \left[\left((\mathds{1}_C-\rho_C)\otimes V_{C'}\right)\phi^+\left((\mathds{1}_C-\rho_C)\otimes V_{C'}\right)\right]\otimes\rho_{\tilde B_2}\right).$$
	Here, $V$ is a unitary such that $\tr(\mathds{1}_C-\rho_C)V_C^T=0$. It is easy to see that such a unitary always exists, the existence is equivalent to the fact that any $|C|$-tuple of real numbers is the ordered list of side lengths of a polygon in the complex plain. Note that $\rho'_{CC'\tilde B_1\tilde B_2}\ket{\phi^+}_{CC'}=0$, and $\rho'_{C'}=\tau_{C'}$.
	Together with the triangle inequality, equation \eqref{eq:eps-like2des1} implies therefore that
	\begin{eqnarray}
	\frac{1}{|C|}&\bigg\|&\frac{|C|}{|A|}\mathcal{E}(\rho)-\frac{1}{|C|^2-1}\left(|C|^2 D_K(\tau_C)\otimes\tau_A -\phi^+_{AA'}\right)\otimes\rho_{\tilde B_1}\bigg\|_1\nonumber\\
	+&\bigg\|&\frac{|C|}{|A|}\mathcal{E}\left[\left((\mathds{1}_C-\rho_C)\otimes V_{C'}\right)\phi^+\left((\mathds{1}_C-\rho_C)\otimes V_{C'}\right)\right]\nonumber\\
	&&-\frac{|C|-1}{|C|}\frac{1}{|C|^2-1}\left(|C|^2 D_K(\tau_C)\otimes\tau_A -\phi^+_{AA'}\right)\bigg\|_1\le 2\sqrt{2\varepsilon},\nonumber
	\end{eqnarray}
	i.e. in particular
	\begin{equation}
	\bigg\|\frac{|C|}{|A|}\mathcal{E}(\rho)-\frac{1}{|C|^2-1}\left(|C|^2 D_K(\tau_C)\otimes\tau_A -\phi^+_{AA'}\right)\otimes\rho_{\tilde B_1}\bigg\|_1\le 2\sqrt{2\varepsilon}|C|.\nonumber
	\end{equation}
	As $\rho$ was arbitrary we have proven that
	\begin{equation}\label{eq:diamond-E}
	\bigg\|\frac{|C|}{|A|}\mathcal{E}-\left\langle\frac{1}{|C|^2-1}\left(|C|^2 D_K(\tau_C)\otimes\tau_A -\phi^+_{AA'}\right)\right\rangle\bigg\|_\diamond\le 2\sqrt{2\varepsilon}|C|.
	\end{equation}

	The only fact that is left to show is, that $\|\mathcal{E}(\ketbra{\phi^+}{v})\|_1$ is small for all normalized $\ket{v}$ such that $\braket{\phi^+}{v}=0$.
	To this end, observe that $\tr_{A}\circ \mathcal E(\sigma_C\otimes(\cdot)_{C'})=E_K^T$ for all  quantum states $\sigma_C$. Let $\rho_C$ be any quantum state that does not have full rank, note that such states span all of $\opr(\hi_C)$, and for hermitian operators there exists a decomposition into such operators that saturates the triangle inequality. Taking a quantum state $\sigma_C$ such that $\bra{\phi^+}\rho\otimes\sigma\ket{\phi^+}=\frac{1}{|C|}\tr\rho_C\sigma_C^T=0$ (the first equality is the mirror Lemma \ref{lem:genmirr}), we have 
	$$\left\|\mathcal{E}(\rho\otimes\sigma)-\frac{|A|}{|C|}\frac{1}{|C|^2-1}\left(|C|^2 D_K(\tau_C)\otimes\tau_A -\phi^+_{AA'}\right)\right\|_1\le 2\sqrt{2\varepsilon}|A|$$
	according to what we have already proven. 
	Using inequality \eqref{eq:diamond-E} we arrive at
	\begin{equation}
	\left\|E_K^\dagger(X)-\frac{|A|}{|C|}\tau_A\tr(X)\right\|_{1}\le 2\sqrt{2\varepsilon}|A|\|X\|_1
	\end{equation}
	For Hermitian matrices $X$ and therefore
	\begin{equation}\label{eq:EK-bound}
	\left\|E_K^\dagger(X)-\frac{|A|}{|C|}\tau_A\tr(X)\right\|_{1}\le 4\sqrt{2\varepsilon}|A|\|X\|_1
	\end{equation}
	For arbitrary $X$.
	We can write $\ket{v}_{CC'}=X_C\ket{\phi^+}_{CC'}$ for some traceless matrix $X_C$. Now we calculate
	\begin{align}
	\left\|\mathcal E(\ketbra{\phi^+}{v}_{CC'})\right\|_1=&\left\|\frac{|A|}{|C|}\proj{\phi^+}_{AA'}\left(E^\dagger_K(X^\dagger)\right)_A\right\|_1\nonumber\\
	=&\frac{|A|}{|C|}\left\|\left(E^\dagger_K(X)\right)_A\ket{\phi^+}_{AA'}\right\|_2\nonumber\\
	=&\frac{\sqrt{|A|}}{|C|}\left\|E^\dagger_K(X)\right\|_2\nonumber\\
	\le&\frac{\sqrt{|A|}}{|C|}\left\|E^\dagger_K(X)\right\|_1\nonumber\\
	\le&\frac{|A|^{3/2}}{|C|}4\sqrt{2\varepsilon}\|X\|_1\nonumber\\
	\le&4\sqrt{2\varepsilon}|A|^{3/2}.
	\end{align}
	The first equation is Lemma \ref{lem:E-of-offd}, the sencond and third equations are easily verified, the first inequality is a standard norm inequality, the second inequality is Equation \eqref{eq:EK-bound}, and the last inequality follows from the normalization of $\ket v$.
	By the Schmidt decomposition, we get a stabilized version of this inequality,
	\begin{align}
	\left\|\mathcal E(\ket{\phi^+}_{CC'}\ket{\alpha}_{\tilde B_1}\bra{v}_{CC'\tilde B_1})\right\|_1\le&2\sqrt{2\varepsilon}|A|^{3/2},
	\end{align}
	for all $\ket\alpha_{\tilde B_1}$ and all $\ket{v}_{CC'\tilde B}$ such that $\braket{\phi^+}{v}=0$
	Combining everything we arrive at
	\begin{align}
	&\Bigg\|\mathcal{E}(X)-\frac{|A|}{|C|}\bigg[\bra{\phi^+}X\ket{\phi^+}\proj{\phi^+}\nonumber\\
	&+\tr\left(\Pi^-X\right)\frac{1}{|C|^2-1}\left(|C|^2D_K(\tau_C)_A\otimes\tau_{A'}-\phi^+_{AA'}\right)\bigg]\Bigg\|_\diamond\le 2\sqrt{2\varepsilon}|A|\left(4\sqrt{|A|}+1\right).
	\end{align}	
	
\end{proof}

We are now ready to prove a version the characterization theorem (i.e., Theorem \ref{thm:eps-effective-char}) with explicit constants. We remark that the exact setting Theorem \ref{thm:effective-char} is simply the case where $\varepsilon = 0$.
\begin{thm}[Precise version of Theorem \ref{thm:eps-effective-char}]\label{thm:eps-effective-char-app}
	Let $\Pi=(\tau, E,D)$ be a \SKQES.
	\begin{enumerate}
		\item If $\Pi$ is $\varepsilon$-\ITNM, then any attack map $\Lambda_{CB\to C\tilde B}$ results in an effective map $\tilde\Lambda_{AB \to A{\tilde B}}$ fulfilling
		\begin{equation}\label{eq:eps-effective-map}
		\left\|\tilde\Lambda_{AB \to A{\tilde B}}-\tilde \Lambda^{\mathrm{exact}}_{AB\to A\tilde B}\right\|_\diamond\le 2\sqrt{2\varepsilon}|A|^4|C|\left(4\sqrt{|A|}+1\right),
		\end{equation}
		where
		\begin{equation*}
		\tilde \Lambda^{\mathrm{exact}}_{AB\to A\tilde B}=\id_A\otimes \Lambda'_{B\to\tilde B}+\frac{1}{|C|^2-1}\left(|C|^2\left\langle D_K(\tau)\right\rangle-\id\right)_A\otimes \Lambda''_{B\to\tilde B},
		\end{equation*}
		with $\Lambda' =\tr_{CC'}[\phi^+_{CC'}\Lambda(\phi^+_{CC'}\otimes (\cdot))]$ and $\Lambda'' =\tr_{CC'}[\Pi^-_{CC'}\Lambda(\phi^+_{CC'}\otimes (\cdot))].$
		\item Conversely, if for a scheme all effective maps fulfil Equation \eqref{eq:eps-effective-map} with the right hand side replaced by $\varepsilon$, then it is $5\varepsilon(\log(|A|)+r)+3h(\varepsilon)$-\ITNM, where $r$ is a bound on the size of the honest user's side information.
	\end{enumerate}
\end{thm}
\begin{proof}
	We start with \textit{1.} We want to bound the diamond norm distance between the effective map $\tilde\Lambda$ resulting from an attack $\Lambda$ and the idealized effective map $\tilde\Lambda^{\mathrm{exact}}$. Let 
	$$\ket\psi_{AA'BB'}=\sum_{i=0}^{|A|^2-1}\sqrt{p_i}\ket{\alpha_i}_{AA'}\otimes\ket{\beta_i}_{BB'}$$
	be an arbitrary pure state given in its Schmidt decomposition across the bipartition $AA'$ vs. $BB'$. We can Write $\ket{\alpha_i}_{AA'}=X^{(i)}_{A'}\ket{\phi^+}$ for some matrices $X^{(i)}$ satisfying $\|X^{(i)}\|_\infty\le|A|$. We calculate the action of $\tilde \Lambda$ on $\ketbra{\alpha_i}{\alpha_j}_{AA'}\otimes\ketbra{\beta_i}{\beta_j}_{BB'}$,
	\begin{align}
	&\tilde\Lambda^{\mathrm{exact}}_{AB \to A{\tilde B}}(\ketbra{\alpha_i}{\alpha_j}_{AA'}\otimes\ketbra{\beta_i}{\beta_j}_{BB'})=X^{(i)}_{A'}\bigg(\proj{\phi^+}_{AA'}\otimes\Lambda'_{B\to \tilde B}(\ketbra{\beta_i}{\beta_j}_{BB'})\nonumber\\
	&+\frac{1}{|C|^2-1}\left(|C|^2 D_K(\tau)_A\otimes\tau_{A'}-\proj{\phi^+}_{AA'}\right)\otimes\Lambda''_{B\to \tilde B}(\ketbra{\beta_i}{\beta_j}_{BB'})\bigg)X^{(j)}_{A'}.
	\end{align}
	In a similar way we get
	\begin{align}
	\tilde\Lambda_{AB \to A{\tilde B}}(\ketbra{\alpha_i}{\alpha_j}_{AA'}\otimes\ketbra{\beta_i}{\beta_j}_{BB'})&=X^{(i)}_{A'}\tilde\Lambda_{AB \to A{\tilde B}}(\proj{\phi^+}_{AA'}\otimes\ketbra{\beta_i}{\beta_j}_{BB'})X^{(i)}_{A'}\nonumber\\
	&=\frac{|C|}{|A|}X^{(i)}_{A'}\mathcal{E}_{CC'\to AA'}\circ\Lambda_{CB \to C{\tilde B}}(\proj{\phi^+}_{CC'}\otimes\ketbra{\beta_i}{\beta_j}_{BB'})X^{(i)}_{A'}.
	\end{align}
	Using Lemma \ref{lem:eps-like2des} we bound
	\begin{align}
	&\left\|\left(\tilde\Lambda_{AB \to A{\tilde B}}-\tilde\Lambda^{\mathrm{exact}}_{AB \to A{\tilde B}}\right)(\ketbra{\alpha_i}{\alpha_j}_{AA'}\otimes\ketbra{\beta_i}{\beta_j}_{BB'})\right\|_1\nonumber\\
	=&\left\|X^{(i)}_{A'}\left(\tilde\Lambda_{AB \to A{\tilde B}}-\tilde\Lambda^{\mathrm{exact}}_{AB \to A{\tilde B}}\right)(\proj{\phi^+}_{AA'}\otimes\ketbra{\beta_i}{\beta_j}_{BB'})X^{(j)}\right\|_1\nonumber\\
	\le&\left\|X^{(i)}\right\|_\infty\left\|X^{(j)}\right\|_\infty\left\|\left(\tilde\Lambda_{AB \to A{\tilde B}}-\tilde\Lambda^{\mathrm{exact}}_{AB \to A{\tilde B}}\right)(\proj{\phi^+}_{AA'}\otimes\ketbra{\beta_i}{\beta_j}_{BB'})\right\|_1\nonumber\\
	=&\left\|X^{(i)}\right\|_\infty\left\|X^{(j)}\right\|_\infty\bigg\|\frac{|C|}{|A|}\mathcal E_{CC'\to AA'}\circ\Lambda_{CB \to C{\tilde B}}(\proj{\phi^+}_{CC'}\otimes\ketbra{\beta_i}{\beta_j}_{BB'})\nonumber\\
	&~~~~~~~~~~~~~~~~~~~~~~~~~~~-\tilde\Lambda^{\mathrm{exact}}_{AB \to A{\tilde B}}(\proj{\phi^+}_{AA'}\otimes\ketbra{\beta_i}{\beta_j}_{BB'})\bigg\|_1\nonumber\\
	\le&2\sqrt{2\varepsilon}|A|^2|C|\left(4\sqrt{|A|}+1\right).
	\end{align}
	The inequalities result from applying H\"older's inequality twice, and  Lemma \ref{lem:eps-like2des}, respectively.	Using the triangle inequality we get
	\begin{align}
	\left\|\left(\tilde\Lambda_{AB \to A{\tilde B}}-\tilde\Lambda^{\mathrm{exact}}_{AB \to A{\tilde B}}\right)(\proj{\psi}_{AA'BB'})\right\|_1\le&2\sqrt{2\varepsilon}|A|^2|C|\left(4\sqrt{|A|}+1\right)\sum_{i,j=0}^{|A|^2-1}\sqrt{p_ip_j}\nonumber\\
	\le&2\sqrt{2\varepsilon}|A|^4|C|\left(4\sqrt{|A|}+1\right).
	\end{align}
	As $\ket{\psi}$ was arbitrary, we have proven
	\begin{align}
	\left\|\tilde\Lambda_{AB \to A{\tilde B}}-\tilde\Lambda^{\mathrm{exact}}_{AB \to A{\tilde B}}\right\|_\diamond\le&2\sqrt{2\varepsilon}|A|^4|C|\left(4\sqrt{|A|}+1\right).
	\end{align}
	
	Now let us prove $2.$ Let $\Lambda_{CB\to C\tilde B}$ again be an arbitrary attack map, and assume that the resulting effective map is $\varepsilon$-close to $\tilde\Lambda^{\mathrm{exact}}_{AB\to A\tilde B}$. Observe that $p^{=}(\Lambda,\rho)=\tr\Lambda'(\rho_B)$.
	
	By Lemma \ref{lem:fannes} and Lemma \ref{lem:DP-CPTPtensCP}, this implies
	\begin{equation}
	I(AR:\tilde B)_{\tilde\Lambda(\rho)}\le I(AR:B)_\rho+h(p^=(\Lambda,\rho))+5\varepsilon\log(|A||R|)+3h(\varepsilon)
	\end{equation}
	with the help of Lemma \ref{lem:DP-CPTPtensCP}
	
\end{proof}

\subsubsection{Relationship to approximate ABW}
%%%%%

Recall that, in Section \ref{sec:ABW-exact}, we discussed the relationship between our notion of exact non-malleability and that of Ambainis et al.~\cite{Ambainis2009} (i.e., \ABWNM.) As we now briefly outline, our conclusions carry over to the approximate case without any significant changes. 

As described in Equation (3'') of~\cite{Ambainis2009}, one first relaxes the notion of \ABWNM~appropriately by requiring that the containment \eqref{eq:abw} in Definition \ref{def:ABWNM} holds up to $\varepsilon$ error in the diamond-norm distance. In the unitary case, both definitions are equivalent to approximate 2-designs (by the results of~\cite{Ambainis2009}, and our Theorem \ref{thm:eps-USKQES-NM-2design}). In the case of general schemes, the plaintext injection attack described in Example \ref{ex:injection} again shows that approximate $\ABWNM$ is insufficient, and that approximate $\ITNM$ is strictly stronger.

\subsection{Authentication}\label{sec:authentication}
%%%%%%%%%%%%%%%%%%%%%%%%%

We now consider the well-studied task of information-theoretic quantum authentication, and explain its connections to non-malleability.

\subsubsection{Definitions}
%%%%

Our definitions of authentication will be faithful to the original versions in~\cite{dupuis2012actively, garg2016new}, with one slight modification. When decryption rejects, our encryption schemes (Definition \ref{def:SKQES}) output $\bot$ in the plaintext space, rather than setting an auxiliary qubit to a ``reject'' state. These definitions are equivalent in the sense that one can always set an extra qubit to ``reject'' conditioned on the plaintext being $\bot$ (or vice-versa). Nonetheless, as we will see below, this mild change has some interesting consequences.

We begin with the definition of Dupuis, Nielsen and Salvail~\cite{dupuis2012actively}, which demands that the effective average channel of the attacker ignores the plaintext.

\begin{defn}[DNS Authentication~\cite{dupuis2012actively}]\label{def:DNS-auth}
	A \SKQES~ $(\tau_K, E, D)$ is called $\varepsilon$-DNS-authenticating if, for any CPTP-map $\Lambda_{CB\to CB'}$, there exists CP-maps 
	$\Lambda^\acc_{B\to \tilde B}$ and $\Lambda^\rej_{B\to \tilde B}$ such that $\Lambda^\acc + \Lambda^\rej$ is\,\footnote{Note that there is a typographic error in \cite{dupuis2012actively} and \cite{broadbent2016efficient} at this point of the definition. In those papers, the two effective maps are asked to sum to the identity, which is impossible for many obvious choices of $\Lambda$.} TP, and for all $\rho_{AB}$ we have
	\begin{equation}\label{eq:DNS-auth}
	\bigl\| \tr_K D(\Lambda(E(\rho_{AB}\otimes \tau_K))) - (\Lambda^\acc(\rho_{AB}) + \proj{\bot}\otimes \Lambda^\rej(\rho_{B}))\bigr\|_1\le \varepsilon\,.
	\end{equation}
\end{defn}

An alternative definition was recently given by Garg, Yuen and Zhandry~\cite{garg2016new}. It asks that, \emph{conditioned on acceptance}, with high probability the effective channel is close to a channel which ignores the plaintext.

\begin{defn}[GYZ Authentication~\cite{garg2016new}]\label{def:auth}
	A \SKQES~ $(\tau_K, E, D)$ is called $\varepsilon$-GYZ-authenticating if, for any CPTP-map $\Lambda_{CB\to CB'}$, there exists a CP-map $\Lambda^\acc_{B\to \tilde B}$ such that for all $\rho_{AB}$
	\begin{equation}\label{eq:auth}
	\bigl\|\Pi_\acc\,D(\Lambda(E(\rho_{AB}\otimes \tau_K)))\,\Pi_\acc - \Lambda^\acc(\rho_{AB})\otimes \tau_K\bigr\|_1\le \varepsilon\,.
	\end{equation}
	Here $\Pi_\acc$ is the acceptance projector, i.e. projection onto $\hi_A$ in $\hi_A\oplus\C\ket\bot$.
\end{defn}
A peculiar aspect of the original definition in~\cite{garg2016new} is that it does not specify the outcome in case of rejection, and is thus stated in terms of trace non-increasing maps. Of course, all realistic quantum maps must be CPTP; this means that the designer of the encryption scheme must still declare what to do with the contents of the plaintext register after decryption. Our notion of decryption makes one such choice (i.e., output $\bot$) which seems natural.

\subsubsection{GYZ authentication implies DNS authentication}
%%%%
A priori, the relationship between Definition 2.2 in~\cite{dupuis2012actively} and Definition 8 in~\cite{garg2016new} is not completely clear. On one hand, the latter is stronger in the sense that it requires success with high probability (rather than simply on average.) On the other hand, the former makes the additional demand that the ciphertext is untouched even if we reject. As we will now show, with our slight modification, we can prove that GYZ-authentication implies DNS-authentication.
\begin{thm}
	Let  $(\tau, E,D)$ be $\varepsilon$-totally authenticating for sufficiently small $\varepsilon$. Then it is $O(\sqrt{\varepsilon})$-DNS authenticating.
\end{thm}
\begin{proof}
	Let $\Lambda_{CB\to C\tilde B}$ be a CPTP map and $\varepsilon\le 62^{-2}$. By Definition \ref{def:auth} there exists a CP map $\Lambda'_{B\to \tilde B}$ such that for all states $\rho_{AB}$,
	\begin{equation}\label{eq:authdef}
	\left\|\Pi_a D(\Lambda(E(\rho_{AB}\otimes \tau_K)))\Pi_a-\Lambda'(\rho_{AB}\otimes \tau_K)))\right\|_1\le \varepsilon\,.
	\end{equation}
	Assume for simplicity that $D=M_\bot\circ D$, where $M_\bot$ measures the rejection symbol versus the rest. (otherwise we can define a new decryption map that way.) Define the CP maps
	\begin{align*}
	\Lambda^{(1)}_{AB\to \tilde B} 
	&=\tr_A\Pi_a\tilde{\Lambda}(\cdot)\\
	\Lambda^{(2)}_{AB\to \tilde B}
	&=\bra{\bot}_A\tilde\Lambda(\cdot)\ket{\bot}_A\\
	\Lambda''_{B\to \tilde B}
	&=\tr_C\Lambda(E_K(\tau_A)\otimes(\cdot)).
	\end{align*}
	By Theorem 15 in \cite{garg2016new} we have 
	\begin{equation}
	\left|E_K(\rho_{ABR})-E_K(\tau_A)\otimes\rho_{BR}\right\|_1\le 14\sqrt{\varepsilon},
	\end{equation}
	which implies that
	\begin{equation}\label{eq:decomp-1}
	\left\|\tr_A\otimes\Lambda''-\tr_C\circ\Lambda\circ E_K\right\|_\diamond\le\hat\varepsilon := 14\sqrt{\varepsilon}.
	\end{equation}
	Note that
	\begin{align}\label{eq:decomp-2-1}
	\tr_C\!\circ\!\Lambda\!\circ\! E_K 
	&=\tr_{CK}\!\circ\!\Lambda\!\circ\! E((\cdot)\otimes\tau_K) \nonumber\\
	&=\tr_{AK}\!\circ\! D\!\circ\!\Lambda\!\circ\! E((\cdot)\otimes\tau_K)=\tr_A\!\circ\!\tilde{\Lambda}.
	\end{align}
	On the other hand, we also have that, by Equation \eqref{eq:authdef},
	\begin{align}\label{eq:decomp-2-2}
	\bigl\|\tr_A\circ\tilde{\Lambda}-\tr_A\otimes\Lambda'-\Lambda^{(2)}\bigr\|\le
	\bigl\|\tr_A\left(\Pi_a\tilde{\Lambda}(\cdot)\right)-\Lambda'\bigr\|_\diamond\le\varepsilon
	\end{align}
	Combining Equations \eqref{eq:decomp-1}, \eqref{eq:decomp-2-1} and \eqref{eq:decomp-2-2}, we get
	\begin{equation}\label{eq:diamondbound}
	\bigl\|\Lambda^{(2)}-\tr_A\otimes(\Lambda''-\Lambda')\bigr\|_\diamond\le \varepsilon+\hat\varepsilon.
	\end{equation}
	Now observe that
	\begin{equation}\label{eq:doesntdependonA}
	\left[\tr_A\otimes(\Lambda'-\Lambda'')_{B\to\tilde B}\right]\circ\Xi_{A\to A}=\tr_A\otimes(\Lambda'-\Lambda'')_{B\to\tilde B}
	\end{equation}
	For all CPTP maps $\Xi_{A\to A}$. We define $\Lambda'''_{B\to \tilde B}=\Lambda^{(2)}(\tau_A\otimes(\cdot))$ and calculate
	\begin{align*}
	\bigl\|\Lambda^{(2)}-\tr_A\otimes\Lambda'''\bigr\|_\diamond
	&\le \bigl\|\Lambda^{(2)}-\tr_A\otimes(\Lambda''-\Lambda')\bigr\|_\diamond\\
	&~~~+\bigl\|\tr_A\otimes(\Lambda''-\Lambda')-\tr_A\otimes\Lambda'''\bigr\|_\diamond\,,
	\end{align*}
	by the triangle inequality for the diamond norm. Continuing with the calculation, 
	\begin{align}\label{eq:central-bound}
	\bigl\|\Lambda^{(2)}-\tr_A\otimes\Lambda'''\bigr\|_\diamond
	&\le \varepsilon+\hat{\varepsilon}+\bigl\|\tr_A\otimes(\Lambda''-\Lambda')-\tr_A\otimes\Lambda'''\bigr\|_\diamond\nonumber\\
	&= \varepsilon+\hat{\varepsilon}+\bigl\|\tr_A\otimes (\Lambda''-\Lambda')-\Lambda^{(2)}\circ \langle\tau_A\rangle_{A\to A}\bigr\|_\diamond\nonumber\\
	&= \varepsilon+\hat{\varepsilon}+\bigl\|\bigl[\tr_A\otimes (\Lambda''-\Lambda')-\Lambda^{(2)}\bigr]\circ \langle\tau_A\rangle_{A\to A}\bigr\|_\diamond\nonumber\\
	&\le 2(\varepsilon+\hat\varepsilon)=28\sqrt{\varepsilon}+2\varepsilon.
	\end{align}
	The first inequality above is Equation \eqref{eq:diamondbound}. The first equality is just a rewriting of the definition of $\Lambda'''$, and the second equality is Equation \eqref{eq:doesntdependonA}. Finally, the last inequality is due to Equation \eqref{eq:diamondbound} and the fact that the diamond norm is submultiplicative.
	
	We have almost proven security according to Definition \ref{def:DNS-auth}, as we have shown $\tilde\Lambda$ to be close in diamond norm to $\id_A\otimes\Lambda'+\big\langle\proj{\bot}\big\rangle\otimes\Lambda'''$. However, $\Lambda'+\Lambda'''$ is only approximately TP; more precisely, we have that for all $\rho_{ABR}$,
	\begin{align}\label{eq:tracebound}
	|\tr(\Lambda'+\Lambda''')(\rho_{ABR})-1|
	&=|\tr(\Lambda'+\Lambda'''-\Lambda)(\rho_{ABR})|\nonumber\\
	&\le|\tr(\Lambda'-\Lambda^{(1)})(\rho_{ABR})|+|\tr(\Lambda'''-\Lambda^{(2)})(\rho_{ABR})|\nonumber\\
	&\le 28\sqrt{\varepsilon}+3\varepsilon.
	\end{align}
	We therefore have to modify $\Lambda' + \Lambda''$ so that it becomes TP, while keeping the structure required for DNS authentication.
	Let $M_B=(\Lambda'+\Lambda''')^\dagger(\mathds 1_{\tilde B})$, and $\lambda_{\min}$ and $\lambda_{\max}$ its minimal and maximal eigenvalue. Then Equation \eqref{eq:tracebound} is equivalent to $\lambda_{\min}\ge 1-\eta$ and $\Lambda_{\max}\le 1+\eta$, where we have set $\eta := 28\sqrt{\varepsilon}+3\varepsilon$. Now define the corresponding CP-map, i.e., $\mathcal{M}(X)=M^{-1/2}XM^{-1/2}$. Note that $M$ is invertible for $\eta<1$ which follows from $\varepsilon\le 62^{-2}$. We bound
	\begin{align}\label{eq:Mbound}
	\left\|\mathcal{M}-\id\right\|_\diamond&=\sup_{\rho_{BE}}\bigl\|M^{-1/2}_B\rho_{BE} M^{-1/2}_B-\rho_{BE}\bigr\|_1\nonumber\\
	&\le \sup_{\rho_{BE}}\bigl\{\bigl\|M^{-1/2}_B\rho_{BE}\bigl(M^{-1/2}_B-\mathds 1_B\bigr)\bigr\|_1+\bigl\|\bigl(M^{-1/2}_B-\mathds 1_B\bigr)\rho_{AB}\bigr\|_1\bigr\}\nonumber\\
	&\le \bigl(\bigl\|M^{-1/2}\bigr\|_\infty+1\bigr)\bigl\|M^{-1/2}-\mathds 1\bigr\|_\infty\nonumber\\
	&=(1+\lambda_{\min}^{-1/2})\max(1-\lambda_{\max}^{-1/2}, \lambda_{\min}^{-1/2}-1)\nonumber\\
	&\le (1+(1-\eta)^{-1/2})\max\bigl[1-(1+\eta)^{-1/2},(1-\eta)^{-1/2}-1\bigr]\nonumber\\
	&= (1+(1-\eta)^{-1/2})((1-\eta)^{-1/2}-1)
	=\frac{\eta}{1-\eta}\le 2\eta
	\end{align}
	The first inequality is the triangle inequality of the trace norm. The second inequality follows by three applications of H\" older's inequality with $p=1$ and $q=\infty$. The last inequality follows from the assumption $\varepsilon\le 62^{-2}$. The second to last equality holds because $\sqrt{1+x}\le 1+\frac{x}{2}$ and $(1+x/2)^{-1}\ge 1-x/2$ for $x\in[-1,1]$ imply
	\begin{align}
	(1-\eta)^{-1/2}-1\ge&(1-\eta/2)^{-1}-1\ge\eta/2
	\end{align}
	and 
	\begin{align}
	1-(1+\eta)^{-1/2}\le&1-(1+\eta/2)^{-1}\le\eta/2.
	\end{align}
	Altogether we have
	\begin{align}
	&\bigl\|\bigl(\id_A\otimes\Lambda'+\big\langle\proj{\bot}\big\rangle\otimes\Lambda'''\bigr)\circ\mathcal M-\tilde{\Lambda}\bigr\|\nonumber\\
	&\le\bigl\|\bigl(\id_A\otimes\Lambda'+\big\langle\proj{\bot}\big\rangle\otimes\Lambda'''-\tilde\Lambda\bigr)\circ\mathcal M\bigr\|_\diamond
	+\bigl\|\tilde\Lambda\circ\bigl(\mathcal M-\id\bigr)\bigr\|_\diamond\nonumber\\
	&\le\bigl\|\id_A\otimes\Lambda'+\big\langle\proj{\bot}\big\rangle\otimes\Lambda'''-\tilde\Lambda\bigr\|_\diamond\|\mathcal M\|_\diamond
	+\bigl\|\mathcal M-\id\bigr\|_\diamond\nonumber\\
	&=\bigl(\bigl\|\id_A\otimes\Lambda'-\Pi_a\tilde\Lambda\Pi_a\bigr\|_\diamond
	+\bigl\|\big\langle\proj{\bot}\big\rangle\otimes\Lambda'''-\proj{\bot}\otimes\Lambda^{(2)}\bigr\|_\diamond\bigr)\|M^{-1}\|_\infty\nonumber\\
	&~~~+\left\|\mathcal M-\id\right\|_\diamond\nonumber\\
	&\le 2(\varepsilon+28\sqrt{\varepsilon}+2\varepsilon)+2\eta=4\eta
	\end{align}
	The first and second inequality are the triangle inequality and the submultiplicativity of the diamond norm. The third inequality is due to Equations \eqref{eq:authdef}, \eqref{eq:central-bound} and \eqref{eq:Mbound}, as well as $\varepsilon\le 62^{-2}$. For the first equality, note that it is easy to check that $\|\mathcal M\|_\diamond=\|M^{-1/2}\|_\infty^2=\lambda_{\min}^{-1}$.

\end{proof}

\subsubsection{Achieving GYZ authentication with two-designs}
%%%%

In \cite{garg2016new}, the authors provide a scheme for their notion of authentication based on unitary eight-designs. We now show that, in fact, an approximate 2-design suffices. This is interesting, as it implies that the well-known Clifford scheme (see e.g \cite{Dupuis2010,Broadbent2016}) satisfies the strong security of Definition \ref{def:auth}. All of the previous results on authentication which use the Clifford scheme thus automatically carry over to this stronger setting. We remark that our proof is inspired by the reasoning based on Schur's lemma used in results on decoupling \cite{Berta2011,Dupuis2014,Majenz2017,Berta2016}.

\begin{thm}\label{thm:2-design-auth}
	Let $\mathrm D=\left\{U_k\right\}_k$ be a $\delta$-approximate unitary 2-design on $\hi_C$. Let $\hi_C=\hi_{A}\otimes\hi_T$ and define
	\begin{align*}
	E_k(X_A) &= U_k\left(X_A\otimes \proj{0}_T\right)\left(U_k\right)^\dagger\\
	D_k(Y_C) &= \bra 0_T\left(U_k\right)^\dagger Y U_k\ket 0_T+\tr((\mathds 1_T-\proj 0_T)\left(U_k\right)^\dagger Y U_k)\proj{\bot}\,.
	\end{align*}
	Then the \SKQES~$(\tau_K, E, D)$ is $4(1/|T| + 3\delta)^{1/3}$-GYZ-authenticating.
\end{thm}

\begin{rem}
	The following proof uses the same simulator as the proof for the 8-design scheme in \cite{garg2016new}, called ``oblivious adversary" there. The construction exhibited there is efficient given that the real adversary is efficient.	
\end{rem}

\begin{proof}
	To improve readability, we will occasionally switch between adding subscripts to operators (indicating which spaces they act on) and omitting these subscripts.	
	
	We begin by remarking that it is sufficient to prove the GYZ condition (specifically, Equation \eqref{eq:auth}) for pure input states and isometric adversary channels. Indeed, for a general state $\rho_{AB}$ and a general map $\Lambda_{CB\to C\tilde B}$, we may let $\rho_{ABR}$ and $V_{CB\to C\tilde BE}$ be the purification and Stinespring dilation, respectively. We then simply observe that the trace distance decreases under partial trace (see e.g. \cite{Nielsen2000}).
	
	Let $\rho_{AB}$ be a pure input state and 
	$$
	\Lambda_{CB\to C\tilde B}(X_{CB}) = V_{CB\to C\tilde B}X_{CB}V_{CB\to C\tilde B}^\dagger
	$$
	an isometry. We define the corresponding ``ideal'' channel $\Gamma_V$, and the corresponding ``real, accept'' channel $\Phi_k$, as follows:
	\begin{align}
	\left(\Gamma_V\right)_{B\to\tilde B}&=\frac{1}{|C|}\tr_CV\text{ and}\nonumber\\
	\left(\Phi_k\right)_{AB\to A\tilde B}&=\bra 0_T(U_k)^\dagger_C V_{CB\to C\tilde B} U_k\ket 0_T.
	\end{align}
	Note that for any matrix $M$ with $\|M\|_\infty\le 1$, the map $\Lambda_M(X)=M^\dagger XM$ is completely positive and trace non-increasing. We have
	\begin{equation}
	\left\|\Gamma_V\right\|_\infty\le \frac{1}{|C|}\sum_i\left\|\bra i V \ket i\right\|_\infty\le 1.
	\end{equation}
	
	We begin by bounding the expectation of $\left\|(\left(\Gamma_V\right)_{B\to\tilde B}-\left(\Phi_k\right)_{AB\to A\tilde B})\ket{\rho}_{AB}\right\|_2^2$, as follows. To simplify notation, we set $\sigma_{ABT} := \proj{\rho}_{AB}\otimes\proj 0_T$ to be the tagged state corresponding to plaintext (and side information) $\rho_{AB}$.
	\begin{align}\label{eq:authbound1}
	\frac{1}{|K|}&\sum_k\left\|(\Gamma_V-\Phi_k)\ket{\rho}\right\|_2^2
	=\frac{1}{|K|}\sum_k\bra{\rho}(\Gamma_V-\Phi_k)^\dagger(\Gamma_V-\Phi_k)\ket{\rho}\nonumber\\
	&=\frac{1}{|K|}\sum_k\tr\left[\sigma_{ABT} (U_k)^\dagger V^\dagger U_k\proj 0(U_k)^\dagger V U_k\right]\nonumber\\
	&~~~~~- 2\frac{1}{|K|}\sum_k\tr\left[\sigma_{ABT} (U_k)^\dagger V^\dagger U_k \Gamma_V\right]
	+ \bra{\rho}\left(\Gamma_V\right)^\dagger \Gamma_V\ket{\rho}\,.
	\end{align}
	First we bound the second term, using the fact that $\Gamma_V$ only acts on $B$.
	\begin{align}\label{eq:authbound4}
	\frac{1}{|K|}&\sum_k\tr\left[\sigma_{ABT} (U_k)^\dagger V^\dagger U_k \Gamma_V\right]
	= \frac{1}{|K|}\sum_k\tr\left[U_k\sigma_{ABT}(U_k)^\dagger V^\dagger  \Gamma_V\right]\nonumber\\
	&= \int\tr\left[\left(U \sigma_{ABT} U^\dagger+\Delta\right) V^\dagger  \Gamma_V\right]
	\ge \int\tr\left[U \sigma_{ABT} U^\dagger V^\dagger  \Gamma_V\right]-\delta\nonumber\\
	&= \int\tr\left[ \sigma_{ABT} U^\dagger V^\dagger  U\Gamma_V\right]-\delta
	= \bra{\rho}\left(\Gamma_V\right)^\dagger \Gamma_V\ket{\rho}-\delta\,.
	\end{align}
	In the above, the operator $\Delta$ is the ``error'' operator in the $\delta$-approximate 2-design. The second equality above follows from $\|\Delta\|_1 \leq \delta$ and the fact that a 2-design is also a 1-design; the inequality follows by H{\"o}lder's inequality, and the last step follows from Schur's lemma. 
	
	The first term of the RHS of Equation \eqref{eq:authbound1} can be simplified as follows. We will begin by applying the swap trick (Lemma \ref{lem:swap-trick}) $\tr [XY]=\tr [F  X\otimes Y]$ in the second line below. The swap trick is applied to register $CC'$, with the operators $X$ and $Y$ defined as indicated below.
	\begin{align}\label{eq:authbound2}
	&\frac{1}{|K|}\sum_k \tr\Bigl[\,\underbrace{\sigma_{ABT}(U_k)^\dagger_{C} V^\dagger_{C\tilde B\to CB} (U_k)_C\proj 0_T}_{X}\,\underbrace{(U_k)^\dagger_C V_{CB\to C\tilde B} (U_k)_C}_{Y}\,\Bigr]\nonumber\\
	&= \frac{1}{|K|}\sum_k \tr\left[\left(\sigma_{ABT}\otimes\proj 0_{T'}\right)\left(U_k^{\otimes 2}\right)_{CC'}V^\dagger_{C\tilde B\to CB}V_{C'B\to C'\tilde B}\left(U_k^{\otimes 2}\right)_{CC'}^\dagger F_{CC'}\right]\nonumber\\
	&= \frac{1}{|K|}\sum_k \tr\left[\left(U_k^{\otimes 2}\right)_{CC'}^\dagger\left(\sigma_{ABT}\otimes\proj 0_{T'}\right)\left(U_k^{\otimes 2}\right)_{CC'}V^\dagger_{C\tilde B\to CB}V_{C'B\to C'\tilde B} F_{CC'}\right]\nonumber\\
	&\le \int\tr\left[\left(U^{\otimes 2}\right)_{CC'}^\dagger\left(\sigma_{ABT}\otimes\proj 0_{T'}\right)U^{\otimes 2}_{CC'}V^\dagger_{C\tilde B\to CB}V_{C'B\to C'\tilde B} F_{CC'}\right]+\delta\nonumber\\
	&= \int \tr\left[\left(\sigma_{ABT}\otimes\proj 0_{T'}\right)U^{\otimes 2}_{CC'}V^\dagger_{C\tilde B\to CB}V_{C'B\to C'\tilde B}\left(U^{\otimes 2}\right)_{CC'}^\dagger F_{CC'}\right]+\delta.
	\end{align}
	The inequality above follows the same way as in Equation \eqref{eq:authbound4}. Let $d=|C|$. We calculate the integral above using Lemma \ref{lem:Usquared}, as follows.
	\begin{equation}\label{eq:authbound3}
	\int U^{\otimes 2}V^\dagger_{C\tilde B\to CB}V_{C'B\to C'\tilde B}\left(U^{\otimes 2}\right)^\dagger\D U = \mathds 1_{CC'}\otimes R^{\mathds 1}_B+F_{CC'}\otimes R^F_B,
	\end{equation}
	where we have set
	\begin{align}
	R^{\mathds 1}_B=&\frac{1}{d(d^2-1)}\left(d^3\Gamma_V^\dagger \Gamma_V -d\mathds 1\right)\nonumber
	=\frac{1}{(d^2-1)}\left(d^2\Gamma_V^\dagger \Gamma_V -\mathds 1\right)\nonumber\\
	R^{F}_B=&\frac{1}{d(d^2-1)}\left(d^2\mathds 1-d^2\Gamma_V^\dagger \Gamma_V\right)\nonumber
	=\frac{d}{(d^2-1)}\left(\mathds 1-\Gamma_V^\dagger \Gamma_V\right).
	\end{align}
	plugging \eqref{eq:authbound3} into \eqref{eq:authbound2} and using Lemma \ref{lem:swap-trick} again, we get
	\begin{align}
	\int \tr&\left[\left(\sigma_{ABT}\otimes\proj 0_{T'}\right)U^{\otimes 2}_{CC'}V^\dagger_{C\tilde B\to CB}V_{C'B\to C'\tilde B}\left(U^{\otimes 2}\right)_{CC'}^\dagger F_{CC'}\right]\nonumber\\
	&= \tr\left[\left(\sigma_{ABT}\otimes\proj 0_{T'}\right)\left(\mathds 1_{CC'}\otimes R^{\mathds 1}_{B^2\to \tilde B^2}+F_{CC'}\otimes R^F_{B^2\to \tilde B^2}\right)F_{CC'}\right]\nonumber\\
	&= \tr\left[\proj{\rho}_{B}\left( R^{\mathds 1}_{B}+|A| R^F_{B}\right)\right]\nonumber\\
	&= \tr\left[\proj{\rho}_{B}\left(\frac{d(d-|A|)}{d^2-1}\left(\Gamma_V^\dagger \Gamma_V\right)_B+\frac{d|A|-1}{d^2-1}\mathds 1_B\right)\right]\,.
	\end{align}
	Now recall that $d=|A||T|$. Using the fact that $(a-1)/(b-1)\le a/b$ for $b \ge a$, we can give a bound as follows.
	\begin{align}
	\tr&\left[\proj{\rho}\left(\frac{d(d-|A|)}{d^2-1}\left(\Gamma_V^\dagger \Gamma_V\right)+\frac{d|A|-1}{d^2-1}\mathds 1\right)\right]\nonumber\\
	&= \frac{d|A|(|T|-1)}{d^2-1}\bra{\rho}\left(\Gamma_V^\dagger \Gamma_V\right)\ket{\rho}+\frac{d|A|-1}{d^2-1}\nonumber\\
	&\le \bra{\rho}\left(\Gamma_V^\dagger \Gamma_V\right)\ket{\rho}+\frac{1}{|T|}\,.
	\end{align}
	
	Putting everything together, we arrive at
	\begin{align}
	\frac{1}{|K|}\sum_k\left\|(\Gamma_V-\Phi_k)\ket{\rho}\right\|_2^2\le\frac{1}{|T|}+3\delta.
	\end{align}
	By Markov's inequality this implies
	\begin{equation}
	\mathbb{P}\left[\bigl\|(\Gamma_V-\Phi_k)\ket{\rho}\bigr\|_2^2>\alpha\left(\frac{1}{|T|}+3\delta\right)\right]\le\frac{1}{\alpha}
	\end{equation}
	which is equivalent to
	\begin{equation}
	\mathbb{P}\left[\bigl\|(\Gamma_V-\Phi_k)\ket{\rho}\bigr\|_2>\alpha^{1/2}\left(\frac{1}{|T|}+3\delta\right)^{1/2}\right]\le\frac{1}{\alpha},
	\end{equation}
	where the probability is taken over the uniform distribution on $\mathrm{D}$. Choosing $\alpha=(1/|T|+3\delta)^{-1/3}$ this yields
	\begin{equation}
	\mathbb{P}\left[\left\|(\Gamma_V-\Phi_k)\ket{\rho}\right\|_2>\left(\frac{1}{|T|}+3\delta\right)^{1/3}\right]\le \left(\frac{1}{|T|}+3\delta\right)^{1/3}.
	\end{equation}
	
	Let $S\subset D$ be such that $|S|/|\mathrm D|\ge 1-(1/|T|+3\delta)^{1/3}$ and 
	$\left\|(\Gamma_V-\Phi_k)\ket{\rho}\right\|_2\le(1/|T|+3\delta)^{1/3}$ for all $U_k\in S$. Using the easy-to-verify inequality $\|\proj{\psi}-\proj{\phi}\|_1\le 2\|\ket{\psi}-\ket{\phi}\|_2$, given as Lemma \ref{lem:1-norm-2-norm} in the supplementary material, we can bound
	\begin{align}
	\frac{1}{|K|}&\sum_{U_k\in\mathcal D}\left\|\Phi_k\proj{\rho}\left(\Phi_k\right)^\dagger-\Gamma_V\proj{\rho}\Gamma_V^\dagger\right\|_1\nonumber\\
	&\le\frac{1}{|S|}\sum_{U_k\in\mathcal S}\left\|\Phi_k\proj{\rho}\left(\Phi_k\right)^\dagger-\Gamma_V\proj{\rho}\Gamma_V^\dagger\right\|_1+2\left(\frac{1}{|T|}+3\delta\right)^{1/3}\nonumber\\
	&\le\frac{2}{|S|}\sum_{U_k\in\mathcal S}\left\|(\Gamma_V-\Phi_k)\ket{\rho}\right\|_2+2|T|^{-1/3}\nonumber\\
	&\le4\left(\frac{1}{|T|}+3\delta\right)^{1/3}.
	\end{align}
	This completes the proof for pure states and isometric adversary channels. As noted above, the general case follows.
	
\end{proof}

As an example, one may set $|T|=2^{s}$  (i.e. $s$ tag qubits) and take an approximate unitary 2-design of accuracy $2^{-s}$. The resulting scheme would then be $\Omega(2^{-s/3})$-GYZ-authenticating.

A straightforward corollary of the above result is that, in the case of unitary schemes, adding tags to non-malleable schemes results in GYZ authentication. We leave open the question of whether this is the case for general (not necessarily unitary) schemes.

\begin{cor}\label{thm:cor-NM-GYZ}
	Let $(\tau, E, D)$ be a $2^{-rn}$-non-malleable unitary \SKQES~with plaintext space $A$. Define a new scheme $(\tau, E', D')$ with plaintext space $A'$ where $A = TA'$ and
	\begin{align*}
	E'(X)
	&= E(X\otimes\proj 0_T)\\
	D'(Y)
	&= \bra{0}_TD(Y)\ket 0_T+\tr\left[(\mathds 1_T-\proj 0_T)D(Y)\right]\proj \bot\,.
	\end{align*}
	Then there exists a constant $r>0$ such that $(\tau, E', D')$ is $2^{-\Omega(n)}$-GYZ-authenticating if $|T|=2^{\Omega(n)}$.
\end{cor}
The proof is a direct application of Theorem \ref{thm:eps-USKQES-NM-2design} (approximate non-malleability is equivalent to approximate 2-design) and Theorem \ref{thm:2-design-auth} (approximate 2-designs suffice for GYZ authentication.) We emphasize that, by Remark \ref{rem:efficient}, exponential accuracy requirements can be met with polynomial-size circuits.

\subsubsection{DNS authentication from non-malleability}
%%%%

We end with a theorem concerning the case of general (i.e., not necessarily unitary) schemes. We show that adding tags to a non-malleable scheme results in a \DNS-authenticating scheme. In this proof we will denote the output system of the decryption map by $\overline A$ to emphasize that it is $A$ enlarged by the reject symbol.

\begin{thm}
	Let $r$ be a sufficiently large constant, and let $(\tau, E, D)$ be an $2^{-rn}$-\ITNM~\SKQES~with $n$ qubit plaintext space $A$, and choose an integer $d$ dividing $|A|$. Then there exists a decomposition $A=TA'$ and a state $\ket{\psi}_T$ such that $|T| = d$ and the scheme $(\tau, E', D')$ defined by 
	\begin{align*}
	E^t(X)
	&= E(X\otimes\proj \psi_T)\\
	D^t(Y)
	&= \bra{\psi}_TD(Y)\ket \psi_T+\tr\left[(\mathds 1_T-\proj \psi_T)D(Y)\right]\proj \bot\,.
	\end{align*}
	is $(4/|T|)+2^{-\Omega(n)}$-DNS-authenticating. 
\end{thm}
\begin{proof}
	We prove the statement for $\varepsilon=0$ for simplicity, the general case follows easily by employing Theorem \ref{thm:eps-effective-char} instead of Theorem \ref{thm:effective-char}.
	
	By Theorem \ref{thm:effective-char}, for any attack map $\Lambda_{CB\to C\tilde B}$, the effective map is equal to
	\begin{equation}
	\tilde{\Lambda}_{AB\to \overline A\tilde B}=\id_A\otimes\Lambda'_{B\to\tilde B}+\frac{1}{|C|^2-1}\left(|C|^2 \langle D_K(\tau_C)\rangle-\id\right)_{\overline A}\otimes\Lambda''_{B\to\tilde B}
	\end{equation}
	for CP maps $\Lambda'$ and $\Lambda''$ whose sum is TP. The effective map under the tagged scheme is therefore
	\begin{align*}
	\tilde{\Lambda}^t_{A'B\to \overline A'\tilde B}
	&= \bra{\psi}_T\tilde{\Lambda}_{AB\to \overline A\tilde B}((\cdot)\otimes \psi_T)\ket{\psi}_T\\
	&~~~+\tr\bigl[(\mathds 1_T-\psi_T)\tilde{\Lambda}_{AB\to \overline A\tilde B}((\cdot)\otimes \psi_T)\bigr]\proj \bot\\
	&=\left(\id_{A'}\right)_{A'\to\overline{A}'}\otimes\Lambda'_{B\to\tilde B}\\
	&~~~+\bigl(|C|^2 \big\langle \bigl(\bra{\psi}_TD_K(\tau_C)\ket{\psi}_T\bigr)_{A'}\oplus\beta\proj{\bot}\big\rangle-\id_{A'}\bigr)_{A\to \overline A'}\otimes \frac{\Lambda''_{B\to\tilde B}}{|C|^2-1}\\
	\end{align*}
	with $\beta= \tr\left[(\mathds 1-\psi)_TD_K(\tau_C)\right]$. We would like to say that, unless the output is the reject symbol, the effective map on $A$ is the identity. We do not know, however, what $D_K(\tau_C)$ looks like.
	Therefore we apply a standard reasoning that if a quantity is small \emph{in expectation}, then there exists at least one small instance. We calculate the expectation of $\tr\bra{\psi}_TD_K(\tau_C)\ket{\psi}_T$ when the decomposition $A=TA'$ is drawn at random according to the Haar measure,
	\begin{align}
	\int \tr\bra{\psi}U_A^\dagger D_K(\tau_C)U_A\ket{\psi}_T \D U_A
	&=\tr \left[\left(\int U_A\ket{\psi}_T\otimes\mathds 1_{A'}\psi U_A^\dagger \D U_A\right)D_K(\tau_C)\right] \nonumber\\
	&=\frac{\tr\mathds 1_A}{\tr\Pi_\acc}\tr \Pi_{\acc}D_K(\tau_C)\nonumber\\
	&\le1/|T|.
	\end{align}
	Hence there exists at least one decomposition $A=TA'$ and a state $\ket{\psi}_T$ such that $\hat{\gamma}:=\tr\bra{\psi}_TD_K(\tau_C)\ket{\psi}_T\le 1/|T|$. Define $\gamma=\max(\hat{\gamma}, |C|^{-2})$. For the resulting primed scheme, let 
	$$
	\Lambda_{\rej}:=\frac{(1-\gamma)|C|^2}{|C|^2-1}\Lambda''
	\qquad \text{and} \qquad
	\Lambda_{\acc}=\Lambda'+\frac{\gamma |C|^2-1}{|C|^2-1}\Lambda''\,.
	$$
	We calculate the diamond norm difference between the real effective map an the ideal effective map,
	\begin{align}
	&\bigl\|\tilde{\Lambda}^t-\id\otimes \Lambda_\acc-\langle\proj{\bot}\rangle\otimes\Lambda_\rej\bigr\|_\diamond\nonumber\\
	&\le \bigl\|\id\otimes\Lambda'+\frac{1}{|C|^2-1}\bigl(|C|^2 \big\langle \bigl(\bra{\psi}D_K(\tau)\ket{\psi}\bigr)\big\rangle-\id\bigr)
	\otimes\Lambda''-\id\otimes \Lambda_\acc\bigr\|_\diamond\nonumber\\
	&~~~~~+\bigl\|\langle\proj{\bot\rangle}\otimes (1-\hat\gamma)|C|^2\Lambda'' / (|C|^2-1)-\langle\proj{\bot}\rangle\otimes\Lambda_\rej\bigr\|_\diamond\nonumber\\
	&\le(1+|C|^{-2})(|T|^{-1}+2|C|^{-2})\nonumber\\
	&= |T|^{-1}(1+(|A'||T|)^{-2})(1+2|A'|^{-2})\nonumber\\
	&\le 4|T|^{-1}
	\end{align}
	as desired\,.
	
\end{proof}

\fbox{\begin{minipage}{\textwidth}\vspace{.3cm}
		\begin{center}\large{\textbf{Summary of Chapter \ref{chap:crypto}}}\end{center}\vspace{.3cm}
		\begin{itemize}
			\item A strengthened definition of quantum non-malleability (\ITNM) was introduced
			\item Contrary to the definition from \cite{Ambainis2009}, the new definition takes adversaries with prior side-information into account
			\item It removes the vulnerability for a certain attack we call plain text injection attack
			\item It implies secrecy, a result analogous to the fact that quantum authentication implies secrecy
			\item The encryption scheme that applies a random element from a two-design, like, e.g., the Clifford group, is \ITNM
			\item Adding a fixed tag to the message before encrypting with a \ITNM~ yields a quantum authentication scheme
			\item If a unitary \ITNM~scheme is used for constructing an authentication scheme, a stronger notion of authentication is achieved that allows for full key recycling
		\end{itemize}
		
\end{minipage}}

 % Symmetric key cryptography for quantum data

\hchapterstar{Conclusion}\label{chap:conclusion}

In this thesis, a collection of results on quantum entropy and its applications, as well as port based teleportation, have been presented. What follows is a short informal summary of all results to provide an overview over them. A new inequality for the von Neumann entropy has been proven in Chapter \ref{chap:entropy}. In Chapter \ref{chap:one-shot}, two topics in one-shot quantum information theory have been discussed. The decoupling technique has been generalized to be applicable to the one-shot setting. As the resulting notion of decoupling requires some ancillary helper state that can be handed back approximately unaltered, we have baptized this new notion to the name \emph{catalytic decoupling}. Subsequently we have explored different techniques to find lower bounds on the resource requirement for port based teleportation. On the way, we have proven a lower bound for the necessary size of the program register of an approximate universal quantum processor with given parameters. In the last chapter, Chapter \ref{chap:crypto}, we have generalized a classical entropic definition of the cryptographic security notion of information-theoretic non-malleability to the quantum  setting. This new definition strengthens, at the same time, a previous definition of quantum non-malleability. Furthermore, schemes that fulfill it can be used as a primitive to build quantum authentication schemes.

This is a wide spectrum of results, leaving an array of open question for future research. The problem of finding an unconstrained non-von-Neumann type entropy inequality is still wide open in general, and in particular the question whether any more constraints can be removed from the inequality \eqref{eq:genLiWi}. 

In catalytic decoupling, we have no proof that the catalyst is actually necessary. Particularly interesting is the finite block length regime. Here, standard decoupling is getting better and better. When applying catalytic decoupling, however, the necessary ancilla size is increasing exponentially in the block length. Therefore, if the ancilla were to be necessary, a trade off between ancilla size and remainder system size would be interesting.

For port based teleportation, the main open question concerns the gap between the achievability result from References \cite{Ishizaka2009,Beigi2011} and our best lower bound, Equation \eqref{eq:finallowerbound}. Which of the results is not tight? Is the standard protocol using maximally entangled states and the pretty good measurement optimal, or can entanglement \emph{across the ports} help?

Finally, there are some important open questions concerning non-malleability and authentication in the quantum symmetric key setting. There are two main future directions of research that could be considered opposite. On the one hand, it will be interesting to relax the security requirements to the computational setting, but strengthen the security at the same time by considering adversaries with oracle access to the encryption and/or decryption function (chosen plaintext and chosen ciphertext security). On the other hand, it might be good to get even better information-theoretic security guarantees. This could be in the form of composability, or a security requirement that has to hold with high probability over the choice of the key. % Conclusion

%% ----------------------------------------------------------------
% Now begin the Appendices, including them as separate files

\addtocontents{toc}{\vspace{1.3em}} % Add a gap in the Contents, for aesthetics

\appendix % Cue to tell LaTeX that the following 'chapters' are Appendices

\addtocontents{toc}{\vspace{1em}}  % Add a gap in the Contents, for aesthetics
\backmatter

%% ----------------------------------------------------------------
\label{Bibliography}
\lhead{\emph{Bibliography}}  % Change the left side page header to "Bibliography"
\bibliographystyle{alpha}  % Use the "unsrtnat" BibTeX style for formatting the Bibliography
%\bibliography{bibdb,bib}  % The references (bibliography) information are stored in the file named "Bibliography.bib"
\newcommand{\etalchar}[1]{$^{#1}$}

\end{document}